\documentclass[11pt]{article}
\usepackage[titletoc,title]{appendix}
\usepackage{amsmath,amsfonts,amssymb,bm,amsthm}
\usepackage[colorlinks,citecolor=blue,linkcolor=blue,pagebackref]{hyperref}
\usepackage{tikz}
\usetikzlibrary{shapes,arrows}
\usepackage[all,cmtip]{xy}
\usetikzlibrary{intersections,shapes.arrows}
\usetikzlibrary{decorations.pathreplacing}
\usetikzlibrary{arrows.meta}
\usepackage{xcolor}
\definecolor{bostonred}{rgb}{0.8, 0.0, 0.0}
\usepackage{hyperref}
\hypersetup{
    colorlinks=true,
    linkcolor=bostonred,
    urlcolor=bostonred,
    citecolor=bostonred
}

\numberwithin{equation}{section}
\newtheorem{theorem}{Theorem}[section]
\newtheorem{lemma}[theorem]{Lemma}
\newtheorem{proposition}[theorem]{Proposition}
\newtheorem{corollary}[theorem]{Corollary}

\theoremstyle{definition}
\newtheorem{example}[theorem]{Example}
\newtheorem{definition}[theorem]{Definition}
\newtheorem{assumption}[theorem]{Assumption}
\theoremstyle{remark}
\newtheorem{remark}[theorem]{Remark}

\DeclareFontFamily{U}{mathx}{}
\DeclareFontShape{U}{mathx}{m}{n}{<-> mathx10}{}
\DeclareSymbolFont{mathx}{U}{mathx}{m}{n}
\DeclareMathAccent{\widecheck}{0}{mathx}{"71}

\usepackage{multicol}
\usepackage{longtable}
\usepackage{slashed}
\usepackage{accents}
\usepackage[T1]{fontenc}
\usepackage{pifont}
\usepackage[top=1in, bottom=1in, left=1in, right=1in]{geometry}
\usepackage{fancyhdr}
\definecolor{amber}{rgb}{0.2, 0.2, 0.6}

\def\mcM{{\mathcal{M}}}

\def\bH{{\mathbb{H}}}
\def\bR{{\mathbb{R}}}
\usepackage{enumerate}
\usepackage{xcolor,cancel}
\usepackage{relsize}

\renewcommand{\tilde}{\widetilde}
\usepackage{orcidlink}

\allowdisplaybreaks

\usepackage{pifont}
\renewcommand*{\backrefalt}[4]{%
\ifcase #1 %
No citations%
\or
\ding{43}~p.~#2%
\else
\ding{43}~pp.~#2%
\fi}
\usepackage{soul}
\begin{document}

\title{The Hadamard parametrix on globally hyperbolic spacetimes with Robin boundary conditions: Fundamental solutions and Hadamard states}
\author{
Beatrice Costeri\,\orcidlink{0009-0004-6594-5926}\thanks{
Dipartimento di Fisica ``Alessandro Volta'',
Universit\`a degli Studi di Pavia \& INFN, Sezione di Pavia \& INdAM, Sezione di Pavia, 
Via Bassi 6, 
I-27100 Pavia, 
Italia;
beatrice.costeri01@universitadipavia.it; \url{https://sites.google.com/view/beatricecosteri/}
}
\and
Claudio Dappiaggi\,\orcidlink{0000-0002-3315-1273}\thanks{Dipartimento di Fisica ``Alessandro Volta'',
Universit\`a degli Studi di Pavia \& INFN, Sezione di Pavia \& INdAM, Sezione di Pavia, 
Via Bassi 6, 
I-27100 Pavia, 
Italia;
claudio.dappiaggi@unipv.it,
\url{https://claudiodappiaggi.com}} 
\and 
Benito A. Ju\'arez-Aubry\,\orcidlink{https://orcid.org/0000-0002-7739-4293}
\thanks{Department of Mathematics, University of York, Ian Wand Building, Deramore Lane, York YO10 5GH, UK; benito.juarezaubry@york.ac.uk. \url{https://bajuarezaubry.wordpress.com/}}
}


\date{}

\maketitle

\vspace{-.6cm}

\begin{abstract}
On a $d$-dimensional, $d\geq 2$, globally hyperbolic spacetime with a timelike boundary $(\mcM,g)$, we investigate the Klein-Gordon equation with Robin boundary conditions. First of all, we prove existence and uniqueness of the advanced and retarded fundamental solutions, discussing in addition their structural properties. Secondly, under the hypothesis of infinitesimal convexity of $\partial\mcM$ and of local finiteness of boundary reflections for broken null bicharacteristics, we characterize their wavefront set in the interior $\mathring{\mcM}=\mcM\setminus\partial\mcM$, using the propagation of singularities theorem of Melrose and Sj\"ostrand \cite[Thm. 4.1]{Melrose_1978}.

In the second part of the paper, we introduce a notion of Robin-Hadamard two-point correlation function as a positive distribution on $\mcM\times\mcM$ with prescribed antisymmetric part and singular structure, the latter being characterized in terms of its wavefront set. We complement this definition with a local Hadamard form. Up to a smooth remainder, this consists of the standard directed Hadamard parametrix together with additional contributions associated with the reflected null rays. Their coefficients are determined by transport equations supplemented by Robin matching conditions at the boundary. The construction applies to an arbitrary, but finite number of reflections. We prove that this local formulation is equivalent to the wavefront set condition, thereby establishing a counterpart of Radzikowski's theorem \cite{Radzikowski_1996, Radzikowski_1996_1} in the presence of a timelike boundary. Finally, starting from the ground state on suitable static backgrounds and using a deformation argument, we prove existence of Robin-Hadamard two-point correlation functions on the class of spacetimes considered.

\


\

{\bf 2020 MSC classes: } 81T20, 35E05

\end{abstract}

\tableofcontents

\allowdisplaybreaks

\section{Introduction}\label{Sec: Introduction} 
The aim of this work is to investigate the Klein-Gordon equation subject to Robin boundary conditions, on globally hyperbolic spacetimes with timelike boundary, see \cite{Ak_Hau_2020}. More precisely, given a $d$-dimensional, $d\geq 2$, globally hyperbolic spacetime $(\mcM,g)$ with $\partial\mcM\neq\emptyset$, we consider the mixed initial-boundary value problem
\begin{equation}
\label{Eq: mixed problem intro}
\begin{cases}
Pu\doteq (\Box_g-m^2)u=f, \qquad f \in C^\infty_0(\mcM),\\
B_\kappa u\doteq(\nabla_n+\kappa)u\vert_{\partial\mcM}=0,\\
u \vert_{\Sigma_{t_0}} = u_0,\;
\nabla_{\partial_t}u\vert_{\Sigma_{t_0}}=u_1, \qquad u_0, u_1 \in C^\infty_0(\mathring{\Sigma}_{t_0}),
\end{cases}
\end{equation}
where $m^2\doteq m_0^2+\xi R(x)$. Here $R$ is the scalar curvature of $\mcM$ and $\xi\in\mathbb R$, while $n$ denotes the inward-pointing unit normal to $\partial\mcM$. Throughout most of the analysis we take $\kappa\in\mathbb R$ to be constant. This restriction is mainly adopted for simplicity. As a matter of fact, most of the analysis extends, \emph{mutatis mutandis}, to $\kappa\in C^\infty(\partial\mcM)$, while additional restrictions arise only when static backgrounds and ground states are considered.

Boundary conditions of this kind arise naturally whenever a classical or quantum field is confined to a region with a timelike boundary and its interplay with the latter has to be incorporated directly into the dynamics. Robin conditions are particularly useful in this respect. They provide a versatile class of local boundary conditions, encompassing those of Neumann type as the special case $\kappa=0$. while they formally approach the Dirichlet ones in the large-$|\kappa|$ limit. At the same time, by imposing a linear relation between the field and its normal derivative, they model several physically relevant situations, including imperfectly reflecting interfaces \cite{Bellucci_2015} and localized surface interactions \cite{Albuq_2004, Bordag_2002}.

From the perspective of algebraic quantum field theory \cite{brunetti2, Fewster:2013lqa, Rejzner}, a first fundamental step towards quantization is the construction of the advanced and retarded fundamental solutions $G^\pm_\kappa\in\mathcal D'(\mcM\times\mcM)$ associated to Equation \eqref{Eq: mixed problem intro} and, hence, compatible with the prescribed boundary condition. As a matter of fact, their difference,
\begin{equation*}
G_\kappa\doteq G^-_\kappa-G^+_\kappa,
\end{equation*}
is the causal, or Pauli--Jordan, propagator. This is one of the key ingredients in every covariant quantization scheme for Bosonic field theories since it codifies the canonical commutation relations of the underlying field algebra through the induced symplectic structure on the space of classical solutions. Hence the analysis of the mixed problem in Equation~\eqref{Eq: mixed problem intro} and the construction of its fundamental solutions provide the classical analytic backbone of the quantum theory.

The Cauchy problem for the Klein-Gordon equation with Robin boundary conditions has also been investigated in \cite{Ginoux_2022} within the framework of Friedrichs systems. In particular, Robin boundary conditions are shown there to define admissible boundary conditions for the forward Cauchy problem under suitable positivity and sign assumptions. Previous constructions of the corresponding propagators in settings relevant to quantum field theory have instead relied on additional geometric assumptions, most notably the existence of an irrotational timelike Killing vector fields, see \cite{Dappiaggi-Drago_2019}. One of the purposes of the present work is indeed to remove this restriction.

The presence of a timelike boundary, however, affects considerably more than the well-posedness of the classical problem. It modifies the spectral properties of the theory and it gives rise to boundary-induced effects such as vacuum polarisation and Casimir-type phenomena \cite{Fang_2022, Strohmaier_2026}. Most notably for the analysis developed here, it alters the propagation of singularities. More precisely, a singularity travelling along a null bicharacteristic can reach $\partial\mcM$, be reflected, and return to the interior. The singular structure of the propagators and of the correlation functions must therefore retain information not only on direct null propagation but also on the reflected dynamics at the boundary. This has been already investigated in \cite{Costeri_25} though in the simplified setting of a $d$-dimensional half-Minkowski spacetime.

Once the algebra of observables has been constructed, one is faced with the further problem of selecting a physically meaningful class among the infinitely many algebraic states compatible with the dynamics. On globally hyperbolic spacetimes without boundary, this role is played by quasifree Hadamard states, whose two-point correlation functions possess a universal singular structure \cite{Brunetti_1996}. Besides providing the curved-spacetime counterpart of the ultraviolet behaviour of the Minkowski vacuum, the Hadamard condition is the microlocal requirement underlying both the construction of local Wick polynomials \cite{Hollands_2001,Hollands_2002} and of renormalized observables such as the stress-energy tensor \cite{Moretti2003,CosteriDappiaggiGoi2026} and the formulation of perturbatively interacting quantum field theories \cite{Rejzner}.

In the presence of a timelike boundary, the standard Hadamard singularity associated with direct null propagation in the bulk is no longer sufficient. Singularities reaching $\partial\mcM$ are reflected back into the interior and propagate along broken null bicharacteristics. This phenomenon has been thoroughly studied in the literature since the late seventies, see, for instance, \cite{Melrose_1978,Melrose_1982,Taylor_1978}. In the present setting, it gives rise to two related problems. The first is at microlocal level since one has to identify precisely the singular relation generated by the reflected null dynamics in order to determine the wavefront set of the fundamental solutions. The second is instead at a local level since one has to understand whether the usual Hadamard parametrix can be supplemented by additional contributions associated with reflected null propagation and how their coefficients are determined by the geometry and by the Robin boundary condition. We refer to \cite{Costeri_25} for a preliminary investigation on half-Minkowski spacetime, while an earlier attempt focusing on Dirichlet boundary conditions on spatial domains of the form $\bR\times\Omega$ with $\Omega\subset\bR^n$ and $\partial\Omega\neq\emptyset$ can be found in \cite{Chazarin}. These questions are relevant independently of the applications to quantum field theory, since they concern the propagation and local structure of singularities for hyperbolic boundary value problems. From the viewpoint of quantum field theory, however, they acquire an additional significance. The reflected Hadamard coefficients encode the universal, state-independent part of the boundary-induced singularity and therefore control, in particular, the universal contribution to the divergence of local observables such as the improved stress-energy tensor when the boundary is approached. The existence of quasifree Hadamard states for the Klein-Gordon field on globally hyperbolic spacetimes with timelike boundary, subject to Dirichlet boundary conditions, has recently been established in \cite{Contini_2026}, together with a propagation result for positive-energy microlocal splittings. Related microlocal formulations in the presence of timelike or conformal boundaries have also been developed in the context of anti-de Sitter and asymptotically anti-de Sitter spacetimes, see, \textit{e.g.}, \cite{Dappiaggi:2017wvj, Dappiaggi-Marta_2021, Dybalski:2018egv,  Enciso_2026, Gannot_2022, Wrochna}.

\vskip .2cm

Moving from half-Minkowski spacetime as in \cite{Costeri_25} to a generic globally hyperbolic spacetime with timelike boundary requires a geometric description of reflected propagation which does not rely on the existence of a distinguished reflection symmetry. To this end, throughout most of the paper we work under a hierarchy of geometric assumptions, whose roles are different and which we shall invoke only when they are needed.

The first is infinitesimal convexity of $\partial\mcM$, see Assumption \ref{Ass: Infinitesimally convex}. In the setting considered here, this is equivalent to requiring that no geodesic reaching $\partial\mcM$ from the interior does so tangentially, see Proposition~\ref{Prop: no-glancing}. In particular, every reflection is transversal and the corresponding reflection times and reflection points depend smoothly on the initial data. This allows us to describe the reflected dynamics locally in terms of smooth families of broken geodesics. We stress that infinitesimal convexity plays no role in the existence of the advanced and retarded fundamental solutions. It enters instead when studying their microlocal properties and, subsequently, when constructing a local Hadamard form adapted to reflected propagation. We choose this weaker condition rather than strict convexity also because it includes relevant examples such as half-Minkowski spacetime. 

A second assumption concerns the possibility of accumulating reflections. Even if every individual interaction with the boundary is transversal, a broken null bicharacteristic could in principle undergo infinitely many reflections within a bounded interval of the global time coordinate. We exclude this possibility by imposing the local finiteness condition of Assumption \ref{Ass: No Zeno points}. Together with infinitesimal convexity, it guarantees that the reflected propagation relevant to the wavefront set can be described by finite broken null bicharacteristics. Notice that these two assumptions are logically distinct. Yet, under stronger hypotheses, such as strict null convexity as discussed in \cite{HintzUhlmann}, local finiteness is automatically guaranteed. This implication is no longer true under the weaker convexity condition adopted here and, therefore, we must impose it independently.

These assumptions already suffice for the analysis of propagation of singularities carried out in Section \ref{Sec: Wavefront Set of the Propagators}. The construction of a local Hadamard parametrix requires instead an additional ingredient. A pair of points in the interior may be connected by several broken geodesics, possibly undergoing different numbers of reflections, while reflected conjugate points may occur along them. We therefore introduce Assumption \ref{Ass: Finite regular reflections}, a regularity condition which excludes reflected conjugate points along the null trajectories under consideration. In turn, this guarantees that the geometric quantities associated with reflected propagation remain smooth up to the reflection points. Although this hypothesis is not needed to characterize the wavefront set of the propagators, it becomes necessary for the local Hadamard construction, where such smooth geometric control is required to describe the singularities associated with reflected null propagation.

More precisely, we organize the reflected dynamics locally into branches. By a "branch" we mean a smooth family of broken geodesics connecting nearby pairs of interior points and undergoing the same fixed number of reflections. The reflection points vary smoothly with the endpoints. Different branches may connect the same pair of points, corresponding to distinct reflected trajectories. For every regular branch with $N$ reflections we construct a corresponding reflected Synge world function, see Definition~\ref{Def: Regular reflected branch chart} and Equation~\eqref{Eq: Reflected Synge world function}. As in the standard case, this function satisfies the expected eikonal identities, while at each reflection point its derivatives obey matching relations dictated by the reflection law and by the geometry of $\partial\mcM$. From a microlocal viewpoint, the reflected Synge functions encode locally the singular relation associated with reflected null propagation. From the viewpoint of the Hadamard construction, they replace the ordinary Synge world function in the singular terms generated by the interplay with the boundary. This branchwise description is closely related to the familiar treatment of reflected rays through billiard and endpoint maps, see, {\it e.g.}, \cite{Petkov, Wunsch}.

\vskip .2cm

Focusing on the content of the paper, our investigation start from an analysis of Equation \eqref{Eq: mixed problem intro}. More precisely, a first main result of this work is that the construction of the advanced and retarded fundamental solutions does not require any of the additional geometric assumptions introduced above, that is global hyperbolicity of the spacetime with timelike boundary suffices.

\begin{theorem}[\textit{cf.,} Theorem~\ref{Thm: Existence of Propagators}]
Let $(\mcM,g)$ be a $d$-dimensional, $d\geq2$, globally hyperbolic spacetime with timelike boundary and let $P \doteq \Box_g-m^2$ be the Klein-Gordon operator thereon. Then there exist unique advanced and retarded fundamental solutions
$$G^\pm_\kappa\in\mathcal D'(\mcM\times\mcM),$$
satisfying Robin boundary conditions and the corresponding causal support properties.
\end{theorem}

In particular, infinitesimal convexity of $\partial\mcM$, local finiteness of the reflected dynamics and the regularity assumptions introduced for the reflected branches play no role in this result. The proof starts considering Equation~\eqref{Eq: mixed problem intro} in the case of compact Cauchy hypersurfaces. Existence of solutions is obtained through a Galerkin approximation scheme, while suitable energy estimates yield uniqueness and finite propagation speed. These properties provide the causal support needed for the construction of the advanced and retarded Green operators. The restriction to compact Cauchy hypersurfaces is subsequently removed by an exhaustion argument.

This differs from the construction in \cite{Dappiaggi-Drago_2019}, where the propagators are obtained by spectral methods on standard static backgrounds, reducing the Klein-Gordon operator to a time-independent self-adjoint spatial operator. The existence of an irrotational, timelike Killing field will be required only later, as an auxiliary tool in the construction of a distinguished Hadamard state.

Having established the existence of the fundamental solutions, we turn to their microlocal properties. At this stage one has to account for the fact that singularities can reach the timelike boundary, be reflected, and then continue to propagate through the interior along broken null bicharacteristics. Under the infinitesimal convexity and local finiteness assumptions discussed above, a second main result of this work gives a complete characterization of the singular structure of the Pauli-Jordan propagator in the interior of $\mcM$.

\begin{theorem}[\textit{cf.}, Theorem \ref{Thm: WF of Propagators}]
Let $(\mcM,g)$ be a $d$-dimensional, $d\geq2$, globally hyperbolic spacetime with timelike boundary abiding by Assumptions~\ref{Ass: Infinitesimally convex} and~\ref{Ass: No Zeno points}. Let $G_\kappa \doteq G^-_\kappa-G^+_\kappa$ be the Pauli-Jordan propagator associated with the Klein-Gordon operator with Robin boundary conditions as in Equation \eqref{Eq: mixed problem intro}. Then
\begin{equation}
\label{Eq: WFset intro}
\operatorname{WF}\left(G_\kappa\vert_{\mathring{\mcM}\times\mathring{\mcM}}\right)=\mathcal{C}_{\mathrm{b}}(\mathring{\mcM}),
\end{equation}
where
$$\mathcal{C}_{\mathrm{b}}(\mathring{\mcM})\doteq\left\{(x,k_x,y,-k_y)\in T^*(\mathring{\mcM}\times\mathring{\mcM})\setminus\{\boldsymbol{0}\}\;\middle|\;(x,k_x)\sim_{\mathrm{b}}(y,k_y)\right\},$$
while $\sim_{\mathrm b}$ stands for being connected by a finite broken null bicharacteristic.
\end{theorem}

Hence the singularities of $G_\kappa$ in the interior are precisely those generated by propagation along null directions together with its successive reflections at the boundary. Notice in particular that the additional regularity assumption used for the branchwise Hadamard construction does not enter this result, namely reflected conjugate points and the associated caustics are allowed at this stage.

The proof relies on the propagation of singularities theorem for hyperbolic boundary value problems of Melrose and Sj\"ostrand \cite{Melrose_1978,Melrose_1982}. Infinitesimal convexity ensures that every null bicharacteristic reaching $\partial\mcM$ does so transversally, namely with a non-vanishing component in the normal direction to the boundary. Hence each reflection is well defined and the trajectory can be continued uniquely as a broken null bicharacteristic. Local finiteness, instead, guarantees that only finitely many such reflections occur in every bounded time interval. Starting from the known singularity of the propagator along the diagonal, propagation along the reflected Hamiltonian flow yields the relation in Equation~\eqref{Eq: WFset intro}. Conversely, causal support and propagation back to the diagonal exclude any additional singular directions. The same analysis provides the corresponding characterization of the wavefront sets of the advanced and retarded fundamental solutions.

Having characterized the singular structure of the propagators, we turn to introducing a notion of Robin-Hadamard two-point correlation function $\omega_{2,\kappa}\in\mathcal{D}^\prime(\mcM\times\mcM)$ by requiring positivity, the canonical commutation relations determined by $G_\kappa$, compatibility with the Klein--Gordon equation and the Robin boundary condition. Furthermore we require that the singular structure is codified by 
\begin{equation}
\label{Eq: Robin Hadamard WF intro}
\operatorname{WF}\left(\omega_{2,\kappa}\vert_{\mathring{\mcM}\times\mathring{\mcM}}\right)=
\mathcal{C}_{\mathrm{b}}^{\triangleright}(\mathring{\mcM}),
\end{equation}
where $\mathcal{C}_{\mathrm{b}}^{\triangleright}(\mathring{\mcM})$ denotes the future-directed part of the broken null relation introduced above. In comparison with the standard Hadamard condition on globally hyperbolic spacetimes without boundary, the usual null-geodesic relation is replaced by the one generated by broken null bicharacteristics, thereby accounting for the reflection of singularities at $\partial\mcM$. This structure is not novel and it has already appeared in the literature, though mainly in the analysis of field theoretic models on asymptotically anti-de Sitter spacetimes, see, {\it e.g.}, \cite{Costeri_25, Dappiaggi-Drago_2019, Dappiaggi:2017wvj, Dappiaggi-Marta_2021, Dappiaggi-Nosari_2016, Gannot_2022, Wrochna}.

We follow this analysis by an investigation of the corresponding local description. In a geodesically convex region away from the boundary, nothing changes and the singular structure is described by the standard directed Hadamard parametrix \cite{Radzikowski_1996, Radzikowski_1996_1}. On a regular reflected branch, instead, an additional singular contribution has to be included. Hence, up to a smooth remainder, the local form of a Robin-Hadamard two-point function consists of the usual directed term together with a reflected contribution associated with the corresponding reflected Synge world function.

The coefficients entering the reflected contribution are not arbitrary. They are determined recursively by transport equations along the reflected null trajectories and satisfy matching conditions at every reflection point which encode the Robin boundary condition. Generalizing the analysis on half-Minkowski spacetime \cite{Costeri_25}, we carry out the construction for a single reflection, treating separately even and odd spacetime dimensions. Subsequently, we iterate the procedure along the successive reflection points. In this way, every regular branch with an arbitrary but finite number of reflections gives rise to its own local Hadamard contribution. The role of Assumption~\ref{Ass: Finite regular reflections} becomes essential precisely at this stage. It guarantees that the reflected geometric data entering the local expansion remain smooth and that the singular contribution associated with each reflected null trajectory can be treated independently. The microlocal condition in Equation~\eqref{Eq: Robin Hadamard WF intro}, instead, continues to make sense without excluding reflected conjugate points. The local construction therefore requires stronger geometric control than the wavefront-set characterization itself.

Having both the microlocal and the local formulations at our disposal, we can compare them. Our third main result is a counterpart, for Klein-Gordon fields with Robin boundary conditions, of the celebrated equivalence theorem of Radzikowski \cite{Radzikowski_1996,Radzikowski_1996_1}. Besides the geometric assumptions discussed above, the statement requires a microlocal admissibility condition on the family of reflected parametrices, see Definition \ref{Def: Admissible reflected Hadamard assembly but not yet an Avenger}, which ensures that their singularities are compatible with the broken null relation.

\begin{theorem}[\textit{cf.}, Theorem~\ref{Thm: Global-to-Local}]
Let $(\mcM,g)$ be a globally hyperbolic spacetime with timelike boundary satisfying the standing geometric assumptions, and let $\omega_{2,\kappa}\in\mathcal D'(\mcM\times\mcM)$ satisfy the Klein--Gordon equation, the Robin boundary condition, positivity and the canonical commutation relations. Assuming in addition microlocal admissibility of the associated family of reflected parametrices, the following conditions are equivalent:
\begin{enumerate}
\item $\omega_{2,\kappa}$ satisfies the microlocal Robin-Hadamard condition in Equation~\eqref{Eq: Robin Hadamard WF intro};
\item $\omega_{2,\kappa}$ is locally of Robin-Hadamard form, namely its singular structure is given, up to a smooth remainder, by the directed Hadamard parametrix together with the contributions associated with the regular reflected branches.
\end{enumerate}
\end{theorem}

Hence the wavefront-set condition and the local reflected Hadamard expansion provide two equivalent descriptions of the same class of two-point correlation functions. This generalizes to globally hyperbolic spacetimes with timelike boundary the corresponding result obtained in \cite{Costeri_25} for half-Minkowski spacetime. We complement the previous analysis by addressing the existence of Robin-Hadamard states. On suitable static backgrounds, and under  required assumptions on the spectrum of the underlying operators, we construct the two-point correlation function of the ground state and we show that it satisfies the Robin-Hadamard condition introduced above. Starting from this distinguished representative, a deformation argument then yields existence on the class of backgrounds which can be connected to such a static spacetime while preserving the geometric assumptions needed in the construction. We postpone the precise hypotheses and the details of this argument to Section \ref{Sec: Existence oF Hadamard States}.
We emphasize that, in the late stages of realization of this work, an interesting procedure to construct Hadamard states with Dirichlet boundary conditions has been discussed in \cite{Contini_2026}. A generalization to our setting might offer the possibility of dropping the limitations which are naturally inherent to the deformation argument.

\paragraph*{Synopsis --} A synopsis of the work is as follows. 

In Section~\ref{Sec: Geometric setting}, we introduce the geometric framework used throughout the paper, working on globally hyperbolic spacetimes with timelike boundary in the sense of \cite{Ak_Hau_2020}. To retain a regular geometric description of reflected propagation, we impose the infinitesimal convexity condition of Assumption~\ref{Ass: Infinitesimally convex}, which is equivalent in the present setting to the absence of glancing geodesics reaching the boundary from the bulk, see Proposition \ref{Prop: no-glancing}. We also discuss the standard static subclass with Cauchy surfaces of bounded geometry, which will later provide the setting for the spectral construction of the ground state. In Section~\ref{Sec: Geodesic and Reflected Geodesic Distances Claudio}, we recall the definition and basic properties of Synge world function and introduce its branchwise reflected counterpart. Since a pair of points in $\mathring{\mcM}$ may be connected by several broken geodesics, possibly undergoing different numbers of reflections, we introduce regular $N$-reflected branch charts, see Definition~\ref{Def: Regular reflected branch chart}, and the associated branchwise $N$-reflected Synge world functions, \textit{cf.}, Definition~\ref{Def: Reflected Synge's world function}. We establish their main geometric properties, including symmetry and the corresponding eikonal identities. Finally, we introduce Assumption~\ref{Ass: Finite regular reflections}, excluding reflected conjugate points and the associated caustics along the null branches under consideration. This hypothesis provides the regular branchwise geometry required for the local Hadamard construction developed in Section~\ref{Sec: Hadamard states with Robin boundary conditions}.

Section~\ref{Sec: BV problem} is devoted to the Klein-Gordon equation with Robin boundary conditions. We first study the associated mixed initial-boundary value problem and prove existence and uniqueness of smooth solutions, initially for compact Cauchy hypersurfaces, see Proposition \ref{Prop: Existence of Solutions}, and subsequently in the general case by an exhaustion argument. This yields the existence and uniqueness of advanced and retarded fundamental solutions on general globally hyperbolic spacetimes with timelike boundary, see Theorem \ref{Thm: Existence of Propagators}. We then recall, for comparison and for later use in the construction of ground states, their spectral realization on standard static backgrounds. Finally, under Assumption~\ref{Ass: Infinitesimally convex} and the local finiteness condition of Assumption~\ref{Ass: No Zeno points}, we characterize the wavefront set of the Pauli-Jordan propagator in terms of finite broken null bicharacteristics, together with the corresponding wavefront sets of the advanced and retarded propagators, see Theorem~\ref{Thm: WF of Propagators}.

In Section~\ref{Sec: Hadamard states with Robin boundary conditions}, we introduce a boundary-adapted notion of Hadamard two-point correlation function for the Klein--Gordon field with Robin boundary conditions. Definition~\ref{Def: Robin Hadamard two-point function} formulates the microlocal Hadamard condition by requiring the singularities of the two-point function to lie on the positive-frequency part of the broken characteristic relation.  We then turn to the local formulation of the Hadamard condition. On regular reflected branch domains, the standard directed Hadamard parametrix is supplemented by a reflected contribution whose coefficients are determined by transport and recursion equations subject to Robin matching conditions at the reflection points. Distinguishing between even and odd spacetime dimensions, Sections~\ref{Sec: Hadamard Recurstion Relations even case} and~\ref{Sec: Hadamard Recursion Relations: Odd dimensions} establish the smoothness and asymptotic properties of the corresponding reflected Hadamard coefficients in the one-reflection case. The construction is subsequently extended by iteration to an arbitrary, though finite, number of regular reflections. Finally, we prove a boundary-adapted counterpart of Radzikowski's theorems \cite{Radzikowski_1996,Radzikowski_1996_1}, establishing the equivalence between the microlocal Robin-Hadamard condition and its branchwise local formulation, under the additional microlocal admissibility requirement on the family of reflected parametrices, see Definition~\ref{Def: Admissible reflected Hadamard assembly but not yet an Avenger}. This extends the corresponding result obtained in \cite{Costeri_25} for half-Minkowski spacetime. In Section \ref{Sec: Existence oF Hadamard States} first we prove that, on static spacetimes and under appropriate spectral assumptions, the ground state is of Hadamard form. Subsequently, in Section \ref{Sec: Deformation Argument} we apply a deformation argument to infer existence of Hadmard states in a large class of admissible backgrounds. Appendix~\ref{App: A} contains the proof of Proposition~\ref{Prop: Causal Support of Solutions}, while Appendix~\ref{App: B} is devoted to the proof of Proposition~\ref{Prop: Existence of Solutions}.

\section{Geometric setting}\label{Sec: Geometric setting}
The goal of this section is twofold. On the one hand, we fix the main notations and conventions employed throughout this work. On the the other hand, we introduce the key geometric data at the heart of our work. These can be divided between global and local structures and our presentation shall reflect such difference. At the same time we will limit ourselves to outlining only those properties which will be used in the following sections. Hence, a reader, who is interested in learning more details about these topics, is referred to \cite{Lee_2018}, \cite{ONeill83} and \cite{Ak_Hau_2020}.

\subsection{Lorentzian manifolds with a timelike boundary}

We call {\em spacetime} a pair $(\mathcal{M},g)$, where $\mcM$ is a $d$-dimensional, oriented and connected manifold, $d \ge 2$, endowed with a smooth Lorentzian metric $g$ of signature $(1, d-1)$, that is $(-,+,\dots, +)$, \cite[Ch. 2]{Lee_2018}. Unless stated otherwise, we assume that $\partial\mcM$ might be non empty. Whenever the Lorentzian metric is clear from the context, we shall simply write $\mcM$ in place of $(\mcM,g)$, with a slight abuse of notation. In addition, we require that $(\mcM, g)$ is also {\em time-oriented}, namely there exists an everywhere non-vanishing timelike vector field $\chi\in\Gamma(T\mcM)$ which has been selected once and for all. 

\begin{definition}
	\label{Def: time-like boundary}
	Given a spacetime $(\mathcal{M},g)$, we call $\partial \mathcal{M}$ a {\bf timelike boundary} if $\partial\mathcal{M}\neq\emptyset$ and, denoting by $\iota: \partial \mathcal{M} \hookrightarrow \mathcal{M}$ the natural inclusion map, the pullback $\iota^* g$ identifies a Lorentzian metric on $\partial \mathcal{M}$. In this case, $(\mathcal{M}, g)$ is said to be a spacetime with timelike boundary and we denote by $\mathring{\mathcal{M}}=\mathcal{M}\setminus\partial\mathcal{M}$ its interior.
\end{definition}

Among the plethora of spacetimes with a timelike boundary as per Definition \ref{Def: time-like boundary}, we confine our attention to those which are \emph{globally hyperbolic}, see, \textit{e.g.}, \cite{Ak_Hau_2020}.

\begin{definition} \label{Def: Globally Hyperbolic}
	\label{Def: ghm}
	A spacetime $(\mathcal{M},g)$ with a non-empty timelike boundary $\partial \mathcal{M}$ as per Definition \ref{Def: time-like boundary} is said to be {\bf globally hyperbolic} if the following two conditions are met: 
	\begin{itemize}
		\item[(i)] $(\mathcal{M
		}, g)$ is \emph{strongly causal} in the sense of \cite[Def. 1.3.7]{Baer_2007},
		\item[(ii)] For any pair of points $p, q \in \mathcal{M}$, the set $J^+(p) \cap J^-(q)$ is either compact or empty. Here $J^{\pm}(\cdot)$ denote the causal future and past, respectively.
	\end{itemize}
\end{definition}

A notable property of globally hyperbolic spacetimes has been proven in \cite[Thm 1.1]{Ak_Hau_2020} as a direct consequence of $(\mathcal{M},g)$ possessing a Cauchy orthogonal splitting. Here we report part of this result, limited to those aspects which will be of interest in our analysis.

\begin{proposition}\label{Prop: Globally Hyperbolic}
	Let $(\mcM,g)$ be a globally hyperbolic spacetime as per Definition \ref{Def: Globally Hyperbolic}. There exists an isometry $\psi: \mathbb{R} \times \Sigma \rightarrow \mathcal{M}$ such that the line element associated to $\psi^* g$ is of the form
	\begin{equation}
		\label{Eq: line element psi*g}
		ds^2 = -\beta dt^2 + h_t,
	\end{equation}
	where $\{t\}\times\Sigma$ is a smooth, spacelike, Cauchy hypersurface for all $t\in\bR$, while $t: \mathbb{R} \times \Sigma \rightarrow \mathbb{R}$ is a temporal function whose gradient $\nabla t$ is tangent to $\partial\mcM$. In addition $\beta \in C^{\infty}(\mathbb{R} \times \Sigma;(0,\infty))$, while $\{h_t\}_{t \in \mathbb{R}}$ is a family of smooth Riemannian metrics on $\{t\} \times \Sigma$, varying smoothly with $t$.
\end{proposition}

\noindent Observe that, since per hypothesis $\partial\mathcal{M}\neq\emptyset$, then also $\partial\Sigma\neq\emptyset$ and we denote the interior of a Cauchy surface as $\mathring{\Sigma}=\Sigma\setminus\partial\Sigma$. We report here a notable result, see \cite[Cor. 5.8]{Ak_Hau_2020}, which will play a key role in our analysis.
\begin{corollary}\label{Cor: Globally Hyperbolic Extension}
	Let $(\mcM,g)$ be a globally hyperbolic spacetime as per Definition \ref{Def: Globally Hyperbolic}. Then there exists a second globally hyperbolic spacetime $(\mcM^\prime,g^\prime)$ with $\partial\mcM^\prime=\emptyset$ together with an isometric embedding $\iota:\mcM\to\mcM^\prime$.
\end{corollary}

In the following we shall impose a restriction on the class of globally hyperbolic spacetimes, related to the geometry of $\partial\mcM$. The reason for requiring it is twofold. On the one hand, we will need to establish a notion of geodesically convex normal neighbourhood of $\mcM$ intersecting $\partial\mcM$. A most convenient way consists of defining these open sets as the restriction to $\mcM$ of their counterparts on $\mcM^\prime$. Yet one has to make sure that a geodesic connecting two points in $\mcM$ is not allowed to escape through $\mcM^\prime\setminus\mcM$. On the other hand, in order to establish the main result of this work, see Section \ref{Sec: Hadamard states with Robin boundary conditions}, we require that every null geodesic stemming from a point in the interior of $\mcM$ hits the boundary with a velocity vector which is not tangent to $\partial\mcM$. We will make these heuristic considerations precise in the next sections. 

\begin{assumption}\label{Ass: Infinitesimally convex}
Let $(\mcM,g)$ be a $d$-dimensional, $d\geq 2$, globally hyperbolic spacetime with a timelike boundary as per Definition \ref{Def: Globally Hyperbolic} and Proposition \ref{Prop: Globally Hyperbolic}. We require $(\mcM, g)$ to be {\bf infinitesimally convex} at $\partial\mcM$. This means that, for every $p\in\partial\mcM$, given a neighbourhood $\mathcal{U}\subset\mcM$ including $p$ and a boundary defining function $z:\mathcal{U}\to [0,\infty)$ such that $z^{-1}(0)=\mathcal{U}\cap\partial\mcM$, while $dz|_{\mathcal{U}\cap\partial\mcM}\neq 0$, the metric induced Hessian 
$$\mathrm{Hess}_g(z)(X,X)\leq 0,\quad\forall X\in T_p\mcM,\;p\in\partial\mcM\;\mathrm{and}\;g(X,n)=0.$$
Here $n$ is the inward pointing normal vector to $\partial\mcM$ at $p$.
\end{assumption}

\begin{remark}\label{Rem: Weak vs Infinitesimal Convexity}
We observe that, in the definition of infinitesimal convexity, as shown in \cite{Caponio13}, particularly Theorem 3.16 and its proof, generalizing a classical result by Bishop in the Riemannian setting, the condition is independent from the choice of the boundary function under the given assumptions. Furthermore it is equivalent to weak local convexity, namely that, for all $p\in\partial\mcM$, there exists a neighbourhood $\mathcal{U}$ including $p$ such that, if a geodesic $\gamma$ is tangent to $\partial\mcM$ at $q\in\mathcal{U}\cap\partial\mcM$, then there exists $\varepsilon>0$ such that $\gamma(t)\notin\mathring{\mcM}$ for all $t\in(-\varepsilon,\varepsilon)$. Therefore, we use interchangeably the adjectives {\em infinitesimal} and {\em weakly local} when referring to convexity. 
\end{remark}

In the following we prove a vital result, which establishes the equivalence between Assumption \ref{Ass: Infinitesimally convex} and the impossibility for a geodesic, stemming from a point in the bulk, to intersect tangentially the boundary. 

\begin{proposition}\label{Prop: no-glancing}
	Let $(\mcM,g)$ be a globally hyperbolic spacetime. Then it is infinitesimally convex if and only if it is {\bf no glancing}, that is, for any geodesic $\gamma:I\subseteq\bR\to M$ for which there exists $s,s_0\in I$ such that $\gamma(s)\in\mathring{\mcM}$, while $\gamma(s_0)\in\partial\mcM$, then
	\begin{equation}\label{Eq: No Glancing}
		g(\dot{\gamma}(s_0),n)\neq 0,
	\end{equation}
	where $\dot{\gamma}$ is the tangent vector to $\gamma$ while $n$ is the normal vector, inward pointing from $\partial\mcM$ at $s_0$. 
\end{proposition}

\begin{proof}
$\Longrightarrow$ Given $p \in \partial \mcM$, we consider $\mathcal{U}_p\subset\mcM$ and Gaussian normal coordinates $(\mathsf{y}, z)$ so that the metric $g$ takes the form
\begin{equation*}
ds^2 = dz^2 + h(\mathsf{y},z)=dz^2+h_{ij}(\mathsf{y},z)d\mathsf{y}^i d\mathsf{y}^j,
\end{equation*}
where $h$ is a Lorentzian metric. Assume, by contradiction, that there exists a geodesic $\gamma: [0,1] \to \mcM$ with $\gamma(0) \in \mathring{\mcM}$ and $s_0\in (0,1]$ such that $\gamma(s_0)=p$ while $g(\dot{\gamma}(s_0), \partial_z) =0$. Since $\partial_z$ is normal to the boundary, this entails that $\dot{\gamma}(s_0)$ is tangent to $\partial\mcM$. Since $\Gamma^z_{zz}=\Gamma^z_{iz}=0$ and $\Gamma^z_{ij}=-\frac{1}{2}\partial_z h_{ij}$, it follows that 
$$\mathrm{Hess}_g(z)_{ij} = \nabla_i \nabla_j z = \partial_i \partial_j z - \Gamma^z_{ij} \partial_z z = - \Gamma^z_{ij}= \frac{1}{2}\partial_z h_{ij},$$ 
while the $z$-components of the geodesic equations reads
\begin{equation*}
\ddot{z}(s) + \Gamma^z_{k j} \dot{\gamma}^k(s) \dot{\gamma}^j(s) = \ddot{z}(s) - \frac{1}{2} \partial_z h_{kj} \, \dot{\mathsf y}^k(s) \dot{\mathsf y}^j(s) = 0.
\end{equation*}
Since, per hypothesis, at $p\in\partial\mcM$, $\mathrm{Hess}_g(z)_{ij}(s_0) = \frac{1}{2} \partial_z h_{ij}(s_0) \le 0$, it holds that
\begin{equation*}
\ddot{z}(s) = \frac{1}{2}\partial_z h_{kj} \, \dot{\mathsf y}^k(s) \dot{\mathsf y}^j(s) \le \frac{1}{2} \underbrace{[\partial_z h_{kj} - \partial_z h_{kj}(s_0)]}_{\Gamma^z_{kj}(s) - \Gamma^z_{kj}(s_0)} \, \dot{\mathsf y}^k(s) \dot{\mathsf y}^j(s).
	\end{equation*}
Smoothness of the Christoffel symbols entails that
\begin{equation*}
\bigg \vert \Gamma^z_{kj}(s) - \Gamma^z_{kj}(s_0) \bigg \vert = \bigg \vert \int_{s_0}^s \partial_\lambda \Gamma^z_{kj}(\lambda) \, d \lambda \bigg \vert = \bigg \vert \int_0^{z(s)} \partial_z \Gamma^z_{kj} \, dz \bigg \vert \le C |z(s)|,
\end{equation*}
where $C>0$ can be chosen uniformly on a sufficiently small compact subset of $\mathcal{U}_p$ containing both $p$ and $\gamma(s)$. Since, on a sufficiently small interval of the affine parameter $s$, we can also claim that $\dot{\mathsf{y}^k}$ is bounded for all $k=1,\dots, d-1$, $d=\dim\mcM$, this entails that there exists a second positive constant $C^\prime$ such that  $\ddot{z}(s)\leq C^\prime z(s)$, where we have dropped the absolute value being $z$ positive in $\mcM$. Recalling that $\dot{z}(s_0) = 0$, this inequality yields upon repeated integration
\begin{flalign*}
\dot{z}(s) & \leq C^\prime\int_{s_0}^s  z(\lambda) \, d\lambda, \\
z(s) &\leq C^\prime\int_{s_0}^s \int_{\lambda}^s z(\lambda) \, d\lambda \, d\eta =  C^\prime \int_{s_0}^s (s-\lambda) z(\lambda) \, d\lambda \leq C^\prime s \underbrace{\int_{s_0}^s z(\lambda) \, d\lambda}_{\mathsf R(s)}.
	\end{flalign*}
Notice that $\mathsf R (s_0) = 0$ and $\mathsf R'(s) = z(s) \le C' s \mathsf R(s)$. Since $\mathsf{R}(s)\geq 0$ being the integral of a positive quantity, by Gr\"onwall lemma, we obtain that $\mathsf R(s) = 0$ and, thus, $z(s) = 0$. This is a contradiction since the geodesic is evolving along $\partial\mcM$.

\vskip .2cm

$\Longleftarrow$ Assume the no glancing hypothesis and that the boundary is not infinitesimally convex. Then, there exists $p \in \partial \mcM$ and a vector field $X \in T_p \mcM$ such that $g(X,\partial_z)=0$ and $\mathrm{Hess}_g(z)(X,X) > 0$. As above, $z$ is a boundary function defined in an open subset $\mathcal{U}_p$ containing $p$ where Gaussian normal coordinates have been implicitly introduced. Let $\gamma: I \to \mcM$, $I=[0,1]$, be the geodesic such that $\gamma(0)=p$, while $\dot{\gamma}(0)=X$. Denoting by $z(s) \doteq z(\gamma(s))$ for any $s \in I$ we have that  $z(0)= 0$ and $\dot{z}(0) = d z(X) = 0$.  The geodesic equation entails in turn that
$$\ddot{z}(0)= -\Gamma^z_{ij} {X}^i {X}^j = \mathrm{Hess}_g(z) {X}^i {X}^j \equiv \mathrm{Hess}_g(z) (X, X) > 0.$$ 
Expanding $z(s)$ in Taylor series around $s = 0$, we therefore obtain $z(s) = \frac{1}{2} \mathrm{Hess}_g(z) (X,X) s^2 + \mathcal{O}(s^2)$. Therefore, neglecting the infinitesimal higher order contributions, there exists $s^\prime>0$ such that $z(s) > 0$ for all $s\in(0,s^\prime)$, which implies $\gamma(s)\in\mathring{\mcM}$. At the same time $dz(\dot{\gamma}(0))=g(\dot{\gamma}(0), \partial_z) =0$, which contradicts the no glancing hypothesis and hence it concludes the proof.
\end{proof}

\noindent We observe that, in Proposition \ref{Prop: no-glancing}, the condition on the existence of a point $s\in I$ for which $\gamma(s)$ lies in the interior of $\mcM$ is necessary, since, being $\partial\mcM$ timelike, if $\dim\mcM=d> 2$, there are lightlike curves entirely contained in $\partial\mcM$. 

\begin{definition}\label{Def: Cauchy neighbourhood}
	Let $(\mcM,g)$ be an infinitesimally convex, globally hyperbolic spacetime with timelike boundary as per Assumption \ref{Ass: Infinitesimally convex} and let us identify it with $\bR\times\Sigma$, leaving the corresponding isometry implicit. We denote by $\Sigma_0\doteq\{0\}\times\Sigma$ the Cauchy hypersurface at $t=0$. For every $T>0$, we set
	$$\mathcal{N}_T\doteq t^{-1}((-T,T))=(-T,T)\times\Sigma.$$
	We call $\mathcal{N}_T$ a {\bf Cauchy neighbourhood} of $\Sigma_0$. In the following we fix one such open set, denoting it simply by $\mathcal{N}$, and we set $\mathring{\mathcal{N}}\doteq\mathcal{N}\cap\mathring{\mcM}$.
\end{definition}

\noindent Recall that $t$ is a temporal function, hence strictly monotone along every causal curve. Hence, as in a globally hyperbolic spacetime with empty boundary, every causal curve whose endpoints lie in $\mathcal{N}_T$ is entirely contained therein. This entails that $\mathcal{N}_T$ is causally convex.

\subsection{Static Globally Hyperbolic Spacetimes with a Timelike Boundary}\label{Sec: Static Spacetimes}

In our work, we want to make strong contact with \cite{Dappiaggi-Drago_2019}, where advanced and retarded fundamental solutions for the Klein--Gordon operator were constructed by means of spectral techniques on {\bf (standard) static}, globally hyperbolic spacetimes with timelike boundary. In this case, the representation in Equation \eqref{Eq: line element psi*g} can be chosen so that $\partial_t$ is a timelike Killing field, \textit{i.e.}, the lapse function $\beta$ is $t$-independent and $h_t=h_0\equiv h$ for all $t\in\bR$.
In \cite{Dappiaggi-Drago_2019}, additional constraints on the underlying geometry have been imposed and we outline them succinctly. More precisely, we are interested in working with Lorentzian manifolds whose Cauchy surfaces are of bounded geometry. In the following we combine the setting outlined in \cite{Gérard} with the definitions, given in \cite{AmmanGrosseNistor}, of a Riemannian manifold of bounded geometry with a non-empty boundary. As starting point, we recall the following standard characterization, see \cite{Eich91}.

\begin{definition}\label{Def: Bounded_geometry}
	A Riemannian manifold $(\mathsf N,h)$ with $\partial \mathsf N=\emptyset$ is called of {\bf bounded geometry} if the injectivity radius $r_{\mathrm{inj}}(\mathsf N)>0$ and if $T \mathsf N$ is of totally bounded curvature, that is $\|\nabla^k \mathrm{Riem}\|_{L^\infty}<\infty$ for all $k\in\mathbb{N}\cup\{0\}$, $\mathrm{Riem}$ being the curvature tensor and $\nabla$ the Levi-Civita connection associated with the Riemannian metric $h$.
\end{definition}

\noindent To avoid the problem that $r_{\mathrm{inj}}(\mathsf N)$ vanishes if we drop the assumption $\partial \mathsf N=\emptyset$, one must first consider the following generalization.

\begin{definition}\label{Def: Bounded_geometry_submanifold}
	Let $(\mathsf N,h)$ be a Riemannian manifold of bounded geometry as per Definition \ref{Def: Bounded_geometry} and let $(Y,\iota_Y)$ be a codimension $1$ closed, embedded, smooth submanifold with an inward pointing, unit normal vector field $n_Y$. We say that $(Y,\iota^*_Y h)$ is a {\bf bounded geometry submanifold} if the following two conditions hold:
	\begin{enumerate}
		\item the second fundamental form $K_Y$ of $Y$ in $\mathsf N$ together with all its covariant derivatives along $Y$ is bounded,
		\item there exists $\epsilon_Y>0$ such that the map $$\varphi_{n_Y}:Y\times(-\epsilon_Y,\epsilon_Y)\to \mathsf N, \qquad (x,z)\mapsto\varphi_{n_Y}(x,z)\doteq\exp_x(zn_{Y,x}),$$ is injective, where $\exp_x$ is the exponential map of $\mathsf N$ at $x \in Y$.
	\end{enumerate}
\end{definition} 

\noindent We observe that, as proven in \cite{AmmanGrosseNistor}, Definition \ref{Def: Bounded_geometry_submanifold} entails that $(Y,\iota^*_Y g)$ is automatically a Riemannian manifold of bounded geometry in the sense of Definition \ref{Def: Bounded_geometry}. We can thus introduce the class of Riemannian manifolds relevant to us. To disambiguate between Definition \ref{Def: Bounded_geometry} and to make contact with the structures already introduced, we denote by $\Sigma$ Riemannian manifolds with a non-empty boundary.

\begin{definition}\label{Def: Bounded_geometry_boundary}
	Let $(\Sigma,h)$ be a Riemannian manifold with $\partial\Sigma\neq\emptyset$. We say that $(\Sigma,h)$ has {\bf bounded geometry} if there exists a Riemannian manifold of bounded geometry $(\widehat{\Sigma},\widehat{h})$, with $\partial \widehat{\Sigma} = \emptyset$ and of the same dimension of $\Sigma$ such that 
	\begin{enumerate}
		\item $\Sigma$ is a closed subset of $\widehat{\Sigma}$ and $h=\widehat{h}|_\Sigma$,
		\item $(\partial\Sigma,\iota^*\widehat{h})$ is a bounded geometry submanifold of $\widehat{\Sigma}$, where $\iota:\partial\Sigma\to\widehat{\Sigma}$ is the embedding map.
	\end{enumerate}
\end{definition}

\begin{remark}
	We remark that in \cite{Schick} a definition of a Riemannian manifold with boundary and of bounded geometry has been given without making any reference to the ambient space $(\widehat{\Sigma},\widehat{h})$. On the one hand, Definition \ref{Def: Bounded_geometry_boundary} is equivalent to the one given in this reference. On the other hand, in the ensuing discussion, the existence of $\widehat{\Sigma}$ will play an important role and, therefore, we highlight it from the very beginning.
\end{remark}

\noindent Henceforth, we shall restrict our attention to the distinguished subclass of standard static globally hyperbolic spacetimes, for which the orthogonal splitting in Equation \eqref{Eq: line element psi*g} can be chosen with $t$-independent coefficients.

\begin{remark}\label{Rem: Comparison to GOW}
	We highlight that our assumptions on the underlying background are slightly different from those in \cite{Gérard}. In this reference Lorentzian manifolds of bounded geometry are introduced and, under the additional assumption of global hyperbolicity, an associated notion of Cauchy hypersurface of bounded geometry is developed, see \cite[Def. 3.1 \& 3.3]{Gérard}. These assumptions are motivated by the goal of establishing a global, pseudodifferential construction of Hadamard states. In our framework bounded geometry is a tool to guarantee the existence on each Cauchy surface of a good notion of Sobolev spaces and of trace maps, see \cite{AmmanGrosseNistor}. Yet, if we impose in addition that, in Equation \eqref{Eq: line element psi*g}, $\beta\in C^\infty_b(\Sigma,h)$, namely $\beta$ is smooth and bounded with all its covariant derivatives up to all orders, then a standard static, globally hyperbolic spacetime $(\mathcal{M},g)$ is also admissible according to \cite{Gérard}.
\end{remark}

\begin{example}
	The prototypical example of a standard static, globally hyperbolic spacetime with a non-empty timelike boundary is $d$-dimensional Minkowski half-space $(\mathbb{H}^d, \eta)$, which is in addition infinitesimally convex as per Assumption \ref{Ass: Infinitesimally convex}. It serves as testing ground for the investigation of quantum field theories on more general backgrounds, but we refrain from discussing it here, since a detailed account of this scenario can be found in \cite{Costeri_25}. Yet it is important to mention it, since accounting for this scenario is one of the main reasons why we do not restrict attention to strictly convex, globally hyperbolic spacetimes, a restriction that would otherwise simplify some of our constructions.
\end{example}

\begin{example}\label{Rem: no AdS}
	In the mathematical physics literature, the notion of a spacetime with timelike boundary is typically exemplified by the $d$-dimensional Anti-de Sitter spacetime $\mathbb{A}d\mathbb{S}_{d}$, $d>2$. Upon considering the restriction to the Poincaré patch $P\mathbb{A}d\mathbb{S}_{d}$, a conformal rescaling of the metric identifies it with the interior of the $d$-dimensional Minkowski half-space, endowed with a timelike conformal boundary, see \cite[Rmk. 2.2]{Costeri_25}. The main difficulty in this setting is that, under this conformal transformation, mass terms such as the one in the Klein-Gordon equation give rise on the half-space to a potential, singular at the boundary. This behaviour requires techniques different from those discussed in the present work in order to construct Hadamard states compatible with Robin boundary conditions, \textit{cf.}, \cite{Dappiaggi_26}.
\end{example}

\noindent We conclude this first introductory part of the paper by establishing an ancillary geometric property possessed by standard static, globally hyperbolic spacetime with a timelike boundary. This is a generalization of Corollary \ref{Cor: Globally Hyperbolic Extension} which will simplify some of the following proofs. For related results, see also \cite{Grosse}.

\begin{proposition}\label{Prop: Extension of Globally Hyperbolic}
	Let $(\mathcal{M},g)$ be a standard static, globally hyperbolic spacetime with timelike boundary such that, referring to the orthogonal splitting as per Proposition \ref{Prop: Globally Hyperbolic}, the Cauchy surface $(\Sigma,h)$ is of bounded geometry, \textit{cf.}, Definition \ref{Def: Bounded_geometry_boundary}. Then, there exists a standard static, globally hyperbolic spacetime $(\mcM^\prime,g^\prime)$ with empty boundary such that $(\mcM,g)$ is isometrically embedded therein.
\end{proposition}

\begin{proof}
Under the given assumptions, Definition \ref{Def: Bounded_geometry_boundary} guarantees that there exists $(\widehat{\Sigma},\widehat{h})$, a Riemannian manifold of bounded geometry with $\partial\widehat{\Sigma}=\emptyset$ extending $(\Sigma, h)$. Up to the identification of $\mcM$ with $\bR\times\Sigma$ which is left implicit, we set $\mcM^\prime=\bR\times\widehat{\Sigma}$. We endow it with the metric $g^\prime$ whose associated line element reads
	$$ds^2=-\beta^\prime dt^2+\widehat{h}.$$
	Here $\beta^\prime\in C^\infty(\widehat{\Sigma};(0,\infty))$ must be chosen so that $\beta^\prime|_\Sigma=\beta$. In order for $(\mcM^\prime,g')$ to be globally hyperbolic, the Riemannian manifold $\widehat{\Sigma}$ endowed with the optical metric $(\beta')^{-1}\widehat{h}$ must be complete, see, {\it e.g.}, \cite[Thm. 4.4]{Caponio} and \cite[Rem. 7.3]{Sanchez:2007rx}. This result also entails that, being $(\mcM, g)$ globally hyperbolic per hypothesis,  $q_\Sigma\doteq\beta^{-1}h$ is complete on $\Sigma$. Furthermore, since $\Sigma$ is per hypothesis a closed subset of $\widehat{\Sigma}$, there exists $\tilde{\beta}\in C^\infty(\widehat{\Sigma};(0,\infty))$ such that $\tilde{\beta}|_{\Sigma}=\beta$. However, there is no guarantee that the optical metric $q \doteq \tilde{\beta}^{-1}\widehat{h}$ is complete. 
    
    To circumvent this issue, we invoke a classical result in differential geometry, see \cite{NomizuOzeki}, which asserts that, for every Riemannian manifold $(\widehat{\Sigma},q)$, there exists a smooth function $\Omega\in C^\infty(\widehat{\Sigma};(0,\infty))$ such that the conformally rescaled manifold $(\widehat{\Sigma},\Omega^2 q)$ is complete. We prove in the following that such $\Omega$ can be chosen so that $\Omega|_\Sigma=1$. Since $(\Sigma,q_\Sigma)$ is complete, using \cite{Gordon73,Gordon74} we can infer that there exists a smooth proper function $r\in C^\infty(\Sigma;[0,\infty))$ such that $|dr|_{q_\Sigma}\doteq \sqrt{q_\Sigma^{-1}(dr,dr)}\leq\frac{1}{2}$. Observe that the proof in these two references applies also to complete Riemannian manifolds with non-empty boundary. We extend $r$ to a smooth proper function $R\in C^\infty(\widehat{\Sigma};[0,\infty))$ such that $R|_{\Sigma}=r$ and on $\Sigma$, $|dR|_q=|dr|_{q_\Sigma}\leq\frac{1}{2}$. First of all, one can extend $r$ smoothly to a neighbourhood of $\Sigma$, using a collar neighbourhood of $\partial\Sigma$ in $\widehat\Sigma$, see Definition \ref{Def: Bounded_geometry_submanifold}. This extension is then patched with a smooth positive exhaustion of $\widehat\Sigma$, whose existence follows from \cite[Prop.~2.28]{Lee_2013}. The patching can be arranged so that the resulting function agrees with $r$ on $\Sigma$ and, by construction  $|dR|_q=|dr|_{q_\Sigma}$ on $\Sigma$. Herein we can extend $r$ smoothly across the boundary and, then, we proceed with an iteration building a smooth exhaustion function of $\widehat{\Sigma}$, see \cite[Prop 2.28]{Lee_2013}. Setting for simplicity $F\doteq|dR|_q^2$, we choose $\Theta\in C^\infty([0,\infty);[0,\infty))$ such that 
	$$\Theta(s)=0\;\mathrm{if}\; 0\leq s\leq \frac{1}{2}\quad\mathrm{and}\quad 1+\Theta(s)\geq \sqrt{s}\; \, \mathrm{for \, all}\; s\geq0.$$
	Fixing $\Omega\doteq1+\Theta\circ F$ this is per construction smooth and strictly positive. Furthermore, $|dr|_{q_\Sigma}\leq\frac{1}{2}$ entails $F|_\Sigma\leq\frac{1}{4}$ and thus $\Omega|_\Sigma=1$. We have found the sought candidate for the conformal factor and we verify that it enjoys the necessary properties. Hence, setting $q_\Omega\doteq\Omega^2 q$, for each $x\in\widehat{\Sigma}$
	$$|dR|_{q_\Omega}(x)=\Omega^{-1}(x)|dR|_{q}(x)=\frac{\sqrt{F(x)}}{1+\Theta(F(x))}\leq 1.$$
	Thus, for every piecewise differentiable curve $\gamma:[a,b]\to\widehat{\Sigma}$, with $a,b$ finite, it holds that
$$|R(\gamma(b))-R(\gamma(a))|	\leq\int_a^b |dR_{\gamma(t)}(\dot\gamma(t))|\,dt \leq	\int_a^b |\dot\gamma(t)|_{q_\Omega}\,dt=L_{q_\Omega}(\gamma),$$
where $L_{q_\Omega}(\gamma)$ denotes the length of the curve $\gamma$ with respect to the metric $q_\Omega$. In turn, this entails that, for every $x,y\in\widehat{\Sigma}$, 
$$|R(x)-R(y)|\leq d_{q_\Omega}(x,y),$$
namely $R$ is $1$-Lipschitz with respect to the Riemannian distance $d_{q_\Omega}$ induced by $q_\Omega$. Having established this property, we can consider any but fixed point $x_0\in\widehat{\Sigma}$ and $L>0$. Denoting by $B_{q_\Omega}(x_0,L)$ the ball of radius $L$, centred at $x_0$, with distance induced by $q_\Omega$, it holds that, for any $y\in B_{q_\Omega}(x_0,L)$, 
$$R(y)\leq R(x_0)+L.$$
In other words $\overline{B_{q_\Omega}(x_0,L)}\subset	R^{-1}\bigl([0,R(x_0)+L]\bigr)$, which is a compact set, being $R$ proper. We have thus proven that any open ball is relatively compact. Consequently, since any bounded subset of $\widehat{\Sigma}$ is contained in one of such balls, it is also relatively compact. This entails that we can now invoke the Hopf-Rinow theorem, see \cite[Thm  7.1]{Petersen}, to conclude that $(\widehat{\Sigma}, q_\Omega)$ is complete. Setting $\beta^\prime\doteq\Omega^{-2}\tilde\beta$ is the last step in order to conclude the proof.
\end{proof}

\begin{remark}
We observe that Proposition \ref{Prop: Extension of Globally Hyperbolic} does not entail that $(\mcM^\prime,g^\prime)$ is a Lorentzian manifold of bounded geometry in the sense of \cite{Gérard} since there is no control on the boundedness of $\beta^\prime$ and of all its covariant derivatives up to all orders. 
\end{remark}

\subsection{On the Synge world function and its reflected counterpart}\label{Sec: Geodesic and Reflected Geodesic Distances Claudio} 

In this section we introduce a few additional geometric structures, which will play a key role in the identification of a local form for Hadamard states in presence of a timelike boundary. First and foremost, we will recall the standard notion of \emph{Synge world function}, and subsequently we will introduce a counterpart which accounts for the reflection of light rays at $\partial\mcM$. 

\paragraph{Synge World Function --} We start our discussion by defining a key ingredient in the analysis of the structural properties of Hadamard states. On account of our standing assumptions, we can guarantee the existence of a globally hyperbolic spacetime with empty boundary, $(\mcM^\prime,g^\prime)$ containing $(\mcM,g)$, see Corollary \ref{Cor: Globally Hyperbolic Extension}. Let $\mathcal{U}^\prime\subset\mcM^\prime$ be a geodesically convex subset and let $\mathcal{U}\doteq \mathcal{U}^\prime\cap\mcM$. For any pair of points $x,y\in \mathcal{U}$, let $\gamma^{(x,y)}:[0,1]\to \mathcal{U}'$ be the unique geodesic in $\mcM^\prime$ connecting them, \textit{i.e.}, $\gamma^{(x,y)}(0) = x$ and $\gamma^{(x,y)}(1) = y$. Then we denote by $\sigma$ the {\bf Synge world function}, that is the halved squared geodesic distance defined by
	\begin{equation}\label{Eq: Synge World Function Interior}
		\sigma(x,y)\doteq\frac{1}{2}\int\limits_0^1 g^\prime(\dot{\gamma}^{(x,y)}(s),\dot{\gamma}^{(x,y)}(s)) \, ds,
	\end{equation}
where $\dot{\gamma}^{(x,y)}$ is the tangent vector to $\gamma^{(x,y)}$.

\begin{remark}\label{Rem: Convexity of the boundary}
	Equation \eqref{Eq: Synge World Function Interior} codifies in a single formula two distinct possibilities that we have to account for:
	\begin{enumerate}
		\item $\mathcal{U}^\prime\subset\mathring{\mcM}$. In this case $\mathcal{U}^\prime=\mathcal{U}$ and we could have avoided referring to the ambient globally hyperbolic spacetime $\mcM^\prime$ replacing $g^\prime$ with $g$ in Equation \eqref{Eq: Synge World Function Interior} since $g^\prime|_\mcM=g$.
		\item $\mathcal{U}^\prime\cap\partial\mcM\neq\emptyset$. In this case $\mathcal{U}^\prime\neq \mathcal{U}$ and one should account for two different scenarios, namely, given any pair of points $x,y\in \mathcal{U}$, the connecting geodesic $\gamma^{(x,y)}$ is entirely contained in $\mathcal{U}$ or not. In the first instance, one can replace once more $g^\prime$ with $g$, while, in the second case, this is not possible. As a matter of fact, in the latter scenario, the definition of $\sigma(x,y)$ depends on the choice of $\mcM^\prime$ which is highly non-unique. If we require that $\mcM$ is infinitesimally convex at $\partial\mcM$ as per Assumption \ref{Ass: Infinitesimally convex}, one can fine tune $\mathcal{U}$ so that all geodesics connecting two points $x,y\in \mathcal{U}$ lie entirely therein, see \cite{Caponio13}. This is the reason why we have taken this condition as a standing assumption.
	\end{enumerate}
\end{remark}

For later convenience, we recall some standard properties of Synge world function, as defined in Equation \eqref{Eq: Synge World Function Interior}. We omit their proof and we refer the interested reader to \cite{Poisson_2011}.
As above, we denote by $\mathcal{U} \doteq \mathcal{U}^\prime\cap\mcM$ an open subset such that $\mathcal{U}^\prime$ is geodesically convex:
\begin{itemize}
	\item[\ding{104}] $\sigma\in C^\infty(\mathcal{U}\times \mathcal{U})$, it is symmetric and it vanishes if $x$ is connected to $y$ by a lightlike curve.
	\item[\ding{104}] $\sigma$ satisfies an {\em eikonal equation} in both variables, that is, 
	\begin{equation}\label{Eq: Eikonal Equation for Sigma}
	g^{-1}(d_x\sigma,d_x\sigma)=g^{-1}(d_y\sigma,d_y\sigma)=2\sigma,
	\end{equation}
	the subscripts $x$ and $y$ denoting respectively whether the differential is acting on the first or on the second entry.
\item[\ding{104}] Denoting the coincidence point limit by 
	$$[\cdot]:C^\infty(\mathcal{U}\times \mathcal{U})\to C^\infty(\mathcal{U}),\quad f(x,y)\mapsto [f](x)\doteq f(x,x),$$
	it holds that 
	\begin{equation}\label{eq: Coincidence Point Limit for Sigma}
	[\sigma]=0,\quad[\sigma_{\nu}]=0,\quad[\sigma_{\mu\nu}]=g_{\mu\nu},\quad[\Box_g \sigma]=d,
\end{equation}
where $\dim\mcM=d$. Observe that the subscript $\sigma_{\nu}\doteq \nabla_{\nu,x}\sigma$, the $x$ entailing that the covariant derivative acts on the first variable. Should we need to act on the second entry of $\sigma$, we would employ primed indices, {\it e.g.}, $\sigma_{\mu^\prime}\doteq\nabla_{\mu,y}\sigma$. With reference to Remark \ref{Rem: Convexity of the boundary}, these identities do not depend on the choice of $\mcM^\prime$ and they hold true also at $\partial\mcM$.
\end{itemize}

\paragraph{Reflected Synge World Function --} In order to characterize the singularity structure of the two-point correlation function of a Hadamard state, as well as that of the advanced and retarded propagators associated with a Klein-Gordon operator on $(\mcM,g)$, an additional geometric structure is needed. Intuitively, this should encode the fact that geodesics, upon reaching $\partial\mcM$, are reflected back into the bulk. In \cite{Costeri_25}, considering half-Minkowski spacetime, this structure was constructed directly from Synge world function by exploiting the existence of a discrete isometry, namely $z\mapsto -z$, where $z$ denotes the coordinate orthogonal to the boundary, located at $z=0$. In the more general setting considered here, no such special geometric assumption is available. Therefore, a different construction is required, which we now outline. Heuristically speaking, despite the assumption that $\partial\mcM$ is infinitesimally convex, we need to account for the possibility that a geodesic has multiple reflections, that is it hits several time $\partial\mcM$. In the following, in order to make the role of reflections more transparent, we split the regime where only one reflection is possible, from that where an arbitrary number is allowed. The reason is that, in the first case all formulae and, more importantly, their interpretation are more transparent and accessible, especially if one seeks a comparison with \cite{Costeri_25}.

\vskip .2cm

Let $(\mcM,g)$ be an infinitesimally convex and globally hyperbolic spacetime with a timelike boundary as per Assumption \ref{Ass: Infinitesimally convex}. For any $p\in\partial\mcM$, we define the {\bf reflection map}
\begin{equation}\label{Eq: Reflection Vector}
	\mathcal{R}_p:T_p\mcM\to T_p\mcM,\quad v\mapsto\mathcal{R}(v)\doteq v-2g(v,n)n,
\end{equation}
where $n\in T_p\mcM$ is the inward pointing unit vector, normal to $\partial\mcM$.

\begin{definition}\label{Def: N-reflected initial data}
	Let $N\geq 1$. We denote by $\Gamma_N\subset T\mathring{\mcM}$ the collection of those points $(x,v)$ for which there exist $\tau_0,\ldots,\tau_{N-1}\in(0,1)$, with $\tau_0<\cdots<\tau_{N-1}$, and a continuous, piecewise geodesic $\gamma_{(x,v)}:[0,1]\to\mcM$ of the form
	\begin{equation}\label{Eq: gamma map}
		\gamma_{(x,v)}(\tau)=\left\{
		\begin{array}{ll}
			\gamma_0(\tau) & \tau\in[0,\tau_0],\\
			\gamma_1(\tau) & \tau\in[\tau_0,\tau_1],\\
			\ldots & \\
			\gamma_N(\tau) & \tau\in[\tau_{N-1},1],
		\end{array}
		\right.
	\end{equation}
	such that the following conditions hold:
	\begin{itemize} 
		\item[\ding{104}] $\gamma_{(x,v)}(0)=x$ and $\dot{\gamma}_{(x,v)}(0)=v$. Furthermore, setting $x_i\doteq\gamma_{(x,v)}(\tau_i)$, it holds that $x_i\in\partial\mcM$ for all $i=0,\ldots,N-1$, while $\gamma_{(x,v)}(1)\in\mathring{\mcM}$,
		\item[\ding{104}] each $\gamma_i$ is affinely parametrized and, setting $v_i\doteq\dot{\gamma}_i(\tau_i)$, it holds by continuity of $\gamma_{(x,v)}$ that
		$$\gamma_i(\tau_i)=\gamma_{i+1}(\tau_i)=x_i,\qquad \dot{\gamma}_{i+1}(\tau_i)=\mathcal{R}_{x_i}(v_i),$$
		for every $i=0,\ldots,N-1$,
		\item[\ding{104}] the curve has no further intersections with the boundary, namely, setting $\tau_{-1}= 0$ and $\tau_N= 1$, it holds that $\gamma_i[(\tau_{i-1},\tau_i)]\subset\mathring{\mcM}$ for every $i=0,\ldots,N$.
	\end{itemize}
	We call $\gamma_{(x,v)}$ the {\bf $N$-reflected broken geodesic} generated by $(x,v)$, while $\tau_i$ and $x_i$ are called respectively its \emph{reflection times} and \emph{reflection points}.
\end{definition}

\noindent Observe that, on account of Proposition \ref{Prop: no-glancing}, $g(v_i,n)\neq 0$ at every reflection point. In the following we establish smooth dependence of the reflection times and points with respect to $(x,v)\in\Gamma_N$. The following result is closely related in spirit to \cite[Lem.~2.7]{Wunsch}, although its proof proceeds by an independent argument.

\begin{lemma}\label{Lem: Smooth dependence of reflected geodesics}
	For every $N\geq 1$, $\Gamma_N$ is an open subset of $T\mathring{\mcM}$. Furthermore, the maps
	$$(x,v)\mapsto\tau_i(x,v),\qquad (x,v)\mapsto x_i(x,v),\qquad (x,v)\mapsto v_i(x,v),$$
	are smooth on $\Gamma_N$, for every $i=0,\ldots,N-1$. In addition, also the endpoint map $(x,v)\mapsto\gamma_{(x,v)}(1)$ is smooth.
\end{lemma}

\begin{proof}
	Let $(\bar{x},\bar{v})\in\Gamma_N$ and denote by $\bar{\tau}_0,\ldots,\bar{\tau}_{N-1}$ the corresponding reflection times. Consider first the geodesic segment stemming from $(\bar{x},\bar{v})$ and let $z$ be a boundary-defining function in a neighbourhood of its first reflection point. Denoting by $\phi_s$ the geodesic flow of the ambient spacetime $(\mcM^\prime,g^\prime)$ as per Proposition \ref{Prop: Extension of Globally Hyperbolic} and by $\pi:T\mcM^\prime\to\mcM^\prime$ the canonical projection, we set
	$$F_0(x,v,s)\doteq z\bigl(\pi\circ\phi_s(x,v)\bigr).$$
	At $(\bar{x},\bar{v},\bar{\tau}_0)$ it holds that $F_0=0$, while $\partial_sF_0\neq 0$ on account of the transversality condition in Definition \ref{Def: N-reflected initial data}. The implicit function theorem therefore identifies, in a neighbourhood of $(\bar{x},\bar{v})$, a unique smooth first reflection time $\tau_0(x,v)$. Consequently, using Equation \eqref{Eq: gamma map}, both the first reflection point and the incoming and outgoing velocity vectors depend smoothly on $(x,v)$. The same argument can be applied iteratively starting from the geodesic segment stemming from the reflected initial datum at $x_0(x,v)$. This yields inductively the smoothness of all reflection times and points as well as of the incoming velocity vectors. One can draw a similar conclusion for the endpoint at $\tau=1$.
	
	We are left with one last property to establish, namely that, up to shrinking the neighbourhood of $(\bar{x},\bar{v})$, the broken geodesics obtained varying $(x,v)$ have no additional intersections with $\partial\mcM$. To this end, choose pairwise disjoint open intervals $I_i\subset(0,1)$, where each $I_i$ contains only the reflection time $\bar{\tau}_i$, but none of the others. The complement $K_0\doteq[0,1]\setminus\bigcup_{i=0}^{N-1}I_i$ is compact and, by Definition \ref{Def: N-reflected initial data}, its image $K\doteq\gamma_{(\bar{x},\bar{v})}(K_0)\subset\mathring{\mcM}$ is also compact. Since $\partial\mcM$ is a closed submanifold of $\mcM$, picking an auxiliary Riemannian metric on $\mcM$, the distance between $K$ and $\partial\mcM$ is strictly positive. Having already established the smooth dependence of the geodesic segments on their initial data, varying them in a sufficiently small neighbourhood of $(\bar{x},\bar{v})$, does not make this distance shrink to $0$. Hence they cannot intersect $\partial\mcM$ outside the intervals $I_i$. Consequently, every initial datum in this neighbourhood generates a broken geodesic with exactly $N$ reflections. Therefore we have found that, for every $(\bar{x},\bar{v})\in\Gamma_N$, there exists an open subset $U_{(\bar{x},\bar{v})}$ of $T\mathring{\mcM}$ which lies entirely in $\Gamma_N$. Consequently, we can cover the latter using these subsets, which is tantamount to saying that $\Gamma_N$ is open.
\end{proof}

Observe that, under Assumption \ref{Ass: Infinitesimally convex} and in view of Proposition \ref{Prop: no-glancing}, we know that, since $\dot{\gamma}_0(\tau_0)$ is outward pointing, then $\mathcal{R}(\dot{\gamma}_0(\tau_0))$ is directed towards the interior of $\mcM$. Iterating, the same applies at every reflection point.  We can collect this procedure in an {\em N-reflected exponential map}
\begin{equation}\label{Eq: Reflected Exp Map}
\mathrm{Exp}_{-,N}:\Gamma_N\to\mathring{\mcM},\quad (x,v)\mapsto\mathrm{Exp}_{-,N}(x,v)\doteq\gamma_{(x,v)}(1).
\end{equation}
In view of Lemma \ref{Lem: Smooth dependence of reflected geodesics}, $\mathrm{Exp}_{-,N}$ is smooth. These and some of the following concepts have appeared already in the literature with a slightly different nomenclature, especially in the analysis of billiard flows, see, {\it e.g.}, \cite{Petkov, Sogge}. We associate to it the {\em $N$-reflected endpoint map}, which serves the purpose of encoding both the initial and the endpoint of a broken geodesic:
\begin{equation}\label{Eq: N-reflected endpoint map}
	\Phi_N:\Gamma_N\to\mathring{\mcM}\times\mathring{\mcM},\quad (x,v)\mapsto\Phi_N(x,v)\doteq\bigl(x,\mathrm{Exp}_{-,N}(x,v)\bigr).
\end{equation}
It is important to keep in mind that $\Phi_N$ need not be injective, since distinct $N$-reflected broken geodesics may connect the same pair of points. Before identifying suitable local restrictions of $\Phi_N$, we single out those initial data at which its differential is invertible, so that $\Phi_N$ is a local diffeomorphism. The following definition is inspired by a similar analysis in \cite{Wunsch}, see in particular Section 4.

\begin{definition}\label{Def: Regular reflected datum}
	Let $N\geq 1$ and let $(x,v)\in\Gamma_N$, as per Definition \ref{Def: N-reflected initial data}. We call $\Gamma_N^{\rm reg}\subset\Gamma_N$ the collection of all {\bf regular $N$-reflected data} $(x,v)$, that is the differential
	$$d_v\mathrm{Exp}_{-,N}|_{(x,v)}:T_x\mcM\to T_{\mathrm{Exp}_{-,N}(x,v)}\mcM$$
	is an isomorphism. Here the subscript $v$ denotes differentiation along the fibres of $T\mathring{\mcM}$. At the same time, we say that $y\in\mathring{\mcM}$ is {\em reflected-conjugate} to $x$ (along $\gamma_{(x,v)}$) if $y=\mathrm{Exp}_{-,N}(x,v)$ but $d_v\mathrm{Exp}_{-,N}|_{(x,v)}:T_x\mcM\to T_y\mcM$ is not invertible.
\end{definition}

\begin{lemma}\label{Lem: Regularity endpoint map}
	Let $N\geq 1$ and let $(x,v)\in\Gamma_N$, as per Definition \ref{Def: N-reflected initial data}. Then $(x,v)\in\Gamma_N^{\rm reg}$ if and only if
	$$d\Phi_N|_{(x,v)}:T_{(x,v)}\Gamma_N\to T_{\Phi_N(x,v)}(\mathring{\mcM}\times\mathring{\mcM})$$
	is an isomorphism. Furthermore, $\Gamma_N^{\rm reg}$ is an open subset of $\Gamma_N$ and, hence, also of $T\mathring{\mcM}$.
\end{lemma}

\begin{proof}
Choosing a local trivialization of $T\mathring{\mcM}$ around $x$, the differential of $\Phi_N$ at $(x,v)$ takes the block form
$$d\Phi_N|_{(x,v)}=\left(\begin{array}{cc}
\mathrm{id}_{T_x\mcM} & 0 \\ 
d_x\mathrm{Exp}_{-,N}|_{(x,v)} & d_v\mathrm{Exp}_{-,N}|_{(x,v)}
\end{array}\right).$$
Hence, $d\Phi_N|_{(x,v)}$ is an isomorphism if and only if $d_v\mathrm{Exp}_{-,N}|_{(x,v)}$ enjoys the same property. Furthermore smoothness of $\mathrm{Exp}_{-,N}$ entails that $\Gamma_N^{\mathrm{reg}}$ is an open subset of $\Gamma_N$ and, hence, of $T\mathring{\mcM}$ on account of Lemma \ref{Lem: Smooth dependence of reflected geodesics}.
\end{proof}

\noindent As already mentioned, even if we restrict our attention to regular initial points for a broken geodesic with $N$ reflections, that is $(x,v_1)\in\Gamma_N^{\mathrm{reg}}$, it can happen that there exists $(x,v_2)\in\Gamma_N^{\mathrm{reg}}$ such that $\mathrm{Exp}_{-,N}(x,v_1)=\mathrm{Exp}_{-,N}(x,v_2)$, see Equation \eqref{Eq: Reflected Exp Map}. In other words the initial and final points in $\mathring{\mcM}$ are the same. In the following, we show that it is possible to subdivide $\Gamma_N^{\mathrm{reg}}$ in a countable family of open subsets, labelled by a parameter $\alpha\in\mathbb{N}$, where this phenomenon cannot occur. The statement that $\alpha$ is countable is a byproduct of two features. On the one hand, the base manifold is second countable, while, on the other hand, on each fibre $d\Phi_N|_{(x,v)}$ is an isomorphism. Hence, given an arbitrary but fixed base point $x\in\mathring{\mcM}$ it is not possible to find infinitely many pairs $\{(x,v_j)\}_{j \in \mathbb{N}}\in T_x\mathring{\mcM}$ with the same endpoint under $\mathrm{Exp}_{-,N}$.

\begin{definition}\label{Def: Regular reflected branch chart}
Let $N\geq 1$. A {\bf regular $N$-reflected branch chart over $\mcM$} is a pair $(\mathcal{V}_{N,\alpha},\Omega_{N,\alpha})$, indexed by $\alpha \in \mathbb{N}$, such that both $\mathcal{V}_{N,\alpha}\subset\Gamma_N^{\rm reg}$ and $\Omega_{N,\alpha}\subset\mathring{\mcM}\times\mathring{\mcM}$ are open while
$$\Phi_{N,\alpha}\doteq\Phi_N|_{\mathcal{V}_{N,\alpha}}:\mathcal{V}_{N,\alpha}\to\Omega_{N,\alpha}$$
is a diffeomorphism. 
\end{definition}	

Observe that, since in the first component $\Phi_N$ acts as the identity and it is per definition invertible on $\mathcal{V}_{N,\alpha}$, we can identify a smooth map
	$$(x,y)\mapsto v_{N,\alpha}(x,y)\in T_x\mcM\;\text{such that}\;\Phi^{-1}_{N,\alpha}(x,y)=\bigl(x,v_{N,\alpha}(x,y)\bigr).$$
	For every $(x,y)\in\Omega_{N,\alpha}$, we set
	\begin{equation}\label{Eq: GammaN}
	\gamma_{N,\alpha}^{(x,y)}\doteq\gamma_{(x,v_{N,\alpha}(x,y))},
	\end{equation}
	and call it the {\em broken geodesic associated with the branch $(N,\alpha)$}.

\begin{lemma}\label{Lem: Existence of regular reflected branch charts}
	Given $N\geq 1$ and $(x_0,v_0)\in\Gamma_N^{\rm reg}$, there exist $\alpha \in \mathbb{N}$ and a regular $N$-reflected branch chart $(\mathcal{V}_{N,\alpha},\Omega_{N,\alpha})$ over $\mathring{\mcM}$ such that $(x_0,v_0)\in\mathcal{V}_{N,\alpha}$.
\end{lemma}

\begin{proof}
	Since on account of Lemma \ref{Lem: Regularity endpoint map} the differential $d\Phi_N|_{(x_0,v_0)}$ is an isomorphism, we can apply the inverse function theorem. This yields open neighbourhoods $\mathcal{V}\subset\Gamma_N^{\rm reg}$ of $(x_0,v_0)$ and $\Omega\subset\mathring{\mcM}\times\mathring{\mcM}$ of $\Phi_N(x_0,v_0)$ such that
	$$\Phi_N|_{\mathcal{V}}:\mathcal{V}\to\Omega$$
	is a diffeomorphism. We identify $\mathcal{V}$ with the sought $\mathcal{V}_{N,\alpha}$.
\end{proof}

\begin{remark}\label{Rem: Overlapping reflected branches}
	Observe that different branch charts are not necessary disjoint and they might intersect. Hence, a pair $(x,y)$ belonging to such an intersection can be connected by distinct regular $N$-reflected broken geodesics. The branch label $\alpha \in \mathbb{N}$ serves the purpose of avoiding confusion due to this possibility. Hence we shall retain it wherever necessary.
\end{remark}

\begin{remark}\label{Rem: From M to N}	
Definition \ref{Def: Regular reflected branch chart} and the following Lemma \ref{Lem: Existence of regular reflected branch charts} have been formulated on $\mathring{\mcM}\times\mathring{\mcM}$. Yet, for later purposes, we will be interested in considering only $\mathcal{N}$, a causally convex Cauchy neighbourhood as per Definition \ref{Def: Cauchy neighbourhood}. The above statements can be adapted accordingly and, in this case, we shall replace $\Gamma_N^{\mathrm{reg}}$ with  $\Gamma^{\mathrm{reg}}_{N,\mathcal{N}}\doteq \Gamma^{\mathrm{reg}}_N\cap\Phi^{-1}_N[\mathring{\mathcal{N}}\times\mathring{\mathcal{N}}]$, that is, we want to consider $N$-reflected broken geodesics starting and ending in $\mathring{\mathcal{N}}$.
\end{remark}

\begin{definition}\label{Def: Reflected Synge's world function}
	Let $N\geq 1$ and let $(\mathcal{V}_{N,\alpha},\Omega_{N,\alpha})$ be a regular $N$-reflected branch chart as per Definition \ref{Def: Regular reflected branch chart}. We call {\bf branchwise $N$-reflected Synge's world function} associated with $(N,\alpha)$ the map $\sigma_{-,N,\alpha}:\Omega_{N,\alpha}\to\bR$ such that
	\begin{equation}\label{Eq: Reflected Synge world function}
		\sigma_{-,N,\alpha}(x,y)\doteq\frac{1}{2}\int_0^1 g\left(\dot{\gamma}_{N,\alpha}^{(x,y)}(s),\dot{\gamma}_{N,\alpha}^{(x,y)}(s)\right)\,ds,
	\end{equation}
	where $\gamma_{N,\alpha}^{(x,y)}$ is the broken geodesic associated with the branch $(N,\alpha)$ as per Equation \eqref{Eq: GammaN}.
\end{definition}

\begin{remark}\label{Rem: Reversed reflected branch}
	For $(x,v)\in\Gamma_N$, let
	\begin{equation}\label{Eq: rN map}
	\mathfrak{r}_N:\Gamma_N\to\Gamma_N\quad|\quad(x,v)\mapsto\mathfrak{r}_N(x,v)\doteq\left(\mathrm{Exp}_{-,N}(x,v),-\dot{\gamma}_{(x,v)}(1)\right).
	\end{equation}
	This is the initial datum generating the reversed broken geodesic $s\mapsto\gamma_{(x,v)}(1-s)$ and, in addition, $\mathfrak{r}^2_N=\mathrm{id}|_{\Gamma_N}$. Furthermore,
	$$\Phi_N\circ\mathfrak{r}_N=\mathfrak{s}\circ\Phi_N,$$
	where $\mathfrak{s}(x,y)\doteq(y,x)$. Hence every regular branch chart $(\mathcal{V}_{N,\alpha},\Omega_{N,\alpha})$ determines a reversed regular branch chart, denoted by $(\mathcal{V}_{N,\alpha^{\mathrm{op}}},\Omega_{N,\alpha^{\mathrm{op}}})$, where
	$$\mathcal{V}_{N,\alpha^{\mathrm{op}}}\doteq\mathfrak{r}_N(\mathcal{V}_{N,\alpha}),\qquad \Omega_{N,\alpha^{\mathrm{op}}}\doteq\mathfrak{s}(\Omega_{N,\alpha}).$$
	This is a generalization of the intuitive feature occurring when $N=1$. Suppressing the branch index, if we consider a $1$-reflected broken geodesic $\gamma^{(x,y)}_1$ connecting two points $x,y\in\mathring{\mcM}$, we can construct an inverse one starting from $y$ with initial velocity $-\dot{\gamma}^{(x,y)}_1$ and reaching $x$. 
\end{remark}

\noindent Having established a branchwise notion of reflected Synge world function, we investigate some of its structural properties and the relation with Equation \eqref{Eq: Synge World Function Interior}, limiting the attention only to those aspects which will play a role in our analysis. The proof of the following lemma is variation is spirit of that of \cite[Lem 3.1]{Wunsch}.

\begin{lemma}\label{Lem: Properties of Reflected Geodesic Distance}
Let $N\geq 1$, let $(\mathcal{V}_{N,\alpha},\Omega_{N,\alpha})$ be a regular $N$-reflected branch chart and let $\sigma_{-,N,\alpha}$ be as per Definition \ref{Def: Reflected Synge's world function}. The following properties hold true:
\begin{enumerate}
\item $\sigma_{-,N,\alpha}\in C^\infty(\Omega_{N,\alpha})$.
\item Denoting by $\alpha^{\mathrm{op}}$ the reversed branch as per Remark \ref{Rem: Reversed reflected branch}, it holds that
$$\sigma_{-,N,\alpha}(x,y)=\sigma_{-,N,\alpha^{\mathrm{op}}}(y,x),\qquad (x,y)\in\Omega_{N,\alpha}.$$
\item $\sigma_{-,N,\alpha}$ satisfies an eikonal equation separately in both entries:
\begin{equation}\label{Eq: Eikonal Equation for sigma_-}
g^{\mu\nu}\nabla_\mu\sigma_{-,N,\alpha}\nabla_\nu\sigma_{-,N,\alpha}=2\sigma_{-,N,\alpha}\quad\mathrm{and}\quad g^{\mu^\prime\nu^\prime}\nabla_{\mu^\prime}\sigma_{-,N,\alpha}\nabla_{\nu^\prime}\sigma_{-,N,\alpha}=2\sigma_{-,N,\alpha},
\end{equation}
where unprimed and primed indices refer respectively to derivatives acting on the first or on the second entry.
\end{enumerate}
\end{lemma}

\begin{proof}
We prove the result first for $N=1$ and for a fixed branch label $\alpha$. To simplify the notation, throughout this proof we suppress $\alpha$ and write $v(x,y)\doteq v_{1,\alpha}(x,y)$, $\sigma_-\doteq\sigma_{-,1,\alpha}$ and $\gamma_-^{(x,y)}\doteq\gamma_{1,\alpha}^{(x,y)}$. We parametrize the broken geodesic $\gamma_-^{(x,y)}:[0,1]\to\mcM$ as
\begin{equation}\label{Eq: Broken Geodesic}
	\gamma_{-}^{(x,y)}(\tau)=\left\{\begin{array}{ll}
		\gamma_{(x,v(x,y))}(\tau) & \tau\in [0,\tau_\partial]\\
		\gamma_{(x_\partial,\mathcal{R}(v^{(x,y)}_\partial))}(\tau-\tau_\partial) & \tau\in [\tau_\partial,1]
	\end{array}
	\right. ,
\end{equation}
where $\tau_\partial\in (0,1)$ is such that $x_\partial \doteq \gamma_{(x,v(x,y))}(\tau_\partial)\in\partial\mcM$, while $\dot{\gamma}_{(x,v(x,y)}(\tau_\partial)\doteq v^{(x,y)}_\partial$. The map $\mathcal{R}$ is as per Equation \eqref{Eq: Reflection Vector}.

\vskip .2cm
	
{\em 1.} On account of Definition \ref{Def: Regular reflected branch chart}, the map $\Omega_{1,\alpha}\ni(x,y)\mapsto v(x,y)$ is smooth, where $v(x,y)\in T_x\mcM$ is such that $\mathrm{Exp}_{-}(x, v(x,y))\equiv\mathrm{Exp}_{-, 1}(x,v(x,y))=y$. The reflected broken geodesic $\gamma_-^{(x,y)}$ in Equation \eqref{Eq: Broken Geodesic} is constituted by two bits, each of which is a standard geodesic, hence of constant length. Focusing on the first one, this entails that, for $s\in [0,\tau_\partial)$,
\begin{flalign*}
    g(\dot{\gamma}^{(x,y)}_-(s), \dot{\gamma}^{(x,y)}_-(s)) &=g_x(v(x,y),v(x,y)) \\ &\Longrightarrow \frac{1}{2}\int_0^{\tau_\partial} g(\dot{\gamma}^{(x,y)}_-(s), \dot{\gamma}^{(x,y)}_-(s))\, ds=\frac{\tau_\partial}{2}g_x(v(x,y),v(x,y)).
\end{flalign*}
Similarly the second bit of $\gamma_-^{(x,y)}$ yields for $s\in [\tau_\partial,1]$
\begin{flalign*}
	g(\dot{\gamma}^{(x,y)}_-(s), \dot{\gamma}^{(x,y)}_-(s))&=g_{x_\partial}(\mathcal{R}(v_\partial^{(x,y)}),\mathcal{R}(v_\partial^{(x,y)})) \\ &\Longrightarrow \frac{1}{2}\int_{\tau_\partial}^1 g(\dot{\gamma}^{(x,y)}_-(s), \dot{\gamma}^{(x,y)}_-(s))\, ds=\frac{1-\tau_\partial}{2}g_{x_\partial}(\mathcal{R}(v_\partial^{(x,y)}),\mathcal{R}(v_\partial^{(x,y)})).
\end{flalign*}
A direct calculation shows that the reflection map in Equation \eqref{Eq: Reflection Vector} is an isometry, \textit{i.e.},  
\begin{equation}\label{Eq: Reflection Isometry}
	g_{x_\partial}(\mathcal{R}(v_\partial^{(x,y)}),\mathcal{R}(v_\partial^{(x,y)}))=g_{x_\partial}(v_\partial^{(x,y)},v_\partial^{(x,y)}).
\end{equation}
Yet, constancy of the geodesic length entails that 
\begin{gather}
	g_{x_\partial}(v_\partial^{(x,y)},v_\partial^{(x,y)})=g_x(v(x,y),v(x,y))\Longrightarrow\notag\\
	\sigma_-(x,y)=\frac{1}{2}\int_0^1 g(\dot{\gamma}^{(x,y)}_-(s), \dot{\gamma}^{(x,y)}_-(s))\, ds=\frac{1}{2}g_x(v(x,y),v(x,y)).\label{Eq: Geodesic identity}
\end{gather}
Yet, since the metric smoothly depends on the base point $x$ and the assignment $(x,y)\mapsto v(x,y)$ is smooth, so must be $\sigma_-$. 

\vskip .2cm

{\em 2.} Given the reflected broken geodesic $\gamma_-^{(x,y)}(\tau)$, $\tau\in[0,1]$, the curve $\gamma_-^{(x,y)}(1-\tau)$ belongs to the reversed branch $\alpha^{\mathrm{op}}$ and connects $y$ to $x$. Replacing it in Equation \eqref{Eq: Reflected Synge world function} yields
$$\sigma_{-,1,\alpha}(x,y)=\sigma_{-,1,\alpha^{\mathrm{op}}}(y,x),$$
where we reinstate the subscripts for clarity.

\vskip .2cm

{\em 3.} By the previous item, the eikonal equation in the second entry follows from the corresponding equation in the first entry applied to the reversed branch. Therefore, it suffices to prove the statement for derivatives acting on the first entry. Let $(x,y)\in\Omega_{1,\alpha}$ and keep $y$ fixed. We let $x$ vary along a smooth curve $\lambda:I\to\mathring{\mcM}$ such that $I\subseteq\bR$ contains the origin, $\lambda(0)=x$ and $(\lambda(\varepsilon),y)\in\Omega_{1,\alpha}$ for every $\varepsilon\in I$. It descends from Equation \eqref{Eq: Reflected Synge world function} combined with Equation \eqref{Eq: Broken Geodesic} that we need to consider the contribution to $\sigma_-$ coming from both integrals. Focusing on the first one, we need to consider
$$I_1(\varepsilon)\doteq\int_0^{\tau_{\partial,\varepsilon}}g(\dot{\gamma}_{(\lambda(\varepsilon),v(\lambda(\varepsilon),y))}(s),\dot{\gamma}_{(\lambda(\varepsilon),v(\lambda(\varepsilon),y))}(s))\,d s.$$
The variation with respect to $\varepsilon$ evaluated at $\epsilon = 0$ yields two terms: Denoting by $$\Delta^+(s)\doteq\left.\frac{\partial\gamma_{(\lambda(\varepsilon),v(\lambda(\varepsilon),y))}(s)}{\partial\varepsilon}\right|_{\varepsilon=0} \quad \text{and} \quad \dot{\tau}_\partial=\left.\frac{\partial\tau_{\partial,\varepsilon}}{\partial\varepsilon}\right|_{\varepsilon=0},$$ we obtain
$$\left.\frac{dI_1}{d\varepsilon}\right|_{\varepsilon=0}=2\int_0^{\tau_\partial}g\left(\frac{D\Delta^+(s)}{ds},\dot{\gamma}_{(x,v(x,y))}(s)\right)\,ds+g_{x_\partial}(v_\partial^{(x,y)},v_\partial^{(x,y)})\dot{\tau}_\partial,$$
where $v_\partial\doteq\dot{\gamma}_{(x,v(x,y))}(\tau_\partial)$ while $x_\partial$ is the point hit at $\partial\mcM$ by the geodesic. Since the velocity of an affinely parametrized geodesic is covariantly constant, using the Leibniz rule, metric compatibility and the geodesic equation, we obtain
\begin{gather*}
g\left(\frac{d\Delta^+(s)}{ds},\dot{\gamma}_{(x,v(x,y))}(s)\right)+g\left(\Delta^+(s),\frac{d\dot{\gamma}_{(x,v(x,y))(s)}}{ds}\right)=\\
g\left(\frac{d\Delta^+(s)}{ds},\dot{\gamma}_{(x,v(x,y))}(s)\right)=\frac{d}{ds}g\left(\Delta^+(s),\dot{\gamma}_{(x,v(x,y))}(s)\right).
\end{gather*}
Consequently, 
\begin{align*}
	\left.\frac{dI_1}{d\varepsilon}\right|_{\varepsilon=0}=&2g\left(\Delta^+(\tau_\partial),\dot{\gamma}_{(x,v(x,y))}(\tau_\partial)\right)-2g\left(\Delta^+(0),\dot{\gamma}_{(x,v(x,y))}(0)\right)+g_{x_\partial}(v_\partial^{(x,y)},v_\partial^{(x,y)})\dot{\tau}_\partial\\
	&=2g\left(\Delta^+(\tau_\partial),v_\partial^{(x,y)}\right)-2g\left(\Delta^+(0),v(x,y)\right)+g_{x_\partial}(v_\partial^{(x,y)},v_\partial^{(x,y)})\dot{\tau}_\partial.
\end{align*}
Observe that, in the equalities above, we have used the dominated convergence theorem to justify passing the limit $\varepsilon\to 0$ under the integral sign. Consider now
$$I_2(\varepsilon)\doteq\int_0^{\tau^\prime_{\partial,\varepsilon}}g\left(\dot{\gamma}_{(\lambda(\varepsilon)_\partial,\mathcal{R}(v^{(\lambda(\varepsilon),y)}_\partial))}(s^\prime),\dot{\gamma}_{(\lambda(\varepsilon)_\partial,\mathcal{R}(v^{(\lambda(\varepsilon),y)}_\partial))}(s^\prime)\right)\,d s^\prime.$$	
where we have replaced $s$ by $s^\prime=s-\tau_{\partial,\varepsilon}$ and we have set $\tau^\prime_{\partial,\varepsilon}=1-\tau_{\partial,\varepsilon}$. In addition $\lambda(\varepsilon)_\partial$ is the starting point of the $\varepsilon$-dependent family of geodesics at $\partial\mcM$. Denoting by $$\Delta^-(s^\prime)\doteq\left.\frac{\partial\gamma_{(\lambda(\varepsilon)_\partial,\mathcal{R}(v^{(\lambda(\varepsilon),y)}_\partial))}(s^\prime)}{\partial\varepsilon}\right|_{\varepsilon=0},$$ and repeating the same procedure as for $I_1$ we end up with
\begin{align*}
	\left.\frac{dI_2}{d\varepsilon}\right|_{\varepsilon=0}=&2g\left(\Delta^-(\tau^\prime_\partial),\dot{\gamma}_{(x_\partial,\mathcal{R}(v^{(x,y)}_\partial))}(\tau^\prime_\partial)\right)-2g\left(\Delta^-(0),\dot{\gamma}_{(x_\partial,\mathcal{R}(v^{(x,y)}_\partial))}(0)\right)-g_{x_\partial}(v_\partial^{(x,y)},v_\partial^{(x,y)})\dot{\tau}_\partial\\
	&=2g\left(\Delta^-(\tau^\prime_\partial),\dot{\gamma}_{(x_\partial,\mathcal{R}(v^{(x,y)}_\partial))}(\tau^\prime_\partial)\right)-2g\left(\Delta^-(0),\mathcal{R}(v^{(x,y)}_\partial)\right)-g_{x_\partial}(v_\partial^{(x,y)},v_\partial^{(x,y)})\dot{\tau}_\partial.
\end{align*}
To derive the last term on the right-hand side we used implicitly that geodesics are of constant length, Equation \eqref{Eq: Reflection Isometry} and the identity $\frac{\partial\tau^\prime_{\partial,\varepsilon}}{\partial\epsilon}=-\frac{\partial\tau_{\partial,\varepsilon}}{\partial\varepsilon}$. Putting together these variational identities we obtain
\begin{align}
	\frac{1}{2}\left.\frac{d(I_1+I_2)}{d\varepsilon}\right|_{\varepsilon=0}=&g\left(\Delta^+(\tau_\partial),v_\partial^{(x,y)}\right)-g\left(\Delta^+(0),v(x,y)\right)+\notag\\
	&g\left(\Delta^-(\tau^\prime_\partial),\dot{\gamma}_{(x_\partial,\mathcal{R}(v^{(x,y)}_\partial))}\right)-g\left(\Delta^-(0),\mathcal{R}(v^{(x,y)}_\partial)\right).\label{Eq: Aux1}
\end{align}
Consider now the family $\lambda(\varepsilon)_\partial$ which parametrizes where each geodesic $\lambda(\varepsilon)$ intersects $\partial\mcM$ and let us denote by $B\doteq \left.\frac{\partial\lambda(\varepsilon)_\partial}{\partial\varepsilon}\right|_{\varepsilon=0}\in T_{x_\partial}\mcM$. Per construction $B$ is tangent to the boundary, that is $g_{x_\partial}(B,n)=0$. At the same time, since $\lambda(\varepsilon)_\partial=\gamma_{(\lambda(\varepsilon),v(\lambda(\varepsilon),y))}(\tau_{\partial,\varepsilon})$, the chain rule entails
$$B=\left.\frac{\partial\lambda(\varepsilon)_\partial}{\partial\varepsilon}\right|_{\varepsilon=0}=\Delta^+(\tau_\partial)+v_\partial^{(x,y)}\dot{\tau}_\partial.$$ 
As a consequence of $\gamma_{(x,v(x,y))}(\tau_\partial)=\gamma_{(x_\partial,\mathcal{R}(v^{(x,y)}_\partial))}(0)$ in Equation \eqref{Eq: Broken Geodesic}, it holds that $\Delta^-(0)=B$. Furthermore, since the endpoint $y$ is fixed, it turns out that, if we take the variation with respect to $\varepsilon$ of
$$y=\gamma_{(\lambda(\varepsilon)_\partial,\mathcal{R}(v^{(\lambda(\varepsilon),y)}_\partial))}(\tau^\prime_{\partial,\epsilon}),$$
applying the chain rule, we end up with
$$0=\Delta^-(\tau^\prime_\partial)+\dot{\eta}(\tau^\prime_\partial)\dot{\tau}^\prime_\partial=\Delta^-(\tau^\prime_\partial)-\dot{\eta}(\tau^\prime_\partial)\dot{\tau}_\partial.$$
Here we have already evaluated at $\varepsilon=0$ and we have set $\eta(s^\prime)\doteq\gamma_{(x_\partial,\mathcal{R}(v_\partial^{(x,y)}))}(s^\prime)$. Inserting these identities in Equation \eqref{Eq: Aux1} we obtain
\begin{gather*}
\frac{1}{2}\left.\frac{d(I_1+I_2)}{d\varepsilon}\right|_{\varepsilon=0}=\\
g\left(B-v_\partial^{(x,y)}\dot{\tau}_\partial,v_\partial^{(x,y)}\right)-g\left(\Delta^+(0),v(x,y)\right)+g\left(\dot{\eta}(\tau^\prime_\partial)\dot{\tau}_\partial,\dot{\eta}(\tau^\prime_\partial)\right)-g\left(B,\mathcal{R}(v^{(x,y)}_\partial)\right).
\end{gather*}
Observe that $g(B,v_\partial^{(x,y)}-\mathcal{R}(v^{(x,y)}_\partial))=0$ since the difference $v_\partial^{(x,y)}-\mathcal{R}(v^{(x,y)}_\partial)$ is proportional to $n$, the vector normal to $\partial\mcM$, while $B$ is tangent to the boundary. In addition, $$g\left(\dot{\eta}(\tau^\prime_\partial)\dot{\tau}_\partial,\dot{\eta}(\tau^\prime_\partial)\right)-g\left(v_\partial^{(x,y)}\dot{\tau}_\partial,v_\partial^{(x,y)}\right)=\dot{\tau}_\partial\left(g\left(\dot{\eta}(\tau^\prime_\partial),\dot{\eta}(\tau^\prime_\partial)\right)-g\left(v_\partial^{(x,y)},v_\partial^{(x,y)}\right)\right)=0,$$
since we use once more that geodesics have tangent vectors of constant length. To summarize 
$$\left.\frac{d\sigma_-(\lambda(\varepsilon),y)}{d\varepsilon}\right|_{\varepsilon=0} = \frac{1}{2}\left.\frac{d(I_1+I_2)}{d\varepsilon}\right|_{\varepsilon=0}=-g\left(\Delta^+(0),v(x,y)\right),$$
or equivalently $d_x\sigma_-(\Delta^+(0))=-g\left(\Delta^+(0),v(x,y)\right)$. Since $\Delta^+(0)$ is arbitrary, we can infer that, at the level of components, this identity entails
\begin{equation}\label{Eq: Derivative of Sigma-}
	\nabla_\mu\sigma_-(x,y)=-g_{\mu\nu}v^\nu(x,y).
\end{equation}
In turn this yields together with Equation \eqref{Eq: Geodesic identity}
$$g^{\mu\nu}(x)(\nabla_\mu\sigma_-)(x,y)(\nabla_\nu\sigma_-)(x,y)=g_x(v(x,y),v(x,y))=2\sigma_-.$$

The generalization to an arbitrary branch $(N,\alpha)$ follows inductively on $N$. Since every geodesic segment is affinely parametrized and the reflection map in Equation \eqref{Eq: Reflection Vector} is an isometry, the squared norm of the tangent vector is constant along the whole broken geodesic. Consequently,
\begin{equation}\label{Eq: Branchwise reflected action}
\sigma_{-,N,\alpha}(x,y)=\frac{1}{2}g_x\left(v_{N,\alpha}(x,y),v_{N,\alpha}(x,y)\right).
\end{equation}
This entails smoothness. For the variational identity discussed above, one can decompose the action into the $N+1$ geodesic segments as in Equation \eqref{Eq: gamma map}. At each reflection point the endpoint contribution of the incoming segment cancels with that of the outgoing counterpart because
$$\dot{\gamma}_{i+1}(\tau_i)=\mathcal{R}_{x_i}(v_i).$$
Recall that the variation of the reflection point is tangent to $\partial\mcM$. Repeating the one-reflection cancellation at $x_0,\ldots,x_{N-1}$ leaves only the two endpoint terms. This proves the first-entry eikonal equation, while the second-entry identity follows by applying the same argument to the reversed branch $\alpha^{\mathrm{op}}$.
\end{proof}

\noindent We associate to every reflection point of a fixed regular branch two partial Synge world functions obtained by restricting the broken geodesic immediately before and after the reflection. Let $(\mathcal{V}_{N,\alpha},\Omega_{N,\alpha})$ be a regular $N$-reflected branch chart, let $(x,y)\in\Omega_{N,\alpha}$ and write
$$
\gamma\doteq\gamma_{N,\alpha}^{(x,y)}.
$$
We denote by $\tau_0<\cdots<\tau_{N-1}$ and $x_0,\ldots,x_{N-1}\in\partial\mcM$ its reflection times and reflection points, respectively. For every $j=0,\ldots,N-1$, choose $s<\tau_j$ on the geodesic segment immediately preceding $x_j$ and set $w\doteq\gamma(s)$. The restriction of $\gamma$ from $w$ to $y$, reparametrized on $[0,1]$ as
$$
\widehat{\gamma}_s(u)\doteq\gamma\bigl(s+(1-s)u\bigr),
\qquad u\in[0,1],
$$
is generated by the initial datum
$$
\bigl(w,(1-s)\dot\gamma(s)\bigr)\in\Gamma_{N-j}.
$$
Likewise, if $s^\prime>\tau_j$ lies on the segment immediately following $x_j$ and $w^\prime\doteq\gamma(s^\prime)$, the restriction from $w^\prime$ to $y$ is generated by
$$
\bigl(w^\prime,(1-s^\prime)\dot\gamma(s^\prime)\bigr)
\in\Gamma_{N-j-1},
$$
for $j=0,\ldots,N-2$, whereas for $j=N-1$ it is the ordinary final geodesic segment. Whenever the reflected truncated data are regular, choose regular branch charts containing them and denote the corresponding branch labels by $\alpha_j^{\rm in}$ and $\alpha_j^{\rm out}$. Using Definition \ref{Def: Reflected Synge's world function}, we set
\begin{subequations}\label{Eq: Incoming outgoing reflected Synge functions}
	\begin{equation}\label{Eq: sigmajin}
		\sigma_{j,N,\alpha}^{\rm in}(w,y)\doteq\sigma_{-,N-j,\alpha_j^{\rm in}}(w,y),
	\end{equation}
	\begin{equation}\label{Eq: sigmajout}
		\sigma_{j,N,\alpha}^{\rm out}(w^\prime,y)\doteq
		\begin{cases}
			\sigma_{-,N-j-1,\alpha_j^{\rm out}}(w^\prime,y),&j=0,\ldots,N-2,\\
			\sigma(w^\prime,y),&j=N-1.
		\end{cases}
	\end{equation}
\end{subequations}
We call these respectively the {\em incoming} and {\em outgoing branchwise Synge functions} at the $j$-th reflection point. In the following we will be interested in extending the reflected Synge world functions, including their branch versions in Equations \eqref{Eq: sigmajin} and \eqref{Eq: sigmajout}, so that the points in the first entry are allowed also to lie at $\partial\mcM$. We shall employ the adjective {\em one-sided extension} to denote that the first entry can approach a boundary point $x_j$ along the incoming or outgoing side of the reflected branch. To avoid burdening the notation we will still employ the same symbols $\sigma_{j,N,\alpha}^{\rm in/out}$ and smoothness will be understood in the sense of $\mcM\times\mathring{\mcM}$.

\begin{lemma}\label{Lem: Matching N-reflected Synge world functions}
	Let $(\mathcal{V}_{N,\alpha},\Omega_{N,\alpha})$, $\alpha \in \mathbb{N}$, be a regular $N$-reflected branch chart and let $(x,y)\in\Omega_{N,\alpha}$. Suppose that, for every $j=0,\ldots,N-1$, the branchwise Synge functions
	$\sigma^{\rm in}_{j,N,\alpha}$ and $\sigma^{\rm out}_{j,N,\alpha}$ in Equations \eqref{Eq: sigmajin} and \eqref{Eq: sigmajout} are defined and admit smooth one-sided extensions to a common neighbourhood of $(x_j,y)$ in $\mcM\times\mathring{\mcM}$, denoted by the same symbols. Let $\{x^a\}_{a=0,\ldots,d-2}$ be local coordinates on $\partial\mcM$ such that the vector fields $\partial_a$ are tangent to the boundary, and let $\widetilde n\in\Gamma(T\mcM)$ be any smooth extension of the inward pointing unit normal vector field $n$ to a neighbourhood of the reflection points. For every $j=0,\ldots,N-1$, the following identities hold at $(x_j,y)$:
	\begin{enumerate}
		\item The incoming and outgoing functions coincide:
		$$\sigma_{j,N,\alpha}^{\rm in}(x_j,y)=\sigma_{j,N,\alpha}^{\rm out}(x_j,y).$$
		\item Their tangential derivatives coincide, namely,
		$$(\nabla_a\sigma_{j,N,\alpha}^{\rm  in})(x_j,y)=(\nabla_a \sigma_{j,N,\alpha}^{\rm out})(x_j,y),\qquad a=0,\ldots,d-2,$$
		while their normal derivatives have opposite sign:
		$$(\nabla_n\sigma_{j,N,\alpha}^{\rm in})(x_j,y)=-(\nabla_n\sigma_{j,N,\alpha}^{\rm out})(x_j,y),$$
		where
		$$(\nabla_n\sigma_{j,N,\alpha}^{\rm in/out})(x_j,y)\doteq\lim_{w\to x_j}(\nabla_{\widetilde n}\sigma_{j,N,\alpha}^{\rm in/out})(w,y).$$
		\item The following additional identities hold true:
		\begin{subequations}
			\begin{equation}
				\label{Eq: N-reflected_2nd_order_1}
				(\nabla_a \nabla_b \sigma_{j,N,\alpha}^{\rm in})(x_j,y)=(\nabla_a \nabla_b \sigma_{j,N,\alpha}^{\rm out})(x_j,y)+2K_{a b}(x_j)(\nabla_n\sigma_{j,N,\alpha}^{\rm out})(x_j,y),
			\end{equation}
			\begin{equation}\label{Eq: N-reflected_2nd_order_2}
				(\nabla_a \nabla_n\sigma_{j,N,\alpha}^{\rm in})(x_j,y)=-(\nabla_a \nabla_n\sigma_{j,N,\alpha}^{\rm out})(x_j,y)+2K_a{}^b(x_j)(\nabla_b\sigma_{j,N,\alpha}^{\rm out})(x_j,y),
			\end{equation}
			\begin{equation}
				\label{Eq: N-reflected_2nd_order_3}
				(\nabla_n\nabla_n\sigma_{j,N,\alpha}^{\rm in})(x_j,y)=(\nabla_n\nabla_n\sigma_{j,N,\alpha}^{\rm out})(x_j,y)+2K^{a b}(x_j)\frac{(\nabla_a \sigma_{j,N,\alpha}^{\rm out})(x_j,y)		(\nabla_b\sigma_{j,N,\alpha}^{\rm out})(x_j,y)}{(\nabla_n\sigma_{j,N,\alpha}^{\rm out})(x_j,y)}.
			\end{equation}
		\end{subequations}
		Here $K$ denotes the second fundamental form of $\partial\mcM$. 
	\end{enumerate}
	In particular, for $j=N-1$ one has $\sigma_{N-1,N,\alpha}^{\rm out}=\sigma$ and the identities above compare the function associated with the last reflected branch directly with the ordinary Synge world function as in Equation \eqref{Eq: Synge World Function Interior}.
\end{lemma}

\begin{proof}
As in Lemma \ref{Lem: Properties of Reflected Geodesic Distance}, we prove the result for $N=1$ and for a fixed branch $\alpha$, suppressing the branch label in the proof. In this case $j=0$ and

$$\sigma_{0,1,\alpha}^{\rm in}=\sigma_{-,1,\alpha}\equiv\sigma_-,\qquad \sigma_{0,1,\alpha}^{\rm out}=\sigma.$$
	
\vskip .2cm

{\em 1.} This is a direct consequence of the fact that extending Equation \eqref{Eq: Broken Geodesic} when $x=x_0\in\partial\mcM$ is tantamount to considering only $\gamma_{-}^{(x_0,y)}(\tau)=\gamma_{(x_0,v(x_0,y))}(\tau)$ with $\tau\in [0,1]$. Hence the reflected geodesic comprises only the ordinary geodesic connecting $x_0$ to $y$. This entails automatically that $\sigma(x_0,y)=\sigma_-(x_0,y)$.

\vskip .2cm

{\em 2.} Let us consider the function $F(x,y)=\sigma(x,y)-\sigma_-(x,y)$, which, on account of the previous item vanishes at $x=x_0$. This entails that $\nabla_a F(x_0,y)=\partial_a F(x_0,y)=0$ for all $a=0,\dots,d-2$ since $\partial_a$ is tangent to $\partial\mcM$. This proves that $(\nabla_a\sigma)(x_0,y)=(\nabla_a\sigma_-)(x_0,y)$. Focusing on the normal derivatives we observe that Equation \eqref{Eq: Derivative of Sigma-} entails $\nabla_{\tilde{n}}\sigma_-(x,y)=-g_{x}(\tilde{n},v(x,y))$ where $\tilde{n}\in\Gamma(T\mcM)$ is as per hypothesis. Consequently $\nabla_n\sigma_-(x_0,y)=-\lim\limits_{x\to x_0}g_{x}(\tilde{n},v(x,y))$. Let us denote now by $\mathfrak w$ the velocity vector of the geodesic connecting the initial point $x$ to $y$. Using the identity $\nabla_\mu\sigma(x,y)=-g_{\mu\nu}(x){\mathfrak w}^{\nu}$, see \cite{Poisson_2011}, it holds that $\nabla_{\tilde{n}}\sigma(x,y)=-g_x(n,{\mathfrak w})$. Taking the limit as $x\to x_0$, we recall that $v=\mathcal{R}(\mathfrak w)$ where $\mathcal{R}$ is defined in Equation \eqref{Eq: Reflection Isometry}. This reflection map is constructed in such a way to invert the normal component of the velocity vector, namely $g_{x_0}(n,v)=-g_{x_0}(n,\mathfrak w)$ from which the sought conclusion descends.

\vskip .2cm

{\em 3.} Using once more the shorthand notation $F(x,y)=\sigma(x,y)-\sigma_-(x,y)$, since $F$ vanishes at $x=x_0\in\partial\mcM$, it holds that, for all $a, b=0,\dots,d-2$ such that $\partial_a,\partial_b$ is tangent to $\partial\mcM$, then $(\partial_a\partial_b F)(x_0,y)=0$ or, equivalently 
$$(\nabla_a\nabla_b F)(x_0,y)=\Gamma^\lambda_{ab}(x_0)(\nabla_\lambda F)(x_0,y)=\Gamma^n_{ab}(x_0)(\nabla_n F)(x_0,y),$$
where only the component normal to the boundary can contribute in view of item {\em 2.} Using the standard identity in Riemannian and Lorentzian geometry $\Gamma^n_{ab}=K_{ab}$, see, {\it e.g.}, \cite[Lemma 8.3]{Lee_2018}, together with $(\nabla_n F)(x_0,y)=2(\nabla_n\sigma)(x_0,y)$, also consequence of item {\em 2.}, Equation \eqref{Eq: N-reflected_2nd_order_1} descends. 

In order to prove the second identity, we introduce the function $$G(x,y)=(\nabla_{\tilde n}\sigma)(x,y)+(\nabla_{\tilde n} \sigma_-)(x,y).$$ In view of item {\em 2.} it holds that $G(x_0,y)=0$. Hence, for all $a=0,\dots,d-2$ such that $\partial_a$ is tangent to $\partial\mcM$, we have that $(\partial_a G)(x_0,y)=0$. This entails that
$$(\nabla_a G)(x_0,y)=-\Gamma_{a n}^b((\nabla_b\sigma)(x_0,y)+(\nabla_b\sigma_-)(x_0,y)),$$
where only tangential derivatives contribute to the right hand-side. Using that $\Gamma_{a n}^b=-K_a^b$ and item {\em 2.} once more, it descends $(\nabla_a G)(x_0,y)=-2K_a^b(x_0)(\nabla_b\sigma)(x_0,y)$, which is Equation \eqref{Eq: N-reflected_2nd_order_2}. We are left with the last point. Both $\sigma$ and $\sigma_-$ satisfy the eikonal equation in the first entry as proven in item {\em 2.}
Taking a covariant derivative in the normal direction and then evaluating at $(x_0,y)$ yields
$$(\nabla^\mu\sigma)(x_0,y)(\nabla_\mu\nabla_n\sigma)(x_0,y)=(\nabla_n\sigma)(x_0,y),$$
and
$$(\nabla^\mu\sigma_-)(x_0,y)(\nabla_\mu\nabla_n\sigma_-)(x_0,y)=(\nabla_n\sigma_-)(x_0,y).$$
The latter becomes
$$(\nabla^a\sigma)(x_0,y)(\nabla_a\nabla_n\sigma_-)(x_0,y)-(\nabla_n\sigma)(x_0,y)(\nabla_n\nabla_n\sigma_-)(x_0,y)=-(\nabla_n\sigma)(x_0,y),$$
where we recall that Latin indices denote directions tangential to $\partial\mcM$ and where we employed the identities proven in item {\em 2.}
Using Equation \eqref{Eq: N-reflected_2nd_order_2} we obtain
\begin{gather*}
-(\nabla^a\sigma)(x_0,y)(\nabla_a\nabla_n\sigma)(x_0,y)-2(\nabla^a\sigma)(x_0,y)K_a{}^c(x_0)(\nabla_c\sigma)(x_0,y)\\
-(\nabla_n\sigma)(x_0,y)(\nabla_n\nabla_n\sigma_-)(x_0,y)=-(\nabla_n\sigma)(x_0,y).
\end{gather*}
To conclude it suffices to observe that differentiating Equation \eqref{Eq: Eikonal Equation for Sigma} gives
$$(\nabla^a\sigma)(x_0,y)(\nabla_a\nabla_n\sigma)(x_0,y)+(\nabla_n\sigma)(x_0,y)(\nabla_n\nabla_n\sigma)(x_0,y)=(\nabla_n\sigma)(x_0,y).$$
We remark that Equation \eqref{Eq: N-reflected_2nd_order_3} rests on the no-glancing property as per Proposition \ref{Prop: no-glancing} which guarantees that $(\nabla_n\sigma)(x_0,y)\neq 0$. 

For an arbitrary branch $(N,\alpha)$, the statement follows by backward induction on $j$. Here backward refers to the fact that initial step of the induction procedure refers to $j=N-1$, for which $\sigma_{N-1,N,\alpha}^{\rm out}=\sigma$. In this case the one-reflection argument above applies verbatim. 
\end{proof}

\begin{remark}
We stress that Lemmas \ref{Lem: Properties of Reflected Geodesic Distance} and \ref{Lem: Matching N-reflected Synge world functions} generalize the same relations discussed in \cite{Costeri_25} for half-Minkowski spacetime, though the form of Equations \eqref{Eq: N-reflected_2nd_order_1}, \eqref{Eq: N-reflected_2nd_order_2} and \eqref{Eq: N-reflected_2nd_order_3} is simpler in this reference since the boundary is totally geodesic, hence $K_{a b}=0$.
\end{remark}

\noindent Having established the structural properties of the Synge world function on every reflected branch, in the following we spell out the geometric assumptions which will allow us to establish in Section \ref{Sec: Comparison of Hadamard states} the equivalence between the microlocal characterization of Hadamard states and their local counterpart. We stress that they represent sufficient conditions and they could be dropped if one wishes only to focus on the microlocal structure. We will comment more on this point in Section \ref{Sec: Wavefront Set of the Propagators}.

\begin{assumption}\label{Ass: Finite regular reflections}
	Let $\mathcal{N}$ be a causally convex Cauchy neighbourhood as per Definition \ref{Def: Cauchy neighbourhood}. For every $N\geq 1$, set
	\begin{equation}\label{Eq: Gamma0N}
	\Gamma^{0}_{N,\mathcal{N}}\doteq\left\{(x,v)\in\Gamma_N\; \vert \;\Phi_N(x,v)\in\mathring{\mathcal{N}}\times\mathring{\mathcal{N}},\ \ g_x(v,v)=0\right\},
\end{equation}
	and define
	$$\Gamma^{0}_{\mathcal{N}}\doteq\bigsqcup_{N\geq 1}\Gamma^{0}_{N,\mathcal{N}}.$$
We assume that the following conditions hold:
	\begin{enumerate}
		\item Every reflected null datum with endpoints in $\mathring{\mathcal{N}}$ is regular, namely
		$$\Gamma^{0}_{N,\mathcal{N}}\subseteq\Gamma_N^{\rm reg},\qquad N\geq 1,$$
        where $\Gamma^{\mathrm{reg}}_N$ is as per Definition \ref{Def: Regular reflected datum}.
	
		\item Let $(x,v)\in\Gamma^{0}_{N,\mathcal{N}}$ and let $x_0,\ldots,x_{N-1}$ be the reflection points of $\gamma_{(x,v)}$. Setting $y\doteq\mathrm{Exp}_{-,N}(x,v)$, for every $j=0,\ldots,N-1$, the incoming and outgoing restrictions whose first endpoint is sufficiently close to $x_j$ are described by fixed branch labels $\alpha_j^{\rm in}$ and $\alpha_j^{\rm out}$ as in Equations \eqref{Eq: sigmajin} and \eqref{Eq: sigmajout}. The resulting branchwise Synge functions admit smooth one-sided extensions to a common open neighbourhood of $(x_j,y)$ in $\mcM\times\mathring{\mcM}$.
	\end{enumerate}
\end{assumption}

\begin{remark}\label{Rem: Meaning finite regular reflections}
	The following comments on Assumption \ref{Ass: Finite regular reflections} are in due course:
	\begin{itemize}
		\item[\ding{104}] The first condition excludes reflected conjugate points along the lightlike branches. We focus only on these since they are the only curves relevant to the singular structure of the propagators and of Hadamard states, see Section \ref{Sec: Wavefront Set of the Propagators} and \ref{Sec: Hadamard states with Robin boundary conditions}. Observe that this a statement concerning points lying in the interior of the manifold. To make contact with a nomenclature often used on globally hyperbolic spacetimes with empty boundary, in this way we are avoiding caustics in our analysis. 
		\item[\ding{104}] The second condition guarantees instead that the matching relations of Lemma \ref{Lem: Matching N-reflected Synge world functions} hold true at every reflection point. In other words, we are avoiding that $y$ is reflected-conjugate to the boundary point $x_j$ or viceversa if we consider the curve connecting them with the opposite orientation. Furthermore Condition {\em 2.} also excludes a loss of smoothness.
	\end{itemize}
\end{remark}

The next proposition organizes the possibly many reflected null trajectories into a countable family of smooth local branches. This provides the geometric framework needed later to describe the reflected singularities of the Hadamard parametrix branch by branch, hence allowing us to establish a counterpart of Radzikowski's theorem in our setting, see Section \ref{Sec: Comparison of Hadamard states}.

\begin{proposition}\label{Prop: Reflected null branch atlas}
	Let $(\mcM,g)$ be globally hyperbolic spacetime with timelike boundary abiding by Assumptions \ref{Ass: Infinitesimally convex} and \ref{Ass: Finite regular reflections}. Furthermore, consider $\mathcal{N}$, a causally convex Cauchy neighbourhood as per Definition \ref{Def: Cauchy neighbourhood}. For every $N\geq1$, we set
	$$\mathfrak{R}^{0}_{N,\mathcal{N}}\doteq\Phi_N\bigl(\Gamma^{0}_{N,\mathcal{N}}\bigr)\subset\mathring{\mathcal{N}}\times\mathring{\mathcal{N}},$$
	where $\Gamma^0_{N,\mathcal{N}}$ is as per Equation \eqref{Eq: Gamma0N}. Then there exist a countable index set $\mathcal{I}_N$ and a family of regular $N$-reflected branch charts, see Definition \ref{Def: Regular reflected branch chart},
	$$\left\{(\mathcal V_{N,\alpha},\Omega_{N,\alpha})\right\}_{\alpha\in\mathcal{I}_N},$$
	such that 
	\begin{enumerate}
		\item For every $\alpha\in\mathcal{I}_N$,
		$$\mathcal V_{N,\alpha}\subset\Gamma_N^{\rm reg}\cap\Phi_N^{-1}\bigl(\mathring{\mathcal{N}}\times\mathring{\mathcal{N}}\bigr),\quad\mathrm{and}\quad\Gamma^{0}_{N,\mathcal{N}}\subseteq\bigcup_{\alpha\in\mathcal{I}_N}\mathcal V_{N,\alpha}\quad\mathrm{with}\quad\Omega_{N,\alpha}\subset\mathring{\mathcal{N}}\times\mathring{\mathcal{N}}.$$
		\item The family can be chosen closed under reversal, that is, if $\alpha\in\mathcal{I}_N$, then $\alpha^{\mathrm{op}}\in\mathcal{I}_N$. Here the opposite index is associated to the $N$-reflected branch chart
		$$\mathcal V_{N,\alpha^{\rm op}}=\mathfrak r_N(\mathcal V_{N,\alpha}),\quad\Omega_{N,\alpha^{\rm op}}=\mathfrak s(\Omega_{N,\alpha}),$$
		where $\mathfrak r_N$ and $\mathfrak s$ are defined in Remark \ref{Rem: Reversed reflected branch}.		
		\item For every $\alpha\in\mathcal{I}_N$, the set
\begin{equation}\label{Eq: Z}		
\mathcal{Z}_{N,\alpha}\doteq\left\{(x,y)\in\Omega_{N,\alpha}\;|\;\sigma_{-,N,\alpha}(x,y)=0\right\}
\end{equation}
is a smooth hypersurface of $\Omega_{N,\alpha}$ and
$$\mathfrak{R}^{0}_{N,\mathcal{N}}=\bigcup_{\alpha\in\mathcal{I}_N}\mathcal{Z}_{N,\alpha}.$$
	\end{enumerate}
\end{proposition}

\begin{proof}
	Letting $N\geq1$ be arbitrary but fixed, we observe that, on account of condition {\em 1.} in Assumption \ref{Ass: Finite regular reflections}, $\Gamma^0_{N,\mathcal{N}}\subseteq\Gamma_N^{\rm reg}$. Hence we can apply Lemma \ref{Lem: Existence of regular reflected branch charts} to establish the existence of a regular $N$-reflected branch chart around each $(x_0,v_0)\in\Gamma^0_{N,\mathcal{N}}$ and, without loss of generality, we can choose it so that its image under $\Phi_N$ is contained in $\mathring{\mathcal{N}}\times\mathring{\mathcal{N}}$. Using \cite[Prop. A.16]{Lee_2013}, we can extract a countable family of these charts covering $\Gamma^{0}_{N,\mathcal{N}}$, thereby proving item $1.$ Since Equation \eqref{Eq: rN map} entails that $\mathfrak r_N\bigl(\Gamma^0_{N,\mathcal{N}}\bigr)=\Gamma^0_{N,\mathcal{N}}$, given any $\mathcal{O}$ in this family, we can construct $\mathfrak{r}_N[\mathcal{O}]$ adding it to the cover, if not already included. This procedure preserves the property of being countable, ultimately yielding item {\em 2.}
	
	Focusing on the last statement, for $(x,y)\in\Omega_{N,\alpha}$, Equation \eqref{Eq: Branchwise reflected action} entails that $\sigma_{-,N,\alpha}(x,y)=0$ if and only if $v_{N,\alpha}(x,y)$ is lightlike, where $y=\mathrm{Exp}_-(v_{N,\alpha}(x,y))$. Since $\Omega_{N,\alpha}\subset\mathring{\mathcal{N}}\times\mathring{\mathcal{N}}$, this implies $\bigl(x,v_{N,\alpha}(x,y)\bigr)\in\Gamma^{0}_{N,\mathcal{N}}$. We can now prove $	\mathfrak{R}^{0}_{N,\mathcal{N}}=\bigcup_{\alpha\in\mathcal{I}_N}\mathcal Z_{N,\alpha}$. We start from the inclusion $\subseteq$. Given $(x,y)\in\mathfrak R^0_{N,\mathcal{N}}$, by definition, there exists $(x,v)\in\Gamma^0_{N,\mathcal{N}}$ such that $\Phi_N(x,v)=(x,y)$. The covering property in item {\em 1.} entails that there exists $\alpha\in\mathcal{I}_N$ such that $(x,v)\in\mathcal{V}_{N,\alpha}$. 
	Since $\left.\Phi_N\right|_{\mathcal V_{N,\alpha}}\colon\mathcal V_{N,\alpha}\longrightarrow\Omega_{N,\alpha}$ is injective, it follows that $v_{N,\alpha}(x,y)=v$.
	Moreover, $(x,v)\in\Gamma^0_{N,\mathcal{N}}$ entails that $g_x(v,v)=0$. Hence $\sigma_{-,N,\alpha}(x,y)=\frac{1}{2}g_x(v,v)=0$ and therefore $(x,y)\in\mathcal Z_{N,\alpha}$. We conclude that
	$$\mathfrak R^0_{N,\mathcal{N}}\subseteq\bigcup_{\alpha\in\mathcal{I}_N}\mathcal Z_{N,\alpha}.$$
	
	We focus now on the reverse inclusion $\supseteq$. Hence, let $(x,y)\in\mathcal Z_{N,\alpha}$ with $\alpha\in\mathcal{I}_N$. Then $\bigl(x,v_{N,\alpha}(x,y)\bigr)\in\mathcal V_{N,\alpha}$,
	and, by item {\em 1.}, its endpoints lie in $\mathring{\mathcal{N}}$. Furthermore, $\sigma_{-,N,\alpha}(x,y)=\frac{1}{2}g_x\bigl(v_{N,\alpha}(x,y),v_{N,\alpha}(x,y)\bigr)=0$, so that $\bigl(x,v_{N,\alpha}(x,y)\bigr)\in\Gamma^0_{N,\mathcal{N}}$. Since $\Phi_N\bigl(x,v_{N,\alpha}(x,y)\bigr)=(x,y)$, it follows that $(x,y)\in\Phi_N\bigl(\Gamma^0_{N,\mathcal{N}}\bigr)=\mathfrak R^0_{N,\mathcal{N}}$.
	Thus
	$$\bigcup_{\alpha\in\mathcal{I}_N}\mathcal Z_{N,\alpha}\subseteq\mathfrak R^0_{N,\mathcal{N}}.$$
	
	In order to show that $\mathcal Z_{N,\alpha}$ is a smooth hypersurface we employ the same strategy as in item {\em 3.} of the proof of Lemma \ref{Lem: Properties of Reflected Geodesic Distance}. More precisely we let the endpoint of $\sigma_{-,N,\alpha}(x,y)$ vary, though allowing for an arbitrary number of reflections. It entails $d_x\sigma_{-,N,\alpha}(x,y)=-v_{N,\alpha}(x,y)^\flat$, where $\flat$ still denotes the musical isomorphism.
	On $\mathcal Z_{N,\alpha}$ the vector $v_{N,\alpha}(x,y)$, and hence $d_x\sigma_{-,N,\alpha}(x,y)$, cannot vanish, since it generates a broken geodesic with $N\geq1$ reflections. Therefore the regular level-set theorem, see \cite[Cor. 5.14]{Lee_2013} yields the claimed hypersurface property.
\end{proof}

\noindent While Proposition \ref{Prop: Reflected null branch atlas} focuses on pairs of points connected by reflected null rays, at a microlocal level, we need also to codify the corresponding covectors. The next result serves exactly this purpose and it is preparatory to the whole analysis in Section \ref{Sec: Hadamard states with Robin boundary conditions}.
\begin{corollary}\label{Cor: Branchwise conormal resolution}
Under the same assumptions of Proposition \ref{Prop: Reflected null branch atlas}, for every $N\geq 1$ and every $\alpha\in\mathcal{I}_N$, set
\begin{equation}\label{Eq: Branchwise reflected conormal relation}
\Lambda^{\triangleright}_{N,\alpha}\doteq\left\{\left(x,\lambda k_{x;N,\alpha},y,-\lambda k_{y;N,\alpha}			\right)\in T^*\left(\mathring{\mathcal{N}}\times\mathring{\mathcal{N}}\right)\setminus\{0\}\;\middle|\;(x,y)\in\mathcal{Z}_{N,\alpha},\lambda\bigl(t(x)-t(y)\bigr)>0\right\},
\end{equation}
where $k_{x;N,\alpha}\doteq d_x\sigma_{-,N,\alpha}$, while $k_{y;N,\alpha}\doteq	-d_y\sigma_{-,N,\alpha}$, while $t:\mcM\to\bR$ is the underlying time function. Then
$$\Lambda^{\triangleright}_{N,\alpha}\subseteq N^*\mathcal Z_{N,\alpha}\setminus\{\boldsymbol{0}\},$$
where $N^*\mathcal Z_{N,\alpha}$ denotes the conormal bundle of $\mathcal Z_{N,\alpha}$. Furthermore, defining
\begin{equation}\label{Eq: N-reflected cotangent relation}
\begin{aligned}
\Lambda^{\triangleright}_{N,\mathcal{N}}\doteq\Bigl\{\,&(x,k_x,y,-k_y)\in T^*\left(\mathring{\mathcal{N}}\times\mathring{\mathcal{N}}\right)\setminus\{\boldsymbol{0}\}\;\Big|\;\exists\,(x,v)\in\Gamma^0_{N,\mathcal{N}}\text{such that} y=\mathrm{Exp}_{-,N}(x,v),\\			&\exists\,\lambda\in\bR\setminus\{0\}\text{ such that }	k_x=-\lambda v^\flat,\quad k_y=-\lambda\dot\gamma_{(x,v)}(1)^\flat,\quad	k_x\triangleright 0\,\Bigr\},
\end{aligned}
\end{equation}
where $\triangleright$ denotes that the covector is future pointing, it holds that
\begin{equation}\label{Eq: Cotangent branch decomposition}
\Lambda^{\triangleright}_{N,\mathcal{N}}=\bigcup_{\alpha\in\mathcal{I}_N}\Lambda^{\triangleright}_{N,\alpha}.
\end{equation}
\end{corollary}

\begin{proof}
Since $d\sigma_{-,N,\alpha}$ does not vanish on $\mathcal Z_{N,\alpha}$, see the proof of Proposition \ref{Prop: Reflected null branch atlas}, its conormal bundle reads
$$N^*\mathcal Z_{N,\alpha}\setminus\{\boldsymbol{0}\}=\left\{\left(x,\lambda k_{x;N,\alpha}, y, -\lambda k_{y;N,\alpha}\right)\; \vert \; (x,y)\in\mathcal Z_{N,\alpha},\ \lambda\neq 0\right\}.$$
	Observe that the sign condition in Equation \eqref{Eq: Branchwise reflected conormal relation} selects the component for which the covector in the first entry is future pointing since, if $y\in J^+(x)$, then $t(x)-t(y)<0$ and $k_{x;N,\alpha}$ is past pointing, which entails $\lambda<0$. If instead $x\in J^+(y)$, both signs are reversed. Using the same endpoint variation argument as in the proof of Lemma \ref{Lem: Properties of Reflected Geodesic Distance} or of Proposition \ref{Prop: Reflected null branch atlas}, it turns out that 
	\begin{equation}\label{Eq: Branchwise endpoint derivatives}
		k_{x;N,\alpha}=-v_{N,\alpha}(x,y)^\flat,
	\end{equation}
	where, for simplicity of the notation, we drop the $(x,y)$-dependence on the left-hand side. Using the reversed branch $\alpha^{\rm op}$ as per Remark \ref{Rem: Reversed reflected branch}, one also obtains
	\begin{align*}
		-k_{y;N,\alpha}=d_y\sigma_{-,N,\alpha}(x,y)=d_x\sigma_{-,N,\alpha^{\rm op}}(y,x)=-v_{N,\alpha^{\rm op}}(y,x)^\flat=\dot\gamma_{N,\alpha}^{(x,y)}(1)^\flat.
	\end{align*}
	Consequently, every element of $\Lambda^{\triangleright}_{N,\alpha}$ can be written as
	$$\left(x,-\lambda v_{N,\alpha}(x,y)^\flat,y,\lambda\dot\gamma_{N,\alpha}^{(x,y)}(1)^\flat\right)\doteq (x,k_x, y,-k_y),$$
	where $k_x=-\lambda v_{N,\alpha}(x,y)^\flat$ and $k_y=-\lambda\dot\gamma_{N,\alpha}^{(x,y)}(1)^\flat$. Since $(x,y)\in\mathcal Z_{N,\alpha}$, Proposition \ref{Prop: Reflected null branch atlas} entails that $\bigl(x,v_{N,\alpha}(x,y)\bigr)\in\Gamma^0_{N,\mathcal{N}}$. The condition $\lambda(t(x)-t(y))>0$ is equivalent to $k_x$ being future pointing. Hence
	$$
	\bigcup_{\alpha\in\mathcal{I}_N}\Lambda^{\triangleright}_{N,\alpha}
	\subseteq
	\Lambda^{\triangleright}_{N,\mathcal{N}}.
	$$
	
	Conversely, let $(x,k_x,y,-k_y)\in\Lambda^{\triangleright}_{N,\mathcal{N}}$. By definition there exists $(x,v)\in\Gamma^0_{N,\mathcal{N}}$ generating a null broken geodesic from $x$ to $y$. Item \emph{1.} of Proposition \ref{Prop: Reflected null branch atlas} entails that there exists $\alpha\in\mathcal{I}_N$ such that $(x,v)\in\mathcal V_{N,\alpha}$. Therefore
	$$v=v_{N,\alpha}(x,y),\quad (x,y)\in\mathcal Z_{N,\alpha}.$$
	Equations \eqref{Eq: Branchwise endpoint derivatives} and the corresponding identity in the second entry shows that
	$$
	(x,k_x,y,-k_y)
	=
	\left(x,\lambda d_x\sigma_{-,N,\alpha}(x,y),
	y,\lambda d_y\sigma_{-,N,\alpha}(x,y)\right).
	$$
	Since $k_x$ is future pointing, $\lambda(t(x)-t(y))>0$, and thus $(x,k_x,y,-k_y)\in\Lambda^{\triangleright}_{N,\alpha}$. This proves the reverse inclusion and, thus, Equation \eqref{Eq: Cotangent branch decomposition}.
\end{proof}

\noindent Corollary \ref{Cor: Branchwise conormal resolution} establishes the branchwise conormal description of null broken geodesics with finitely many reflections. Taking the union over all $N\geq 1$ yields the corresponding description for all null broken geodesics in $\mathcal{N}$ having finitely many reflections.

\begin{corollary}\label{Cor: Countable reflected conormal resolution}
Under the same assumptions of Proposition \ref{Prop: Reflected null branch atlas}, let
\begin{equation}\label{Eq: Full reflected relation}
\mathfrak{R}_{\mathcal{N}}^{0,\mathrm{ref}}\doteq\bigcup_{N\geq1}\mathfrak{R}^0_{N,\mathcal{N}},\quad\mathrm{and}\quad\Lambda_{\mathcal{N}}^{\triangleright,\mathrm{ref}}\doteq\bigcup_{N\geq1}\Lambda^{\triangleright}_{N,\mathcal{N}}.
\end{equation}
Then 
\begin{equation}\label{Eq: Full reflected pair branch decomposition}
\mathfrak R_{\mathcal{N}}^{0,\mathrm{ref}}=\bigcup_{(N,\alpha)\in\mathcal{I}_{\mathrm{ref}}}\mathcal Z_{N,\alpha},\quad\mathrm{and}\quad\Lambda_{\mathcal{N}}^{\triangleright,\mathrm{ref}}=\bigcup_{(N,\alpha)\in\mathcal{I}_{\mathrm{ref}}}\Lambda^{\triangleright}_{N,\alpha}.
\end{equation}
where $\mathcal{I}_{\mathrm{ref}}\doteq\bigsqcup_{N\geq1}\bigl(\{N\}\times\mathcal{I}_N\bigr)$ is an index set closed under branch reversal, that is $(N,\alpha)^{\rm op}\doteq(N,\alpha^{\rm op})$. Furthermore, for every $(N,\alpha)\in\mathcal{I}_{\mathrm{ref}}$,
\begin{equation}\label{Eq: Cotangent branch reversal}
(x,k_x,y,-k_y)\in\Lambda^{\triangleright}_{N,\alpha}\quad\Longleftrightarrow\quad(y,k_y,x,-k_x)\in\Lambda^{\triangleright}_{N,\alpha^{\rm op}}.
\end{equation}
\end{corollary}

\begin{proof}
	Observe that, since $\mathcal{I}_N$ is countable on account of Proposition \ref{Prop: Reflected null branch atlas}, so is $\mathcal{I}_{\mathrm{ref}}$. Equation \eqref{Eq: Full reflected pair branch decomposition} follows from taking the union over $N\geq1$ in item \emph{3.} of Proposition \ref{Prop: Reflected null branch atlas} and in Equation \eqref{Eq: Cotangent branch decomposition}. Focusing on Equation \eqref{Eq: Cotangent branch reversal}, we observe that, on account of Remark \ref{Rem: Reversed reflected branch}, $\sigma_{-,N,\alpha^{\rm op}}(y,x)=\sigma_{-,N,\alpha}(x,y)$. Consequently, following Equation \eqref{Eq: Branchwise endpoint derivatives} and restoring the $(x,y)$-dependence,
	$$k_{y;N,\alpha^{\rm op}}(y,x)=-k_{y;N,\alpha}(x,y),\quad k_{x;N,\alpha^{\rm op}}(y,x)=-k_{x;N,\alpha}(x,y).$$
	Equation \eqref{Eq: Branchwise reflected conormal relation} entails that, if $(x,k_x,y,-k_y)\in\Lambda^{\triangleright}_{N,\alpha}$, there exists $\lambda\neq0$ such that $k_x=\lambda k_{x;N,\alpha}(x,y)$ and $k_y=\lambda k_{y;N,\alpha}(x,y)$ with $\lambda\bigl(t(x)-t(y)\bigr)>0$. Here $t$ is still the underlying time function. Setting $\lambda^{\rm op}\doteq-\lambda$, the preceding identities yield
	$$(y,k_y,x,-k_x)=\left(y,\lambda^{\rm op}k_{y;N,\alpha^{\rm op}}(y,x),x,-\lambda^{\rm op}k_{x;N,\alpha^{\rm op}}(y,x)\right),	$$
	while $\lambda^{\rm op}\bigl(t(y)-t(x)\bigr)=\lambda\bigl(t(x)-t(y)\bigr)>0$. Therefore $(y,k_y,x,-k_x)\in\Lambda^{\triangleright}_{N,\alpha^{\rm op}}$. The converse follows by applying the same argument to the reversed branch.
\end{proof}

\begin{remark}\label{Rem: Useful Formula}
For later convenience, we highlight that, for any regular $1$-reflected branch chart $(\mathcal{V}_{1,\alpha},\Omega_{1,\alpha})$ and any open subset of $\Omega_{1,\alpha}$ on which the ordinary Synge world function $\sigma$ is defined as per Equation \eqref{Eq: Synge World Function Interior}, it holds that, for every smooth scalar function $F=F(\sigma,\sigma_-)$, where $\sigma_-\equiv\sigma_{-,1,\alpha}$, 
\begin{equation}\label{Eq: Derivative operators with sigma, sigma-}
\begin{cases}
\partial_{\mu}F = \sigma_\mu \frac{\partial F}{\partial \sigma} + (\sigma_-)_\mu \frac{\partial F}{\partial \sigma_-}, \\
\Box_g F= \sigma^{\mu}{}_\mu \frac{\partial F}{\partial \sigma} + 2\sigma \frac{\partial^2 F}{\partial^2 \sigma} + 2 g^{\mu \, \nu} \sigma_\mu (\sigma_-)_\nu \frac{\partial^2 F}{\partial \sigma \partial \sigma_-} + (\sigma_{-})^{\mu}{}_\mu \frac{\partial F}{\partial \sigma_-} + 2\sigma_- \frac{\partial^2 F}{\partial^2 \sigma_-},
\end{cases}
\end{equation}
where we employed the notation $\sigma_{\mu \, \nu}\doteq \nabla_{\mu} \nabla_{\nu} \sigma$ and $g^{\mu \, \nu} \sigma_{\mu \, \nu} \doteq \sigma^{\mu}{}_{\mu}$. Observe that an analogous notation has been adopted also for $\sigma_-$.
\end{remark}

\section{Fundamental Solutions on Spacetimes with Timelike Boundary}\label{Sec: BV problem}

In this section we introduce the field theory under consideration, together with the Robin boundary conditions and we investigate the corresponding solution theory. In the following, we consider $(\mcM,g)$ to be a $d$-dimensional, $d\geq 2$, globally hyperbolic spacetime with a timelike boundary, but we do not require that it is infinitesimally convex as per Assumption \ref{Ass: Infinitesimally convex}. This constraint plays a role only in Section \ref{Sec: Wavefront Set of the Propagators}. On top of $(\mcM,g)$, we consider a real scalar field $u:\mcM\to\bR$ satisfying the Klein-Gordon equation
\begin{equation}\label{Eq: KG equation}
	Pu=(\Box_g-m^2)u=0,
\end{equation}
where $\Box_g \doteq g^{\mu \nu} \nabla_{\mu} \nabla_{\nu}$ denotes the D'Alembert operator built out of the Lorentzian metric $g$, while 
\begin{equation}\label{Eq: Effective mass}
	m^2=m^2_0+\xi R.
\end{equation}
Here $m^2_0\geq 0$ is the squared mass of the field, $R$ the scalar curvature while $\xi\in\bR$ is an arbitrary coupling parameter. Since $\mcM$ possesses a non-empty timelike boundary, the Cauchy problem for Equation \eqref{Eq: KG equation} must be supplemented by boundary conditions at $\partial\mcM$. Following the same rationale of \cite{Costeri_25}, we consider here those of Robin type, namely 
\begin{equation}\label{Eq: Robin Boundary Conditions}
	(\nabla_n u)|_{\partial\mcM} \equiv (\partial_n u) \vert_{\partial \mcM} =-\kappa u|_{\partial\mcM}, \qquad \kappa \in \mathbb{R},
\end{equation} 
where $n$ is the inward pointing, unit, normal vector field to $\partial\mcM$.

\begin{remark}\label{Rmk: range of variability of kappa}
Observe that $\kappa$ is chosen to be constant in Equation \eqref{Eq: Robin Boundary Conditions} only for simplicity of the presentation. All the results obtained could be extended to $\kappa\in C^\infty(\partial\mcM)$ with the due exception of Section \ref{Sec: Fundamental Solutions on Static Spacetimes}. Herein one would be forced to consider $\kappa$ time-independent. Even more generally, one might also try to replace $\kappa$ with a pseudo-differential operator $\Theta\in\Psi^k(\partial\mcM)$ as done in \cite{Dappiaggi-Marta_2020, Gannot_2022} asymptotically anti-de Sitter spacetimes. This specific generalization would require instead a significant departure from our analysis and therefore we do not consider it in this work.

At last, in contrast to \cite{Costeri_25}, in Equation~\eqref{Eq: Robin Boundary Conditions} we do not restrict the Robin parameter $\kappa$ to take values in $(-\infty,0]$. In this reference if $\kappa>0$ bound states do occur, leading to problems in the construction of Hadamard states. Such difficulties do not appear in this work and we observe that, setting $\kappa = 0$ in Equation \eqref{Eq: Robin Boundary Conditions} corresponds to considering \emph{Neumann} boundary conditions, while those of \emph{Dirichlet} type can be encoded by formally taking the limit $\kappa \rightarrow \pm \infty$.   
\end{remark}

To fix our nomenclature, we give a definition of advanced and retarded propagators associated to the Klein-Gordon operator $P$ with Robin boundary conditions as in Equation \eqref{Eq: Robin Boundary Conditions}. For more information on these structures for globally hyperbolic spacetimes with empty boundary, we refer to \cite{Baer_2007}.

\begin{definition}\label{Def: SolFond on boundary}
Let $(\mcM, g)$ be a $d$-dimensional, $d\geq 2$, globally hyperbolic spacetime with a timelike boundary. Let $P$ be the Klein-Gordon operator as per Equation \eqref{Eq: KG equation}. We call {\bf advanced $(+)$} and {\bf retarded $(-)$ propagators with Robin boundary conditions} any $G^\pm_\kappa \in \mathcal{D}^\prime(\mcM\times\mcM)$ such that, denoting by $\mathring{\mcM} \doteq \mcM\setminus\partial\mcM$,
	\begin{equation}
		\label{Eq: Boundary and PDE}
		(P\otimes\mathbb{I})G^\pm_\kappa\vert_{\mathring{\mcM}\times\mathring{\mcM}}=(\mathbb{I}\otimes P)G^\pm_\kappa\vert_{\mathring{\mcM}\times\mathring{\mcM}}=\delta\vert_{\mathring{\mcM}\times\mathring{\mcM}},
	\end{equation}
	and
	\begin{equation}\label{Eq: Boundary conditions}
		(\nabla_{n}\otimes\mathbb{I})G^\pm_\kappa|_{\partial\mcM}=(-\kappa\otimes\mathbb{I})G^\pm_\kappa|_{\partial\mcM},
	\end{equation}
	where $\nabla_{n} \doteq n^{\mu} \nabla_{\mu}$ is such that $n$ is the unit, inward pointing, normal vector field to $\partial\mcM$. Furthermore, $(\cdot)\vert_{\partial \mathcal{M}}$ denotes the pull back to $\partial\mcM$ of the bidistribution within brackets with respect to the first entry. Denoting by $\mathcal{G}_{\kappa}^{\pm}: \mathcal{D}(\mcM) \rightarrow C^{\infty}(\mcM)$ the {\bf advanced $(+)$} and {\bf retarded $(-)$ Green operators} corresponding to $G^{\pm}_{\kappa}$, we have that
	\begin{equation}\label{Eq: Support at the boundary}
		\mathrm{supp}(\mathcal{G}^\pm_\kappa(f))\subseteq J^\mp(\mathrm{supp}(f)), \quad \forall f\in\mathcal{D}(\mathring{\mcM})
	\end{equation}
	where the partial evaluation is on the second entry.
\end{definition}

\begin{remark}\label{Rem: Well-defined pull-back}
	Note that the pull back of $G^{\pm}_{\kappa}$ to the boundary is well-defined on account of \cite[Thm. 8.2.4]{Hormander_1990}. As a matter of fact, denoting by 
    $$\mathrm{Char}(P) \doteq
\{(x,k_x)\in T^*\mcM\setminus\{0\}\;|\; \sigma_P(x,k_x) = g_x^{-1}(k_x, k_x) = 0\},$$ the characteristic set of $P$, Equation \eqref{Eq: Boundary and PDE} entails that $$\mathrm{WF}(G^\pm_\kappa)\subset N^*\operatorname{Diag}_2(\mcM)\cup(\mathrm{Char}(P)\times\mathrm{Char}(P)),$$ where 
    $$N^*\operatorname{Diag}_2(\mcM)=\{(x,k_x, x,-k_x)\in T^*(\mcM\times\mcM)\setminus\{\boldsymbol{0}\}\},$$
while restriction at the boundary is codified by the canonical injection $\iota_1:\partial\mcM\times\mcM\to\mcM\times\mcM$ whose associated set of normal maps is 
	$$N_{\iota_1}\doteq\{(x,\lambda n_x^\flat, y, 0)\in T^*(\partial\mcM\times\mcM)\setminus\{\boldsymbol{0}\}\},$$
	where $n_x^\flat$ is the metric induced element of $T^*_x\mcM$ starting from $n$, the inward pointing normal vector to $\partial\mcM$ at $x$. By direct inspection, we have that $N_{\iota_1}\cap N^*\operatorname{Diag}_2(\mcM)=\emptyset$. Elements lying in $\mathrm{Char}(P)\times\mathrm{Char}(P)$ are of the form $(x,k_x,y,k_y)$ constrained by $g^{-1}_x(k_x,k_x)=g^{-1}_y(k_y,k_y)=0$. Comparison with $N_{\iota_1}$ entails that $k_y=0$, but if $x\in\partial\mcM$, then $k_x$ must be lightlike. We can split $k_x=k^\partial_x+\eta_k n_x^\flat$, where $\eta_k=g^{-1}(k_x,n_x^\flat)$ while $k^\partial_x \doteq k_x-\eta_k n_x^\flat$. The no-glancing hypothesis entails that $\eta_k\neq 0$. Hence 
	$$0=g^{-1}_x(k_x,k_x)\Longrightarrow g^{-1}_x(k^\partial_x,k^\partial_x)+\eta^2_k=0.$$
	This identity cannot hold true if $k^\partial_x =0$. Hence, $k_x$ cannot be proportional to $n_x^\flat$ entailing that $N_{\iota_1}\cap (\mathrm{Char}(P)\times\mathrm{Char}(P))=\emptyset$.
\end{remark}

\subsection{Existence of the fundamental solutions}\label{Sec: Existence of the Fundamental Solutions}

In this section we prove one of the main results of this work, namely the existence of fundamental solutions associated to Equation \eqref{Eq: KG equation}. To this end, we consider the associated initial value-boundary problem
\begin{equation}\label{Eq: Cauchy problem on non-static spacetime}
	\begin{cases}
		P u \doteq (\Box_g - m^2) u = f, \\
		\mathcal{B}_\kappa u \doteq (\partial_n + \kappa)u \vert_{\partial \mcM} = 0, \quad \kappa \in \mathbb{R},\\
		u|_{\Sigma_{t_0}}=u_0\quad\textrm{and}\quad\partial_tu|_{\Sigma_{t_0}}=u_1
	\end{cases}
\end{equation}
where $n$ denotes the inward pointing unit normal vector field to $\partial\mcM$, $u_0, u_1 \in C^\infty_0(\mathring{\Sigma}_{t_0})$, while $f\in C^\infty_0(\mcM)$.

The analysis can be divided conceptually in three separate statements. The first establishes causal propagation and it is based on energy estimates. Here we strongly rely on the inspiring ideas and results in \cite[App. A]{Fournodavlos:2021eye}. The second establishes existence and uniqueness of a smooth solution of Equation \eqref{Eq: KG equation} for prescribed smooth and compactly supported initial data. This is based on standard PDE techniques, in particular the method of Galerkin as discussed in \cite[Sec. 7.2]{Evans}. The third and last part wraps up the preceding one concluding that the proven result entail existence of advanced and retarded fundamental solutions. We conclude this long subsection by establishing some underlying basic structural properties. 

\begin{proposition}\label{Prop: Causal Support of Solutions}
	Let $(\mcM,g)$ be a globally hyperbolic spacetime with $\dim\mcM\geq 2$ as per Definition \ref{Def: Globally Hyperbolic}. Given $f\in C^\infty_0(\mcM)$ and $u_0,u_1\in C^\infty_0(\mathring{\Sigma}_{t_0})$ where $\Sigma_{t_0}=\{t_0\}\times\Sigma$ is a Cauchy surface of $\mcM$ and, given $u\in C^\infty(\mcM)$, solution of Equation \eqref{Eq: Cauchy problem on non-static spacetime}, 
$$\mathrm{supp}(u)\subseteq J^+(\Omega)\cup J^-(\Omega),$$
where $J^\pm$ denote the causal future and past in $(\mcM,g)$ and where  $\Omega:=\mathrm{supp}(u_0)\cup\mathrm{supp}(u_1)\cup\mathrm{supp}(f)$.
\end{proposition}

\noindent Since also the proof of this proposition is rather lengthy, we feel more appropriate to relegate it to Appendix \ref{App: A}. 

\begin{proposition}\label{Prop: Existence of Solutions}
	Let $(\mcM,g)$ be a globally hyperbolic spacetime with $\dim\mcM\geq 2$ as per Definition \ref{Def: Globally Hyperbolic} with compact Cauchy surfaces. There exists a unique $u\in C^\infty(\mcM)$ solving the initial-boundary value problem in Equation \eqref{Eq: Cauchy problem on non-static spacetime}.
\end{proposition}

\noindent Since the proof of this proposition is rather lengthy and relies on standard energy and coercivity estimates, we deem it more appropriate to relegate it to Appendix \ref{App: B}.

\begin{theorem}\label{Thm: Existence of Propagators}
Let $(\mcM,g)$ be a globally hyperbolic spacetime with $\dim\mcM\geq 2$ as per Definition \ref{Def: Globally Hyperbolic} and let $P=\Box_g-m^2$ be the associated Klein-Gordon operator. Then this admits advanced and retarded propagators $G^\pm_\kappa\in\mathcal{D}^\prime(\mcM\times\mcM)$ with Robin boundary conditions, in the sense of Definition \ref{Def: SolFond on boundary}. 
\end{theorem}

\begin{proof}
Consider Equation \eqref{Eq: Cauchy problem on non-static spacetime} and fix a source $f\in C^\infty_0(\mcM)$. Let $\Sigma_{t_0}$ be a Cauchy surface in the past of $\operatorname{supp} f$ assigning thereon vanishing initial data. Define 
\begin{equation}\label{Eq. Gkappa-def}
\mathcal{G}_\kappa^-: C^\infty_0(\mcM) \to C^\infty(\mcM), \qquad f \mapsto u \doteq \mathcal{G}^-_\kappa f,
\end{equation}
where, on account of Propositions \ref{Prop: Causal Support of Solutions} and \ref{Prop: Existence of Solutions}, $u\in C^\infty(\mcM)$ and $\operatorname{supp} u \subseteq J^+(\operatorname{supp} f)$. Observe that, per construction 
\begin{equation*}
	\begin{cases}
		P \mathcal{G}^-_\kappa f = f, \\
		(\partial_n + \kappa \mathbb{I}) \mathcal{G}^-_\kappa f \vert_{\partial \mcM} = 0, 
	\end{cases}
\end{equation*}
which entails $P\circ\mathcal{G}^-_\kappa = \operatorname{id} \vert_{C^\infty_0(\mcM)}$. To establish uniqueness of $\mathcal{G}^-_\kappa$, assume that there exists $\widetilde{\mathcal{G}}^-_\kappa$ abiding by all the above properties. Then $\Delta \mathcal{G}^-_\kappa \doteq \mathcal G^-_\kappa - \widetilde{\mathcal{G}}^-_\kappa$ satisfies for all $f\in C^\infty_0(\mcM)$
\begin{equation*}
	\begin{cases}
		P \Delta \mathcal{G}^-_\kappa f  = 0, \\
		(\partial_n + \kappa \mathbb{I}) \Delta \mathcal{G}^-_\kappa f  \vert_{\partial \mcM} = 0, \\
		\Delta \mathcal{G}^-_\kappa f \vert_{t=t_0} = \partial_t (\Delta \mathcal{G}^-_\kappa f) \vert_{t=t_0} = 0.
	\end{cases}
\end{equation*}
Uniqueness established in Proposition \ref{Prop: Existence of Solutions}  entails that $\Delta \mathcal{G}^-_\kappa f=0$ for all $f\in C^\infty_0(\mcM)$, that is $\Delta\mathcal{G}^-_\kappa=0$. Consider 
\begin{equation}
\label{Eq: C infty kappa}
    C^\infty_\kappa(\mcM)\doteq\{u\in C^\infty(\mcM)\;|\;(\nabla_n+\kappa)u|_{\partial\mcM}=0\},
\end{equation}
where $n$ is the inward pointing, unit vector, normal to $\partial\mcM$. Set $w \doteq \mathcal{G}^-_\kappa P h - h$ for $h \in C^\infty_\kappa(\mcM)\cap C^\infty_0(\mcM)$. Then, $P w = 0$ and $(\nabla_n+\kappa)w \vert_{\partial \mcM} = 0$. Since $\operatorname{supp} w \subseteq J^+(\operatorname{supp} h)$, there exists $t_0 \in \mathbb{R}$ such that $w_0 \doteq w \vert_{t=t_0} = 0$ and $w_1 \doteq \partial_t w \vert_{t=t_0} = 0$. Uniqueness of the solution as per Proposition \ref{Prop: Existence of Solutions} entails that $w = 0$. As a consequence, $\mathcal{G}^-_\kappa\circ P = \mathrm{id}$ on $C^\infty_\kappa(\mcM)$.

In order to conclude the proof we are left with establishing continuity of $\mathcal{G}^-_\kappa$. Given an arbitrary but fixed compact set $K\subset\mcM$, define 
\begin{equation*}
	\mathcal{D}_K (\mcM) \doteq \{f \in C^\infty_0(\mcM) \, \vert \, \operatorname{supp} f \subset K \}.
\end{equation*}
To prove that $\mathcal{G}^-_\kappa: C^\infty_0(\mcM) \to C^\infty(\mcM)$ is continuous in the relevant topologies, we need to control the semi-norms on $C^\infty(\mcM)$, namely,
\begin{equation*}
	p_{L, m}(u) \doteq \max_{|\alpha| \le m} \sup_{x \in L} |\nabla^\alpha u(x)|, 
\end{equation*}
where $L$ is an arbitrary compact set and $m \in \mathbb{N}_0$. Let us cover 
\begin{equation*}
	L \cap J^+(K) = \bigcup_{i=1}^{N < \infty} D_i \subset \bigcup_{i=1}^{N < \infty} \mathcal{C}_{I,i}, 
\end{equation*}
where each $D_i=J^-(p_i)\cap(I\times\Sigma)$, $I=[t^\prime_0,t_1]$, while $\mathcal{C}_{I, i} \equiv \mathcal{C}_{I, i}(K, L) \simeq I \times \Sigma_i$ is a cylinder. Observe that these are defined to coincide with the sets used in the proofs of Propositions \ref{Prop: Causal Support of Solutions} and \ref{Prop: Existence of Solutions}. Since we are considering vanishing initial data $u_0 = u_1 = 0$, we can apply the coercive estimates in Equation \eqref{Eq: coercive estimate 1} combined with $\mathsf K(u,u)\leq C_0\|u\|_{L^2_t}$ to infer a counterpart of Equation \eqref{Eq: estimate of the finite energy from below}, namely
$$\mathrm{E}[u](t) \doteq \frac{1}{2} \mathrm{K}(\partial_t u, \partial_t u) + \frac{1}{2} \mathrm{B}(u, u)\leq \mathrm{C} (\|\partial_t u\|^2_{L^2_t} + \|u\|^2_{H^1_t}).$$
where $\mathrm{C}>0$. Since $\mathrm{E}[u](t_0) = 0$, this estimate combined with the bound on the energy due to Gr\"onwall lemma, see Appendix \ref{App: B}, yields 
\begin{equation}\label{Eq: star step 7}
\sup_{t \in I} \left(\|u(t)\|^2_{H^1_t} + \|\partial_t u(t)\|^2_{L^2_t} \right) \le C_0 \int_I \|f(s)\|^2_{L^2_s} \, ds < \infty, 
\end{equation}
where $C_0 = C_0(K, L, I)$ is a constant which does not depend on the source term $f$. The same kind of coercive estimates applied to the time derivatives of $\mathrm{E}[u](t)$, as in Appendix \ref{App: B}, entail that, being $\mathrm{E}^{(j)}(t_0) = 0$ for all $j \ge 1$,
\begin{equation}\label{Eq: delta step 7}
\sup_{t \in I} \left(\|\partial_t^j u(t)\|^2_{H^1_t} + \|\partial_t^{j+1} u(t)\|^2_{L^2_t} \right) \le C_j \sum_{k=0}^j \int_I \|\partial_s^k f(s)\|^2_{L^2_s} \, ds < \infty, 
\end{equation}
where $C_j\equiv C_j(K,L,I) \ne C_j(f)$. Combining now Equations \eqref{Eq: star step 7} and \eqref{Eq: delta step 7}, we get for all $j\geq 0$
\begin{equation*}
	\sup_{t \in I} \left(\|\partial_t^j u(t)\|^2_{H^1_t} + \|\partial_t^{j+1} u(t)\|^2_{L^2_t} \right) \le C'_j \sum_{k=0}^j \sup_{t \in I} \|\partial_s^k f(t)\|^2_{L^2(\Sigma_i)} < \infty,
\end{equation*}
where $C^\prime_j>0$ encompasses all proportionality constants. By elliptic regularity, as exploited in Step 3. of Appendix \ref{App: B}, we can switch from $H^1$ to $H^r$ for arbitrary $r > 0$ obtaining
\begin{equation*}
	\sup_{t \in I} \|\partial_t^j u(t)\|^2_{H^r_t} \le C(j, r, i) \sum_{|\alpha| \le j < \infty} \sup_{t \in I} \|\nabla^{\alpha} f(t)\|^2_{L^2(\Sigma_i)} < \infty.
\end{equation*}
Choosing $\mathcal{C}_{I, i}$ such that $\operatorname{supp} f \subset K \subseteq \mathcal{C}_{I,i}$ and denoting by $K_i \doteq K \cap \Sigma_i$, it holds that $\operatorname{supp} f(t) \subset K_i$, which entails in turn that there exists $C_K>0$ for which
\begin{equation*}
\|\nabla^{\alpha} f(t)\|^2_{L^2(\Sigma_i)} = \|\nabla^{\alpha} f(t)\|^2_{L^2(K_i)} \le C_K \sup_{K} |\nabla^\alpha f|. 
\end{equation*}
Hence, 
\begin{equation*}
	\sup_{t \in I} \|\partial_t^j u(t)\|^2_{H^r_t} \le C'(j, r, i, K) \sum_{|\alpha| \le q(j,r)} \sup_{t \in I} \sup_{K_i} |\nabla^{\alpha} f(t)|^2, \; \forall r > 0\;\textrm{and}\; j\geq 0,
\end{equation*}
where $q(j,r)\in\mathbb{N}\cup\{0\}$. If $r > \frac{1}{2} \operatorname{dim} \Sigma + m$, then $H^r(\Sigma_i) \hookrightarrow C^m(\Sigma_i)$ by Sobolev embedding theorem, which implies 
\begin{equation*}
	\sup_{t \in I} \|\partial_t^j u(t)\|^2_{C^m(\Sigma_i)} \le \tilde{C}(j, r, i, K) \sum_{|\alpha| \le q(j,r)} \sup_K |\nabla^{\alpha} f(t)|^2, \; \forall j \in \mathbb{N}_0, \, \forall m > 0.
\end{equation*}
We can conclude that
\begin{equation*}
	p_{L, m}(u) \le \tilde{C}(K, L, m) \sum_{|\alpha| \le q(m)} \sup_K |\nabla^{\alpha} f(t)|^2, \qquad \forall m > 0.
\end{equation*}
Recalling that $u \doteq \mathcal{G}^-_\kappa f$, we conclude that, for every compact set $K$, $\mathcal{G}^-_\kappa: C^\infty_K(\mcM) \to C^\infty(\mcM)$ is continuous. Since $C^\infty_0(\mcM)$ can be realized as the inductive limit over $K$ of $\mathcal{D}_K(\mcM)$, we conclude that also $\mathcal{G}^-_\kappa: C^\infty_0(\mcM) \to C^\infty(\mcM)$ is continuous. Eventually, by Schwartz kernel theorem, 
\begin{equation*}
	\exists G^-_\kappa \in \mathcal{D}'(\mcM \times \mcM),
\end{equation*}
which concludes the proof.
\end{proof}

\noindent We establish a notable property relating $\mathcal{G}^+_\kappa$ to $\mathcal{G}^-_\kappa$ which will be useful at a later stage.
\begin{corollary}\label{Cor: G+- Formal Adjoints}
	Let $(\mcM,g)$ be a globally hyperbolic spacetime with timelike boundary and let $\mathcal{G}^\pm_\kappa$ be the advanced and retarded Green operators associated to the Klein-Gordon operator with Robin boundary conditions. Then $\mathcal{G}^-_\kappa$ is the formal adjoint of $\mathcal{G}^+_\kappa$ and viceversa, namely, for every $f,h\in C^\infty_0(\mcM)$,
	$$(\mathcal{G}^+_\kappa f,h)=(f,\mathcal{G}^-_\kappa h),$$
	where $(\cdot,\cdot)$ denotes the standard $L^2$-pairing.
\end{corollary}

\begin{proof}
	Since $\mathcal{G}^+_\kappa f,\mathcal{G}^-_\kappa h \in C^\infty(\mcM)$ on account of Proposition \ref{Prop: Existence Robin propagators}, then both pairing are well-defined and
	$$(\mathcal{G}^+_\kappa f,h)=(\mathcal{G}^+_\kappa f,P\mathcal{G}^-_\kappa h)=(f,\mathcal{G}^-_\kappa h),$$
	where, in the first identity, we used that $P\circ\mathcal{G}^\pm_\kappa = \operatorname{id} \vert_{C^\infty_0(\mcM)}$, while, in the second one, we have integrated by parts exploiting that $\mathcal{G}^\pm_\kappa$ automatically implements Robin boundary conditions.
\end{proof}

To conclude the section we extend to the case in hand the analysis of \cite{Costeri_25} aimed at giving a full characterization of the space of solutions for the Klein-Gordon equation \eqref{Eq: KG equation} in terms of an exact sequence. As a first step, let us recall that $$C^\infty_\kappa(\mcM) \doteq \{f\in C^\infty(\mcM)\;|\;(\nabla_n+\kappa)f|_{\partial\mcM}=0\}.$$
Generalizing \cite{Bar15}, we can infer that one can extend continuously $\mathcal{G}^\pm_\kappa$ as per Definition \ref{Def: SolFond on boundary} to operators 
\begin{equation}\label{Eq: Extend Domain}
	\mathcal{G}^-_\kappa: C^\infty_{\mathrm{pc}}(\mcM)\to C^\infty_\kappa(\mcM)\;\mathrm{and}\; \mathcal{G}^+_\kappa:C^\infty_{\mathrm{fc}}(\mcM)\to C^\infty_\kappa(\mcM),
\end{equation}
where the subscripts $\mathrm{fc}/\mathrm{pc}$ stand respectively for future and past compact, namely
$$C^\infty_{\mathrm{pc}}(\mcM)\doteq\{f\in C^\infty(\mcM)\;|\;\exists t_f\in\bR\;\mathrm{for}\;\mathrm{which}\;f=0\;\forall t<t_f\},$$
while
$$C^\infty_{\mathrm{fc}}(\mcM)\doteq\{f\in C^\infty(\mcM)\;|\;\exists t_f\in\bR\;\mathrm{for}\;\mathrm{which}\;f=0\;\forall t>t_f\}.$$
In addition it holds that
\begin{equation}\label{Eq: Extended Domains 2}
	P\circ\mathcal{G}^\pm_\kappa=\mathrm{id}|_{C^\infty_{\mathrm{fc}/\mathrm{pc}}(\mcM)}\;\mathrm{and}\;\mathcal{G}^\pm_\kappa\circ P|_{C^\infty_{\mathrm{fc}/\mathrm{pc},\kappa}(\mcM)}=\mathrm{id}|_{C^\infty_{\mathrm{fc}/\mathrm{pc},\kappa}(\mcM)},
\end{equation}
where $C^\infty_{\mathrm{pc}/\mathrm{fc},\kappa}(\mcM)=C^\infty_{\mathrm{pc}/\mathrm{fc}}(\mcM)\cap C^\infty_\kappa(\mcM)$.

\begin{proposition}\label{Prop: Exact Sequence}
	Let $P$ be the Klein-Gordon operator as per Equation \eqref{Eq: KG equation} with Robin boundary conditions and denote by $\mathcal{G}_\kappa=\mathcal{G}^-_\kappa-\mathcal{G}^+_\kappa$ the Pauli-Jordan Green operator constructed out $\mathcal{G}^\pm_\kappa$, whose associated bidistribution is as per Proposition \ref{Prop: Existence Robin propagators}. Then, introducing the space of smooth and timelike compact functions, $C^\infty_{\mathrm{tc}}(\mcM)\doteq C^\infty_{\mathrm{pc}}(\mcM)\cap C^\infty_{\mathrm{fc}}(\mcM)$, the following is an exact sequence
	\begin{flalign*}
		\xymatrix{
			0\ar[r] & C^\infty_{\mathrm{tc}, \kappa}(\mcM) \ar[r]^{P} & C^\infty_{\mathrm{tc}}(\mcM)\ar[r]^{\mathcal{G}_\kappa} & C^\infty_{\kappa}(\mcM) \ar[r]^{P}& C^\infty(\mcM)\ar[r] & 0,}
	\end{flalign*}
	where $C^\infty_{\mathrm{tc},\kappa}(\mcM)\doteq C^\infty_{\mathrm{tc}}(\mcM)\cap C^\infty_\kappa(\mcM)$. As a consequence there exists an isomorphism of topological vector spaces between $\frac{C^\infty_{\mathrm{tc},\kappa}(\mcM)}{P[C^\infty_{\mathrm{tc},\kappa}(\mcM)]}$ and $\ker(P)|_{C^\infty_{\kappa}(\mcM)}$, the space of smooth solutions of the Klein-Gordon equation abiding by Robin boundary conditions.
\end{proposition}

\begin{proof}
	The first arrow in the sequence entails that $P$ is injective when acting on $C^\infty_{\mathrm{tc}}(\mcM)$. As a matter of fact, given $f\in C^\infty_{\mathrm{tc}, \kappa}(\mcM)$ such that $Pf=0$, then $\mathcal{G}^+_{\kappa} (Pf)=f=0$. 
	
	The second arrow of the sequence asserts that $\mathrm{Im}(P)=\ker(\mathcal{G}_\kappa)$. The inclusion $\subseteq$ is a direct consequence of $\mathcal{G}^\pm_\kappa \circ P = \mathrm{id}$ on $C^\infty_{\mathrm{tc}, \kappa}(\mcM)$. To prove the converse, suppose that $f\in\ker\mathcal{G}_\kappa$. In view of the support properties of $\mathcal{G}^\pm_\kappa$ as per Definition \ref{Def: SolFond on boundary}, this entails that $h=\mathcal{G}^+_\kappa f=\mathcal{G}^-_\kappa f\in C^\infty_{\mathrm{tc}, \kappa}(\mcM).$ In other words $f=Ph$ which is the sought conclusion. 
	
	The third arrow codifies instead that $\mathrm{Im}(\mathcal{G}_\kappa)=\ker(P)$ on $C^\infty_\kappa(\mcM)$. The inclusion $\subseteq$ is a direct consequence of the defining properties of $\mathcal{G}^\pm_\kappa$. Hence, assume $\Phi\in\ker (P)$ and consider $\chi\in C^\infty(\bR)$ such that $\chi=\chi(t)$ is equal to $1$ for $t>t_1$ while it vanishes for $t<t_0$. Setting $\Phi^+=\chi\Phi$ and $\Phi^-=(1-\chi)\Phi$, we observe that $\Phi^\pm\in C^\infty_{\mathrm{pc}/\mathrm{fc},\kappa}(\mcM)$. Furthermore $P\Phi=0$ entails $P\Phi^+=-P\Phi^-\in C^\infty_{\mathrm{tc}}(\mcM)$. The support properties of $\mathcal{G}^\pm_\kappa$ guarantee both that $\mathcal{G}^-_\kappa(P\Phi^+)=\Phi^+$ and that $\mathcal{G}^+_\kappa(P\Phi^+)=-\Phi^-$. Hence $\Phi=\mathcal{G}_\kappa(P\Phi^+)$, namely $\ker(P)\subseteq\mathrm{Im}(\mathcal{G}_\kappa)$ on $C^\infty_\kappa(\mcM)$. 
	
	The last arrow entails instead that $P$ is surjective on $C^\infty(\mcM)$. Similarly to the preceding discussion, given $\Phi\in C^\infty_\kappa(\mcM)$, we split it as $\Phi=\Phi^++\Phi^-$, where $\Phi^\pm\in C^\infty_{\mathrm{pc}/\mathrm{fc}}(\mcM)$. Setting $h=\mathcal{G}^-_{\kappa} \Phi^++\mathcal{G}^+_{\kappa}\Phi^-\in C^\infty_\kappa(\mcM)$ we can infer by direct inspection that $Ph=\Phi^++\Phi^-=\Phi$. This concludes the proof since the last statement is a direct consequence of the properties of an exact sequence of vector spaces, combined with the property that all arrows are implemented by continuous maps.
\end{proof}

\noindent As a consequence of this proposition we can establish uniqueness of the advanced and retarded propagators. The proof of the following lemma is identical to that of \cite[Prop. 4.22]{Costeri_25} replacing $\bH^d$ with $\mcM$. For this reason we omit it.

\begin{lemma}
    The bi-distributions $G^\pm_\kappa\in\mathcal{D}^\prime(\mcM\times\mcM)$ individuated in Theorem \ref{Thm: Existence of Propagators} are the unique advanced and retarded propagators for the Klein-Gordon operator $P$ with Robin boundary conditions on a globally hyperbolic spacetime with timelike boundary of dimension $d\geq 2$.
\end{lemma}

\subsection{Fundamental solutions on static spacetimes}\label{Sec: Fundamental Solutions on Static Spacetimes}
In this section, we restrict our attention to the class of standard static, globally hyperbolic spacetimes with timelike boundary such that their Cauchy surface is of bounded geometry as per Definition \ref{Def: Bounded_geometry_boundary}. As in the previous section we are not making explicit use of the assumption of infinitesimal convexity, which will play instead a key role when discussing the microlocal properties of the fundamental solutions. Although this is a special case of the general one considered in the previous section, we feel worth summarizing the spectral construction of the advanced and retarded propagators developed in \cite{Dappiaggi-Drago_2019}, since it will allow us to prove the existence of Hadamard states in Section \ref{Sec: Hadamard states with Robin boundary conditions} using a deformation argument.

We shall therefore only recall their definition and state the main results of \cite{Dappiaggi-Drago_2019} in the present geometric setting. Starting from $(\mcM,g)$, it is convenient to consider the conformally rescaled, ultrastatic metric $\widetilde{g}$ whose line element reads
\begin{align}\label{Equation: conformally rescaled metric}
\widetilde{ds}^2=-dt^2+\beta^{-1}h\,.
\end{align}
This entails that, setting $\dim\mcM=d$,
\begin{enumerate}
	\item The Klein-Gordon operator $P$ on $(\mcM,g)$ is mapped on $(\mcM,\widetilde{g})$ to $\widetilde{P}= -(\partial^2_t+A_m)$ such that 
	\begin{flalign}\label{Eq: Conformal Relation Between Wave Operators}
		P\circ\beta^{\frac{2-d}{4}}&=
		\beta^{-\frac{2+d}{4}}\bigg[\partial_t^2+A_m\bigg], \notag \\
		A_m=A+\beta m^2\doteq
		\Delta_{\beta^{-1}h}+
		\frac{2-d}{2}\beta^{-\frac{1}{2}}&\Delta_{\beta^{-1}h}(\beta^{\frac{1}{2}})-\frac{(2-d)(d-4)}{4}\,h^{ij}\nabla_i (\beta)\nabla_j(\beta)+\beta m^2,
	\end{flalign}
	\item The operator $\nabla_n+\kappa\mathbb{I}$ implementing the Robin boundary conditions at $\partial\mcM$ on $(\mcM,g)$ reads on $(\mcM,\widetilde{g})$
	\begin{equation}\label{Eq: Conformal Relation Boundary Condition}
		\nabla_n+\widetilde{\kappa}\mathbb{I},\quad\widetilde{\kappa}=\kappa-\frac{d-2}{4} n(\ln\beta),
	\end{equation}
\end{enumerate}
where $n(\ln \beta) \doteq \nabla_n (\ln \beta) \vert_{\partial \mcM}$ is the directional derivative of $\ln \beta$ along the inward pointing, unit normal vector field $n$, evaluated at the boundary. 
Following \cite{Dappiaggi-Drago_2019} it is convenient to construct advanced and retarded Green operators $\widetilde{\mathsf{G}}^\pm_{\widetilde{\kappa}}$ for $\partial_t^2+A+\beta m^2$. These are related to $G^\pm_\kappa$, those of $P$, via the identity
\begin{align}\label{Eq: Conformal Relation Between Propagators}
	\widetilde{\mathsf{G}}^\pm_{\widetilde{\kappa}}=
	\beta^{\frac{2-d}{4}}\circ G^\pm_\kappa\circ\beta^{\frac{2+d}{4}}\,.
\end{align}

\noindent \noindent In this framework we can formulate a straightforward generalization of \cite[Thm. 30]{Dappiaggi-Drago_2019}. In that reference only the wave operator is considered, but the construction rests on the following notable properties of $A$, see Equation \eqref{Eq: Conformal Relation Between Wave Operators}. This is first seen as a uniformly elliptic, symmetric operator on $C^\infty_0(\Sigma)\subset L^2(\Sigma,\beta^{-1}h)$. Afterwards, using the theory of boundary triples \cite{Grubb68} as well as \cite[Prop. 19]{Dappiaggi-Drago_2019}, to each self-adjoint operator $\vartheta$ on the space of square integrable functions on $\partial\Sigma$, one associates a self-adjoint extension of $A$, denoted by $A_\vartheta$. Here $\vartheta$ plays the role of codifying a boundary condition and, in our case, $\vartheta=\widetilde{\kappa}^{-1}\mathbb{I}$.

Here we need to apply this scheme to the operator $A_m$. On the one hand, since $m^2=m_0^2+\xi R$ is a zeroth-order contribution, $A_m$ is uniformly elliptic. This a property governed only by the principal part of the differential operator, which is the same as that of $A$. On the other hand, the identification of the self-adjoint extensions of $A_m$ in terms of assigned boundary conditions strongly requires the existence of a good notion of underlying Sobolev spaces, as discussed in \cite[Sec. 1.2]{Dappiaggi-Drago_2019}. For this reason, we impose the additional assumption, see Remark \ref{Rem: Comparison to GOW}:
\begin{center}
	$\beta,\beta^{-1}\in C^\infty_b(\Sigma,h)$.
\end{center}
Observe that this requirement entails both that $(\mcM,g)$ is a Lorentzian manifold of bounded geometry in the sense of \cite{Gérard}, as well as that $(\Sigma,q)$, where $q \doteq \beta^{-1}h$, is a complete Riemannian manifold. Under these assumptions one can introduce for every $s\in\mathbb{R}$ a notion of Sobolev space $H^s(\Sigma,q)$ and, if $s>\frac{1}{2}$, also a continuous surjective map 
\begin{equation}\label{Eq: Gamma0}
\gamma_0:H^s(\Sigma,q)\to H^{s-\frac{1}{2}}(\partial\Sigma,q_\partial),
\end{equation}
where $q_\partial$ denotes the pull-back of $q$ to $\partial\Sigma$. If we denote by $\nu_q$ the outward pointing unit normal to $\partial\Sigma$ with respect to $q$, it turns out that $\nu_q=-\beta^{\frac{1}{2}}n$ and for $s>\frac{3}{2}$ we can also introduce the continuous and surjective map
$$\gamma_1=-\gamma_0\circ\nabla_{\nu_q}:H^s(\Sigma,q)\to H^{s-\frac{3}{2}}(\partial\Sigma,q_\partial).$$
The Robin boundary condition in Equation \eqref{Eq: Robin Boundary Conditions} is encoded by the operator $\gamma_1+\beta^{\frac{1}{2}}\widetilde{\kappa}\gamma_0$. Having established the relevant notation we can state the following existence theorem for the propagators, whose proof is omitted being a straightforward generalization of that in \cite[Thm 30]{Dappiaggi-Drago_2019}.

\begin{proposition}	\label{Prop: Existence Robin propagators}
Let $(\mcM, g)$ be a $d$-dimensional, $d\geq 2$, standard static, globally hyperbolic spacetime with a timelike boundary. Assume in addition that
	\begin{enumerate}
		\item $\beta,\beta^{-1}\in C^\infty_b(\Sigma,h)$, see Remark \ref{Rem: Comparison to GOW},
		\item Considering $A_m$ with $D(A_m)=C^\infty_0(\Sigma)\subset L^2(\Sigma,q)$, $q \doteq \beta^{-1}h$, and denoting by $A_{m,{\widetilde{\kappa}}}$ the self-adjoint extension individuated by the rescaled Robin boundary condition codified by Equation \eqref{Eq: Conformal Relation Boundary Condition} with
		$$D(A_{m,{\widetilde{\kappa}}})=\{f\in H^2(\Sigma, q)\;|\;\gamma_1 f=- \beta^{\frac{1}{2}}\widetilde{\kappa}\gamma_0 f\},$$
		its spectrum is bounded from below.
	\end{enumerate}
Under the above assumptions, there exist unique advanced and retarded Green propagators $G^\pm_\kappa\in\mathcal{D}^\prime(\mcM\times\mcM)$ as per Definition \ref{Def: SolFond on boundary}
such that
\begin{equation}\label{Eq: Gpm from Gkappa}
	G^+_\kappa=-\Theta(t^\prime-t)G_\kappa, \quad G^-_\kappa=\Theta(t-t^\prime)G_\kappa.
\end{equation}
Here $G_\kappa=G^-_\kappa-G^+_\kappa$ is the associated Pauli-Jordan propagator where 
\begin{equation}\label{Eq: Beta Gkappa}
G_\kappa=\beta^{\frac{d-2}{4}}\circ\widetilde{G}_\kappa\circ\beta^{-\frac{d+2}{4}},
\end{equation}
and where, for all $f_1,f_2\in\mathcal{D}(\mcM)$, it holds that
\begin{equation}\label{Eq: Spectral Gkappa}
\widetilde{G}_\kappa(f_1,f_2)=\int_{\mathbb R^2} \left(f_1(t),	A_{m,\widetilde{\kappa}}^{-\frac{1}{2}}\sin\left(A_{m,\widetilde{\kappa}}^{\frac{1}{2}}(t-t')	\right)f_2(t')\right)_{L^2(\Sigma, q)}  dt \, dt',
\end{equation}
where the operator in the integral is defined in terms of its spectral resolution.	
\end{proposition}

\subsection{Wavefront Set of the Propagators}\label{Sec: Wavefront Set of the Propagators}
 Having established both the existence of the advanced and retarded fundamental solutions for the Klein-Gordon operator $P$ as in Equation \eqref{Eq: KG equation} with Robin boundary conditions and their causal structure, we can investigate their wavefront set. In the following we rely instead on the seminal works by Melrose and Sj\"ostrand \cite{Melrose_1978,Melrose_1982}, see also \cite{Taylor_1978}. Herein the authors consider a class of geometries which is larger than the one we are interested in, namely they allow both for grazing and glancing rays. These are situations which we discard in this work since we are mainly interested in finding a counterpart of Radzikowski's theorem, see Section \ref{Sec: Comparison of Hadamard states}. It is unclear whether their occurrence is compatible with the result that we are seeking.

\vskip .2cm

In the following we introduce some additional geometric and analytic data, necessary for establishing the main result of this section. It is worth observing that none of the structures introduced below depends on the Robin parameter $\kappa$. To start with, following \cite{Hormander_1990}, we denote by $\sigma_P$ the principal symbol of the Klein-Gordon operator $P$ in Equation~\eqref{Eq: KG equation}, namely
$$\sigma_P(x,k_x)\doteq g_x^{-1}(k_x,k_x),\quad\forall (x,k_x)\in T^*\mcM.$$
This comes accompanied by the characteristic set of $P$, that is
$$\mathrm{Char}(P)
\doteq
\{(x,k_x)\in T^*\mcM\setminus\{0\}\;|\; \sigma_P(x,k_x)=0\},$$
as well as by an Hamiltonian vector field $H_{\sigma_P}$. Its integral lines $\Gamma$ are also referred to as {\em null bicharacteristics}.

\begin{remark}\label{Rem: Connection to No Glancing}
It is worth rewriting the no-glancing hypothesis in Equation~\eqref{Eq: No Glancing} in this language. More precisely, for $(x,k_x)\in\mathrm{Char}(P)$, $x\in\partial\mcM$, we set
\begin{equation}
    \label{Eq: eta}
    \eta(x,k_x)\doteq g_x^{-1}(k_x,n_x^\flat)=k_x(n_x),
\end{equation}
where $n_x$ is the inward pointing normal vector to $\partial\mcM$ at $x$ while $^\flat$ denotes the standard musical isomorphism. Equation \eqref{Eq: No Glancing} entails that every element of $\mathrm{Char}(P)$ with $x\in\partial\mcM$, reached by a null bicharacteristic issued from the interior, satisfies $\eta(x,k_x)\neq0$. For every $x\in\partial\mcM$, let $\mathcal{R}_x$ be the reflection map in Equation~\eqref{Eq: Reflection Vector}. Its dual counterpart is
\begin{equation}\label{Eq: Dual Reflection}
\mathcal{R}_x^*:T_x^*\mcM\rightarrow T_x^*\mcM,\qquad\mathcal{R}_x^*k\doteq	k_x-2\eta(x,k_x)n_x^\flat,
\end{equation}
which abides by the identities
$$\sigma_P(x,\mathcal{R}_x^*k_x)=\sigma_P(x,k_x),\qquad\eta(x,\mathcal{R}_x^*k_x)=-\eta(x,k_x),\qquad	(\mathcal{R}_x^*k_x)|_{T_x\partial\mcM}=k_x|_{T_x\partial\mcM}.$$
It is worth emphasizing, $\mathcal{R}_x^*$ preserves the characteristic cone $\mathcal{N}^*_x=\{k_x\in T^*_x\mcM\setminus\{0\}\;|\;\sigma_P(x,k_x)=0\}$ and the time orientation of null covectors.
\end{remark}

\noindent As already mentioned in the previous sections, we have to account for the reflection of light rays hitting the boundary $\partial\mcM$. Yet, even including Assumption \ref{Ass: Infinitesimally convex}, there is no reason to assume that, starting from $(x,k_x)\in T^*\mathring{\mcM}\setminus\{0\}$ with $\sigma_P(x,k_x)=0$, the Hamiltonian flow encounters only at most once the boundary. Multiple reflections can and need to be accounted for. The following definition serves the purpose of giving a rigorous translation of this heuristic argument.

\begin{definition}\label{Def: Finite broken null bicharacteristic}
A {\bf finite broken null bicharacteristic (BNB)} is a map
$$\Gamma:[a,b]\setminus\{s_1,\ldots,s_N\}\longrightarrow\mathrm{Char}(P),\qquad s\mapsto\Gamma(s)=(\gamma(s),k(s)),$$
where $N\in\mathbb{N}$ and $a<s_1<\cdots<s_N<b$, such that
\begin{enumerate}
\item the base curve $\gamma$ extends continuously to $[a,b]$ and $\gamma(s_j)\in\partial\mcM$ for every $j=1,\ldots,N$,
\item on each connected component of $[a,b]\setminus\{s_1,\ldots,s_N\}$, $\Gamma$ is an integral curve of $H_{\sigma_P}$,
\item the one-sided limits $k(s_j^-)$ and $k(s_j^+)$ exist and satisfy
$$\eta(\gamma(s_j),k_x(s_j^-))\neq 0\quad\textrm{and}\quad k_x(s_j^+)=\mathcal{R}_{\gamma(s_j)}^*k_x(s_j^-).$$
\end{enumerate}
Two points $(x,k_x)$ and $(x^\prime,k_{x^\prime})\in\mathrm{Char}(P)$ which are connected by a BNB are denoted by
$$(x,k_x)\sim_{\mathrm b}(x^\prime,k_{x^\prime}).$$
By convention we set $(x,k_x)\sim_{\mathrm b}(x,k_x)$ and we write $(x,k_x)\sim (x^\prime, k_{x^\prime})$ in the limiting case where no reflection occurs, hence $N=0$. In this case we say that these points are connected by a regular bicharacteristic. 
\end{definition}

\begin{remark}\label{Rem: BNG}
	Under Assumption \ref{Ass: Infinitesimally convex} and in view of Proposition \ref{Prop: no-glancing}, considering the Klein-Gordon operator as per Equation \eqref{Eq: KG equation}, the curve $\gamma$ subordinated to a BNB is a piecewise smooth null geodesic and we shall refer to it as {\bf broken null geodesic}.
\end{remark}

Notice that the integer $N$ codifies the number of reflections of the broken null bicharacteristic. Furthermore, being the underlying manifold time-oriented, a BNB can be called future-directed ({\em resp.},  past-directed) if such is the covector $k(a)$ and we denote it by $k(a)\triangleright 0$ ({\em resp.}, $k(a)\triangleleft 0$). This property is preserved along the flow by all $k(s)$, $s\in [a,b]$.

\vskip .2cm

\noindent We can now establish a more direct connection between the geometric construction of Section \ref{Sec: Geodesic and Reflected Geodesic Distances Claudio} and the broken null relation governing the microlocal analysis below. To this end, first of all, we introduce the relevant notation. Let $(\mcM,g)$ be an globally hyperbolic and infinitesimally convex spacetime with timelike boundary abiding by Assumptions \ref{Ass: Infinitesimally convex} and \ref{Ass: Finite regular reflections} while $\mathcal{N}$ is a causally convex Cauchy neighbourhood as per Definition \ref{Def: Cauchy neighbourhood}. We set
\begin{equation}\label{Eq: Restricted future broken relation}
\mathcal{C}_{\mathrm b,\mathcal{N}}^{\triangleright}\doteq\left\{(x,k_x,y,-k_y)\in T^*(\mathring{\mathcal{N}}\times\mathring{\mathcal{N}})\setminus\{\boldsymbol{0}\}\;|\;(x,k_x)\sim_{\mathrm b}(y,k_y),\; k_x\triangleright 0\right\}.
\end{equation}
and the direct component as
\begin{equation}\label{Eq: Directed null relation in N}
\mathcal{C}_{\mathrm b,\mathcal{N}}^{\triangleright,\mathrm{dir}}\doteq{}\Delta_{\mathcal{N}}^{\triangleright}\cup
\Bigl\{(x,k_x,y,-k_y)\in T^*(\mathring{\mathcal{N}}\times\mathring{\mathcal{N}})\setminus\{\boldsymbol{0}\}\;|\;(x,k_x)\sim (y,k_y),\; x\neq y,\; k_x\triangleright 0\Bigr\},
\end{equation}
where
\begin{equation}\label{Eq: Future characteristic diagonal in N}
\Delta_{\mathcal{N}}^{\triangleright}\doteq\left\{(x,k_x,x,-k_x)\;|\; x\in\mathring{\mathcal{N}},\; (x,k_x)\in\mathrm{Char}(P),\; k_x\triangleright 0\right\},
\end{equation}	

\noindent We start by establishing an auxiliary result in the following lemma.

\begin{lemma}\label{Lem: Direct null conormal relation}
Let $(\mcM,g)$ be a globally hyperbolic spacetime with timelike boundary, abiding by Assumptions \ref{Ass: Infinitesimally convex} and \ref{Ass: Finite regular reflections}. In addition, let $\mathcal{N}$ be a causally convex Cauchy neighbourhood as per Definition \ref{Def: Cauchy neighbourhood}.	Let $\mathcal U\subset\mathring{\mathcal{N}}$ be a causally convex, geodesic, open neighbourhood and let thereon $\sigma$ be defined as per Equation \eqref{Eq: Synge World Function Interior}. Set
\begin{equation}\label{Eq: Direct null hypersurface}
\mathcal Z_{0,\mathcal U}\doteq	\left\{(x,y)\in\mathcal U\times\mathcal U\;|\;x\neq y,\; \sigma(x,y)=0\right\},\; k_{x,0}\doteq d_x\sigma(x,y)\;\mathrm{and}\; k_{y,0}\doteq-d_y\sigma(x,y),
\end{equation}
where we dropped in $k_{x,0}$ and $k_{y,0}$ the explicit dependence on $(x,y)$ for simplicity. Then $\mathcal Z_{0,\mathcal U}$ is a smooth hypersurface of $(\mathcal U\times\mathcal U)\setminus\operatorname{Diag}_2(\mathcal U)$ where $\operatorname{Diag}_2(\mathcal U)=\{(x,x) \; \vert \; x \in\mathcal{U}\}$. In addition, denoting by $t$ the underlying time function,
\begin{equation}\label{Eq: Direct branch conormal relation}
\Lambda_{0,\mathcal U}^{\triangleright}\doteq\left\{\left(x,\lambda k_{x,0},y,-\lambda k_{y,0}\right)\;|\;	(x,y)\in\mathcal Z_{0,\mathcal U},\;\lambda\bigl(t(x)-t(y)\bigr)>0\right\},
	\end{equation}
is the future-directed half of $N^*\mathcal Z_{0,\mathcal U}\setminus\{\boldsymbol{0}\}$. Furthermore,
\begin{equation}\label{Eq: Local direct branch identification}
\Lambda_{0,\mathcal U}^{\triangleright}=\mathcal{C}_{\mathrm b,\mathcal{N}}^{\triangleright,\mathrm{dir}}\cap	T^*\left((\mathcal U\times\mathcal U)\setminus\operatorname{Diag}_2(\mathcal U)\right).
\end{equation}
\end{lemma}

\begin{proof}
Since $\mathcal U$ is geodesically convex, every pair $(x,y)\in\mathcal U\times\mathcal U$ is connected by a unique geodesic $\gamma^{(x,y)}$ lying therein. For $x\neq y$, using the same endpoint variation argument as for example in the proof of item {\em 3.} of Proposition \ref{Prop: Reflected null branch atlas}, it turns out that
	$$d_x\sigma(x,y)=-\dot\gamma^{(x,y)}(0)^\flat\;\mathrm{and}\;d_y\sigma(x,y)=\dot \gamma^{(x,y)}(1)^\flat,$$
	where $\flat$ denotes the standard musical isomorphism. Hence $d\sigma$ does not vanish on $\mathcal Z_{0,\mathcal U}$ and, consequently, $0$ is a regular value of the restriction of $\sigma$ to
	$(\mathcal U\times\mathcal U)\setminus\operatorname{Diag}_2(\mathcal U)$. Using the regular level-set theorem, see \cite[Cor. 5.14]{Lee_2013}, entails that $\mathcal Z_{0,\mathcal U}$ is a smooth hypersurface therein. Its conormal bundle is
	$$N^*\mathcal Z_{0,\mathcal U}\setminus\{\boldsymbol{0}\}=\left\{\left(x,\lambda k_{x,0}, y,-\lambda k_{y,0}\right)\;|\;(x,y)\in\mathcal Z_{0,\mathcal U},\;\lambda\neq 0\right\}.$$
	As in the proof of Corollary \ref{Cor: Branchwise conormal resolution}, the condition $\lambda(t(x)-t(y))>0$ selects the component for which the covector in the first entry is future pointing. To conclude, observe that, since $\mathcal U$ is a causally convex, geodesic neighbourhood, every BNB as per Definition \ref{Def: Finite broken null bicharacteristic} with no reflection points and distinct endpoints in $\mathcal U$ has as its base curve the unique null geodesic joining them. Furthermore its endpoint covectors are proportional to $d_x\sigma$ and to $-d_y\sigma$. Conversely, every element of $\Lambda_{0,\mathcal U}^{\triangleright}$ determines a null bicharacteristic whose base curve is the unique null geodesic joining its endpoints. This is nothing but the sought Equation \eqref{Eq: Local direct branch identification}.
\end{proof}

\begin{proposition}\label{Prop: Resolution of broken null relation}
Let $(\mcM,g)$ be a globally hyperbolic spacetime with timelike boundary, abiding by Assumptions \ref{Ass: Infinitesimally convex} and \ref{Ass: Finite regular reflections}. In addition let $\mathcal{N}$ be a causally convex Cauchy neighbourhood as per Definition \ref{Def: Cauchy neighbourhood}. Using Equations \eqref{Eq: Restricted future broken relation} it holds that
\begin{equation}\label{Eq: Full broken relation resolution}
\mathcal{C}_{\mathrm b,\mathcal{N}}^{\triangleright}=\mathcal{C}_{\mathrm b,\mathcal{N}}^{\triangleright,\mathrm{dir}}\cup\Lambda_{\mathcal{N}}^{\triangleright,\mathrm{ref}}=\mathcal{C}_{\mathrm b,\mathcal{N}}^{\triangleright,\mathrm{dir}}\cup\bigcup_{(N,\alpha)\in\mathcal{I}_{\mathrm{ref}}}\Lambda_{N,\alpha}^{\triangleright},
\end{equation}
where $\Lambda_{N,\alpha}^\triangleright$ is as per Equation \eqref{Eq: Branchwise reflected conormal relation}, while $\mathcal{I}_{\mathrm{ref}} \doteq \bigsqcup_{N \ge 1} (\{N\} \times \mathcal{I}_N)$, see Corollary \ref{Cor: Countable reflected conormal resolution}.
\end{proposition}

\begin{proof}
Given $(x,k_x,y,-k_y)\in\mathcal{C}_{\mathrm b,\mathcal{N}}^{\triangleright}$, Definition \ref{Def: Finite broken null bicharacteristic} entails the existence of a BNB connecting $(x,k_x)$ to $(y,k_y)$. On account of our assumptions this has a finite number of reflection points, say $N\geq 0$. The associated curve connecting $x$ to $y$ is piecewise continuous, lightlike and its time orientation is preserved at every reflection. Hence it lies entirely in $\mathcal{N}$ since it is causally convex.
	
Assuming that $N=0$, we have two options. In the first one $x=y$, which implies $k_y=k_x$ or, equivalently,	$(x,k_x,y,-k_y)=(x,k_x,x,-k_x)\in\Delta_{\mathcal{N}}^{\triangleright}$, In the second one, $x\neq y$, which entails
$(x,k_x)\sim(y,k_y)$, or, equivalently $(x,k_x, y,-k_y)\in\mathcal{C}_{\mathrm b,\mathcal{N}}^{\triangleright,\mathrm{dir}}$.
	
Suppose instead that $N\geq1$. Up to an affine reparametrization of the underlying curve on $[0,1]$, we can identify a point $(x,v)\in\Gamma^0_{N,\mathcal{N}}$ and $\lambda\in\bR\setminus\{0\}$ such that
$$y=\mathrm{Exp}_{-,N}(x,v),\quad k_x=-\lambda v^\flat,\quad k_y=-\lambda\dot\gamma_{(x,v)}(1)^\flat.$$
Since, per hypothesis, $k_x\triangleright0$, it follows from Equation \eqref{Eq: N-reflected cotangent relation}, combined with Corollary \ref{Cor: Countable reflected conormal resolution} that
$$(x,k_x,y,-k_y)\in\Lambda_{N,\mathcal{N}}^{\triangleright}\subseteq\Lambda_{\mathcal{N}}^{\triangleright,\mathrm{ref}}=\bigcup_{(N,\alpha)\in\mathcal{I}_{\mathrm{ref}}}\Lambda_{N,\alpha}^{\triangleright}.$$
This proves that
$$\mathcal{C}_{\mathrm b,\mathcal{N}}^{\triangleright}\subseteq\mathcal{C}_{\mathrm b,\mathcal{N}}^{\triangleright,\mathrm{dir}}\cup\bigcup_{(N,\alpha)\in\mathcal{I}_{\mathrm{ref}}}\Lambda_{N,\alpha}^{\triangleright}.$$
To establish the opposite inclusion, first of all we observe that, per construction, $\mathcal{C}_{\mathrm b,\mathcal{N}}^{\triangleright,\mathrm{dir}}\subseteq\mathcal{C}_{\mathrm b,\mathcal{N}}^{\triangleright}$. Hence, let us consider $(x,k_x,y,-k_y)\in\Lambda_{\mathcal{N}}^{\triangleright,\mathrm{ref}}$. It follows that there exists $N\geq 1$ and $(x,v)\in\Gamma^0_{N,\mathcal{N}}$ and $\lambda\in\bR\setminus\{0\}$ such that
$$y=\mathrm{Exp}_{-,N}(x,v),\quad k_x=-\lambda v^\flat,\quad k_y=-\lambda\dot\gamma{(x,v)}(1)^\flat.$$
Setting $\gamma\doteq\gamma_{(x,v)}$ and $k(s)\doteq-\lambda\dot\gamma(s)^\flat$, the map $s\longmapsto\bigl(\gamma(s),k(s)\bigr)$ is, up to an affine reparameterization, an integral curve of the Hamiltonian vector field $H_{\sigma_P}$, which satisfies in addition all the hypotheses of Definition \ref{Def: Finite broken null bicharacteristic}. In other words $(x,k_x)\sim_{\mathrm{b}} (y,k_y)$ and, in addition $k_x\triangleright 0$ per definition of $\Lambda_{\mathcal{N}}^{\triangleright,\mathrm{ref}}$. We have thus shown
$$\mathcal{C}_{\mathrm b,\mathcal{N}}^{\triangleright,\mathrm{dir}}\cup\bigcup_{(N,\alpha)\in\mathcal{I}_{\mathrm{ref}}}\Lambda_{N,\alpha}^{\triangleright}\subseteq\mathcal{C}_{\mathrm b,\mathcal{N}}^{\triangleright},$$
which is the sought conclusion.
\end{proof}

\begin{remark}\label{Rem: Full Directed and Reflected Singularities}
For future convenience we highlight that Equation \eqref{Eq: Full broken relation resolution} can be generalized dropping the requirement encoded by $\triangleright$ that the first covector is future-pointing. As a matter of fact, it holds that 
\begin{equation}\label{Eq: Full broken relation resolution no future directed}
\mathcal{C}_{\mathrm b,\mathcal{N}}=\mathcal{C}_{\mathrm b,\mathcal{N}}^{\mathrm{dir}}\cup\Lambda_{\mathcal{N}}^{\mathrm{ref}}=\mathcal{C}_{\mathrm b,\mathcal{N}}^{\mathrm{dir}}\cup\bigcup_{(N,\alpha)\in\mathcal{I}_{\mathrm{ref}}}\Lambda_{N,\alpha},
\end{equation}
where $\mathcal{C}^{\mathrm{dir}}_{\mathrm b,\mathcal{N}}$ and $\Lambda_{N, \alpha}$ are defined respectively as in Equations \eqref{Eq: Directed null relation in N} and \eqref{Eq: Branchwise reflected conormal relation} dropping the requirement $k_x\triangleright 0$.
\end{remark}

Definition \ref{Def: Finite broken null bicharacteristic} combined with the no-glancing hypothesis in Equation \eqref{Eq: No Glancing} seems to suggest that in our construction we allow only for generalized broken bicharacteristics which can hit as many times as they want the boundary $\partial\mcM$. Each time they are reflected towards the interior $\mathring{\mcM}$ and, furthermore, for every bounded time interval, only a finite number of reflections can occur. Yet, we are still failing to control a limiting scenario, individuated in \cite{Melrose_1978}. In the following we employ a terminology borrowed from impact dynamics, see \cite{ClarkBloch_2023}.

\begin{definition}\label{Def: Zeno point}
Under the same assumptions of Definition \ref{Def: Finite broken null bicharacteristic}, we say that $p\in\partial\mcM$ is a {\bf Zeno (or accumulation) point} associated to the Klein-Gordon operator $P$ on $(\mcM,g)$ if there exists a sequence $\{x_j\}_{j\in\mathbb{N}}$ of points in $\partial\mcM$ such that 
\begin{enumerate}
	\item $x_j\neq x_i$ for all $j\neq i$ and, denoting by $t$ the time coordinate on $\mcM$ associated to the underlying timelike Killing field, $t(x_j)>t(x_i)$ for all $j>i$,
	\item there exists a BNB whose reflection points are $\{x_j\}_{j \in \mathbb{N}}$,
	\item there exists $p\in\partial\mcM$ such that $\lim_{j\to\infty}x_j=p$.
\end{enumerate}
\end{definition}

\noindent Observe that, if a Zeno point is present, infinitely many reflections accumulate at a finite value of the global time function, that is $t(p)$. Furthermore, recalling that the underlying BNB can be parametrized as $\Gamma(s)=(\gamma(s),k(s))$ at $p\in\partial\mcM$, it turns out that $\eta(p,k(s_p))=0$, see Equation \eqref{Eq: eta}, where $s_p\in\bR$ is such that $\gamma(s_p)=p$. This can be realized as follows: Consider a sequence $\{x_j\}_{j\in\mathbb{N}}$ of points in $\partial\mcM$ converging to a Zeno point $p\in\partial\mcM$ and let $\gamma_j$ denote the curve connecting $x_j$ to $x_{j+1}$. Being the underlying background static, without loss of generality, we parametrize the curve with the coordinate time $t$. This entails, that, as $j\to\infty$, $\tau_j \doteq t(x_{j+1})-t(x_j)\to 0$, whereas there exist $v, v_j \in T_p\mcM$ such that $v \doteq \dot{\gamma}(t(p))$ and $v_j \doteq \dot{\gamma}(t(x_j))$. Denoting by $z:\partial\mcM\times [0,\varepsilon) \to [0, \varepsilon)$, $\varepsilon>0$, a boundary defining function, it holds that $z(x_j)=0$ for all $j\in\mathbb{N}$. Hence the difference quotient $\frac{z(\gamma(\tau_j))-z(\gamma(0))}{\tau_j}$ tends to $0$ as $\tau_j\to 0$, which entails $dz(v_j)=0$. By continuity this entails that $dz(v)=0$, that is, at the Zeno point, the curve $\gamma$ is tangent to the boundary. In other words at $p$ we end up with a glancing lightlike curve and this does not contradict Equation \eqref{Eq: No Glancing} since the phenomenon occurs only in the limit, while, at each finite reflection step, the no glancing hypothesis holds true.

In the following, in order to avoid the occurrence of such accumulation points, we impose a local finiteness condition on the reflected null dynamics. Taking into account Definition \ref{Def: Finite broken null bicharacteristic}, for every
$$
\rho \doteq (x,k_x)\in\mathrm{Char}(P)\cap T^*\mathring{\mcM},
$$
let $\mathfrak{B}(\rho)$ denote the family of all finite broken null bicharacteristics whose image contains $\rho$. Given $\Gamma\in\mathfrak{B}(\rho)$, with reflection parameters
$$
s_1<\cdots<s_N,
$$
and a compact interval $I\subset\bR$, we set
$$
N_I(\Gamma)
\doteq\#
\left\{
	j\in{1,\ldots,N}\;
	\middle|\;
	t(\gamma(s_j))\in I
	\right\}.
$$

\begin{assumption}[Local finiteness of the reflected dynamics]\label{Ass: No Zeno points}
For every
$$\rho \doteq (x,k_x)\in\mathrm{Char}(P)\cap T^*\mathring{\mcM}$$
and every compact interval $I\subset\bR$, it holds that
$$\sup_{\Gamma\in\mathfrak{B}(\rho)}N_I(\Gamma)<+\infty.	$$
\end{assumption}

In other words, along the reflected null trajectory through any interior characteristic point, only finitely many boundary interactions can occur in a bounded interval of the global time coordinate. The corresponding bound is allowed to depend both on $\rho$ and on $I$. Furthermore, no uniform estimate over the characteristic set is assumed since only the lack of Zeno points is necessary. 

\begin{remark}
We stress that Assumption \ref{Ass: No Zeno points} is strictly necessary in view of our hypothesis of infinitesimal convexity of $\partial\mcM$ as per Assumption \ref{Ass: Infinitesimally convex}. On the contrary, under the slightly stronger request of strict null-convexity of $\partial\mcM$, every inextensible broken null geodesic is tame, namely its reflection parameters have no accumulation in any finite time interval  see \cite[Proposition~2.12]{HintzUhlmann}. We recall also that, restricting to this class of backgrounds, is not a viable option since we want to be able to consider half-Minkowski spacetime among the underlying manifolds.
\end{remark}

At last we can establish the sought result concerning the wavefront set of the propagators under investigation.

\begin{theorem}\label{Thm: WF of Propagators}
Let $P$ be the Klein-Gordon operator as in Equation \eqref{Eq: KG equation} on $(\mcM, g)$ a $d$-dimensional, $d\geq 2$, globally hyperbolic spacetime with a timelike boundary abiding by Assumptions \ref{Ass: Infinitesimally convex} and \ref{Ass: No Zeno points}. Then, it holds that 
\begin{equation}\label{Eq: WF of Gkappa}
\mathrm{WF}(\left.G_\kappa\right|_{\mathring{\mcM}\times\mathring{\mcM}})=\mathcal{C}_{\mathrm b}(\mathring{\mcM}),
\end{equation}
where $\left.G_\kappa\right|_{\mathring{\mcM}\times\mathring{\mcM}}$ is the restriction to $\mathring{\mcM}\times\mathring{\mcM}$ of the Pauli-Jordan propagator with Robin boundary conditions, while
\begin{equation}\label{Eq: CbM}
\mathcal{C}_{\mathrm b}(\mathring{\mcM})
\doteq
\left\{
(x,k_x,y,-k_y)\in \mathring{T}^*(\mathring{\mcM}\times\mathring{\mcM})
\;\middle|\;
(x,k_x)\sim_{\mathrm b}(y,k_y)
\right\},
\end{equation}
where the symbol $\sim_{\mathrm b}$ entails that the two points are connected by a finite broken null bicharacteristic, see Definition \ref{Def: Finite broken null bicharacteristic}. 
\end{theorem}

\begin{proof}
The proof is divided into two parts, one for each inclusion necessary to establish Equation \eqref{Eq: WF of Gkappa}. 

\vskip .4cm

\noindent{\bf Part 1: Establishing $\mathrm{WF}(\left.G_\kappa\right|_{\mathring{\mcM}\times\mathring{\mcM}})\subseteq\mathcal{C}_{\mathrm b}(\mathring{\mcM})$} -- With a slight abuse of notation, in the following, we write $G_\kappa$ in place of $\left.G_\kappa\right|_{\mathring{\mcM}\times\mathring{\mcM}}$. When propagation in the first variable up to $\partial \mcM$ is considered, we shall regard $G_\kappa$ instead as the kernel $G_\kappa \vert_{\mcM \times \mathring{\mcM}}$. We further split the analysis into four different steps.

\vskip .2cm

{\bf Step 1:} Since $P\otimes\mathbb{I}$ and $\mathbb{I}\otimes P$ are properly supported differential operators on $\mathring{\mcM}\times\mathring{\mcM}$, we can conclude that 
$$(P\otimes\mathbb{I})G_\kappa=(\mathbb{I}\otimes P)G_\kappa=0,$$
entails that $\textrm{WF}(G_\kappa)\subseteq\textrm{Char}(P\otimes\mathbb{I})\cap\textrm{Char}(\mathbb{I}\otimes P)$ or, equivalently,
$$(x,k_x,y,k_y)\in\mathrm{WF}(G_\kappa) \Longrightarrow \begin{cases}
    k_x = 0 \quad \text{or} \quad (x,k_x) \in \operatorname{Char}(P), \\
    k_y = 0 \quad \text{or} \quad (y,k_y) \in \operatorname{Char}(P),
\end{cases}$$
where either $k_x,k_y$ may \textit{a priori} vanish, although they cannot vanish simultaneously. We exclude this option in the Step 4.

\vskip .2cm

{\bf Step 2:} In this part of the proof we establish that
$$(x,k_x,y,k_y)\in\operatorname{WF}(G_\kappa) \; \text{such that} \; (x,k_x)\sim_{\mathrm b} (w,k_w)\Rightarrow (w,k_w,y,k_y)\in\operatorname{WF}(G_\kappa),$$ see Definition \ref{Def: Finite broken null bicharacteristic}.
Considering the operator $(P\otimes\mathbb{I})$, acting on $\mcM \times\mathring{\mcM}$, its principal symbol reads $\widetilde{\sigma}(x,y;k_x,k_y)=\sigma_P(x,k_x)$ and the associated Hamilton vector field is $H_{\widetilde{\sigma}}=(H_{\sigma_P},0)$.
We can thus apply \cite[Thm.~4.1]{Melrose_1978} being the hypotheses fulfilled. In the case under scrutiny, in view of Assumption \ref{Ass: No Zeno points}, we can work with the ordinary wavefront set. The principal symbol $\widetilde{\sigma}$ of the Klein-Gordon operator is of real principal type on $$\textrm{Char}_0(P\otimes\mathbb{I})=\{(x,k_x,y,k_y)\;|\; \sigma_P(x,k_x)=0\;\textrm{and}\; k_x\neq 0\}.$$ As a matter of fact, thereon, it is smooth, real-valued and, moreover, $d\widetilde{\sigma}$ and $\alpha$, the canonical $1$-form on $T^*\mcM$, are linearly independent. In addition $\partial\mcM$ is timelike
hence, non-characteristic for $P$, while the Robin boundary condition are of real Neumann type in the sense of \cite[Eq. (0.7)]{Melrose_1978}. Also the condition on the order-one principal symbol of $P-P^*$ in \cite[(0.8)]{Melrose_1978} is trivial since $P$ is formally self-adjoint. We also observe that, on account of the no-glancing hypothesis in Equation \eqref{Eq: No Glancing}, the bicharacteristic flow also intersects the boundary in hyperbolic points. More precisely let $\pi_\partial:T^*\mcM|_{\partial\mcM}\to T^*\partial\mcM$ be the standard projection map which cancels the component orthogonal to the boundary of a covector. A point $\xi\in T^*\partial\mcM$ is called hyperbolic if $\pi^{-1}_\partial(\xi)\cap\textrm{Char}(P)$ consists of two distinct points. At any such point the broken Hamiltonian flow of \cite{Melrose_1978}, here generated by $H_{\widetilde{\sigma}}$ identifies the incoming and outgoing null covectors $k^\pm_x$ having the same tangential component that is, using Equation \eqref{Eq: Dual Reflection}, $k^+_x=\mathcal R_x^*k^-_x$. In other words, we can conclude that 
$$(x,k_x)\sim_{\mathrm b}(w,k_w)\Longrightarrow(w,k_w,y,k_y)\in\operatorname{WF}(G_\kappa).$$

\vskip .3cm

{\bf Step 3:} Here we establish that, if $y\in\mathring{\mcM}$ then $(y,k^\prime_y,y,-k_y)\in\textrm{WF}(G_\kappa)$ if and only if $k^\prime_y=k_y$. To this end, we can choose a globally hyperbolic, causally convex, open neighbourhood of $y$, $\mathcal{U} \subset\mcM$ such that $\mathcal{U} \cap\partial\mcM=\emptyset$. Hence, recalling that $G_\kappa=G^-_\kappa-G^+_\kappa$ and the advanced and retarded propagators of $P$ on a globally hyperbolic spacetime are unique, we can infer that $G_\kappa|_{\mathcal{U}\times \mathcal{U}}$ must coincide with the Pauli-Jordan propagator of $P$ constructed intrinsically on $(\mathcal{U},g|_\mathcal{U})$. Thereon, it holds that \cite{Duistermaat_1972}
$$\operatorname{WF}(G_\kappa|_{\mathcal{U}\times \mathcal{U}})=\{(x,k_x,y,-k_y)\in T^*(\mathcal{U}\times \mathcal{U}) \setminus \{\boldsymbol{0}\}\;|\; (x,k_x)\sim (y,k_y)\},$$
where $\sim$ here entails that $x$ and $y$ are connected by a lightlike geodesic to which $k_x$ is tangent and $k_y$ is obtained as the parallel transport of $k_x$. By choosing $x=y$ we are forced to set $k_x=k_y$. Hence we obtain the sought statement. 

\vskip .2cm

{\bf Step 4:} To conclude, consider the Cauchy surface $\Sigma_y\doteq\{t(y)\}\times \Sigma$ of $\mcM$. Per construction $y\in\Sigma_y$. Assume that $(x,k_x,y,-k_y)\in\textrm{WF}(G_\kappa)$. Starting from $(x,k_x)$, we consider the Hamiltonian flow generated by $H_{\widetilde{\sigma}}$ and the no-glancing hypothesis combined with Assumption \ref{Ass: No Zeno points} entail that such flow must intersect $\Sigma_y$. Let $y^\prime$ denote such point and $k_{y^\prime}$ the corresponding covector. By construction $(x,k_x)\sim_{\mathrm{b}} (y^\prime,k_{y^\prime})$ and $t(y)=t(y^\prime)$, where $t$ is the time function of $(\mcM,g)$ associated to the underlying timelike Killing field. From Step 3. we also know that $(y^\prime,k_{y^\prime},y, -k_y)\in\operatorname{WF}(G_\kappa)$. Since $G_{\kappa}$ is causally supported it must hold that $y^\prime\in J^+(y)\cup J^-(y)$, but, having the same time coordinate, the only possibility is $y^\prime=y$. At this point we can use Step 3. to conclude that $(y,k_{y^\prime},y, - k_y)\in\operatorname{WF}(G_\kappa)$ if and only if $k_{y^\prime}= k_y$. As a consequence we have shown that, if $(x,k_x,y,-k_y)\in\operatorname{WF}(G_\kappa)$ then $(x,k_x)\sim_b (y,k_y)$. Observe that this is excluding from the wavefront set also points of the form $(x,k_x,y,0)$ and $(x,0,y,k_y)$ with $k_x,k_y\neq 0$. Indeed, Corollary \ref{Cor: G+- Formal Adjoints} entails that, at the level of integral kernels $G^*_\kappa(x,y)=-G_\kappa(y,x)$. As a consequence, if $(x,k_x,y,0)\in\textrm{WF}(G_\kappa)$, then also $(y,0,x,k_x)\in\textrm{WF}(G_\kappa)$. We can focus only on the first case. If $(x,k_x,y,0)\in\operatorname{WF}(G_\kappa)$, with $k_x\neq 0$, we could propagate $(x,k_x)$ to the same Cauchy surface where $y$ lies. As above, causal support of the propagators entails that $x=y$ and consequently, once more by Step 3, $k_x=0$.  This concludes the first half of the proof, namely $\operatorname{WF}(\left.G_\kappa\right|_{\mathring{\mcM}\times\mathring{\mcM}})\subseteq\mathcal{C}_{\mathrm b}(\mathring{\mcM})$.
	
\vskip .4cm

\noindent{\bf Part 2: Establishing $\mathcal{C}_{\mathrm b}(\mathring{\mcM})\subseteq\mathrm{WF}(\left.G_\kappa\right|_{\mathring{\mcM}\times\mathring{\mcM}})$} -- Consider a point $$(x,k_x,y,-k_y)\in T^*(\mathring{\mcM}\times\mathring{\mcM}) \setminus \{\boldsymbol{0}\} \quad \text{such that} \quad (x,k_x) \sim_{\mathrm b} (y,k_y)$$ and suppose that it does not lie in $\operatorname{WF}(G_\kappa)$. Yet, every point $(w,k_w)\in T^*\mathring{\mcM} \setminus \{0\}$ lying on the bicharacteristic flow generated by $H_{\sigma_P}$ is such that $(w,k_w,y,-k_y)\notin\operatorname{WF}(G_\kappa)$. If not, one could apply the reasoning in the first part of the proof claiming that, by backtracking the Hamiltonian flow, $(x,k_x,y,-k_y)$ should lie in $\operatorname{WF}(G_\kappa)$. Yet, using Assumption \ref{Ass: Infinitesimally convex}, we can follow the bicharacteristic flow up to the Cauchy surface of $\mcM$ containing $(y,-k_y)$. The same reasoning as in Step 5 above entails that $(w,k_w)=(y,k_y)$ and, in turn, that $(y,k_y,y,-k_y)\notin\operatorname{WF}(G_\kappa)$. This is not possible on account of Step 4 and, hence, it concludes the proof.
\end{proof}

\begin{corollary}\label{Cor: WF of Gpmkappa}
	Under the same assumptions of Theorem \ref{Thm: WF of Propagators}, denoting by $G^\pm_\kappa$ the advanced and retarded propagators of the Klein-Gordon operator $P$ with Robin boundary conditions, it holds that
	$$\mathrm{WF}(G^\pm_\kappa)=\mathcal{C}^\mp_\mathrm{b}(\mathring{\mcM})\cup N^*\operatorname{Diag}_2(\mathring{\mcM}),$$
	where 
$$\mathcal{C}_{\mathrm b}^{\pm}(\mathring{\mcM})\doteq\left\{(x,k_x,y,-k_y)\in\mathcal{C}_{\mathrm b}(\mathring{\mcM})\;|\; x\in J^{\mp}(y)\right\},$$
while 
$$N^*\operatorname{Diag}_2(\mathring{\mcM})=\{(x,k_x,x,-k_x)\in T^*(\mathring{\mcM}\times\mathring{\mcM})\setminus\{\boldsymbol{0}\}\}.$$
\end{corollary}

\begin{proof}
We establish the sought result for $G^-_\kappa \doteq G_\kappa^{-} \vert_{\mathring{\mcM} \times \mathring{\mcM}}$, the other following suit. Using Equation \eqref{Eq: Gpm from Gkappa} we write $G^-_\kappa=\Theta(t-t^\prime)G_k$. Hence, if $t<t^\prime$, $G^-_\kappa=0$ and its wavefront set is empty. On the contrary, if $t>t^\prime$, $G^-_\kappa=G_\kappa$ and, hence, thereon $\mathrm{WF}(G^-_\kappa)=\mathrm{WF}(G_\kappa)$. Combining this identity with the causal support of the retarded propagator, we have obtained $\mathcal{C}_{\mathrm b}^+(\mathring{\mcM})$. We are left with analysing the singular structure as $t=t^\prime$. Yet, it suffices to recall that
 $(P\otimes\mathbb{I})G^-_\kappa \vert_{\mathring{\mcM} \times \mathring{\mcM}}=\delta_{\operatorname{Diag}_2(\mathring{\mcM})}$ and, being $P\otimes\mathbb{I}$ a differential operator on $\mathring{\mcM}\times\mathring{\mcM}$ we can conclude that
 $N^*\operatorname{Diag}_2(\mathring{\mcM})\subseteq\mathrm{WF}(G^-_\kappa)$. Since $G^-_\kappa$ is causally supported, at $t=t^\prime$, only points along the diagonal lie in the singular support and thus $N^*\operatorname{Diag}_2(\mathring{\mcM})$ exhausts the whole set of possibilities. This concludes the proof.
\end{proof}

\section{Hadamard states with Robin boundary conditions}\label{Sec: Hadamard states with Robin boundary conditions}
The goal of this section is twofold. On the one hand we give a definition of Hadamard two-point correlation function for the Klein-Gordon operator with Robin boundary conditions assuming that the underlying geometry abides by the standing assumptions outlined in the previous sections. On the other hand, we individuate a local Hadamard condition that is we constrain the form of the integral kernel of an admissible two-point correlation function on every locally convex geodesic neighbourhood, both reflected and non, see Definition \ref{Def: Regular reflected branch chart}. Our ultimate goal is to prove a counterpart of the celebrated Radzikowski's theorem on globally hyperbolic spacetimes with empty boundary \cite{Radzikowski_1996, Radzikowski_1996_1}, which establishes that the equivalence between the local and microlocal notion of Hadamard states. We observe that this generalizes the content of \cite{Costeri_25} where this assertion has been established on half-Minkowski spacetime. The analysis developed in this section will be complemented in Section \ref{Sec: Existence oF Hadamard States}, where we establish the existence of a Robin Hadamard state in a large class of admissible backgrounds using a deformation argument.

\paragraph{Microlocal Formulation of the Hadamard Condition --} We address the first question, starting from a generalization of \cite[Def. 5.1]{Costeri_25}.

\begin{definition}\label{Def: Robin Hadamard two-point function}
Let $P$ be the Klein-Gordon operator as in Equation \eqref{Eq: KG equation} on $(\mcM, g)$ a $d$-dimensional, $d\geq 2$, globally hyperbolic spacetime with a timelike boundary abiding by Assumptions \ref{Ass: Infinitesimally convex} and \ref{Ass: No Zeno points}. We say that $\omega_{2, \kappa} \in\mathcal{D}^\prime(\mcM\times\mcM)$ is a \textbf{Robin-Hadamard two-point correlation function} if 
\begin{itemize}
\item[\ding{104}] $\omega_{2,\kappa}$ is such that
\begin{equation}\label{Eq: Robin eqs of motion}
\begin{cases}
(P \otimes \mathbb{I}) \omega_{2,\kappa} = 0,  \quad \text{in} \quad \mathring{\mcM} \times \mathring{\mcM}, \\
\left.((\nabla_n+\kappa \mathbb{I}) \otimes \mathbb{I}) \omega_{2,\kappa}\right|_{\partial\mcM} = 0, \, \, \kappa\in \mathbb{R},
\end{cases}        
\end{equation}
where $\nabla_{n} \doteq n^{\mu} \nabla_{\mu}$ is such that $n$ is the unit, inward-pointing, normal vector field to $\partial\mcM$, while $(\cdot)\vert_{\partial \mathcal{M}}$ denotes the pull back to $\partial\mcM$ with respect to the first entry.
\item[\ding{104}] For all $f, f' \in \mathcal{D}(\mcM)$, $\omega_{2,\kappa}$ encodes
\begin{equation}\label{Eq: positivity + CCR}
\omega_{2, \kappa} (\bar{f}, f) \ge 0 \quad [\text{\textbf{Positivity}}] \quad \text{and} \quad \omega_{2, \kappa}(f,f') - \omega_{2, \kappa}(f',f) = i G_\kappa(f,f') \quad \text{\textbf{[CCR]}},
\end{equation}
where $G_\kappa \doteq G^-_\kappa - G^+_\kappa$ is the retarded-minus-advanced Robin propagator as per Proposition \ref{Prop: Existence Robin propagators}. 
\item[\ding{104}] The restriction to $\mathring{\mcM}\times\mathring{\mcM}$ of $\omega_{2,\kappa}$, still denoted by the same symbol with a mild abuse of notation, has the following wavefront set:  
\begin{equation}\label{Eq: Hadamard WF set Robin}
\operatorname{WF}(\omega_{2,\kappa})=\mathcal{C}_{\mathrm b}^\triangleright(\mathring{\mcM})\doteq \left\{(x,k_x, y,k_y)\in\mathcal{C}_{\mathrm b}(\mathring{\mcM})\;\mathrm{and}\; k_x \triangleright 0\right\},
\end{equation}
where $\mathcal{C}_{\mathrm b}(\mathring{\mcM})$ is as per Equation \eqref{Eq: CbM}, while $\triangleright$ entails that the covector is future-pointing.
\end{itemize}
\end{definition}

\begin{remark}
We observe that, in Equation \eqref{Eq: Robin eqs of motion} both the equation of motion and the boundary condition have been only imposed on the first entry since this is automatically inherited also by the second one. As a matter of fact, working at the level of integral kernel, Definition \ref{Def: Robin Hadamard two-point function} entails that $\omega_{2,\kappa}(x,y)=\mu_{\kappa}(x,y)+\frac{i}{2}G_\kappa(x,y)$. It turns out that $G_\kappa$ codifies the antisymmetric part of the two-point correlation function, since it is on-shell and it abides by the boundary conditions on both entries, see Proposition \ref{Prop: Existence Robin propagators}. Consequently, only $\mu_\kappa$, is not \emph{a priori} determined. Yet, being symmetric, since it must abide by the equation of motion and the boundary condition in the first entry, so it must in the second one.
\end{remark}

\noindent From a microlocal viewpoint, the existence of a single Robin-Hadamard two-point correlation function provides control over the whole class of such distributions, as the following result shows. The existence of this class will be established in Section \ref{Sec: Deformation Argument}, where a Robin-Hadamard two-point correlation function is constructed by means of a deformation argument starting from the ground-state representative on a standard static spacetime. For the remainder of this section, we therefore assume that the class of Robin-Hadamard two-point correlation functions is non-empty and fix one representative, denoted by $\omega_{2,\kappa}$. The ensuing result is a straightforward, yet useful, generalization of the corresponding statement for globally hyperbolic spacetimes without boundary.

\begin{corollary}
	Let $\omega^\prime_{2,\kappa}\in\mathcal{D}^\prime(\mcM\times\mcM)$ be a Robin-Hadamard two-point correlation function for the Klein-Gordon operator $P$ as per Definition \ref{Def: Robin Hadamard two-point function}. Then, under the same assumptions of Proposition \ref{Prop: Ground 2-point function is Hadamard}, $\omega^\prime_{2,\kappa}-\omega_{2,\kappa}\in C^\infty(\mcM\times\mcM)$ where $\omega_{2,\kappa}$ is a fixed Robin two-point correlation function, see Section \ref{Sec: Deformation Argument}.
\end{corollary}

\begin{proof}
Consider any Robin two-point correlation function $\omega^\prime_{2,\kappa}\neq\omega_{2,\kappa}$ as per Definition \ref{Def: Robin Hadamard two-point function}. Equation \eqref{Eq: positivity + CCR} entails that the antisymmetric part of $\omega^\prime_{2,\kappa}$ and $\omega_{2,\kappa}$ coincide. Hence $\Lambda \doteq \omega^\prime_{2,\kappa}-\omega_{2,\kappa}$ is a symmetric distribution. Furthermore, using Equation \eqref{Eq: Hadamard WF set Robin}, it holds that 
$$\mathrm{WF}(\Lambda)\subseteq\mathrm{WF}(\omega^\prime_{2,\kappa})\cup\mathrm{WF}(\omega_{2,\kappa})=\{(x,k_x,y,k_y)\in\mathcal{C}_{\mathrm b}(\mathring{M})\;\mathrm{and}\;k_x\triangleright 0\}.$$
Yet, being $\Lambda$ symmetric, if $(x,k_x,y,k_y)\in\mathrm{WF}(\Lambda)$ then $(y,k_y,x,k_x)\in\mathrm{WF}(\Lambda)$. At the same time the condition $(x,k_x,y,k_y)\in\mathcal{C}_{\mathrm b}(\mathring{M})$ entails that $k_x\triangleright 0\Longrightarrow -k_y\triangleright 0$. Thus, $(y,k_y,x,k_x)$ cannot lie in the wavefront set of $\Lambda$ and the only option left is $\mathrm{WF}(\Lambda)=\emptyset$, which is the desired conclusion.
\end{proof}

\subsection{Local Hadamard form and Recursion Relations}\label{Sec: Local Hadamard form and Recursion Relations}

In this section we investigate the local formulation of a Robin-Hadamard two-point correlation function for a Klein-Gordon field, generalizing the standard one on globally hyperbolic spacetimes with empty boundary \cite{Kay_Wald:1991,Radzikowski_1996,Radzikowski_1996_1} and the half-Minkowski analysis of \cite{Costeri_25}. We will rely mainly on the geometric structures introduced in Section \ref{Sec: Geodesic and Reflected Geodesic Distances Claudio} and, in particular, directed propagation along null geodesics will be described exploiting that the underlying background can be realized as a submanifold of a globally hyperbolic extension $(\mcM^\prime,g^\prime)$. Observe that the latter is not necessarily required to be static as the construction outlined in Proposition \ref{Prop: Extension of Globally Hyperbolic} is highly non-unique, \textit{cf.}, Corollary \ref{Cor: Globally Hyperbolic Extension}. Reflected propagation is subordinated instead to the existence of regular reflected branch charts as per Definition \ref{Def: Regular reflected branch chart}. For clarity of the exposition, where possible, we specialize each analytic construction on a single regular one-reflected branch and only afterwards extend the procedure to an arbitrary finite reflected branch by iteration. This choice has also the additional merit of making more transparent the comparison with our benchmark, the half-Minkowski scenario, see \cite{Costeri_25}.

In the remainder of this section $(\mcM,g)$ denotes a $d$-dimensional, $d\geq2$, globally hyperbolic spacetime with timelike boundary which abides by Assumptions \ref{Ass: Infinitesimally convex}, \ref{Ass: Finite regular reflections} and \ref{Ass: No Zeno points}. 

\begin{definition}\label{Def: One-reflection Hadamard branch domain}
	Let $(\mathcal V_{1,\alpha},\Omega_{1,\alpha})$ be a regular $1$-reflected branch chart as per Definition \ref{Def: Regular reflected branch chart}. Let $\mathcal U^\prime\subset\mcM^\prime$ be a geodesically convex normal neighbourhood and set $\mathcal U\doteq\mathcal U^\prime\cap\mcM$. We call an open subset
	$$\Omega^{\mathrm H}_{1,\alpha}\subseteq\Omega_{1,\alpha}\cap(\mathring{\mathcal U}\times\mathring{\mathcal U})$$
	a {\bf one-reflection Hadamard branch domain} if, for every $(x,y)\in\mathcal Z_{1,\alpha}\cap\Omega^{\mathrm H}_{1,\alpha}$, the incoming reflected Synge function and the outgoing ordinary Synge function admit a smooth one-sided extensions at the unique reflection point as per Lemma \ref{Lem: Matching N-reflected Synge world functions}.
\end{definition}

Observe that the existence of a one-reflection Hadamard branch domain is not an additional assumption, as it is guaranteed, for every regular $1$-reflected branch chart, by Condition \emph{2.} of Assumption \ref{Ass: Finite regular reflections}.

\paragraph{Directed Hadamard parametrix --} Before delving into a detailed analysis of the structure of the modifications to the Hadamard parametrix due to the reflections at $\partial\mcM$, we exploit the existence of this global geometric structure to characterize the directed component. More precisely, given $\mathcal U^\prime\subset\mcM^\prime$, a geodesically convex normal neighbourhood, we call $\mathcal U\doteq\mathcal U^\prime\cap\mcM$. To avoid confusion we denote by $P^\prime$ the counterpart on $(\mcM', g')$ of the Klein-Gordon operator $P$ as in Equation \eqref{Eq: KG equation} and $P^\prime|_{\mathring{\mcM}}=P$. Since $(\mcM^\prime, g^\prime)$ has an empty boundary, we can rely on the theory of Hadamard parametrices in this setting, see,  {\it e.g.}, \cite{Friedlander_1975,Kay_Wald:1991,Radzikowski_1996,Radzikowski_1996_1}. Hence, calling $\sigma^\prime$ the Synge world function on $\mathcal{U}^\prime$, it holds that, on $\mathcal{U}^\prime\times\mathcal{U}^\prime$, the integral kernel of the {\em Hadamard parametrix} $\widetilde{H}\equiv\widetilde{H}^+$ reads
\begin{gather}\label{Eq: Boundaryless Parametrix}
\widetilde{H}^\pm(x,y) = \lim_{\epsilon \rightarrow 0^+}\frac{\Gamma(\frac{d}{2}-1)}{2(2\pi)^{\frac{d}{2}}}\left( \frac{\widetilde{U}(x,y)}{\sigma^\prime_{\pm\epsilon}(x,y)^{\frac{d-2}{2}}} + \delta_d\widetilde{V}(x,y) \ln \left(\frac{\sigma^\prime_{\pm\epsilon}(x,y)}{\lambda^2}\right)\right),
\end{gather}
where $\lambda\in \mathbb{R}\setminus \{0\}$ is a reference energy scale, $\Gamma$ is the Euler Gamma function, $\widetilde{\sigma}_{\pm\epsilon} (x,y) \doteq\widetilde{\sigma}(x,y) \pm i \epsilon (t(x) - t(y)) + \epsilon^2$, with $\sigma$ the Synge world function as per Equation \eqref{Eq: Synge World Function Interior}. Furthermore, $t: \mcM^\prime \to \mathbb{R}$ is the underlying time function, and
\begin{equation*}
\delta_d =
\begin{cases}
1 \, \, \text{if} \, \, d \, \, \text{is even}, \\
0\, \, \text{if} \, \, d \, \, \text{is odd}.
\end{cases}
\end{equation*}

\begin{remark}
	We highlight that we have introduced for later convenience the disambiguation between the $\pm i\epsilon$ contribution in Equation \eqref{Eq: Boundaryless Parametrix}. From a physical and microlocal viewpoint, the $+i\epsilon$-prescription is tantamount to considering future pointing covectors in the first entry of the wavefront set of $\widetilde{H}$, see Equation \eqref{Eq: Restricted future broken relation}. This accounts for the condition of requiring positive frequencies in quantum field theory, more precisely referred to as the positive spectrum condition.
\end{remark}

\noindent In the following we shall define the component of the Hadamard parametrix on $(\mcM,g)$ which codifies the singular structure of an admissible two-point correlation function as the restriction of the one on the ambient space. Observe that this procedure can be read as the counterpart at the level of states of the notion of {\em F-locality} discussed in \cite{Kay1992FLocality}.

\begin{definition}\label{Def: Directed Hadamard Parametrix}
Given $(\mathcal{M}, g)$ satisfying the aforementioned assumptions, if $\mathcal{U}$ is a causally convex, geodesic neighbourhood such that $\mathcal{U}\cap\partial\mathcal{M}=\emptyset$, we call {\bf directed (component of the) Hadamard parametrix} $H\equiv H^+ \doteq \widetilde{H}^+|_{\mathcal{U}\times\mathcal{U}}\in\mathcal{D}^\prime(\mathcal{U}\times\mathcal{U})$. Similarly we set $H^- \doteq \widetilde{H}^- |_{\mathcal{U}\times\mathcal{U}}\in\mathcal{D}^\prime(\mathcal{U}\times\mathcal{U})$.
\end{definition}

\noindent Given $H$ as per Definition \ref{Def: Directed Hadamard Parametrix}, its integral kernel descends from Equation \eqref{Eq: Boundaryless Parametrix} as
\begin{equation}\label{Eq: Directed Hadamard Parametrix}
	H^\pm(x,y) =  \lim_{\epsilon \rightarrow 0^+}\frac{\Gamma(\frac{d}{2}-1)}{2(2\pi)^{\frac{d}{2}}}\left( \frac{U(x,y)}{\sigma_{\pm\epsilon}(x,y)^{\frac{d-2}{2}}} + \delta_d V(x,y) \ln \left(\frac{\sigma_{\pm\epsilon}(x,y)}{\lambda^2}\right)\right),
\end{equation}
where $\sigma$ is the Synge world function which, by construction, is such that $\sigma=\widetilde{\sigma}|_{\mathcal{U}\times\mathcal{U}}$. Furthermore it also holds that $U=\widetilde{U}|_{\mathcal{U}\times\mathcal{U}}$ and $V=\widetilde{V}|_{\mathcal{U}\times\mathcal{U}}$.

Yet, as already mentioned several times in this paper and also in view of the analysis in \cite{Costeri_25}, in presence of a manifold with a non-empty timelike boundary, we must expect that the underlying parametrix catches the existence of a singularity when two-points $x$ and $y$ are connected by a broken null geodesic, see Remark \ref{Rem: BNG}. This enforces the existence of a new contribution to the Hadamard parametrix for every reflected null branch. The following definition serves the purpose to codify this observation starting from the case of just one reflection. 

\begin{definition}\label{Def: Local Hadamard State}
	We say that $\omega_{2,\kappa}\in\mathcal{D}^\prime(\mcM\times\mcM)$ is a {\bf two-point function of local Robin-Hadamard form} associated to the Klein-Gordon operator $P$ as per Equation \eqref{Eq: KG equation}, if, in addition to Equations \eqref{Eq: Robin eqs of motion} and \eqref{Eq: positivity + CCR}, it satisfies the following conditions: 
	\begin{itemize}
		\item[\ding{104}] if $\mathcal{U}$ is a causally convex, geodesic neighbourhood such that $\mathcal{U}\cap\partial\mathcal{M}=\emptyset$, there exists $W_\kappa\in C^\infty(\mathcal{U}\times\mathcal{U})$ such that the integral kernel of $\omega_{2,\kappa}$ restricted to $\mathcal{U}\times\mathcal{U}$ reads
		\begin{equation*}
			\omega_{2, \kappa} \vert_{\mathcal{U} \times \mathcal{U}}(x,y)  \doteq H(x,y) + W_\kappa(x,y), 
		\end{equation*}
		where $H \equiv H^+$ is the directed Hadamard parametrix as per Equation \eqref{Eq: Directed Hadamard Parametrix}.
		
		\item[\ding{104}] if $\Omega^{\mathrm H}_{1,\alpha}$ is a one-reflection Hadamard branch domain as per Definition \ref{Def: One-reflection Hadamard branch domain}, there exist $U^\prime,V^\prime,W^\prime_\kappa\in C^\infty(\Omega^{\mathrm H}_{1,\alpha})$ such that
		\begin{gather}
			\omega_{2,\kappa}|_{\Omega^{\mathrm H}_{1,\alpha}}(x,y)=H|_{\Omega^{\mathrm H}_{1,\alpha}}(x,y)+H^{\mathrm{ref}}_{\kappa;1,\alpha}(x,y)+W^\prime_\kappa(x,y),\quad H^{\mathrm{ref}}_{\kappa;1,\alpha}\equiv H^{\mathrm{ref},+}_{\kappa;1,\alpha}\notag\\
			H^{\mathrm{ref},\pm}_{\kappa;1,\alpha}(x,y)=\lim_{\epsilon\to 0^+}\frac{\Gamma(\frac d2-1)}{2(2\pi)^{\frac d2}}\left(\frac{U^\prime(x,y)}{\sigma_{-,1,\alpha,\pm\epsilon}(x,y)^{\frac{d-2}{2}}}+\delta_d V^\prime(x,y)\ln\left(\frac{\sigma_{-,1,\alpha,\pm\epsilon}(x,y)}{\lambda^2}\right)\right),\label{Eq: Reflected Hadamard Parametrix}
		\end{gather}
		where $H \equiv H^+$ is still as per Equation \eqref{Eq: Directed Hadamard Parametrix}. In addition
		$$\sigma_{-,1,\alpha,\pm\epsilon}(x,y)\doteq\sigma_{-,1,\alpha}(x,y)\pm i\epsilon\bigl(t(x)-t(y)\bigr)+\epsilon^2,$$
		where $\sigma_{-,1,\alpha}$ is the branchwise reflected Synge world function of Definition \ref{Def: Reflected Synge's world function}. By convention we set $\sigma_{-,1,\alpha,+\epsilon}=\sigma_{-,1,\alpha,\epsilon}$. We refer to $H^{\mathrm{ref}}_{\kappa;1,\alpha}$ as the {\bf reflected component of the Hadamard parametrix} on the branch $(1,\alpha)$. 
	\end{itemize}
\end{definition}

\noindent With a slight abuse of nomenclature, we follow \cite{Costeri_25} and use the terminology {\bf (local) Robin Hadamard parametrix}. On an interior geodesically convex neighbourhood this denotes the directed kernel $H$, while, on a one-reflection Hadamard branch domain, it denotes $H+H^{\mathrm{ref}}_{\kappa;1,\alpha}$. In the following we shall also discuss the case with arbitrary reflections adopting a similar nomenclature. Distinguishing between all possible scenarios would lead to introducing a complicated terminology on top of an already rather daunting notation.

In the following we investigate how to determine the functions $U,V$ and  $U^\prime, V^\prime$ as power series in $\sigma$ and $\sigma_{-,1,\alpha}$ respectively. The coefficients satisfy suitable recursion relations which are determined as solution of suitable transport equations. Following the same spirit of \cite{Costeri_25} we discuss separately even and odd dimensions, starting from the former.  

\subsubsection{Hadamard Recursion Relations: Even Dimensional Case}\label{Sec: Hadamard Recurstion Relations even case}
In this subsection we assume that the dimension $d$ of underlying manifold $\mcM$ is even. We work on a fixed one-reflection Hadamard branch domain $\Omega^{\mathrm H}_{1,\alpha}$ and, in order to simplify the notation, throughout  we suppress the branch label $(1,\alpha)$ and we set
$$
\sigma_-\doteq\sigma_{-,1,\alpha}.
$$
In this section the $\epsilon$-prescription in the definition of the directed or reflected component of the Hadamard parametrix plays no role and hence we shall not disambiguate between $\pm i\epsilon$. Mirroring the standard construction of the Hadamard recursion relations on globally hyperbolic spacetimes without boundary, see, {\it e.g.}, \cite{Garabedian_1964, Decanini:2005eg}, we consider on $\Omega^{\mathrm H}_{1,\alpha}$, the following expansions
\begin{equation} \label{Eq: expansion of U,V,U',V'}
    \begin{cases}
    U(x,y) = \sum_{j=0}^{\frac{d-4}{2}} u_j (\underline{x},z, \underline{x}^\prime,z^\prime) \sigma^j \\
        V(x,y) \sim \sum_{k=0}^\infty v_k (\underline{x},z,\underline{x}^\prime,z^\prime) \sigma^k \\
        U^\prime(x,y) = \sum_{j=0}^{\frac{d-4}{2}} u^\prime_j (\underline{x},z, \underline{x}^\prime,z^\prime) \sigma_-^j \\
        V^\prime(x,y) \sim \sum_{k=0}^\infty v^\prime_k (\underline{x},z,\underline{x}^\prime,z^\prime) \sigma_-^k
    \end{cases},
\end{equation}
where $\sim$ denotes an asymptotic expansion. The sum both of $U$ and $U'$ runs over a finite number of indices, since higher terms would yield a smooth contribution in Equation \eqref{Eq: Directed Hadamard Parametrix} which could be reabsorbed in the definition of $W_\kappa$ or of $W^\prime_\kappa$. If $d=2$, we set $U=U^\prime=0$.

\paragraph{Directed Hadamard coefficients --} In view of our definition of the directed Hadamard parametrix, the coefficients $u_j$, $j=0,\ldots,\frac{d-4}{2}$, and $v_k$, $k\in\mathbb N\cup\{0\}$, appearing in Equation \eqref{Eq: expansion of U,V,U',V'} are the restrictions of the Hadamard coefficients of $H$ on a geodesically convex normal neighbourhood $\mathcal{U}^\prime$ such that $\mathcal{U} \doteq \mathcal{U}^\prime\cap\mcM\neq\emptyset$. 

More precisely, we obtain a family of transport equations with prescribed initial data. These can be expressed in terms of the coinciding point limits of the underlying coefficients, namely $[u_j]$ and $[v_j]$, where we recall that 
\begin{equation}\label{Eq: Coinding point limit}
		[\cdot]:C^\infty(\mathcal{U}\times\mathcal{U})\to\mathcal{U},\quad F(x,y)\mapsto [F](x) \doteq F(x,x).
\end{equation}
In turn, since the equation for each of the coefficients $u_j, v_k$ depends on the preceding term $u_{j-1},v_{k-1}$ hence forming a set of recursion relations, consistency entails that only $[u_0]$ is free to be chosen. Covariance and physical motivations require to select its value so to be the same as in Minkowski spacetime, that is $[u_0]=1$. The outcome is known as the {\em (directed) Hadamard recursion relations}.  

In the case under scrutiny, these are identical to the corresponding ones known for the Hadamard parametrix on globally hyperbolic spacetimes with empty boundary since the directed component is not sensible to the presence of the boundary, contrary to the reflected counterpart. For this reason we do not derive them leaving an interested reader to the literature for further details, see, {\it e.g.}, \cite{Decanini:2005eg, Garabedian_1964}. We limit ourselves to reporting them for completeness:

\begin{flalign} 
	&\begin{cases}\label{Eq: relations for U even dimension}
		2\sigma^\mu \partial_\mu u_0 + (\sigma^{\mu}{}_{\mu}-d) u_0 = 0, \\
		[u_0] = 1, \\\\
		Pu_j + (2j+4-d)\sigma^\mu \partial_\mu u_{j+1} + \big[(j+1)(2j+\sigma^\mu{}_{\mu} + 4-2d) + \left(1-\frac{d}{2}\right) (\sigma^{\mu}{}_{\mu} -d)\big] u_{j+1} = 0, \\
		[u_{j+1}] = -\frac{[Pu_j]}{2j+4-d}, \hspace{0.2cm} 0\le j \le \frac{d}{2}-3, 
	\end{cases}&
\end{flalign}    

\begin{flalign} \label{Eq: relations for V even dimension}  
	&\begin{cases}
		Pu_{\frac{d}{2}-2} + 2\sigma^\mu \partial_\mu v_0 + (\sigma^\mu{}_{\mu}-2)v_0 = 0, \\
		[v_0] = -\frac{[Pu_{\frac{d}{2}-2}]}{d-2}, \\\\
		Pv_k + 2(k+1) \sigma^\mu \partial_\mu v_{k+1} + (k+1)(\sigma^\mu{}_{\mu} +2k) v_{k+1} = 0, \\
		[v_{k+1}] = - \frac{[Pv_k]}{(k+1)(d+2k)}, \hspace{0.2cm} k \in \mathbb{N}\cup\{0\}.
	\end{cases}&
\end{flalign}

Observe that, if $\dim\mcM=d=2$, there is no contribution from the $U$-component of $H$ and, hence, only the $V$-term survives. In this case one must replace the initial condition in Equation \eqref{Eq: relations for V even dimension} with $[v_0]=1$. In the following corollary we summarize two remarkable properties of the Hadamard coefficients of the directed parametrix, the first proven in \cite{Moretti_2000}.

\begin{corollary}\label{Cor: Symmetry Directed Hadamard Coefficients}
Let $\mathcal{U}^\prime$ be a geodesically convex normal neighbourhood of $\mcM'$ and let $\mathcal{U} \doteq \mathcal{U}^\prime\cap\mcM\neq\emptyset$. Let $\{u_j\}_{j=0,\dots,\frac{d-4}{2}}$, $\{v_k\}_{k\in\mathbb{N}\cup\{0\}}$ be the Hadamard recursion coefficients as per Equations \eqref{Eq: relations for U even dimension} and \eqref{Eq: relations for V even dimension}. Then they are smooth on $\mathcal{U}\times\mathcal{U}$ and they are symmetric, that is 
		$$u_j(x,y)=u_j(y,x)\quad\mathrm{and}\quad v_k(x,y)=v_k(y,x).$$
\end{corollary}

\noindent Observe that this last corollary applies also to neighbourhoods with a non-empty intersection with $\partial\mcM$ since these properties are all inherited from the corresponding ones in the ambient spacetime $(\mcM^\prime,g^\prime)$.

\begin{remark}\label{Rem: Convergence of U and V Series}
	As already mentioned in \cite{Costeri_25} in the analysis of half-Minkowski spacetime, the series in Equation \eqref{Eq: expansion of U,V,U',V'} are asymptotic. Focusing on $U$ and $V$, if the underlying manifold is analytic and if $\mathcal{U}$ is the restriction to $\mcM^\prime$ of a causally convex, geodesic neighbourhood $\mathcal{U}^\prime$, then one can establish convergence to an analytic function \cite[Thm. 4.3.1]{Friedlander_1975}. 
\end{remark}

For the reflected coefficients we cannot draw automatically the same conclusions of Remark \ref{Rem: Convergence of U and V Series}. Therefore, in the following, we establish an ancillary result which will allow us to conclude that the expansion for $V^\prime$ in Equation \eqref{Eq: expansion of U,V,U',V'} is asymptotic. Observe that the series for $U^\prime$ consists of finitely many terms, hence well-posedness is guaranteed \emph{a priori}. We also stress that the proof of the following lemma is an adaptation to the case in hand of a Borel summation argument as discussed for example in \cite[Thm. 1.2.6]{Hormander_1990}.

\begin{lemma}\label{Lem: Reflected asymptotic summation}
Let $\Omega^{\mathrm H}_{1,\alpha}$ be a one-reflection Hadamard branch domain as per Definition \ref{Def: One-reflection Hadamard branch domain} and let $\mathcal Z_{1,\alpha}$ be as per Equation \eqref{Eq: Z} setting $N=1$. Fix $(x_0,y_0)\in\mathcal Z_{1,\alpha}\cap\Omega^{\mathrm H}_{1,\alpha}$ and denoting by $p\in\partial\mcM$ the corresponding reflection point, there exists
\begin{enumerate}
\item $\mathcal{K}_0\Subset\Omega^{\mathrm H}_{1,\alpha}$ a relatively compact open neighbourhood such that $(x_0,y_0) \in \mathcal{K}_0$ and $d\sigma_{-,1,\alpha}\neq 0$ thereon. 
\item $\mathcal{K}_\partial\subset\mcM\times\mathring{\mcM}$, a relatively compact open neighbourhood containing $(p,y_0)$
and in turn contained in the open neighbourhood established by item {\em 2.} of Assumption \ref{Ass: Finite regular reflections}.
\end{enumerate}
Denoting by $\mathcal{K}$ either $\mathcal{K}_0$ or $\mathcal{K}_\partial$, for any $\{a_k\}_{k\in\mathbb N\cup\{0\}}\subset C^\infty(\mathcal{K})$, there exists $A\in C^\infty(\mathcal{K})$ such that, for every $M\geq0$,
$$A-\sum_{k=0}^{M}a_k\sigma_{-,1,\alpha}^k=\sigma_{-,1,\alpha}^{M+1}R_M,$$
where $R_M\in C^\infty(\mathcal{K})$.
\end{lemma}

\begin{proof}
Since Proposition \ref{Prop: Reflected null branch atlas} entails that $d\sigma_{-,1,\alpha}\neq0$ on $\mathcal{Z}_{1,\alpha}$, we can choose a relatively compact open subset $\mathcal{K}_0\Subset\Omega^{\mathrm H}_{1,\alpha}$ containing $(x_0,y_0)$ and such that, thereon, $d\sigma_{-,1,\alpha}\neq 0$. Focusing on the reflection point $p\in\partial\mcM$, item {\em 2.} of Assumption \ref{Ass: Finite regular reflections} entails that we can choose $\mathcal{K}_\partial$ in such a way that $\sigma_{-,1,\alpha}$ is smooth thereon. Items {\em 1.} and {\em 2.} of Lemma
\ref{Lem: Matching N-reflected Synge world functions} combined with Assumption \ref{Ass: Infinitesimally convex} yield
$$\sigma_{-,1,\alpha}(p,y_0)=\sigma(p,y_0)=0\quad\textrm{and}\quad\partial_z\sigma_{-,1,\alpha}(p,y_0)=-\partial_z\sigma(p,y_0)\neq 0,$$
where $z$ is a boundary function such that $\partial_z$ is the inward pointing vector at $p$ normal to $\partial\mcM$. Hence, we can choose $\mathcal{K}_\partial$ so that $d\sigma_{-,1,\alpha}\neq 0$ thereon.

Let now $\mathcal{K}$ denote either $\mathcal{K}_0$ or $\mathcal{K}_\partial$. In the latter case, we implicitly employ Gaussian normal coordinates centered at $p\in\partial\mcM$ and we extend both $\sigma_{-,1,\alpha}$ and each $a_k$, $k\in\mathbb{N}\cup\{0\}$, smoothly across $z=0$ hence working in a suitable neighbourhood of $\mcM^\prime\times\mcM$. In either case, on $\mathcal{K}$ we can promote $r \doteq \sigma_{-,1,\alpha}$ to be a local coordinate and $\mathcal{K}\cap\mathcal{Z}_{1,\alpha}$ corresponds to the locus $r=0$. Denoting the remaining coordinates collectively by $\zeta=(\zeta^1,\dots,\zeta^{2d-1})$, $d=\dim\mcM$, for every $j\geq0$ we can set $b_j(\zeta)\doteq\left.\partial_r^j\left(\sum_{k=0}^{j}a_k(r,\zeta)r^k\right)\right|_{r=0}.$
After multiplying by a cut-off $\chi(\zeta)$ equal to one on a smaller relatively compact neighbourhood of $(x_0,y_0)$ we are exactly in the setting of Theorem \cite[Thm.~1.2.6]{Hormander_1990}. This entails the existence of a smooth function
$A(r,\zeta)$ such that
$$\left.\partial_r^jA(r,\zeta)\right|_{r=0}=b_j(\zeta),\qquad j\geq0.$$
For every $M \geq 0$ 
$$S_M(r,\zeta)\doteq\sum_{k=0}^{M}a_k(r,\zeta)r^k\Longrightarrow\left.\partial_r^j(A-S_M)\right|_{r=0}=0,
\qquad 0\leq j\leq M.$$
Using Taylor's formula we can infer the existence of $R_M\in C^\infty(\mathcal{K})$ such that
$$A-S_M=r^{M+1}R_M.$$
This is the sought conclusion.
\end{proof}

\paragraph{Expansion of the Reflected Hadamard Parametrix --}  We now focus on $H^{\mathrm{ref}}_{\kappa;1,\alpha}$ in Equation \eqref{Eq: Reflected Hadamard Parametrix} and we establish transport equations for the coefficients of $U^\prime,V^\prime$ in Equation \eqref{Eq: expansion of U,V,U',V'}. The main difference from the preceding case is that initial conditions will be replaced by compatibility constraints with the coefficients of $U$ and $V$, which codify the Robin boundary conditions. As already specified, we work on a one-reflection Hadamard branch domain $\Omega^{\mathrm H}_{1,\alpha}$ and we suppress the branch labels from $\sigma_{-,1,\alpha}$ and from the reflected coefficients if no confusion can arise. Near the reflection point we choose Gaussian normal coordinates, with $z=0$ on $\partial\mcM$ and $\partial_z$ coinciding with the inward pointing, normalized and normal vector field. Equation \eqref{Eq: Derivative operators with sigma, sigma-} entails that

\begin{flalign}\label{Eq: PHR}
	\notag PH^{\mathrm{ref}}_\kappa &=  \left(1- \frac{d}{2} \right) \big[ 2 (\sigma_-)^\mu \partial_\mu u'_0 + ((\sigma_-)^{\mu}{}_{\mu} -d) u'_0] \sigma_-^{-\frac{d}{2}} \\
	&+ \notag \sum_{j=0}^{\frac{d}{2} - 3}  \big[Pu'_j +(4+2j-d) (\sigma_-)^{\mu}\partial_{\mu} u'_{j+1} + (j+1)(2j + (\sigma_-)^{\mu}{}_{\mu} + 4 - 2d) u'_{j+1} \\
	& \notag + \left(1-\frac{d}{2} \right) ((\sigma_-)^{\mu}{}_{\mu} -d) u'_{j+1}\big] \sigma_-^{j-\frac{d}{2} + 1} + \notag \big[ Pu'_{\frac{d}{2}-2} + 2 \sigma_-^\mu \partial_\mu v'_0 + ((\sigma_-)^{\mu}{}_{\mu} -2) v'_0\big] \sigma_-^{-1} \\ &+  \sum_{k=0}^{\infty} \big[Pv'_k + 2 (k+1) (\sigma_-)^\mu \partial_\mu v'_{k+1}+ (k+1)(2k + (\sigma_-)^{\mu}_{\mu}) v'_{k+1} \big] \sigma_-^k \ln(\sigma_-)
\end{flalign}
In order for $PH^{\mathrm{ref}}_\kappa$ to lie in $C^\infty(\Omega^{\mathrm H}_{1,\alpha})$, one must require that the coefficients of the singular terms in Equation \eqref{Eq: PHR} vanish. Yet, one cannot supplement the ensuing transport equations with initial conditions in terms of coinciding point limits, rather one needs to assign the behaviour of the coefficients at $\partial\mcM$. To this end, we evaluate 

\begin{flalign} 
\label{Eq: derivative BC}
 \notag       \partial_z(H+H^{ref}_\kappa) & \big|_{z=0} \simeq \left(1- \frac{d}{2}\right) (u_0 - u'_0) (\partial_z \sigma) \sigma^{-\frac{d}{2}} \bigg\vert_{z=0} \\ &+ \notag \sum_{j=0}^{\frac{d}{2} - 3} \big[\partial_z(u_j+ u'_j) + \frac{1}{2} (2j+4-d)(u_{j+1} - u'_{j+1}) (\partial_z \sigma)\big] \sigma^{j-\frac{d}{2} + 1} \bigg\vert_{z=0} \\&+ \notag \big[\partial_z (u_{\frac{d}{2} -2} - u'_{\frac{d}{2} -2}) + (v_0 - v'_0) (\partial_z \sigma) \big] \sigma^{-1} \bigg\vert_{z=0} \\ &+ \sum_{k=0}^{\infty}   \big[ \partial_z(v_k + v'_k) + (k+1) (v_{k+1} - v'_{k+1}) (\partial_z \sigma)\big] \sigma^k \ln(\sigma) \bigg\vert_{z=0}, 
\end{flalign}
where $\simeq$ denotes that we are omitting on the right hand side all smooth contributions and
\begin{flalign} 
\label{Eq: k term in BC}        \kappa(H+H^{\mathrm{ref}}_\kappa) \big|_{z=0} &= \kappa \sum_{j=0}^{\frac{d}{2} -2} (u_j + u_j') \sigma^{j - \frac{d}{2} + 1} \bigg\vert_{z=0} + \kappa \sum_{k=0}^{\infty} (v_k + v'_k) \sigma^k \ln(\sigma) \bigg\vert_{z=0}.
\end{flalign}
Here we exploited that, on account of Lemma \ref{Lem: Properties of Reflected Geodesic Distance}, $\sigma \vert_{z=0} = \sigma_- \vert_{z=0}$ and $\partial_z \sigma_- \vert_{z=0} = - \partial_z \sigma \vert_{z=0}$.  To derive the boundary conditions satisfied by the coefficients $\{u_j'\}_{j=0}^{\frac{d-4}{2}}$ and $\{v_k'\}_{k=0}^{\infty}$, we compare the same orders in powers of $\sigma_-$. 

\noindent In view of these considerations we can infer that the Hadamard recursion relations for the coefficients $U^\prime,V^\prime$. The zeroth-order reads 
\begin{equation}\label{Eq: recursion order 0}
	\begin{cases}
	2 \sigma_-^\mu \partial_{\mu} u'_0 + ((\sigma_-)^{\mu}{}_{\mu} -d) u'_0= 0, \\
	u'_0 |_{z=0} = u_0 |_{z=0}.
	\end{cases}
\end{equation}

\begin{remark}\label{Rem: solution of u'0}
Observe that, while Equation \eqref{Eq: recursion order 0} and the following Equation \eqref{Eq: recursion for U', V'} are structurally identically to their counterpart in half-Minkowski spacetime studied in \cite{Costeri_25}, one can expect marked differences at the level of solutions. In particular, on $(\bH^d,\eta)$ it turned out that $u_0=u^\prime_0$ if $d>2$. On a generic background among those admitted by our assumptions, this is no longer guaranteed, as one can realize by a close inspection of Equation \eqref{Eq: relations for U even dimension} combined with Lemma \ref{Lem: Properties of Reflected Geodesic Distance}. As a matter of fact it turns out that
$$(\sigma_-)^\mu{}_\mu=\sigma^\mu{}_\mu-2K\nabla_z\sigma-2K^{\mu\nu}\frac{\nabla_\mu\sigma\nabla_\nu\sigma}{\nabla_z\sigma},$$
where the term at the denominator cannot vanish due to the no-glancing hypothesis. Here $K_{\mu\nu}$ is the extrinsic curvature at $\partial\mcM$ and $K$ its trace. Yet an explicit solution of Equation \eqref{Eq: recursion order 0} can be written in terms of the {\em reflected van-Vleck Morette determinant}, namely 
\begin{equation}\label{Eq: Reflected Van-Vleck-Morette Determinant}
	u^\prime_0(x,y)=\Delta^\frac{1}{2}(x,y),\quad\Delta(x,y) \doteq \frac{|\det(-\nabla_\mu\nabla_{\nu^\prime}\sigma_-(x,y))|} {|g_x|^{\frac{1}{2}}|g_y|^{\frac{1}{2}}},
\end{equation}
where $|g_x|,|g_y|$ are the absolute values of the determinant of the metric at the endpoints $(x,y)\in\Omega^{\mathrm H}_{1,\alpha}$, while $\nabla_{\nu^\prime}$ denotes a derivative in the second entry of $\sigma_-$. First of all, we observe that, in view of the regularity of the underlying $1$-reflected branch chart, namely Definitions \ref{Def: Regular reflected datum} and \ref{Def: Regular reflected branch chart}, together with the endpoint variation identity in Lemma \ref{Lem: Properties of Reflected Geodesic Distance}, the mixed Hessian $\nabla_\mu\nabla_{\nu^\prime}\sigma_-$ is non-degenerate. Secondly, a slight generalization of \cite[Sec. 7.2]{Poisson_2011}, entails that, since $\sigma_-$ satisfies the eikonal equation \eqref{Eq: Eikonal Equation for sigma_-}, then
$$d=(\sigma_-)^\mu{}_\mu+\sigma^\mu\ln(u^\prime_0)_\mu.$$
Hence one can reapply verbatim the derivation in \cite[Section 14.2]{Poisson_2011} to infer that 
$$2 \sigma_-^\mu \partial_{\mu} u'_0 + ((\sigma_-)^{\mu}{}_{\mu} -d) u'_0= 0.$$
In order to verify the boundary condition for $u^\prime_0$ in Equation \eqref{Eq: recursion for U', V'}, we fix $x\in\partial\mcM$ and $y\in\mathring{\mcM}$. Using item {\em 3b.} in Lemma \ref{Lem: Properties of Reflected Geodesic Distance}, we can infer that 
$$|\det(-\nabla_\mu\nabla_{\nu^\prime}\sigma_-)|=|\det(-\nabla_\mu\nabla_{\nu^\prime}\sigma)|.$$
The absolute value cancels the sign switch in the $z$-derivative of the Synge world function and of its reflected counterpart. Since $u_0$ has the same functional form of $u^\prime_0$ with $\sigma_-$ replaced by $\sigma$ this entails the sought relation, namely $\left.u^\prime_0\right|_{z=0}=\left.u_0\right|_{z=0}$. 
\end{remark}

\noindent The higher order coefficients are:

\begin{equation} 
\label{Eq: recursion for U', V'}
    \begin{cases}
        Pu'_j + (2j+4-d) \sigma_-^\mu \partial_\mu u'_{j+1} +\big[ (j+1)(2j+4 + (\sigma_-)^\mu{}_{\mu} -2d) + \left(1 - \frac{d}{2}\right)((\sigma_-)^{\mu}{}_{\mu} -d) \big] u'_{j+1} = 0, \\
        (\partial_z+\kappa) (u_j+u'_j) |_{z=0} + \frac{1}{2}(2j+4-d)(\partial_z \sigma) (u_{j+1}-u'_{j+1})|_{z=0} = 0 ,\hspace{0.2cm} 0 \le j \le \frac{d}{2}-3, \\\\
        Pu'_{\frac{d}{2}-2} + 2\sigma_-^\mu \partial_\mu v'_0 + ((\sigma_-)^\mu{}_{\mu}-2) v'_0 = 0, \\
        (\partial_z+\kappa) (u_{\frac{d}{2}-2} + u'_{\frac{d}{2}-2}) |_{z=0} + (\partial_z \sigma) (v_0-v'_0) |_{z=0} = 0, \\\\
        Pv'_k + 2(k+1) \sigma_-^\mu \partial_\mu v'_{k+1} + (k+1)((\sigma_-)^\mu{}_{\mu} + 2k)v'_{k+1} = 0, \\
        (\partial_z+\kappa)(v_k+v^\prime_k)|_{z=0}+(k+1)\partial_z\sigma(v_{k+1}-v'_{k+1})|_{z=0}=0,\hspace{0.2cm} k\in \mathbb{N}\cup\{0\}.
    \end{cases}
\end{equation}

Having established the Hadamard recursion relations for the reflected branch, we can combine them with Lemma \ref{Lem: Reflected asymptotic summation} to conclude the existence of the functions $U^\prime$ and $V^\prime$ in Equation \eqref{Eq: Reflected Hadamard Parametrix}.

\begin{proposition}\label{Prop: Existence one-reflection reflected coefficients even}
Given $\dim\mcM=d\geq 4$ even, setting $\Omega^{\mathrm H}_{1,\alpha}$ as per Definition \ref{Def: One-reflection Hadamard branch domain} and fixing $(x_0,y_0)\in\mathcal{Z}_{1,\alpha}\cap\Omega^{\mathrm H}_{1,\alpha}$ Equations \eqref{Eq: Reflected Van-Vleck-Morette Determinant} and \eqref{Eq: recursion for U', V'} admit unique smooth solutions 
$$\{u'_j\}_{j=0,\ldots,\frac{d-4}{2}},\qquad\{v'_k\}_{k\in\mathbb N\cup\{0\}}$$ on a suitable relatively compact open neighbourhood $\mathcal{K}\subset\Omega^{\mathrm{H}}_{1,\alpha}$. Hence there exists $U^\prime,V^\prime\in C^\infty(\mathcal{K})$ such that
$$U'\doteq\sum_{j=0}^{\frac{d-4}{2}}u'_j\sigma_-^j\quad\textrm{and}\quad V'\sim\sum_{k=0}^{\infty}v'_k\sigma_-^k,$$
where $\sim$ denotes an asymptotic expansion. If $d=2$, the same conclusion holds for $V'$.
\end{proposition}

\begin{proof}
The proof is logically divided into two separate steps, existence of the coefficients with prescribed properties and identification of the functions $U^\prime, V^\prime$.

\vskip .2cm

{\bf Step 1:} Given $p\in\partial\mcM$ the reflection point associated with $(x_0,y_0)$ in agreement with Definition \ref{Def: One-reflection Hadamard branch domain} and $z$ a boundary function such that the inward pointing normal vector at $p$ is $\partial_z$, item {\em 2.} of Lemma \ref{Lem: Matching N-reflected Synge world functions} entails that 
$\partial_z\sigma(p,y)\neq 0$ and $\partial_z\sigma_-(p,y)=-\partial_z\sigma(p,y)\neq 0$. By continuity and considering the open subset $\mathcal{K}_\partial\subset\mcM\times\mathring{\mcM}$ including the point $(p,y_0)$ chosen as in the proof of Lemma \ref{Lem: Reflected asymptotic summation}, these inequalities hold true also on $\mathcal{K}_\partial$.

This entails that the vector field $X_-\doteq\sigma_-^\mu\partial_\mu$ in Equation \eqref{Eq: recursion for U', V'} is non-characteristic at the boundary corresponding to $z=0$, since $\sigma^z\neq 0$. In other words Equation \eqref{Eq: recursion for U', V'} is a collection of first-order equations with boundary data determined by the Robin matching conditions.
Considering $d\geq 4$ and recalling that Equation \eqref{Eq: recursion order 0} determines  $u^\prime_0$ we proceed by induction. Suppose that $u^\prime_j$ has been constructed for $0\leq j\leq\frac{d}{2}-3$. The corresponding Robin matching condition in Equation \eqref{Eq: recursion for U', V'} yields
$$u^\prime_{j+1}|_{z=0}=u_{j+1}|_{z=0}+\frac{2(\partial_z+\kappa)(u_j+u^\prime_j)|_{z=0}}{(2j+4-d)\partial_z\sigma|_{z=0}}.$$
Since the denominator does not vanish, the transport equation for $u^\prime_{j+1}$ takes the form
$$(2j+4-d)\sigma_-^\mu\partial_\mu(u^\prime_{j+1})+a_j u'_{j+1}=-Pu'_j,$$
where $a_j$ is the smooth coefficient in Equation \eqref{Eq: recursion for U', V'}. Since $\sigma_-^\mu\partial_\mu$ is non-characteristic at $z=0$, this transport equation admits a unique smooth solution with the prescribed boundary value, see, {\it e.g.}, \cite[Chap 1. Sec. 5]{FritzJohn}. Without loss of generality we can choose the domain of smoothness as $\mathcal{K}$, the same set in the proof of Lemma \ref{Lem: Reflected asymptotic summation}. In particular this entails that the iterative procedure yields all coefficients $u'_0,\ldots,u'_{\frac{d}{2}-2}$ and they are smooth in the first entry up to the boundary located at $z=0$. Switching to $v^\prime_0$, the matching boundary condition yields
$$v^\prime_0|_{z=0}=v_0|_{z=0}+\frac{(\partial_z+\kappa)(u_{\frac{d}{2}-2}+u^\prime_{\frac{d}{2}-2})|_{z=0}}	{\partial_z\sigma|_{z=0}}.$$
As above Equation \eqref{Eq: recursion for U', V'} entails existence of a smooth solution for $v^\prime_0$ in $\mathcal{K}$. The same induction procedure used for $u^\prime_j$ can be now applied to infer existence of the smooth coefficients $v^\prime_k$, particularly up to the boundary in the first entry at $z=0$.

\vskip .2cm

{\bf Step 2:} Since the expansion of $U'$ is finite, the function $U^\prime\doteq\sum_{j=0}^{\frac{d-4}{2}}u^\prime_j\sigma_-^j$ exists and it is smooth. We focus hence on the family of smooth coefficients $\{v^\prime_k\}_{k\geq 0}$. Under the given hypotheses we can apply Lemma \ref{Lem: Reflected asymptotic summation} with $a_k=v^\prime_k$, concluding the existence of $V^\prime\in C^\infty(\mathcal{K})$ such that, for every $M\geq0$,
\begin{equation}\label{Eq: VM}
V^\prime-V^\prime_M \doteq V^\prime-\sum_{k=0}^{M}v^\prime_k\sigma_-^k=\sigma_-^{M+1}R_M,
\end{equation}
where $R_M\in C^\infty(\mathcal{K})$. Recall that $\mathcal{K}$ is chosen as in the proof of Lemma \ref{Lem: Reflected asymptotic summation}. This is the sought statement.
Moreover, denoting by $p\in\partial\mcM$ the reflection point associated with $(x_0,y_0)$, the coefficients $u'_j$, $v'_k$ and the amplitudes $U^\prime$, $V^\prime$ can be chosen to admit smooth extensions to a suitable neighbourhood of $(p,y_0)$ in $\mcM\times\mathring{\mcM}$. Thereon, Equation \eqref{Eq: VM} holds true with $R_M$ smooth up and including points at the boundary in the first entry. We conclude by observing that the case $d=2$ can be dealt by following the same procedure and therefore we omit the proof.
\end{proof}

\begin{remark}\label{Rem: Symmetry of the Reflected Hadamard Coefficients}
Given a one-reflection Hadamard branch domain $\Omega^{\mathrm H}_{1,\alpha}$ and assuming that $\mcM$ is even dimensional, Proposition \ref{Prop: Existence one-reflection reflected coefficients even} entails that Equation \eqref{Eq: recursion for U', V'} establishes existence of smooth Hadamard coefficients $\{u^\prime_j\}_{j=1,\dots,\frac{d-4}{2}}$, $\{v^\prime_k\}_{k\in\mathbb{N}\cup\{0\}}$ abiding by the boundary condition in the first entry. Yet, one cannot conclude that these coefficients are symmetric and it does not seem possible to easily adapt the reasoning in \cite{Moretti_2000} due to the the presence of the boundary condition. We do not address the problem of proving the symmetry of the remaining coefficients in this paper, since we expect it to be a long detour from our main goal.
\end{remark}

\subsubsection{Hadamard Recursion Relations: Odd Dimensional Case} \label{Sec: Hadamard Recursion Relations: Odd dimensions}
In this subsection we assume that the dimension $d$ of the underlying manifold $\mcM$ is odd. Since the procedure is, {\em mutatis mutandis}, structurally identical to the even dimensional one, we limit ourselves to reporting the main identities without delving into their derivation. Also in this case there is no reason to disambiguate between the $\pm i\epsilon$-prescription in the directed or reflected component of the Hadamard parametrix. Hence we omit it. As before, given a one-reflection Hadamard branch domain $\Omega^{\mathrm H}_{1,\alpha}$, we consider the following expansions
\begin{equation} \label{Eq: expansion of U,V,U',V' odd}
	\begin{cases}
		U(x,y) \sim \sum_{j=0}^{\infty} u_j (x,y) \sigma^j(x,y) \\
		U'(x,y) \sim \sum_{j=0}^{\infty} u'_j (x,y) \sigma_-^j(x,y)
	\end{cases},
\end{equation}
where $\sim$ denotes an asymptotic expansion. As in the even scenario the coefficients of $U$ are the restriction to $\mathcal{U}=\mathcal{U}^\prime\cap\mcM$ of the standard Hadamard counterparts, where $\mathcal{U}^\prime$ is a geodesically convex, normal neighbourhood of the ambient globally hyperbolic spacetime $(\mcM^\prime,g^\prime)$ with empty boundary. Similarly to the preceding section, the coefficients $\{u_j\}_{j\in\mathbb{N}}$ must abide by the standard Hadamard recursion relations, namely
\begin{equation} \label{Eq: relations for U odd dimension}
	 \begin{cases}
        2\sigma^\mu \partial_\mu u_0 + (\sigma^{\mu}{}_{\mu}-d) u_0 = 0, \\
        [u_0] = 1, \\\\
        Pu_j + (2j+4-d)\sigma^\mu \partial_\mu u_{j+1} + \big[(j+1)(2j+\sigma^\mu{}_{\mu} + 4-2d) + \left(1-\frac{d}{2}\right) (\sigma^{\mu}{}_{\mu} -d)\big] u_{j+1} = 0, \\
        [u_{j+1}] = -\frac{[Pu_j]}{2j+4-d}, \hspace{0.2cm}  j \ge 0.
    \end{cases}
\end{equation}    
Their smoothness and symmetry, see \cite{Moretti_2000}, are inherited directly from the known properties on $\mathcal{U}^\prime$. Focusing instead on $U^\prime$ we end up with
\begin{equation} 
	\label{Eq: recursion for U', V' odd}
	\begin{cases}
		2 \sigma_-^\mu \partial_{\mu} u'_0 + ((\sigma_-)^{\mu}{}_{\mu} -d) u'_0= 0, \\
        u'_0 |_{z=0} = u_0 |_{z=0}, \\\\
        Pu'_j + (2j+4-d) \sigma_-^\mu \partial_\mu u'_{j+1} +\big[ (j+1)(2j+4 + (\sigma_-)^\mu{}_{\mu} -2d) + \left(1 - \frac{d}{2}\right)(\sigma_-^{\mu}{}_{\mu} -d) \big] u'_{j+1} = 0, \\
        (\partial_z+\kappa) (u_j+u'_j) |_{z=0} + \frac{1}{2}(2j+4-d)(\partial_z \sigma) (u_{j+1}-u'_{j+1})|_{z=0} = 0 ,\hspace{0.2cm} j \ge 0, 
	\end{cases}
\end{equation}

\begin{remark}\label{Rem: Symmetry of the Reflected Hadamard Coefficients - odd scenario}
As in the even dimensional case, particularly Remark \ref{Rem: solution of u'0}, we can infer right away that $u^\prime_0$ is the reflected van-Vleck Morette determinant,
$$u^\prime_0(x,y)=\Delta^\frac{1}{2}(x,y),\quad\Delta(x,y) \doteq \frac{|\det(-\nabla_\mu\nabla_{\nu^\prime}\sigma_-(x,y))|} {|g_x|^{\frac{1}{2}}|g_y|^{\frac{1}{2}}}.$$
\end{remark}

\noindent Similarly to the even dimensional case we can deduce the existence both of a smooth solution for each $u^\prime_j$ and of a smooth function $U^\prime$ as per Equation \eqref{Eq: expansion of U,V,U',V' odd} in the sense of an asymptotic expansion. The proof of the following statement is identical to that of Proposition \ref{Prop: Existence one-reflection reflected coefficients even} and, hence, we omit it.

\begin{proposition}\label{Prop: Existence one-reflection reflected coefficients odd}
Given $\dim\mcM=d\geq 3$ odd, setting $\Omega^{\mathrm H}_{1,\alpha}$ as per Definition \ref{Def: One-reflection Hadamard branch domain} and fixing $(x_0,y_0)\in\mathcal{Z}_{1,\alpha}\cap\Omega^{\mathrm H}_{1,\alpha}$ Equation \eqref{Eq: recursion for U', V' odd} admits unique smooth solutions $\{u^\prime_j\}_{j\in\mathbb N\cup\{0\}}$ on a suitable relatively compact open neighbourhood $\mathcal{K}_0\subset\Omega^{\mathrm{H}}_{1,\alpha}$. Furthermore there exists $U^\prime\in C^\infty(\mathcal{U})$ such that
\begin{equation}\label{Eq: UprimeN}
U'\sim \sum_{j=0}^\infty u^\prime_j\sigma_-^j,
\end{equation}
where $\sim$ denotes an asymptotic expansion. Moreover, denoting by $p\in\partial\mcM$ the reflection point associated with $(x_0,y_0)$, both the coefficients $u^\prime_j$, $j\in\mathbb{N}\cup\{0\}$, and $U^\prime$ can be chosen smooth on a suitable relatively open compact subset $\mathcal{K}_\partial\subset\mcM\times\mathring{\mcM}$ including $(p,y_0)$. Thereon, for every $M\geq 0$,
$$U^\prime-\sum_{j=0}^{M}u^\prime_j\sigma_-^j=\sigma_-^{M+1}R_M,$$
with $R_M\in C^\infty(\mathcal{K}_\partial)$.
\end{proposition}

\begin{remark}\label{Rem: No symmetry Hadamard Coefficients odd}
Proposition \ref{Prop: Existence one-reflection reflected coefficients odd} establishes the existence on a one-reflection Hadamard branch domain $\Omega^{\mathrm H}_{1,\alpha}$ of smooth Hadamard coefficients $\{u^\prime_j\}_{j\in\mathbb{N}\cup\{0\}}$ abiding by the boundary condition in the first entry. Yet, as in the even dimensional case, we cannot conclude that they are symmetric.
\end{remark}

\subsubsection{Parametrix and Pauli-Jordan Propagator} 
Combining Equation \eqref{Eq: Directed Hadamard Parametrix} with Propositions \ref{Prop: Existence one-reflection reflected coefficients even} and \ref{Prop: Existence one-reflection reflected coefficients odd} we have established the existence of smooth functions $U^\prime,V^\prime$ with a prescribed asymptotic expansion. They are compatible with Equation \eqref{Eq: Reflected Hadamard Parametrix} in Definition \ref{Def: Local Hadamard State} although we have still to address two key compatibility issues. On the one hand  $H|_{\Omega^{\mathrm H}_{1,\alpha}}(x,y)+H^{\mathrm{ref}}_{\kappa;1,\alpha}(x,y)$ must be shown to be a Robin parametrix modulo smooth terms. Once more we shall study this issue in the case of a single reflection, commenting subsequently on the extension to the general scenario. On the other hand we must establish that the antisymmetric part of this parametrix can be chosen to coincide with the Pauli-Jordan propagator $G_\kappa$ as per Proposition \ref{Prop: Existence Robin propagators}.

Given $(x_0,y_0)\in\mathcal Z_{1,\alpha}\cap\Omega^{\mathrm H}_{1,\alpha}$ we fix $\mathcal{K}\Subset\Omega^{\mathrm H}_{1,\alpha}$ a relatively compact open subset as in the proof of Lemma \ref{Lem: Reflected asymptotic summation}. Herein the functions $U^\prime,V^\prime$ as per Propositions \ref{Prop: Existence one-reflection reflected coefficients even} or \ref{Prop: Existence one-reflection reflected coefficients odd} exist. Henceforth, unless stated otherwise, we shall be working only on $\mathcal{K}$. Furthermore, for convenience of the notation, we set
\begin{equation}\label{Eq: H+}
\mathsf{H}_{\kappa;1,\alpha}^+\doteq H+H^{\mathrm{ref}}_{\kappa;1,\alpha},
\end{equation}
where the $+$ sign on the left hand side is a placeholder to remind us that we have chosen the $i\epsilon$-prescription in Equations \eqref{Eq: Directed Hadamard Parametrix} and \eqref{Eq: Reflected Hadamard Parametrix}. We recall that this choice is also motivated by the connection with the requirement of selecting future-pointing covectors in Equation \eqref{Eq: Hadamard WF set Robin}.

We denote by $\mathsf{H}_{\kappa;1,\alpha}^{-}=H^-+H^{\mathrm{ref},-}_{\kappa;1,\alpha}$ the distributions whose integral kernels are as per Equations \eqref{Eq: Directed Hadamard Parametrix} and \eqref{Eq: Reflected Hadamard Parametrix}, though with the $-i\epsilon$-prescription. Observe that being both $U,V,U^\prime,V^\prime$ as well as $\sigma,\sigma_-$ real-valued functions
$$\mathsf{H}_{\kappa;1,\alpha}^-=\overline{\mathsf{H}_{\kappa;1,\alpha}^+}.$$
We define for later convenience
\begin{equation}\label{Eq: One-reflection causal Hadamard kernel}
E_{\kappa;1,\alpha}\doteq -i\left(\mathsf{H}_{\kappa;1,\alpha}^+-\mathsf{H}_{\kappa;1,\alpha}^-\right)=E^{\mathrm{dir}}_{\kappa;1,\alpha}+E^{\mathrm{ref}}_{\kappa;1,\alpha},
\end{equation}
where the directed and reflected components are the sum of the counterparts in $\mathsf{H}_{\kappa;1,\alpha}^\pm$. 

We start by addressing the first of these two issues, whose proof relies on recollecting the already established results.

\begin{proposition}\label{Prop: Robin Parametrix Verified}
Under the standing assumptions of Section \ref{Sec: Local Hadamard form and Recursion Relations}, there exists a relatively compact open subset $\mathcal{K}\subseteq\mcM\times\mathring{\mcM}$ such that $\mathsf{H}_{\kappa;1,\alpha}^\pm\in\mathcal{D}^\prime(\mathcal{K}\times\mathcal{K})$ and
$$P_x\mathsf{H}_{\kappa;1,\alpha}^\pm\in C^\infty(\mathcal{K}),\qquad\left.\mathcal B_{\kappa,x}\mathsf{H}_{\kappa;1,\alpha}^\pm\right|_{\partial\mcM}\in C^\infty(\mathcal{K}_\partial),$$
where $P_x=(P\otimes\mathbb{I})$ while $\mathcal{B}_{\kappa,x}\doteq(\nabla_n+\kappa \mathbb{I})\otimes\mathbb{I}$ is the operator implementing the Robin boundary condition in the first entry with the convention of Definition \ref{Def: Robin Hadamard two-point function}.	
\end{proposition}

\begin{proof}
We focus on $H_{\kappa;1,\alpha}^+$ being the reasoning for $H_{\kappa;1,\alpha}^-$ identical {\em mutatis mutandis} and we choose $\mathcal{K}$ as in the proof of Lemma \ref{Lem: Reflected asymptotic summation}. Hence all functions involved can be evaluated at $\partial\mcM$ in the first entry being smooth thereon. This applies also to $\sigma$ and $\sigma_-$, see item {\em 1.} and {\em 2.} in Lemma \ref{Lem: Matching N-reflected Synge world functions}. Comparing with Equation \eqref{Eq: H+}, the first contribution to $\mathsf{H}_{\kappa;1,\alpha}^{\pm}$ is the restriction to $\mathcal{K}$ of the standard Hadamard parametrix on the ambient spacetime introduced in Equation \eqref{Eq: Boundaryless Parametrix}. Hence, $P_x H\in C^\infty(\mathcal{K})$, see, {\it e.g.}, \cite{Radzikowski_1996,Radzikowski_1996_1}. 

Focusing on $H^{\mathrm{ref}}_{\kappa;1,\alpha}$, the second contribution, and assuming that $\dim\mcM=d\geq 4$ is even, we denote by $V^\prime_M$ the truncation at order $M$ of the asymptotic expansion of $V^\prime$, see Equation \eqref{Eq: VM}. Equation \eqref{Eq: recursion for U', V'} entails that all singular terms arising from the action of $P_x$ vanish, while the associated Robin matching conditions yield the same conclusion after applying $\mathcal B_{\kappa,x}$ at $\partial\mcM$.
	
By Lemma \ref{Lem: Reflected asymptotic summation}, Equation \eqref{Eq: VM} entails, in particular, that $V^\prime-V^\prime_M=\sigma_-^{M+1}R_M$, with $R_M\in C^\infty(\mathcal{K})$. Hence, applying $P_x$, the worst singular term in $P_xH^{\mathrm{ref}}_{\kappa;1,\alpha}$ is proportional to $R^{M-1}\sigma_-^M\ln\sigma_-$, the preceding ones being cancelled by the Hadamard recursion relations. This term is therefore of class $C^{M-2}$, but $M$ can be chosen arbitrarily large. This implies that the corresponding remainder is of class $C^m$ for every $m\in\mathbb N$. A similar reasoning and conclusion can be drawn working with $\mathcal{B}_{\kappa,x}$. In other words 
$$P_x\mathsf{H}_{\kappa;1,\alpha}^+\in C^\infty(\mathcal{K})\quad\textrm{and}\quad\left.\mathcal B_{\kappa,x}\mathsf{H}_{\kappa;1,\alpha}^+\right|_{\partial\mcM}\in C^\infty(\mathcal{K}_\partial),$$
where $\mathcal{K}_\partial$ denotes the realization of $\mathcal{K}$ introduced in the proof of Lemma \ref{Lem: Reflected asymptotic summation}. For $d$ odd or $d=2$ the same line of reasoning leads to an identical conclusion.
\end{proof}

\noindent In the following key proposition, we establish the existence of a connection between the choice of a Robin Hadamard parametrix and the underlying Pauli-Jordan propagator. We stress the relevance of this consistency condition both for its application to the realm of mathematical physics, in connection to the existence of covariant quantization schemes for Bosonic scalar field theories, and for its key role in establishing our main result, Theorem \ref{Thm: Global-to-Local}.

\begin{proposition}\label{Prop: One-reflection Robin parametrix and causal propagators}
Under the standing assumptions of Section \ref{Sec: Local Hadamard form and Recursion Relations}, there exists a relatively compact open subset $\mathcal{K}\subseteq\mcM\times\mathring{\mcM}$ such that the following two statements hold true:
\begin{enumerate}
\item There exists $K^\pm_{\kappa;1,\alpha}\in\mathcal{D}^\prime(\mathcal{K}\times\mathcal{K})$ such that, at the level of integral kernel 
\begin{equation}\label{Eq: One-reflection causal parametrices}
K^-_{\kappa;1,\alpha}(x,y)\doteq\Theta\bigl(t(x)-t(y)\bigr)E_{\kappa;1,\alpha}(x,y)\quad\textrm{and}\quad		K^+_{\kappa;1,\alpha}\doteq-\Theta\bigl(t(y)-t(x)\bigr)E_{\kappa;1,\alpha}(x,y),
\end{equation}
where $t:\mcM\to\mathbb{R}$ is the underlying time function. These distributions satisfy
$$P_xK^\pm_{\kappa;1,\alpha}=\left.\delta_{\operatorname{Diag}_2(\mathring{\mcM})}\right|_{\mathcal{K}}+R^\pm\quad\textrm{and}\quad
\left.\mathcal B_{\kappa,x}K^\pm_{\kappa;1,\alpha}\right|_{\partial\mcM}=B^\pm,$$
where $R^\pm,B^\pm\in C^\infty(\mathcal{K})$ while $\operatorname{Diag}_2(\mathring{\mcM})=\{(x,x)\;|\;x\in\mathring{\mcM}\}$.
\item Let $G^\pm_\kappa$ be the exact advanced and retarded Robin propagators obtained out of Proposition \ref{Prop: Existence Robin propagators} combined with Equation \eqref{Eq: Gpm from Gkappa}. Given $(x_0,y_0)\in\mathcal Z_{1,\alpha}\cap\Omega^{\mathrm H}_{1,\alpha}$, there exists $\mathcal{K}^\prime\Subset\mathcal{K}$ such that
\begin{equation}\label{Eq: Local Pauli-Jordan one-reflection}
G^\pm_\kappa|_{\mathcal{K}^\prime}-K^\pm_{\kappa;1,\alpha}|_{\mathcal{K}^\prime}\in C^\infty(\mathcal{K}^\prime)\Longrightarrow G_\kappa|_{\mathcal{K}^\prime}-E_{\kappa;1,\alpha}|_{\mathcal{K}^\prime}\in C^\infty(\mathcal{K}^\prime),
\end{equation}
where $G_\kappa=G^-_\kappa-G^+_\kappa$.
\end{enumerate}
\end{proposition}

\begin{proof}
We prove the two items separately.
	
\vskip .2cm
	
{\em Item 1:}  Here we employ a microlocal argument. More precisely, starting from Equation \eqref{Eq: One-reflection causal Hadamard kernel}, for any $\dim\mcM=d\geq 2$, Proposition \ref{Prop: Robin Parametrix Verified} entails that $P_xE_{\kappa;1,\alpha}$ is smooth. This implies that $$\operatorname{WF}(E_{\kappa;1,\alpha})\subseteq\operatorname{Char}(P_x)=\{(x,k_x,y,k_y)\in T^*(\mathring{\mcM}\times\mathring{\mcM})\setminus\{\boldsymbol{0}\}\;|\; g_x(k_x,k_x)=0\}.$$
In other words $k_x$ is a lightlike covector and, since $\operatorname{WF}(\Theta(t(x)-t(y))$ only contains timelike covectors, we can apply \cite[Thm.~8.2.10]{Hormander_1990} to infer existence of the sought $K^\pm_{\kappa;1,\alpha}\in\mathcal{D}^\prime(\mathcal{K}\times\mathcal{K})$. Focusing on the action of $P_x$, we recall Equation \eqref{Eq: H+}. Using the properties of second order normally hyperbolic partial differential operators on a globally hyperbolic spacetime with empty boundary, see, {\it e.g.}, \cite{Baer_2007}, we can infer that, on a relatively compact subset $\mathcal{K}$ as that of the proof of Lemma \ref{Lem: Properties of Reflected Geodesic Distance}, denoting by $E^{\operatorname{dir}}_{\kappa;1,\alpha}$ the contribution to $E_{\kappa;1,\alpha}$ coming from the directed Hadamard parametrix $H$, this coincides up to smooth terms with the Pauli-Jordan propagator $G'$ of the Klein-Gordon operator $P'$ on $(\mcM^\prime,g^\prime$). Hence, 
$$P_xE^{\operatorname{dir}}_{\kappa;1,\alpha}=\delta_{\operatorname{Diag}_2(\mathring{\mcM})}|_{\mathcal{K}} + R_{\operatorname{dir}},$$
where $R_{\operatorname{dir}}\in C^\infty(\mathcal{K})$. Denoting by $E^{\operatorname{ref}}_{\kappa;1,\alpha}$ the contribution to $E_{\kappa;1,\alpha}$ coming from the reflected component, we can observe that Equation \eqref{Eq: Reflected Hadamard Parametrix} combined with the choice of $\mathcal{K}$ as per Lemma \ref{Lem: Reflected asymptotic summation} entails that
$$\operatorname{supp}(E^{\operatorname{ref}}_{\kappa;1,\alpha})\subseteq\{(x,y)\in\mathcal{K}\;|\;\sigma_-(x,y)\leq 0\}.$$
Furthermore, since Proposition \ref{Prop: Robin Parametrix Verified} entails that $P_x(E^{\operatorname{ref}}_{\kappa;1,\alpha})\in C^\infty(\mathcal{K})$, the only singularities in $P_x(\Theta\bigl(t(x)-t(y)\bigr)E_{\kappa;1,\alpha}(x,y))$ can occur when $t(x)=t(y)$. Yet, a direct inspection of Definition \ref{Def: Reflected Synge's world function} unveils that $\sigma_-(x,y)\leq 0$ implies that $t(x)\neq t(y)$. The same conclusion can be drawn for $\Theta\bigl(t(y)-t(x)\bigr)E_{\kappa;1,\alpha}(x,y)$.

Focusing on the boundary condition, Proposition \ref{Prop: Globally Hyperbolic} entails that, on a globally hyperbolic spacetime with a timelike boundary $\nabla t$ is tangent to $\partial\mcM$. Hence $n(t)=0$, where $n$ is the inward pointing normal vector field to $\partial\mcM$.	Consequently 
$$\left.\mathcal B_{\kappa,x}K^\pm_{\kappa;1,\alpha}\right|_{\partial\mcM}=\pm\Theta\bigl(\pm(t(x)-t(y))\bigr)	\left.\mathcal{B}_{\kappa,x}E_{\kappa;1,\alpha}\right|_{\partial\mcM}.$$
On account of Proposition \ref{Prop: Robin Parametrix Verified}, $\left.\mathcal{B}_{\kappa,x}E_{\kappa;1,\alpha}\right|_{\partial\mcM}\in C^\infty(\mathcal{K}_\partial)$, where $\mathcal{K}_\partial$ is one of the realization of $\mathcal{K}$ as per Lemma \ref{Lem: Reflected asymptotic summation}. Moreover, extending the reasoning used in the previous point, we can infer that 
$$\operatorname{supp}(E_{\kappa;1,\alpha})\subseteq\{(x,y)\in\mathcal{K}\;|\;\sigma(x,y)\leq 0\;\textrm{or}\;\sigma_-(x,y)\leq 0\}.$$
Since the restriction to $\partial\mcM$ entails that $x\in\partial\mcM$, while $y\in\mathring{\mcM}$, the support condition implies that $t(x)\neq t(y)$. This establishes that the Heaviside function does not yield a singular contribution, hence the sought conclusion.
	
\vskip .2cm
	
{\em Item 2:} On account of the previous item, working still on $\mathcal{K}$ chosen as above and setting $D^\pm\doteq G^\pm_\kappa-K^\pm_{\kappa;1,\alpha}$, it holds that
$$P_xD^\pm=-R^\pm,\quad\textrm{and}\quad\left.\mathcal B_{\kappa,x}D^\pm\right|_{\partial\mcM}=-B^\pm.$$
We show that we can find suitable smooth correcting terms to $D^\pm$ which make the right-hand side of these two identities vanish. To start with and, shrinking if necessary the underlying domain $\mathcal{K}$, we choose $\chi\equiv\chi(z)\in C^\infty_0([0,\epsilon))$, $\epsilon>0$ so that $\chi=0$ in a neighbourhood of $z=0$ while $\partial_z\chi|_{z=0}=1$. Defining $Q^\pm\doteq-\chi(z)B^\pm$,	it holds that
$$\left.\mathcal B_{\kappa,x}Q^\pm\right|_{z=0}=-B^\pm \Longrightarrow\left.\mathcal B_{\kappa,x}(D^\pm-Q^\pm)\right|_{\partial\mcM}=0.$$
Hence, setting $F^\pm \doteq -R^\pm-P_xQ^\pm$ and introducing $D_Q^\pm \doteq D^\pm-Q^\pm$, we are left with a Robin boundary value problem:
\begin{equation}\label{Eq: Auxiliary Robin Problem}
P_xD_Q^\pm=F^\pm\quad\textrm{and}\quad B_{\kappa,x}D_Q^\pm=0.
\end{equation}
Given $(x_0,y_0)\in\mathcal Z_{1,\alpha}\cap\Omega^{\mathrm H}_{1,\alpha}$, first we fix $\mathcal{K}\subset\mathring{\mcM}\times\mathring{\mcM}$ as in the proof of Lemma \ref{Lem: Reflected asymptotic summation}. Subsequently we choose a relatively compact subset $\mathcal{K}^\prime\Subset\mathcal{K}$ and a cutoff function $\chi\in C^\infty_0(\mathcal{K})$ such that $\chi=1$ on $\mathcal{K}^\prime$. Replacing $F^\pm$ with $F^\pm_\chi \doteq \chi F^\prime$ in Equation \eqref{Eq: Auxiliary Robin Problem}, we obtain a compactly supported, smooth source. Hence a corresponding solution is 
$$Q^\pm_\chi=G^\pm_{\kappa,x}(F^\pm_\chi),$$
where $G^\pm_{\kappa,x}$ is the Robin advanced/retarded propagators as per Proposition \ref{Prop: Existence Robin propagators} combined with Equation \eqref{Eq: Gpm from Gkappa}. The subscript $x$ indicates that the propagator is acting on the first entry of $F^\pm_\chi$. Furthermore $Q^\pm_\chi$ is smooth and it coincides on $\mathcal{K}^\prime$ with $Q^\pm$. Hence we have found a solution of Equation \eqref{Eq: Auxiliary Robin Problem} and the proof is concluded.	
\end{proof}

\begin{remark}\label{Rem: No symmetry needed causal construction}
	Proposition \ref{Prop: One-reflection Robin parametrix and causal propagators} does not require symmetry of the reflected Hadamard coefficients. Equation \eqref{Eq: One-reflection causal Hadamard kernel} compares the two boundary values of the same kernel and it does not involve transposition or branch reversal. 
\end{remark}

\noindent As a byproduct of our analysis we can also establish for later convenience the singular structure both of $\mathsf{H}^\pm_{\kappa;1,\alpha}$ and of $K^\pm_{\kappa;1,\alpha}$.

\begin{corollary}\label{Cor: WF One-Reflection}
Under the standing assumptions of Section \ref{Sec: Local Hadamard form and Recursion Relations}, let $(x_0,y_0)\in\mathcal Z_{1,\alpha}\cap\Omega^{\mathrm H}_{1,\alpha}$. Consider a causally convex Cauchy neighbourhood $\mathcal{N}$ as per Definition \ref{Def: Cauchy neighbourhood} and let $\mathcal{K}\Subset\Omega^{\mathrm H}_{1,\alpha}\cap\mathcal{N}$ be a relatively compact open neighbourhood containing $(x_0,y_0)$ as in the proof of Lemma \ref{Lem: Regularity endpoint map}. Then it holds that
\begin{equation}\label{Eq: WF one-reflection reflected causal kernel}
\operatorname{WF}\!\left(E^{\mathrm{ref}}_{\kappa;1,\alpha}|_{\mathcal{K}}\right)=\left(N^*\mathcal{Z}_{1,\alpha}\setminus\{0\}\right)\cap T^*\mathcal{K}.
\end{equation}
while	
\begin{equation}\label{Eq: WF one-reflection Hadamard kernel}
\operatorname{WF}\!\left(\mathsf{H}^+_{\kappa;1,\alpha}|_{\mathcal{K}}\right)=\left(C^{\triangleright,\mathrm{dir}}_{b,\mathcal{K}}\cup\Lambda^{\triangleright}_{1,\alpha}\right)\cap T^*\mathcal{K}\quad\textrm{and}\quad\operatorname{WF}\!\left(\mathsf{H}^-_{\kappa;1,\alpha}|_{\mathcal{K}}\right)=\left(C^{\triangleleft,\mathrm{dir}}_{b,\mathcal{K}}\cup\Lambda^{\triangleleft}_{1,\alpha}\right)\cap T^*\mathcal{K}.
\end{equation}
where $\Lambda^{\triangleright}_{1,\alpha}$ is as per Equation \eqref{Eq: Branchwise reflected conormal relation} setting $N=1$, while $\Lambda^{\triangleleft}_{1,\alpha}$ is defined similarly, though requiring that $\lambda(t(x)-t(y))<0$. At the same time $C^{\triangleright,\mathrm{dir}}_{b,\mathcal{K}}$ is the restriction to $T^*\mathcal{K}$ of Equation \eqref{Eq: Directed null relation in N}, while $C^{\triangleleft,\mathrm{dir}}_{b,\mathcal{K}}$ is defined similarly, though requiring that $k_x$ is past directed.
\end{corollary}

\begin{proof}
Choosing $\mathcal{K}$ sufficiently small, Equation \eqref{Eq: WF one-reflection reflected causal kernel} is a by-product of item {\em 2.} of Proposition \ref{Prop: One-reflection Robin parametrix and causal propagators}	which entails that, on $T^*\mathcal{K}$, $\operatorname{WF}\!\left(E_{\kappa;1,\alpha}|_{\mathcal{K}}\right)=\operatorname{WF}\!\left(G_\kappa|_{\mathcal{K}}\right)$. The latter has been computed in Theorem \ref{Thm: WF of Propagators}. In view of Remark \ref{Rem: Full Directed and Reflected Singularities} it comprises two contributions, one from the directed and one from the reflected components. On account of Equations \eqref{Eq: One-reflection causal Hadamard kernel} and \eqref{Eq: H+} combined with Definitions \ref{Def: Directed Hadamard Parametrix} and \ref{Def: Local Hadamard State}, we can infer that $\operatorname{singsupp}(E^{\mathrm{dir}}_{\kappa;1,\alpha}|_{\mathcal{K}})\cap\operatorname{singsupp}(E^{\mathrm{ref}}_{\kappa;1,\alpha}|_{\mathcal{K}})=\emptyset$ and hence their wavefront sets are disjoint. Since the directed component accounts for $C^{\mathrm{dir}}_{b,\mathcal{K}}$, the second accounts for $\Lambda^{\mathrm{ref}}_{\mathcal{K}}:=\Lambda^{\mathrm{ref}}_{\mathcal{N}}\cap T^*\mathcal{K}$. Comparing with Equation \eqref{Eq: Branchwise reflected conormal relation}, we obtain Equation \eqref{Eq: WF one-reflection reflected causal kernel}.

We need to prove Equation \eqref{Eq: WF one-reflection Hadamard kernel} and we focus only on $\mathsf{H}^+_{\kappa;1,\alpha}|_{\mathcal{K}}$ since the identity concerning $\mathsf{H}^-_{\kappa;1,\alpha}|_{\mathcal{K}}$ follows suit. Equation \eqref{Eq: H+} entails that $\mathsf{H}^+_{\kappa;1,\alpha}|_{\mathcal{K}}$ is the sum between the directed and the reflected parametrix $H$ and $H^{\mathrm{ref}}_{\kappa;1,\alpha}$. Comparing Equations \eqref{Eq: Directed Hadamard Parametrix} and \eqref{Eq: Reflected Hadamard Parametrix} we can infer that $\operatorname{singsupp}(H)\cap\operatorname{singsupp}(H^{\mathrm{ref}}_{\kappa;1,\alpha})=\emptyset$, the first being subordinated to the equation $\sigma(x,y)=0$, while the second to $\sigma_-(x,y)=0$. Hence on $T^*\mathcal{K}$
$$\operatorname{WF}(\mathsf{H}^+_{\kappa;1,\alpha})=\operatorname{WF}(H)\cup\operatorname{WF}(H^{\mathrm{ref}}_{\kappa;1,\alpha}).$$
The wavefront set of the directed component is nothing but $C^{\triangleright,\mathrm{dir}}_{b,\mathcal{K}}$, the restriction to $\mathcal{K}$ of that of the Hadamard parametrix associated to the two-point correlation function of Klein-Gordon field on $(\mcM^\prime,g^\prime)$, see, {\it e.g.}, \cite{Radzikowski_1996,Radzikowski_1996_1}. We focus on $H^{\mathrm{ref}}_{\kappa;1,\alpha}$. We recall that $\mathcal{K}$ has been devised so that, thereon, we may use coordinates $(r,\zeta)$ such that
$$r=\sigma_-(x,y),$$
while $\zeta=(\zeta^1,\dots,\zeta^{2d-1})$, $d=\dim\mcM$, denotes the remaining ones. Moreover $\tau(x,y)\doteq t(x)-t(y)\neq 0$ on the locus $r=0$.  Recalling that the wavefront set of a distribution is invariant under diffeomorphisms and, hence, under coordinate reparametrization of a distribution, it is legit to work with these coordinates. Furthermore, we observe that, as a direct consequence of \cite[8.1.8]{Hormander_1990} and of the explicit form of the Fourier transform in the $(r,\zeta)$-coordinates, for $\beta>0$, the homogeneous distributions $r_+^{-\beta} \doteq \lim\limits_{\epsilon\to 0^+}(r+i\epsilon\tau)^{-\beta}\in\mathcal{D}^\prime(\mathcal{K}\times\mathcal{K})$ have the following singular structure
\begin{equation}\label{Eq: WFaux1}
\mathrm{WF}\!\left(r_+^{-\beta}\right)=\left\{(0,\lambda,w,0)\in T^*(\mathcal{K}\times\mathcal{K})\setminus\{\boldsymbol{0}\}\;|\;\lambda \tau(0,w)>0\right\}.
\end{equation}
whereas, setting $\ln (r_+) \doteq \lim\limits_{\epsilon\to 0^+}\ln (r+i\epsilon\tau)$,
\begin{equation}\label{Eq: WFaux2}
\mathrm{WF}\!\left(\ln(r_+)\right)=\left\{(0,\lambda,w,0)\in T^*(\mathcal{K}\times\mathcal{K})\setminus\{\boldsymbol{0}\}\;|\;\lambda \tau(0,w)>0\right\}.
\end{equation}
This last identity descends from two observations. The first is that we can read $\ln(r_+)$ as the pull-back to $T^*\mathcal{K}$ of the one-dimensional distribution $\ln (x_+) \doteq \lim\limits_{\epsilon\to 0^+}\ln (x+i\epsilon)$ under the map $\pi:\mathcal{K}\to\mathbb{R}$ such that $\pi(r,\zeta)=r$. Observing that in one dimension $\operatorname{WF}(\ln(x_+))=\operatorname{WF}(x_+^{-1})$, \cite[Thm. 8.2.4]{Hormander_1990} yields the inclusion $\subseteq$ in Equation \eqref{Eq: WFaux2}. Observing that $r_+^{-1}=\partial_r\ln(r_+)$ yields the converse inclusion combining Equation \eqref{Eq: WFaux1} with $\operatorname{WF}(Pu)\subseteq\operatorname{WF}(u)$ for every $u\in\mathcal{D}^\prime(\Omega)$, $\Omega\subseteq\bR^d$, and for $P$ a differential operator. Restoring $r=\sigma_{-,1,\alpha}(x,y)$, it turns out that $(0,\zeta)$ are the pairs $(x,y)\in\mathcal{K}\times\mathcal{K}$ such that $\sigma_{-,1,\alpha}(x,y)=0$, $x\in J^+(y)\cup J^-(y)$. Furthermore $\lambda$ becomes $\lambda d\sigma_{-,1,\alpha}(x,y)=(k_x,-k_y)$ where $k_x=\lambda d_x\sigma_{-,1,\alpha}$ and $k_y=-\lambda d_y\sigma_{-,1,\alpha}$. Observe that, if $y\in J^+(x)$, $\tau(x,y)<0$ and, hence $d_x\sigma_{-,1,\alpha}$ is past-directed, that is $\lambda<0$. This entails that $\lambda d_x\sigma_{-,1,\alpha}$ is future-directed and, with a similar reasoning, the same holds true if $x\in J^+(y)$. 
In other words we have shown that 
$$\mathrm{WF}(H^{\mathrm{ref}}_{\kappa,1,\alpha})\subseteq\Lambda^{\triangleright}_{1,\alpha}.$$
Here the inclusion is byproduct of Equation \eqref{Eq: Reflected Hadamard Parametrix} which entails that, for $d>2$ even
$$\operatorname{WF}(H^{\mathrm{ref}}_{\kappa;1,\alpha})\subseteq\operatorname{WF}\left(\frac{U^\prime}{\sigma_{-,\epsilon}^{\frac{d-2}{2}}}\right)\cup\operatorname{WF}\left(V^\prime\ln\left(\frac{\sigma_{-,\epsilon}}{\lambda^2}\right)\right).$$
Furthermore Proposition \ref{Prop: Existence one-reflection reflected coefficients even} entails that $U'$ is a finite linear combination of positive powers of $\sigma_-$ such that the first coefficient $u_0$ is never vanishing on $\mathcal{K}$, see Remark \ref{Rem: solution of u'0}. Hence $U^\prime\sigma_{-,\epsilon}^{\frac{2-d}{2}}$ contributes to the wavefront set as per Equation \eqref{Eq: WFaux1}. \emph{Mutatis mutandis} a similar conclusion can be drawn for $d=2$ or for $d$ odd. 

In order to prove the opposite inclusion, we can proceed by contradiction as, for example, in the proof of Theorem \ref{Thm: WF of Propagators}. As a matter of fact, assume that there exists $(x,k_x,y,-k_y)\in\Lambda^{\triangleright}_{1,\alpha}$ which does not lie in $\mathrm{WF}(H^{\mathrm{ref}}_{\kappa,1,\alpha})$. Consequently it cannot lie in $\mathrm{WF}(H^+_{\kappa;1,\alpha})$. Since $P_xH^+_{\kappa;1,\alpha}\in C^\infty(\mathcal{K})$ as per Proposition \ref{Prop: Robin Parametrix Verified}, we can use once more the propagation of singularities as proven in \cite{Melrose_1978} to transport $(x,k_x)$ to $(y,k_y)$ inferring that $(y,k_y,y,-k_y)\notin\mathrm{WF}(\mathsf{H}^+_{\kappa;1,\alpha})$. Yet, since this point lies in $\mathrm{WF}(H)$, the directed Hadamard parametrix, and since, by construction, $\mathrm{WF}(H)\cap\mathrm{WF}(H^{\mathrm{ref}}_{\kappa,1,\alpha})=\emptyset$, then we have obtained the sought contradiction.
\end{proof}

\noindent We need one last result concerning the interplay between the kernels studied in Corollary \ref{Cor: WF One-Reflection} and those, one would obtain considering the reversed branch.

\begin{corollary}\label{Cor: One-reflection reversal covariance}
Under the same assumptions of Corollary \ref{Cor: WF One-Reflection}, let $\alpha^{\rm op}$ be the reversed branch as per Remark \ref{Rem: Reversed reflected branch}. Let $\mathcal{K}\Subset\Omega^{\mathrm H}_{1,\alpha}$ be as in Corollary \ref{Cor: WF One-Reflection} and set $\mathcal{K}^{\rm op}\doteq\{(y,x)\mid (x,y)\in\mathcal{K}\}$. It holds that
\begin{equation}\label{Eq: One-reflection reversal covariance}
\left(\mathsf{H}^+_{\kappa;1,\alpha}|_{\mathcal{K}}\right)^T-\mathsf{H}^-_{\kappa;1,\alpha^{\rm op}}|_{\mathcal{K}^{\rm op}}\in C^\infty(\mathcal{K}^{\rm op})\quad\textrm{and}\quad\left(\mathsf{H}^-_{\kappa;1,\alpha}|_{\mathcal{K}}\right)^T	-\mathsf{H}^+_{\kappa;1,\alpha^{\rm op}}|_{\mathcal{K}^{\rm op}}\in C^\infty(\mathcal{K}^{\rm op}),
\end{equation}
where, at the level of integral kernels, $\left(\mathsf{H}^\pm_{\kappa;1,\alpha}\right)^T(x,y)\doteq \mathsf{H}^\pm_{\kappa;1,\alpha}(y,x)$.
\end{corollary}

\begin{proof}
Using item {\em 2.} of Proposition \ref{Prop: One-reflection Robin parametrix and causal propagators}, applied to the branches $\alpha$ and $\alpha^{\rm op}$, we can infer that
$$G_\kappa-E_{\kappa;1,\alpha}\in C^\infty(\mathcal{K})\quad\textrm{and}\quad G_\kappa-E_{\kappa;1,\alpha^{\rm op}}\in C^\infty(\mathcal{K}^{\rm op}).$$
Yet, as a consequence of Corollary \ref{Cor: G+- Formal Adjoints}, the Pauli-Jordan propagator is antisymmetric, that is $G_\kappa^T=-G_\kappa$. In other words $E_{\kappa;1,\alpha}^T+E_{\kappa;1,\alpha^{\rm op}}\in C^\infty(\mathcal{K}^{\rm op})$. Motivated by Equation \eqref{Eq: One-reflection causal Hadamard kernel}, we introduce the kernels
$$\mathsf C\doteq\left(\mathsf{H}^+_{\kappa;1,\alpha}\right)^T-\mathsf{H}^-_{\kappa;1,\alpha^{\rm op}}\quad\textrm{and}\quad	\mathsf D\doteq\left(\mathsf{H}^-_{\kappa;1,\alpha}\right)^T-\mathsf{H}^+_{\kappa;1,\alpha^{\rm op}}.$$
Then $\mathsf{D} - \mathsf{C}=i(E_{\kappa;1,\alpha^{\rm op}}+E^T_{\kappa;1,\alpha})\in C^\infty(\mathcal{K}^{\rm op})$. On account of Corollary \ref{Cor: WF One-Reflection} and Equation \eqref{Eq: Cotangent branch reversal}, $\operatorname{WF}(\mathsf C)$ can contain only past-directed covectors in the first entry, while $\operatorname{WF}(\mathsf D)$ can contain only future-directed ones. In other words $\operatorname{WF}(\mathsf C)\cap\operatorname{WF}(\mathsf D)=\emptyset$. Yet, since $\mathsf C + \mathsf D$ is smooth, the algebraic identities $\mathsf C =(\mathsf C + \mathsf D)- \mathsf D$ and $\mathsf D=(\mathsf C+ \mathsf D)- \mathsf C$ entail that both $\mathsf C$ and $\mathsf D$ have empty wavefront set. Hence they are smooth, thereby proving Equation \eqref{Eq: One-reflection reversal covariance}.
\end{proof}

\begin{remark}\label{Rem: Observation on Ref}
We observe that the conclusion of Corollary \ref{Cor: Symmetry Directed Hadamard Coefficients} holds true also if we restrict the attention only to the reflected component, since the contribution of the directed Hadamard parametrix as per Equation \eqref{Eq: Directed Hadamard Parametrix} cancels out. Indeed, by the symmetry of the directed Hadamard coefficients, we can infer that $(H^\pm)^T=H^\mp$. In other words
$$\left(H^{\mathrm{ref},+}_{\kappa;1,\alpha}\right)^T-H^{\mathrm{ref},-}_{\kappa;1,\alpha^{\rm op}}\in C^\infty(\mathcal{K}^{\rm op})\quad\textrm{and}\quad\left(H^{\mathrm{ref},-}_{\kappa;1,\alpha}\right)^T-H^{\mathrm{ref},+}_{\kappa;1,\alpha^{\rm op}}\in C^\infty(\mathcal{K}^{\rm op}).$$
\end{remark}

\begin{remark}\label{Rem: Reversal without coefficient symmetry}
We stress once more that Corollary \ref{Cor: One-reflection reversal covariance} neither requires nor implies symmetry of the reflected Hadamard coefficients.
\end{remark}

\subsubsection{Finite-reflection Hadamard branch chains} In the preceding part of this section, we assumed the existence of only a single reflection. In the following, instead, we extend our analysis to an arbitrary, though finite, number of reflections, also in agreement with the idea at the heart of Assumption \ref{Ass: Finite regular reflections}. We start by generalizing Definition \ref{Def: One-reflection Hadamard branch domain}.

\begin{definition}\label{Def: N-reflection Hadamard branch domain}
Let $(\mathcal V_{N,\alpha},\Omega_{N,\alpha})$, $\alpha \in \mathbb{N}$, be a regular $N$-reflected branch chart as per Definition \ref{Def: Regular reflected branch chart}. We call an open subset
$$\Omega^{\mathrm H}_{N,\alpha}\Subset\Omega_{N,\alpha}$$
an {\bf $N$-reflection Hadamard branch domain} if, taking into account Equation \eqref{Eq: Z}, for every	$(x,y)\in\mathcal Z_{N,\alpha}\cap\Omega^{\mathrm{H}}_{N,\alpha}$ and for every reflection point $x_j$ of $\gamma^{(x,y)}_{N,\alpha}$, the incoming and outgoing branchwise Synge functions $\sigma^{\mathrm{in}}_{j,N,\alpha}$ and $\sigma^{\mathrm{out}}_{j,N,\alpha}$ as per Equation \eqref{Eq: Incoming outgoing reflected Synge functions} admit smooth one-sided extensions to a common neighbourhood of $(x_j,y)$ as per Lemma \ref{Lem: Matching N-reflected Synge world functions}.
\end{definition}

Observe that the existence of $\Omega^{\mathrm H}_{N,\alpha}$ does not require any additional assumption since, being $N$ finite, Condition \emph{2.} of Assumption \ref{Ass: Finite regular reflections} allows one to shrink $\Omega_{N,\alpha}$ so that the requirement in Definition \ref{Def: N-reflection Hadamard branch domain} holds simultaneously at every reflection point. 

\begin{definition}\label{Def: Branchwise reflected Hadamard kernel}
Given an $N$-reflection Hadamard branch domain $\Omega^{\mathrm H}_{N,\alpha}$ as per Definition \ref{Def: N-reflection Hadamard branch domain} we define
\begin{equation}\label{Eq: Branchwise reflected iepsilon phase}
\sigma_{-,N,\alpha,\pm\epsilon}(x,y)\doteq\sigma_{-,N,\alpha}(x,y)\pm i\epsilon\bigl(t(x)-t(y)\bigr)+\epsilon^2,\quad\epsilon>0,
\end{equation}
where $t$ is the underlying time function. A \textbf{branchwise reflected Hadamard kernel} associated with $(N,\alpha)$ is a distribution $H^{\mathrm{ref}}_{\kappa;N,\alpha}\equiv H^{\mathrm{ref},+}_{\kappa;N,\alpha}\in\mathcal{D}^\prime(\Omega^{\mathrm H}_{N,\alpha})$ such that, at the level of integral kernels
\begin{equation}\label{Eq: Branchwise reflected Hadamard kernel}
H^{\mathrm{ref},\pm}_{\kappa;N,\alpha}(x,y)\doteq\lim_{\epsilon\to 0^+}\frac{\Gamma\!\left(\frac{d}{2}-1\right)}		{2(2\pi)^{\frac{d}{2}}}\left(\frac{U^\prime_{N,\alpha}(x,y)}{\sigma_{-,N,\alpha,\pm\epsilon}(x,y)^{\frac{d-2}{2}}}+		\delta_d V^\prime_{N,\alpha}(x,y)\ln\!\left(\frac{\sigma_{-,N,\alpha,\pm\epsilon}(x,y)}{\lambda^2}\right)\right),
\end{equation}
where $U'_{N,\alpha},V'_{N,\alpha}\in C^\infty(\Omega^{\mathrm H}_{N,\alpha})$, while $\delta_d$ is as in Definition \ref{Def: Local Hadamard State}. If $d=2$, the $U^\prime$-term is omitted. In addition, we set
\begin{equation}\label{Eq: Branchwise reflected causal kernel}
E^{\mathrm{ref}}_{\kappa;N,\alpha}\doteq -i\left(H^{\mathrm{ref},+}_{\kappa;N,\alpha}-	H^{\mathrm{ref},-}_{\kappa;N,\alpha}\right).
\end{equation}
\end{definition}

\vskip .2cm

\noindent Given $(x,y)\in\mathcal Z_{N,\alpha}\cap\Omega^{\mathrm H}_{N,\alpha}$, we denote by $x_0,\ldots,x_{N-1}$ the reflection points of $\gamma^{(x,y)}_{N,\alpha}$, see Equation \eqref{Eq: GammaN}. Using Equations \eqref{Eq: sigmajin} and \eqref{Eq: sigmajout}, at the $j$-th reflection we set
\begin{equation}\label{Eq: Incoming outgoing branchwise Hadamard kernels}
H^{\mathrm{in},\pm}_{\kappa;j,N,\alpha}\doteq H^{\mathrm{ref},\pm}_{\kappa;N-j,\alpha^{\mathrm{in}}_j},\qquad	H^{\mathrm{out},\pm}_{\kappa;j,N,\alpha}\doteq
\begin{cases}
H^{\mathrm{ref},\pm}_{\kappa;N-j-1,\alpha^{\mathrm{out}}_j},&j=0,\ldots,N-2,\\
H^\pm,&j=N-1.
\end{cases}
\end{equation}
where $H^\pm$ are the Hadamard kernels of Equation \eqref{Eq: Directed Hadamard Parametrix}. We also set
$$E^{\mathrm{in/out}}_{\kappa;j,N,\alpha}\doteq -i\left(H^{\mathrm{in/out},+}_{\kappa;j,N,\alpha}-H^{\mathrm{in/out},-}_{\kappa;j,N,\alpha}\right).$$
We call the collection in Equation \eqref{Eq: Incoming outgoing branchwise Hadamard kernels} the \textbf{Hadamard branch chain} generated by $(N,\alpha)$. At the $j$-th reflection the incoming member has $N-j$ remaining reflections, while the outgoing member has one reflection fewer.

Considering the kernels in Equation \eqref{Eq: Incoming outgoing branchwise Hadamard kernels}, we denote their smooth amplitudes by $U^{\mathrm{in/out}}_{j,N,\alpha}$ and $V^{\mathrm{in/out}}_{j,N,\alpha}$ and their formal Hadamard coefficients by $u^{\mathrm{in/out}}_{r;j,N,\alpha}$ and $v^{\mathrm{in/out}}_{r;j,N,\alpha}$. In dimension $d>2$ even,
$$U^{\mathrm{in/out}}_{j,N,\alpha}=\sum_{r=0}^{\frac{d}{2}-2}u^{\mathrm{in/out}}_{r;j,N,\alpha}\left(\sigma^{\mathrm{in/out}}_{j,N,\alpha}\right)^r\quad\textrm{and}\quad V^{\mathrm{in/out}}_{j,N,\alpha}\sim\sum_{r=0}^{\infty} v^{\mathrm{in/out}}_{r;j,N,\alpha}\left(\sigma^{\mathrm{in/out}}_{j,N,\alpha}\right)^r,$$
while for $d=2$ the $U$-term is absent. We shall refer to them as the {\em incoming/outgoing Hadamard coefficients}. In odd dimension the $V$-term is absent and
$$U^{\mathrm{in/out}}_{j,N,\alpha}\sim\sum_{r=0}^{\infty}u^{\mathrm{in/out}}_{r;j,N,\alpha}\left(\sigma^{\mathrm{in/out}}_{j,N,\alpha}\right)^r.$$
Here the symbol $\sim$ denotes that the right-hand side is an asymptotic expansion of the left one. Furthermore there is no need to disambiguate between the $\pm i\epsilon$-prescription since it plays no role. In the following proposition we establish that the Hadamard coefficients entering the reflected members of the branch chain in Equation \eqref{Eq: Incoming outgoing branchwise Hadamard kernels} can be constructed recursively, starting from the directed one and proceeding backwards from the last reflection point $x_{N-1}$ to $x_0$, the first one.

\begin{proposition}\label{Prop: Recursive branchwise Hadamard coefficients}
Given $(N,\alpha)\in\mathcal{I}_{\mathrm{ref}}$ as per Corollary \ref{Cor: Countable reflected conormal resolution} and given $\Omega^{\mathrm H}_{N,\alpha}$, an $N$-reflection Hadamard branch domain as per Definition \ref{Def: N-reflection Hadamard branch domain}, let $(x_0,y_0)\in\mathcal Z_{N,\alpha}\cap\Omega^{\mathrm H}_{N,\alpha}$. We denote by $x_0,\ldots,x_{N-1}\in\partial\mcM$ the reflection points of the associated broken geodesic. 
	
Then, for every $j=0,\ldots,N-1$, the incoming Hadamard coefficients associated with the branch $(N-j,\alpha^{\mathrm{in}}_j)$ admit smooth solutions on a suitable relatively compact open subset $\mathcal{K}^{\mathrm{in}}_{j,0}\Subset\Omega^{\mathrm{H}}_{N-j,\alpha^{\mathrm{in}}_j}$, while, for $j=0,\ldots,N-2$, the same holds true for the outgoing ones associated with $(N-j-1,\alpha^{\mathrm{out}}_j)$ on a suitable relatively compact open subset $\mathcal{K}^{\mathrm{out}}_{j,0}\Subset\Omega^{\mathrm{H}}_{N-j-1,\alpha^{\mathrm{out}}_j}$. In addition 
\begin{itemize}
	\item[\ding{104}] The corresponding Hadamard expansions admit smooth asymptotic expansions on these neighbourhoods, with the same structure as in Propositions \ref{Prop: Existence one-reflection reflected coefficients even} and \ref{Prop: Existence one-reflection reflected coefficients odd}.
	\item[\ding{104}]  It holds that
	\begin{equation}\label{Eq: Branchwise interior transport condition}
	P_xH^{\mathrm{in},\pm}_{\kappa;j,N,\alpha}\in C^\infty\!\left(\mathcal{K}^{\mathrm{in}}_{j,0}\right)\quad\textrm{and}\quad P_xH^{\mathrm{out},\pm}_{\kappa;j,N,\alpha}\in 	C^\infty\!\left(\mathcal{K}^{\mathrm{out}}_{j,0}\right),
	\end{equation}
	where the first statement holds true for $j=0,\ldots,N-1$, the second only up to $j=N-2$.
	\item[\ding{104}] For all $j=0,\ldots,N-1$, there exists a relatively compact open subset $\mathcal{K}_{j,\partial}\subset\mcM\times\mathring{\mcM}$ containing $(x_j,y_0)$ on which both $\sigma^{\mathrm{in}}_{j,N,\alpha}$ and	$\sigma^{\mathrm{out}}_{j,N,\alpha}$ admit the smooth one-sided extensions of Definition \ref{Def: N-reflection Hadamard branch domain}. Setting $\mathcal{K}^\partial_{j,\partial}\doteq\mathcal{K}_{j,\partial}\cap\left(\partial\mcM\times\mathring{\mcM}\right)$,
	it holds that
	\begin{equation}\label{Eq: Compact branchwise Robin matching}
	\left.\mathcal B_{\kappa,x}\left(H^{\mathrm{in},\pm}_{\kappa;j,N,\alpha}+H^{\mathrm{out},\pm}_{\kappa;j,N,\alpha}	\right)\right|_{\mathcal{K}^\partial_{j,\partial}}\in C^\infty\!\left(\mathcal{K}^\partial_{j,\partial}\right),
	\end{equation}
\end{itemize}
contained in the common neighbourhood appearing in Definition \ref{Def: N-reflection Hadamard branch domain}, 
\end{proposition}

\begin{proof}
We proceed backwards along the branch chain and we specify that all relatively compact open subset $\mathcal{K}_{j,\partial}$ and $\mathcal{K}_{j,0}$ are chosen following the same construction as in the proof of Lemma \ref{Lem: Reflected asymptotic summation}. At the last reflection point $x_{N-1}$, the outgoing component is the Hadamard kernel $H^\pm$. Lemma \ref{Lem: Matching N-reflected Synge world functions} yields the relations between $\sigma^{\mathrm{in}}_{N-1,N,\alpha}$ and $\sigma$ which were used in the previous section in the analysis of the single reflection scenario. Hence Propositions \ref{Prop: Existence one-reflection reflected coefficients even} and \ref{Prop: Existence one-reflection reflected coefficients odd} entail the statement of the proposition.
	
Assume now that the outgoing component at $x_j$ has already been constructed. Lemma \ref{Lem: Matching N-reflected Synge world functions} shows that the cancellation of the singular boundary terms is governed by the same one-reflection system after the replacements
\begin{equation}\label{Eq: Branchwise recursion replacement rule}
(\sigma_-,U',V';\sigma,U,V)	\longmapsto	\left(\sigma^{\mathrm{in}}_{j,N,\alpha},U^{\mathrm{in}}_{j,N,\alpha},		V^{\mathrm{in}}_{j,N,\alpha};\sigma^{\mathrm{out}}_{j,N,\alpha},U^{\mathrm{out}}_{j,N,\alpha},		V^{\mathrm{out}}_{j,N,\alpha}\right).
\end{equation}
The already constructed outgoing coefficients therefore provide the boundary data for the incoming transport equations. The infinitesimal convexity of the boundary or, equivalently, the no-glancing hypothesis entails $\nabla_n\sigma^{\mathrm{in}}_{j,N,\alpha}(x_j,y)\neq 0$, where $n$ is the inward pointing, unit vector normal to $\partial\mcM$ at the reflection point $x_j$. Accordingly the boundary is non-characteristic for these first-order equations and the same local existence argument used in the one-reflection case applies.
	
Existence of the functions $U^{\mathrm{out}}_{j,N,\alpha}, V^{\mathrm{out}}_{j,N,\alpha}$ descends from the same Borel summation argument as in the proof of Lemma \ref{Lem: Reflected asymptotic summation}. These properties hold for every incoming branchwise Synge function by Assumption \ref{Ass: Finite regular reflections}, Lemma \ref{Lem: Matching N-reflected Synge world functions} and the no-glancing hypothesis, see Equation \eqref{Eq: No Glancing}. Iterating the procedure from $j=N-1$ up to $j=0$ yields Equations \eqref{Eq: Branchwise interior transport condition} and \eqref{Eq: Compact branchwise Robin matching}.
\end{proof}

\begin{remark}
We observe that, when $d>2$, the leading matching relation at the $j$-th reflection is
$$\left.u^{\mathrm{in}}_{0;j,N,\alpha}\right|_{x_j}=\left.u^{\mathrm{out}}_{0;j,N,\alpha}\right|_{x_j}.$$
As in the case of one single reflection we can build explicitly the solution, namely
$$u^{\mathrm{in}}_{0;j,N,\alpha}(x,y)=\frac{\left|\det\!\left(	-\nabla_\mu\nabla_{\nu'}\sigma^{\mathrm{in}}_{j,N,\alpha}(x,y)\right)\right|^{\frac{1}{2}}}{|g_x|^{\frac{1}{4}}|g_y|^{\frac{1}{4}}},$$
which is non vanishing in the regions under consideration. Furthermore, once more no symmetry of the reflected coefficients is used in Proposition \ref{Prop: Recursive branchwise Hadamard coefficients} and this property is not investigated.
\end{remark}

\noindent The preceding proposition constructs the analytic data along the whole finite branch chain. The next result compares Equation \eqref{Eq: Branchwise reflected causal kernel} with the exact Robin Pauli-Jordan propagator, microlocally along the $N$-reflection branch.

\begin{proposition}\label{Prop: Finite-reflection branchwise causal comparison}
	Let $(N,\alpha)\in\mathcal{I}_{\mathrm{ref}}$ and let $\Omega^{\mathrm H}_{N,\alpha}$ be an $N$-reflection Hadamard branch domain as per Definition \ref{Def: N-reflection Hadamard branch domain}. Denoting by $N^*\mathcal Z_{N,\alpha}$ the conormal bundle of the hypersurface $\mathcal Z_{N,\alpha}$ defined in Equation \eqref{Eq: Z}, for every $q\in\left(N^*\mathcal Z_{N,\alpha}\setminus\{\boldsymbol{0}\}\right)\cap T^*\Omega^{\mathrm H}_{N,\alpha}$ there exists an open conic neighbourhood $\Gamma_q\subset T^*\Omega^{\mathrm H}_{N,\alpha}\setminus\{\boldsymbol{0}\}$ such that
	\begin{equation}\label{Eq: Branchwise microlocal causal comparison}
		\Gamma_q\cap\operatorname{WF}\!\left(\left.G_\kappa\right|_{\Omega^{\mathrm H}_{N,\alpha}}-E^{\mathrm{ref}}_{\kappa;N,\alpha}\right)=\emptyset,
	\end{equation}
	where $E^{\mathrm{ref}}_{\kappa;N,\alpha}$ is as per Equation \eqref{Eq: Branchwise reflected causal kernel}.
\end{proposition}

\begin{proof}
Given $q=(x,k_x,y,-k_y)\in\left( N^*\mathcal Z_{N,\alpha}\setminus\{\boldsymbol{0}\}\right)\cap T^*\Omega^{\mathrm H}_{N,\alpha}$, we denote by $x_0,\ldots,x_{N-1}\in\partial\mcM$ the reflection points of $\gamma^{(x,y)}_{N,\alpha}$, ordered according to the underlying affine parameter. For every $j=0,\ldots,N-1$, let $k^{\mathrm{in}}_j,k^{\mathrm{out}}_j\in T^*_{x_j}\mcM$ be the covectors obtained by propagating $k_x$ along the corresponding broken bicharacteristic immediately before and after the $j$-th reflection. On account of Equations \eqref{Eq: sigmajin} and \eqref{Eq: sigmajout}, there exists $\lambda_j\neq 0$ such that	$$k^{\mathrm{in}}_j=\lambda_jd_x\sigma^{\mathrm{in}}_{j,N,\alpha}(x_j,y)\quad\textrm{and}\quad	k^{\mathrm{out}}_j=\lambda_jd_x\sigma^{\mathrm{out}}_{j,N,\alpha}(x_j,y).$$
We recall that both $d_x\sigma^{\mathrm{in/out}}_{j,N,\alpha}(x_j,y)$ are lightlike covectors due to the eikonal equation \eqref{Eq: Eikonal Equation for sigma_-}. We fix the notation $q^{\mathrm{in}}_j\doteq(x_j,k^{\mathrm{in}}_j,y,-k_y)$ and $q^{\mathrm{out}}_j\doteq(x_j,k^{\mathrm{out}}_j,y,-k_y)$ observing that, due to Lemma \ref{Lem: Matching N-reflected Synge world functions} combined with Assumption \ref{Ass: Infinitesimally convex}, $k^{\mathrm{in}}_j\neq k^{\mathrm{out}}_j$.
	
We start from the last reflection point $x_{N-1}$. By Equation \eqref{Eq: Incoming outgoing branchwise Hadamard kernels}, $E^{\mathrm{out}}_{\kappa;N-1,N,\alpha}$ agrees with the contribution coming from the direct null geodesic, while $E^{\mathrm{in}}_{\kappa;N-1,N,\alpha}$ is its reflected counterpart. Hence the sought statements have already been proven in Proposition \ref{Prop: One-reflection Robin parametrix and causal propagators} and Corollary \ref{Cor: WF One-Reflection}. This entails the existence of a suitable relatively compact open neighbourhood $\mathcal{K}_{N-1,\partial}\subset\mcM\times\mathring{\mcM}$ containing $(x_{N-1},y)$ such that
\begin{equation}\label{Eq: Last reflection causal comparison}
\left.G_\kappa\right|_{\mathcal{K}_{N-1,\partial}}-\left.E^{\mathrm{in}}_{\kappa;N-1,N,\alpha}\right|_{\mathcal{K}_{N-1,\partial}}-\left.E^{\mathrm{out}}_{\kappa;N-1,N,\alpha}\right|_{\mathcal{K}_{N-1,\partial}}\in C^\infty(\mathcal{K}_{N-1,\partial}).
\end{equation}
Furthermore 
\begin{equation}\label{Eq: Initial backward regularity}
q^{\mathrm{in}}_{N-1}\notin\operatorname{WF}\!\left(G_\kappa-E^{\mathrm{in}}_{\kappa;N-1,N,\alpha}\right).
\end{equation}
We proceed inductively in the index $j$, hence, heuristically speaking, backwards along $\gamma^{(x,y)}_{N,\alpha}$.
Suppose that, for some $j_0\in\{1,\ldots,N-1\}$, the backward procedure has been carried out at all reflection points $x_j$ with $j\geq j_0$. In particular,
\begin{equation}\label{Eq: Backward induction hypothesis causal}
	q^{\mathrm{in}}_{j_0}\notin\operatorname{WF}\!\left(G_\kappa-E^{\mathrm{in}}_{\kappa;j_0,N,\alpha}\right).
\end{equation}
We prove that the sought statement holds true also at $x_{j_0-1}$, the following reflection point. By construction of the Hadamard branch chain, see Equation \eqref{Eq: Incoming outgoing branchwise Hadamard kernels}, the incoming member at $x_{j_0}$ coincides with the outgoing member at $x_{j_0-1}$. Hence, there exists $\mathcal{U}_{j_0-1}\subset\mathring{\mcM}\times\mathring{\mcM}$, an open neighbourhood containing this segment of the broken geodesic, so that Proposition \ref{Prop: Recursive branchwise Hadamard coefficients} entails that
$$P_x\left(G_\kappa-E^{\mathrm{out}}_{\kappa;j_0-1,N,\alpha}\right)\in C^\infty(\mathcal{U}_{j_0-1}).$$
Therefore, propagation of singularities applied to Equation \eqref{Eq: Backward induction hypothesis causal} yields 
$$q^{\mathrm{out}}_{j_0-1}\notin\operatorname{WF}\!\left(G_\kappa-E^{\mathrm{out}}_{\kappa;j_0-1,N,\alpha}\right).$$ 
By Definition \ref{Def: N-reflection Hadamard branch domain}, we can choose a relatively compact open neighbourhood $\mathcal{K}_{j_0-1,\partial}\subset\mcM\times\mathring{\mcM}$ including $(x_{j_0-1},y)$ where both $\sigma^{\mathrm{in}}_{j_0-1,N,\alpha}$ and $\sigma^{\mathrm{out}}_{j_0-1,N,\alpha}$ admit smooth one-sided extensions. Setting $\mathcal{K}^\partial_{j_0-1,\partial}\doteq\mathcal{K}_{j_0-1,\partial}\cap(\partial\mcM\times\mathring{\mcM})$, on $\mathcal{K}_{j_0-1,\partial}$ consider 
$$D_{j_0-1}\doteq G_\kappa-E^{\mathrm{in}}_{\kappa;j_0-1,N,\alpha}-E^{\mathrm{out}}_{\kappa;j_0-1,N,\alpha}.$$ 
Proposition \ref{Prop: Recursive branchwise Hadamard coefficients} and Equation \eqref{Eq: Compact branchwise Robin matching} entail, possibly after shrinking $\mathcal{K}_{j_0-1,\partial}$, that 
\begin{equation}\label{Eq: Dj smooth defects} 
	P_xD_{j_0-1}\in C^\infty(\mathcal{K}_{j_0-1,\partial})\quad\textrm{and}\quad\left.\mathcal B_{\kappa,x}D_{j_0-1}\right|_{\mathcal{K}^\partial_{j_0-1,\partial}}\in C^\infty(\mathcal{K}^\partial_{j_0-1,\partial}). 
\end{equation} 
The same proof of Corollary \ref{Cor: WF One-Reflection}, applied to $\sigma^{\mathrm{in}}_{j_0-1,N,\alpha}$, shows that $q^{\mathrm{out}}_{j_0-1}\notin\operatorname{WF}(E^{\mathrm{in}}_{\kappa;j_0-1,N,\alpha})$. Combined with 
$q^{\mathrm{out}}_{j_0-1}\notin\operatorname{WF}(G_\kappa-E^{\mathrm{out}}_{\kappa;j_0-1,N,\alpha})$, this entails that $q^{\mathrm{out}}_{j_0-1}\notin\operatorname{WF}(D_{j_0-1})$. Proceeding as in the proof of Proposition \ref{Prop: One-reflection Robin parametrix and causal propagators}, one reduces to the case of homogeneous Robin boundary conditions without altering the wavefront set. The Melrose-Sj\"ostrand propagation theorem \cite[Thm. 4.1]{Melrose_1978}, applied as in the proof of Theorem \ref{Thm: WF of Propagators}, then yields 
$$q^{\mathrm{in}}_{j_0-1}\notin\operatorname{WF}(D_{j_0-1}).$$ 
At the same time, applying the proof of Corollary \ref{Cor: WF One-Reflection} to $\sigma^{\mathrm{out}}_{j_0-1,N,\alpha}$ entails that 
\begin{equation}\label{Eq: Incoming regularity at j} 
q^{\mathrm{in}}_{j_0-1}\notin\operatorname{WF}\!\left(G_\kappa-E^{\mathrm{in}}_{\kappa;j_0-1,N,\alpha}\right). 
\end{equation} 
This completes the inductive step. Iterating the preceding step down to $j_0=1$ yields
$$q^{\mathrm{in}}_0\notin\operatorname{WF}(G_\kappa-E^{\mathrm{in}}_{\kappa;0,N,\alpha}).$$ Since $E^{\mathrm{in}}_{\kappa;0,N,\alpha}$ is the one-sided realization of $E^{\mathrm{ref}}_{\kappa;N,\alpha}$, ordinary propagation of singularities along the initial segment of $\gamma^{(x,y)}_{N,\alpha}$ entails
$$q\notin\operatorname{WF}\!\left(\left.G_\kappa\right|_{\Omega^{\mathrm H}_{N,\alpha}}-E^{\mathrm{ref}}_{\kappa;N,\alpha}\right).$$
Equation \eqref{Eq: Branchwise microlocal causal comparison} follows since the complement of the wavefront set is open and conic.
\end{proof}

\begin{proposition}\label{Prop: WF branchwise reflected Hadamard kernel}
Under the assumptions of Proposition \ref{Prop: Finite-reflection branchwise causal comparison}, it holds that
\begin{equation}\label{Eq: WF branchwise reflected causal kernel}
\operatorname{WF}\!\left(E^{\mathrm{ref}}_{\kappa;N,\alpha}\right)=N^*\mathcal{Z}_{N,\alpha}		\setminus\{\boldsymbol{0}\}\quad\textrm{and}\quad\operatorname{WF}\!\left(H^{\mathrm{ref},\pm}_{\kappa;N,\alpha}\right)=		\Lambda^{\triangleright/\triangleleft}_{N,\alpha},
\end{equation}
where $\Lambda^{\triangleright}_{N,\alpha}$ is as per Equation \eqref{Eq: Branchwise reflected conormal relation}, while $\Lambda^{\triangleleft}_{N,\alpha}$ is defined by replacing $\lambda(t(x)-t(y))>0$ with $\lambda(t(x)-t(y))<0$.
\end{proposition}

\begin{proof}
We can mimic the proof of Corollary \ref{Cor: WF One-Reflection}. As a matter of fact, replacing $\sigma_{-,1,\alpha}$ with $\sigma_{-,N,\alpha}$ is harmless since Proposition \ref{Prop: Reflected null branch atlas} entails that $	d\sigma_{-,N,\alpha}\neq 0$ on $\mathcal{Z}_{N,\alpha}$. Hence the first part of the proof of Corollary \ref{Cor: WF One-Reflection} yields the inclusions 
$$\operatorname{WF}\!\left(H^{\mathrm{ref},+}_{\kappa;N,\alpha}\right)\subseteq\Lambda^{\triangleright}_{N,\alpha}\quad\textrm{and}\quad\operatorname{WF}\!\left(H^{\mathrm{ref},-}_{\kappa;N,\alpha}\right)\subseteq\Lambda^{\triangleleft}_{N,\alpha}.$$
Consequently Equation \eqref{Eq: Branchwise reflected causal kernel} entails
$$\operatorname{WF}\!\left(E^{\mathrm{ref}}_{\kappa;N,\alpha}\right)\subseteq N^*\mathcal{Z}_{N,\alpha}	\setminus\{\boldsymbol{0}\}.$$
To prove the converse inclusion, take $q\in N^*\mathcal{Z}_{N,\alpha}\setminus\{\boldsymbol{0}\}$. On account of Theorem \ref{Thm: WF of Propagators}, $q\in\operatorname{WF}(G_\kappa)$ and, thus, Proposition \ref{Prop: Finite-reflection branchwise causal comparison} yields $q\in\operatorname{WF}\!\left(E^{\mathrm{ref}}_{\kappa;N,\alpha}\right)$, proving Equation \eqref{Eq: WF branchwise reflected causal kernel}. To conclude, if $q\in\Lambda^{\triangleright}_{N,\alpha}$, then Equation \eqref{Eq: WF branchwise reflected causal kernel} yields $q\in\operatorname{WF}\!\left(	E^{\mathrm{ref}}_{\kappa;N,\alpha}\right)$,  while the already established inclusion entails that $q\notin 	\operatorname{WF}\!\left(H^{\mathrm{ref},-}_{\kappa;N,\alpha}\right)$ since the first covector is future pointing. Equation \eqref{Eq: Branchwise reflected causal kernel} therefore entails $q\in\operatorname{WF}\!\left(H^{\mathrm{ref},+}_{\kappa;N,\alpha}\right)$. The other case follows suit and we have reached the sought conclusion.
\end{proof}

\noindent The next result generalizes Corollary \ref{Cor: One-reflection reversal covariance} and the proof is structurally identical, hence we omit it.

\begin{corollary}\label{Cor: Finite-reflection reversal covariance}
Under the same assumptions of Proposition \ref{Prop: WF branchwise reflected Hadamard kernel}, let $\alpha^{\rm op}$ be the reversed branch as per Remark \ref{Rem: Reversed reflected branch}. The Hadamard branch domains for $\alpha$ and $\alpha^{\rm op}$ can be chosen so that $\Omega^{\mathrm{H}}_{N,\alpha^{\rm op}}= \mathfrak{s}(\Omega^{\mathrm{H}}_{N,\alpha})$, where $\mathfrak{s}(x,y)=(y,x)$. Let $\mathcal{K}\Subset\Omega^{\mathrm H}_{N,\alpha}$ and set $\mathcal{K}^{\rm op}\doteq \mathfrak{s}(\mathcal{K})$.Then
\begin{equation}\label{Eq: Finite-reflection reversal covariance}
\left(H^{\mathrm{ref},+}_{\kappa;N,\alpha}|_{\mathcal{K}}\right)^T-H^{\mathrm{ref},-}_{\kappa;N,\alpha^{\rm op}}|_{\mathcal{K}^{\rm op}}\in	C^\infty(\mathcal{K}^{\rm op})\quad\textrm{and}\quad\left(H^{\mathrm{ref},-}_{\kappa;N,\alpha}|_{\mathcal{K}}\right)^T-H^{\mathrm{ref},+}_{\kappa;N,\alpha^{\rm op}}|_{\mathcal{K}^{\rm op}}\in	C^\infty(\mathcal{K}^{\rm op}),
\end{equation}
where, once more, at the level of integral kernels $\left(H^{\mathrm{ref},\pm}_{\kappa;N,\alpha}|_{\mathcal{K}}\right)^T(x,y)= \left( H^{\mathrm{ref},\pm}_{\kappa;N,\alpha}|_{\mathcal{K}} \right)(y,x)$.
\end{corollary}

\begin{remark}
	Corollary \ref{Cor: Finite-reflection reversal covariance} neither requires nor implies symmetry of the reflected Hadamard coefficients. It is the finite-reflection counterpart of Corollary \ref{Cor: One-reflection reversal covariance} and Remark \ref{Rem: Reversal without coefficient symmetry}.
\end{remark}

\paragraph{Multi-branch local Robin-Hadamard condition --} We combine all the finite-reflection contributions. For every $N\geq1$ we refine the countable family of Proposition \ref{Prop: Reflected null branch atlas} by $N$-reflection Hadamard branch domains, see Definition \ref{Def: N-reflection Hadamard branch domain} and we choose a locally finite partition of unity on a neighbourhood of $\Gamma^0_{N,\mathcal{N}}$ subordinate to the corresponding sets in $\Gamma_N^{\mathrm{reg}}$,
$$\psi_{N,\alpha}\in C^\infty_0(\mathcal{V}_{N,\alpha})\quad\textrm{and}\quad \sum_{\alpha\in\mathcal{I}_N}\psi_{N,\alpha}=1,$$
with $\operatorname{supp}\psi_{N,\alpha}$ contained in the preimage under $\Phi_{N,\alpha}$ of $\Omega^{\mathrm H}_{N,\alpha}$. We choose the refinement closed under branch reversal and we keep the notation $\mathcal{I}_{\mathrm{ref}}$ for the resulting countable index set. On $\Omega^{\mathrm{H}}_{N,\alpha}$ we set
$$\chi_{N,\alpha}\doteq\psi_{N,\alpha}\circ\Phi_{N,\alpha}^{-1},$$
and we let $\widetilde{H}^{\mathrm{ref}}_{\kappa;N,\alpha}$ be the extension by zero to $\mathring{\mathcal{N}}\times\mathring{\mathcal{N}}$ of $\chi_{N,\alpha}H^{\mathrm{ref},+}_{\kappa;N,\alpha}$.

\begin{remark}\label{Rem: No unrestricted branch sum}
Neither $\mathcal{I}_{\mathrm{ref}}$ being countable nor Assumption \ref{Ass: No Zeno points} entails that the family of supports
$$\left\{\operatorname{supp}(\widetilde{H}^{\mathrm{ref}}_{\kappa;N,\alpha})\right\}_{(N,\alpha)\in\mathcal{I}_{\mathrm{ref}}},$$
is locally finite in $\mathring{\mathcal{N}}\times\mathring{\mathcal{N}}$. Assumption \ref{Ass: No Zeno points} controls the number of reflections along each fixed broken null bicharacteristic in a bounded time interval, but, to the best of our understanding, the number of distinct regular branches meeting a given pair of endpoints might not necessarily bounded. In this case, when summing family of distributions $(N,\alpha)$, we might not automatically obtain a distribution, nor one with a compatible wavefront set. In the following we discard this option and, to this end, we recall the following notation \cite[Sec. 8.2]{Hormander_1990}, namely, that, given a smooth manifold $X$ and a closed cone $\Gamma\subset T^*X\setminus\{0\}$, $u\in\mathcal{D}^\prime_\Gamma(X)$ if $u\in\mathcal{D}^\prime(X)$ and $\operatorname{WF}(u)\subseteq\Gamma$.
\end{remark}

\begin{definition}\label{Def: Admissible reflected Hadamard assembly but not yet an Avenger}
The family
$$\left\{\widetilde{H}^{\mathrm{ref}}_{\kappa;N,\alpha}\right\}_{(N,\alpha)\in\mathcal{I}_{\mathrm{ref}}}$$
is called \textbf{microlocally admissible} if
$$H^{\mathrm{ref}}_{\kappa,\mathcal{N}}\doteq\sum_{(N,\alpha)\in\mathcal{I}_{\mathrm{ref}}}	\widetilde{H}^{\mathrm{ref}}_{\kappa;N,\alpha}\in\mathcal{D}^\prime_{\mathcal{C}^{\triangleright}_{\mathrm b,\mathcal{N}}}
\!\left(\mathring{\mathcal{N}}\times\mathring{\mathcal{N}}\right),$$
where convergence of the series is in the H\"ormander topology, see \cite[Def. 8.2.2]{Hormander_1990}, while $\mathcal{C}^{\triangleright}_{\mathrm b,\mathcal{N}}$ is as per Equation \eqref{Eq: Restricted future broken relation}. A two-point correlation function $\omega_{2,\kappa}\in\mathcal D^\prime(\mcM\times\mcM)$ satisfying Equations \eqref{Eq: Robin eqs of motion} and \eqref{Eq: positivity + CCR} is of \textbf{branchwise local Robin-Hadamard form} on $\mathcal{N}$ if the above family can be chosen to be admissible and if, for every causally convex, geodesically convex open subset $\mathcal{U}\Subset\mathring{\mathcal{N}}$,
\begin{equation}\label{Eq: Branchwise local Robin Hadamard form}
\left.\omega_{2,\kappa}\right|_{\mathcal{U}\times\mathcal{U}}=H_{\mathcal{U}\times\mathcal{U}}+\left.H^{\mathrm{ref}}_{\kappa,\mathcal{N}}\right|_{\mathcal{U}\times\mathcal{U}}+W_{\kappa,\mathcal{U}\times\mathcal{U}},
\end{equation}
where $H_{\mathcal{U}\times\mathcal{U}}$ is the directed Hadamard parametrix as per Definition \ref{Def: Local Hadamard State} while $W_{\kappa,\mathcal{U}\times\mathcal{U}}\in C^\infty(\mathcal{U}\times\mathcal{U})$. 
\end{definition}

\noindent The localization procedure outlined above entails that the reflected parametrices $H^{\mathrm{ref}}_{\kappa;N,\alpha}$ need to be localized by means of a cut-off function. The following corollary guarantees that, as expected, this procedure does not alter the underlying microlocal properties.

\begin{corollary}\label{Cor: Union WF localized reflected branches}
Under the same assumptions of Definition \ref{Def: Admissible reflected Hadamard assembly but not yet an Avenger}
\begin{equation}\label{Eq: Union WF localized reflected branches}
\bigcup_{(N,\alpha)\in\mathcal{I}_{\mathrm{ref}}}\operatorname{WF}\!\left(\widetilde{H}^{\mathrm{ref}}_{\kappa;N,\alpha}\right)=\Lambda^{\triangleright,\mathrm{ref}}_{\mathcal{N}}.
\end{equation}
\end{corollary}

\begin{proof}
Multiplication by the smooth functions $\chi_{N,\alpha}$ does not add singularities. Conversely, at every datum in $\Gamma^0_{N,\mathcal{N}}$ at least one function $\psi_{N,\alpha}$ is non-zero. At the corresponding base point multiplication by $\chi_{N,\alpha}$ therefore preserves the branch singularity. Proposition \ref{Prop: WF branchwise reflected Hadamard kernel} and Corollary \ref{Cor: Countable reflected conormal resolution} yield Equation \eqref{Eq: Union WF localized reflected branches}.
\end{proof}

\subsection{Hadamard states}\label{Sec: Comparison of Hadamard states}

In this last section we discuss a counterpart of the celebrated theorem by Radzikowski \cite{Radzikowski_1996, Radzikowski_1996_1}, which establishes the equivalence between the microlocal and local notions of Hadamard two-point functions introduced in Definitions \ref{Def: Robin Hadamard two-point function} and \ref{Def: Local Hadamard State}, hence generalizing the result of \cite{Costeri_25}. 

\begin{theorem}[Global vs Local]\label{Thm: Global-to-Local}
Let $(\mcM,g)$ be a $d$-dimensional, $d\geq2$, globally hyperbolic spacetime with timelike boundary which abides by Assumptions \ref{Ass: Infinitesimally convex}, \ref{Ass: Finite regular reflections} as well as \ref{Ass: No Zeno points}. Given the Klein-Gordon operator $P$ as in Equation \eqref{Eq: KG equation}, let $\mathcal{N}$ be a causally convex Cauchy neighbourhood as per Definition \ref{Def: Cauchy neighbourhood}, and let $\omega_{2,\kappa}\in\mathcal D^\prime(\mcM\times\mcM)$ satisfy Equations \eqref{Eq: Robin eqs of motion} and \eqref{Eq: positivity + CCR}. Assume in addition that there exists a microlocally admissible family $$\left\{\widetilde{H}^{\mathrm{ref}}_{\kappa;N,\alpha}\right\}_{(N,\alpha)\in\mathcal{I}_{\mathrm{ref}}},$$ as per Definition \ref{Def: Admissible reflected Hadamard assembly but not yet an Avenger} and denote the corresponding reflected sum by $H^{\mathrm{ref}}_{\kappa,\mathcal{N}}$. Then the following conditions are equivalent:
\begin{enumerate}
\item $\omega_{2,\kappa}$ is branchwise local on $\mathcal{N}$ as per Definition \ref{Def: Admissible reflected Hadamard assembly but not yet an Avenger};
\item It holds that
\begin{equation}\label{Eq: Radzikowski counterpart WF}
\operatorname{WF}\!\left(\omega_{2,\kappa}|_{\mcM\times\mathring\mcM}\right)=\mathcal{C}_{\mathrm b}^\triangleright(\mathring\mcM)=\left\{(x,k_x, y,k_y)\in\mathcal{C}_{\mathrm b}(\mathring{\mcM})\;\mathrm{and}\; k_x \triangleright 0\right\},
\end{equation}
where $\mathcal{C}_{\mathrm b}(\mathring{\mcM})$ is as per Equation \eqref{Eq: CbM}.
\end{enumerate}
\end{theorem}

\begin{proof}
{\em 1 $\!\Longrightarrow\!$ 2.} Let $\mathcal{U}\Subset\mathring{\mathcal{N}}$ be a causally convex, geodesically convex open subset. Equation \eqref{Eq: Branchwise local Robin Hadamard form} yields
$$\omega_{2,\kappa}|_{\mathcal{U}\times\mathcal{U}}=H_{\mathcal{U}\times\mathcal{U}}+H^{\mathrm{ref}}_{\kappa,\mathcal{N}}|_{\mathcal{U}\times\mathcal{U}}+W_{\kappa, \mathcal{U}\times\mathcal{U}}.$$
Using \cite{Radzikowski_1996, Radzikowski_1996_1} we can infer directly 
$$\operatorname{WF}(H_{\mathcal{U}\times\mathcal{U}})=\mathcal{C}^{\triangleright,\mathrm{dir}}_{\mathrm b,\mathcal{N}}|_{T^*(\mathcal{U}\times\mathcal{U})},$$
where $\mathcal{C}^{\triangleright,\mathrm{dir}}_{\mathrm b,\mathcal{N}}$ is as per Equation \eqref{Eq: Restricted future broken relation}. For the reflected contribution we employ Proposition \ref{Prop: WF branchwise reflected Hadamard kernel}, which, combined with Definition \ref{Def: Admissible reflected Hadamard assembly but not yet an Avenger}, ultimately yields
\begin{equation}\label{Eq: Local WF upper inclusion}
\operatorname{WF}\!\left(\omega_{2,\kappa}|_{\mathcal{U}\times\mathcal{U}}\right)\subseteq\mathcal{C}^{\triangleright}_{\mathrm b,\mathcal{N}}|_{T^*(\mathcal{U}\times\mathcal{U})}.
\end{equation}
For the reverse inclusion, let $\rho=(x,k_x,y,-k_y)\in\mathcal{C}^{\triangleright}_{\mathrm b,\mathcal{N}}|_{T^*(\mathcal{U}\times\mathcal{U})}$. By Theorem \ref{Thm: WF of Propagators}, $\rho\in\operatorname{WF}(G_\kappa)$, whereas Equation \eqref{Eq: Local WF upper inclusion} entails $\rho\notin\operatorname{WF}(\omega_{2,\kappa}^{T})$. Hence Equation \eqref{Eq: positivity + CCR} yields $\rho\in\operatorname{WF}(\omega_{2,\kappa})$. Therefore equality holds in Equation \eqref{Eq: Local WF upper inclusion}. Since we can choose the open subsets $\mathcal{U}$ to cover a neighbourhood of the Cauchy hypersurface contained in $\mathcal{N}$, propagation of singularities as in Theorem \ref{Thm: WF of Propagators}, see also \cite[Thm. 4.1]{Melrose_1978}, extends this equality throughout $\mathring{\mathcal{N}}\times\mathring{\mathcal{N}}$. Assumptions \ref{Ass: Infinitesimally convex} and \ref{Ass: No Zeno points} then allow us to propagate it to the whole $\mcM\times\mathring\mcM$, proving Equation \eqref{Eq: Radzikowski counterpart WF}.
	
\vskip .2cm
	
{\em 2 $\!\Longrightarrow\!$ 1.} Let $\omega_{2,\kappa}$ satisfy Equation \eqref{Eq: Radzikowski counterpart WF} and fix $\mathcal{U}\Subset\mathring{\mathcal{N}}$ as above. Set
$$D_{\kappa,\mathcal{U}}\doteq\omega_{2,\kappa}|_{\mathcal{U}\times\mathcal{U}}-H \vert_{\mathcal{U}\times\mathcal{U}}-H^{\mathrm{ref}}_{\kappa,\mathcal{N}}|_{\mathcal{U}\times\mathcal{U}}.$$
Equation \eqref{Eq: Radzikowski counterpart WF} and Definition \ref{Def: Admissible reflected Hadamard assembly but not yet an Avenger} entail that $\operatorname{WF}(D_{\kappa,\mathcal{U}})\subseteq\mathcal{C}^{\triangleright}_{\mathrm b,\mathcal{N}}|_{T^*(\mathcal{U}\times\mathcal{U})}$. At the same time, Proposition \ref{Prop: Finite-reflection branchwise causal comparison}, combined with Corollary \ref{Cor: Finite-reflection reversal covariance}, entails that the antisymmetric part of $H \vert_{\mathcal{U}\times\mathcal{U}}+H^{\mathrm{ref}}_{\kappa,\mathcal{N}}|_{\mathcal{U}\times\mathcal{U}}$ coincides, up to a smooth remainder, with $iG_\kappa$ as in Equation \eqref{Eq: positivity + CCR}. Hence
$$D_{\kappa,\mathcal{U}}-D^T_{\kappa,\mathcal{U}}\in C^\infty(\mathcal{U}\times\mathcal{U}).$$
If $\operatorname{WF}(D_{\kappa,\mathcal{U}})\neq\emptyset$, the last identity would force the same covectors to belong to $\operatorname{WF}(D^T_{\kappa,\mathcal{U}})$ in order for them to cancel out. This is impossible since the covector of the first entry of $\operatorname{WF}(D_{\kappa,\mathcal{U}})$ is future-pointing while it is past-pointing for $\operatorname{WF}(D^T_{\kappa,\mathcal{U}})$. Therefore, $D_{\kappa,\mathcal{U}}\in C^\infty(\mathcal{U}\times\mathcal{U})$ and Equation \eqref{Eq: Branchwise local Robin Hadamard form} follows with $W_{\kappa, \mathcal{U}\times\mathcal{U}}=D_{\kappa,\mathcal{U}}$. Since $H^{\mathrm{ref}}_{\kappa,\mathcal{N}}$ is microlocally admissible per hypothesis, $\omega_{2,\kappa}$ is branchwise local on $\mathcal{N}$ as per Definition \ref{Def: Admissible reflected Hadamard assembly but not yet an Avenger}.
\end{proof}

\noindent We shall now introduce the notion of Feynman parametrix $H^F_{\kappa}$ associated the Klein-Gordon operator $P$ abiding by Robin boundary conditions, see \cite{Duistermaat_1972, Islam_2024} for the corresponding structures on globally hyperbolic spacetimes with empty boundary. To this avail, first of all we introduce the Feynman component of the broken bicharacteristic relation
\begin{equation*}
    \mathcal{C}_{\mathrm{b}}^F (\mathring{\mcM}) \doteq \left\{ (x,k_x, y, k_y) \in \mathcal{C}_\mathrm{b}(\mathring{\mcM}) \, \vert  \, x \ne y \; \text{such that} \, \begin{cases} y \in J^+(x) \Rightarrow k_x \triangleright0 \\ y \in J^-(x) \Rightarrow k_x \triangleleft 0
    \end{cases} \right\}. 
\end{equation*}
Here $C_b(\mathring{\mcM})$ is the finite broken bicharacteristic relation introduced in Equation \eqref{Eq: CbM}. Notice that, on a Cauchy neighbourhood $\mathcal N$, $\mathcal C_{\mathrm{b}}^F(\mathring{\mcM})$ admits the decomposition
\begin{equation*}
     \mathcal{C}_{\mathrm{b}}^F(\mathring{\mcM}) =  \mathcal{C}^{\mathrm{dir}, F}_{\mathrm{b}}(\mathring{\mcM}) \cup \Lambda_{\mathcal{N}}^{\mathrm{ref}, F}, \qquad \Lambda_{\mathcal{N}, F}^{\mathrm{ref}} \doteq \bigcup_{(N, \alpha) \in \mathcal{I}_\mathrm{ref}} \Lambda_{N, \alpha}^F \,, 
\end{equation*}
where 
\begin{equation*}
    \Lambda_{N, \alpha}^F = (\Lambda_{N, \alpha}^\triangleright \cap \{t(x) < t(y)\}) \cup (\Lambda_{N, \alpha}^\triangleleft \cap \{t(x) > t(y)\}). 
\end{equation*}

\begin{definition}[Robin Feynman parametrix]
\label{Def: Robin Feynman Parametrix}
Let $(\mcM,g)$ be a $d$-dimensional, $d\geq2$, globally hyperbolic spacetime with timelike boundary satisfying Assumptions \ref{Ass: Infinitesimally convex}, \ref{Ass: Finite regular reflections} and \ref{Ass: No Zeno points}. We call $H^{\mathrm F}_\kappa\in\mathcal D'(\mcM\times\mcM)$ a {\bf Robin Feynman parametrix} for $P$ if the following conditions hold:
\begin{enumerate}
\item $H^{\mathrm F}_\kappa$ is a bi-parametrix for the Klein--Gordon operator, namely
\begin{equation*}
    (P\otimes\mathbb I)H^{\mathrm F}_\kappa
    =\delta+R_1,
    \qquad
    (\mathbb I\otimes P)H^{\mathrm F}_\kappa
    =\delta+R_2,
\end{equation*}
where $R_1,R_2\in C^\infty(\mcM\times\mcM)$.
\item $H^{\mathrm F}_\kappa$ is compatible with the Robin boundary condition,
\[
\left.
\big((\nabla_n+\kappa\mathbb I)\otimes\mathbb I\big)
H^{\mathrm F}_\kappa
\right|_{\partial\mcM}
=0,
\qquad \kappa\in\mathbb R,
\]
where $n$ denotes the inward-pointing unit normal vector field to $\partial\mcM$.

\item Its singular structure is given by
\begin{equation}
\label{Eq: WF Robin Feynman}
WF\!\left(
H^{\mathrm F}_\kappa\big|_{\mcM\times\mathring{\mcM}}
\right)
=
\mathcal{C}_{\mathrm{b}}^F (\mathring{\mcM})
\cup
N^*\operatorname{Diag}_2(\mathring{\mcM})\setminus\{\boldsymbol{0}\}.
\end{equation}
\end{enumerate}
\end{definition}

\begin{remark}\label{Rem: Feynmann Parametrix on Curved Backgrounds}
We observe that, generalizing \cite[Prop. 5.17]{Costeri_25} to our setting it turns out that, given a Robin Hadamard two-point function $\omega_{2, \kappa}$ and the unique advanced Robin propagator $G^-_{\kappa}$ associated to $P$ as per Proposition \ref{Prop: Existence Robin propagators}, it is always possible to construct a bi-distribution $H^F_{\kappa}\in\mathcal{D}^\prime(\mcM\times\mcM)$ abiding by Definition \ref{Def: Robin Feynman Parametrix}, setting
\begin{equation}
	H^F_{\kappa} \doteq G^-_{\kappa} + i \omega_{2,\kappa}. 
\end{equation}
\end{remark}

\section{Existence of Hadamard States}\label{Sec: Existence oF Hadamard States}

In this section we establish the existence of Robin-Hadamard states proving that this class is non-empty. We divide our analysis into two steps. In the first, we focus on a suitable class of standard static, globally hyperbolic spacetimes with timelike boundary. Using the results of Section \ref{Sec: Fundamental Solutions on Static Spacetimes} we construct the two-point correlation function of the ground state proving that it is of Hadamard form. In the second step, we employ a deformation argument to drop the assumption of invariance of the metric under time translations.

\subsection{Hadamard States on Static Spacetimes}\label{Sec: Hadamard States on Static Spacetimes}

We construct a distinguished example of Robin-Hadamard two-point correlation function, which is often associated both in the theoretical and in the mathematical physics literature to the concept of ground state. Here we generalize a renown result on static, globally hyperbolic spacetimes with empty boundary, see \cite{Junker, Sahlmann:2000fh}. We highlight that, contrary to what occurs in the construction of the advanced and retarded propagators, we need to introduce an additional spectral condition. This is necessary in order to avoid the insurgence of infrared divergences, a feature which occurs also on globally hyperbolic spacetimes with empty boundary, {\it e.g.}, a Hadamard ground state does not exist for a massless real scalar field on the two-dimensional Minkowski spacetime.

\begin{proposition}\label{Prop: Ground 2-point function is Hadamard}
	Let $P$ be the Klein-Gordon operator as in Equation \eqref{Eq: KG equation} on a $d$-dimensional, $d\geq 2$, globally hyperbolic spacetime with a timelike boundary $(\mcM, g)$ which abides by Assumptions \ref{Ass: Infinitesimally convex} and \ref{Ass: No Zeno points}. Denote by $G_\kappa=G^-_\kappa-G^+_\kappa$ the Pauli-Jordan propagator as per Proposition \ref{Prop: Existence Robin propagators} and by $\Sigma$ a Cauchy surface of $(\mcM,g)$. Then, if
	\begin{enumerate}
		\item the spectrum $\sigma(A_{m,\widetilde{\kappa}})\subseteq [0,\infty)$, 
		\item $C^\infty_0(\Sigma)\subset\mathrm{D}(A^{-1/4}_{m,\widetilde{\kappa}})$,
	\end{enumerate}
	then, it is a Robin-Hadamard two-point correlation function  $\omega_{2,\kappa}\in\mathcal{D}^\prime(\mcM\times\mcM)$ such that for all $f,f^\prime\in C^\infty_0(\mcM)$
	\begin{equation}\label{Eq: Hadamard ground 2-point function}
		\omega_{2,\kappa}(f,f^\prime)=\int\limits_{\bR^2}\left(\beta^{\frac{d+2}{4}}f(t),A^{-\frac{1}{2}}_{m,\widetilde{\kappa}}e^{i(t-t^\prime)A_{m,\widetilde{\kappa}}}\beta^{\frac{d+2}{4}}f^\prime(t^\prime)\right)_{L^2(\Sigma,\beta^{-1}h)} \, dtdt^\prime,
	\end{equation}
	where $A_{m\widetilde{\kappa}}$ is defined in Proposition \ref{Prop: Existence Robin propagators} while $(\cdot,\cdot)_{L^2(\Sigma,\beta^{-1}h)}$ is the $L^2$-inner product on the Cauchy surface endowed with the optical metric $\beta^{-1}h$. Furthermore, denoting by $\tau \doteq t-t^\prime$, it holds that $\omega_{2,\kappa}$ identifies a tempered distribution $\mathcal{W}^\otimes_\kappa$ along the $\tau$-direction and, at the level of the corresponding Fourier transform,
	\begin{equation}\label{Eq: Fourier Identity ground 2-point function}
		\mathcal{W}^\otimes_\kappa(k^0)=i\Theta(k^0)\widehat{G}_\kappa(k^0),
	\end{equation}
	where $\Theta$ is the Heaviside step-function. We call $\omega_{2,\kappa}$ the {\bf ground state two-point correlation function}.
\end{proposition}

\begin{proof}
	Setting for simplicity of notation $A\doteq A_{m,\widetilde{\kappa}}$ and recalling that on $\mathrm{D}(A)\subset L^2(\Sigma,\beta^{-1}h)$, $A=A^*$,  
	$$\omega_{2,\kappa}(f,f^\prime)=\int\limits_{\bR^2} \left(A^{-\frac{1}{4}}\beta^{\frac{d+2}{4}}f(t),e^{i(t-t^\prime)A}A^{-\frac{1}{4}}\beta^{\frac{d+2}{4}}f^\prime(t^\prime)\right)_{L^2(\Sigma,\beta^{-1}h)} \, dtdt^\prime.$$
	Since $C^\infty_0(\Sigma)$ is contained in the domain of $A^{-\frac{1}{4}}$, the $L^2$-pairing is well-defined, while the condition on the spectrum also entails the the right-hand side of Equation \eqref{Eq: Hadamard ground 2-point function} is finite. We subdivide the rest of the proof into 4 parts:
	
	\vskip .2cm
	
	{\em Part 1.} Equation \eqref{Eq: Robin eqs of motion} holds true per construction since $A$ is chosen as the self-adjoint extension encoding Robin boundary conditions. Similarly, focusing on Equation \eqref{Eq: positivity + CCR}, we can observe by direct inspection that the antisymmetric part of Equation \eqref{Eq: Hadamard ground 2-point function} coincides with $i G_\kappa$, {\it cf.}, Equation \eqref{Eq: Spectral Gkappa}. Furthermore, for every $f\in C^\infty_0(\Sigma)$
	$$0\leq \omega_2(f,f)=\left|\int\limits_{\bR}e^{itA}A^{-\frac{1}{4}}\beta^{\frac{d+2}{4}}f(t) \, dt\right|^2\leq\int\limits_{\bR}\|e^{itA}A^{-\frac{1}{4}}\beta^{\frac{d+2}{4}}f(t)\|^2_{L^2(\Sigma,\beta^{-1}h)} \, dt,$$
	where, once more, we have employed that $C^\infty_0(\Sigma)\subset\mathrm{D}(A^{-1/4})$ to infer that the last integral is finite.
	
	\vskip .2cm
	
	{\em Part 2.} Here we establish that $\omega_{2,\kappa}$ identifies an element lying in $\mathcal{S}^\prime(\bR;\mathcal{D}^\prime(\Sigma\times\Sigma))$ where the $\bR$-direction refers, at the level of integral kernels, to $\tau \doteq t-t^\prime$. We denote by $\mathcal{W}(\tau,\underline{x},\underline{y})$, the integral kernel of the linear map such that, for all $\varphi(\tau)\in\mathcal{S}(\bR)$ and for all  $u,v\in C^\infty_0(\Sigma)$, yields
	$$\mathcal{W}_\kappa(\varphi,u,v)=\int\limits_{\bR} \varphi(\tau)\left(A^{-\frac{1}{4}}\beta^{\frac{d+2}{4}}u,e^{i\tau A}A^{-\frac{1}{4}}\beta^{\frac{d+2}{4}}v\right)_{L^2(\Sigma,\beta^{-1}h)} \, d\tau.$$
	It descends that
	\begin{gather}
		|\mathcal{W}_\kappa(\varphi,u,v)|\leq\int\limits_{\bR}\left|\varphi(\tau)\left(A^{-\frac{1}{4}}\beta^{\frac{d+2}{4}}u,e^{i\tau A}A^{-\frac{1}{4}}\beta^{\frac{d+2}{4}}v\right)_{L^2(\Sigma,\beta^{-1}h)}\right|\, d\tau \leq\notag \\
		\leq\int\limits_{\bR} \left|\varphi(\tau)\right|\|A^{-\frac{1}{4}}\beta^{\frac{d+2}{4}}u\|_{L^2(\Sigma,\beta^{-1}h)}\;\|A^{-\frac{1}{4}}\beta^{\frac{d+2}{4}}v\|_{L^2(\Sigma,\beta^{-1}h)} \, d\tau = \notag \\
		=\|\varphi\|_{L^1(\mathbb{R})}\|A^{-\frac{1}{4}}\beta^{\frac{d+2}{4}}u\|_{L^2(\Sigma,\beta^{-1}h)}\;\|A^{-\frac{1}{4}}\beta^{\frac{d+2}{4}}v\|_{L^2(\Sigma,\beta^{-1}h)},\label{Eq: 2-point inequality}
	\end{gather}
	where we used the Cauchy-Schwartz inequality combined with the hypothesis on the spectrum of $A$. On the one hand, we observe that $\varphi\to\|\varphi\|_{L^1}$ is a continuous seminorm on $\mathcal{S}(\bR)$. On the other hand, keeping in mind the hypothesis $C^\infty_0(\Sigma)\subset\mathrm{D}(A^{-\frac{1}{4}})$, we observe that $A^{-\frac{1}{4}}$ is a self-adjoint, positive operator, hence closed, \cite[Thm 5.18]{Moretti_2017}. For every compact subset $K\subseteq\Sigma$, set $C^\infty_K(\Sigma) \doteq \{u\in C^\infty_0(\Sigma)\;|\;\mathrm{supp} \, u\subseteq K\}$, which is a Fr\'echet space. The restriction of $A^{-\frac{1}{4}}$ to $C^\infty_K(\Sigma)$ is also closed in the Fr\'echet topology being the canonical injection $\iota_K:C^\infty_K(\Sigma)\to L^2(\Sigma; \beta^{-1}h)$ continuous. Therefore we can apply \cite[Thm. B16]{Grubb68} to conclude that, on this domain $A^{-\frac{1}{4}}$ is continuous. Taking an exhaustion by compacts of $\Sigma$, that is $\{K_j\}_{j\in\mathbb{N}}$ with $K_j$ compact $\forall j$ and $K_j\Subset K_{j+1}$, while $\bigcup_{j\in\mathbb{N}}K_j=\Sigma$, we have that $C^\infty_0(\Sigma)$ can be recovered as an inductive limit of Fr\'echet spaces, that is $C^\infty_0(\Sigma)=\bigcup_{j\in\mathbb{N}}C^\infty_{K_j}(\Sigma)$. Hence we can apply \cite[Thm B.18, (c)]{Grubb68} to conclude that $A^{-\frac{1}{4}}$ is continuous on $C^\infty_0(\Sigma)$. This observation combined with Equation \eqref{Eq: 2-point inequality} entails that the map
	$$\mathcal{S}(\bR)\times C^\infty_0(\Sigma)\times C^\infty_0(\Sigma)\ni(\varphi,u,v)\mapsto\int\limits_\bR \varphi(\tau)\left(A^{-\frac{1}{4}}\beta^{\frac{d+2}{4}}u,e^{i\tau A}A^{-\frac{1}{4}}\beta^{\frac{d+2}{4}}v\right)_{L^2(\Sigma,\beta^{-1}h)} \, d\tau,$$
	is multilinear and jointly continuous. Equipping $\mathcal{S}(\bR)\times C^\infty_0(\Sigma)\times C^\infty_0(\Sigma)$ with the projective tensor-product topology \cite{Treves}, we end up with $\mathcal{W}^\otimes_\kappa$, linear and continuous on $\mathcal{S}(\bR)\otimes C^\infty_0(\Sigma)\otimes C^\infty_0(\Sigma)$ and such that $\mathcal{W}_\kappa(\varphi,u,v)=\mathcal{W}^\otimes_\kappa(\varphi\otimes u\otimes v)$. This new application extends uniquely to the completion of the tensor product denoted by $\hat{\otimes}$, that is, $\mathcal{W}^\otimes_\kappa$ is continuous on $\mathcal{S}(\bR)\hat{\otimes}C^\infty_0(\Sigma)\hat{\otimes}C^\infty_0(\Sigma)$. Since this space is isomorphic to $\mathcal{S}(\bR)\otimes C^\infty_0(\Sigma\times\Sigma)$, we can conclude that $\mathcal{W}^\otimes_\kappa\in\left(\mathcal{S}(\bR)\otimes C^\infty_0(\Sigma\times\Sigma)\right)^\prime$, which is nothing but $\mathcal{S}^\prime(\bR;\mathcal{D}^\prime(\Sigma\times\Sigma))$.
	
	\vskip .2cm
	
	{\em Part 3.} The next step consists of proving Equation \eqref{Eq: Fourier Identity ground 2-point function}. This is a consequence mainly of the spectral theorem applied to $A^{\frac{1}{2}}$, see for example \cite{Moretti_2017}. To start with, we denote by $$E_{A^{\frac{1}{2}}}:\mathsf{B}(\bR)\to\mathcal{B}(L^2(\Sigma, \beta^{-1}h)),$$ the associated projection-valued measure where $\mathsf{B}(\bR)$ is the Borel $\sigma$-algebra over $\bR$. We can therefore introduce the complex-valued measure over $\mathsf{B}(\bR)$:
	$$\nu_{u,v}(I) \doteq \left(A^{-\frac{1}{4}}\beta^{\frac{d+2}{4}}u,E_{A^{\frac{1}{2}}}(I)A^{-\frac{1}{4}}\beta^{\frac{d+2}{4}}v\right)_{L^2(\Sigma,\beta^{-1}h)}, \quad u,v \in C^\infty_0(\Sigma),$$
	which is well-defined on account of the hypotheses on the spectrum of $A$ and of the inclusion $C^\infty_0(\Sigma)\subset\mathrm{D}(A^{-1/4})$. Accordingly Equations \eqref{Eq: Hadamard ground 2-point function} and \eqref{Eq: Beta Gkappa} combined with Equation \eqref{Eq: Spectral Gkappa} yield via partial evaluation the following integral kernels
	$$G_{\kappa}(\tau,u\otimes v)=\int\limits_0^\infty\sin(\lambda\tau)d\nu_{u,v}(\lambda)\quad\mathrm{and}\quad\omega_{2,\kappa}(\tau,u\otimes v)=\int\limits_0^\infty \frac{e^{i\tau\lambda}}{2}d\nu_{u,v}(\lambda),$$
	where $\lambda$ is the spectral parameter. Since both $G_\kappa$ and $\omega_{2,\kappa}$ are tempered distributions along the $\tau$-variable as per Part 2., we can take the Fourier transform following the convention that, for $F\in L^1(\mathbb{R})$, $\widehat{F}(k^0)=\int\limits_{\bR} \, e^{-ik^0\tau}F(\tau) \, d\tau$. This yields that, working at the level of integral kernels,
	$$G_{\kappa}(k^0, u\otimes v)=-i\pi\int\limits_0^\infty\left(\delta(k^0-\lambda)+\delta(k^0+\lambda)\right)d\nu_{u,v}(\lambda)\;\mathrm{and}\;\omega_{2,\kappa}(k^0, u\otimes v)=\pi\int\limits_0^\infty \delta(k^0-\lambda)d\nu_{u,v}(\lambda).$$
	Equation \eqref{Eq: Fourier Identity ground 2-point function} descends by multiplying the first expression by $i\Theta(k^0)$.

	\vskip .2cm
	
	{\em Part 4.} At last we need to establish that the wavefront set of $\omega_{2,\kappa}$ is of Hadamard form, that is Equation \eqref{Eq: Hadamard WF set Robin} holds true. We follow a reasoning very similar to that of the proof of Theorem \ref{Thm: WF of Propagators} and therefore we will only highlight the key differences. To start with we wish to prove that $\mathrm{WF}(\omega_{2,\kappa})\subseteq\mathcal{C}_{\mathrm b}^\triangleright(\mathring{\mcM})$. In comparison with the proof of Theorem \ref{Thm: WF of Propagators}, the first three steps are identical since they only rely on the underlying bidistribution being in the kernel of $(P\otimes\mathbb{I}), (\mathbb{I}\otimes P)$, on $(\mcM,g)$ being stationary and on the propagation of singularity theorem proved in \cite{Melrose_1978}. The first difference occurs in the fourth step, for which it is necessary to invoke a different argument. More precisely we must show that, for any $y\in\mathring{\mcM}$, $(y,k_y,y,k^\prime_y)\in\mathrm{WF}(\omega_{2,\kappa})$ if and only if $k_y\triangleright 0$ and $k^\prime_y=-k_y$. As a preliminary step, let us focus on the distribution $\mathcal{W}^\otimes_\kappa$ and let $(\tau_0,\underline{x}_0,\underline{y}_0,k_{\tau_0}, k_{\underline{x}_0},k_{\underline{y}_0})\in\mathrm{WF}(\mathcal{W}^\otimes_\kappa)$ with $\underline{x}_0,\underline{y}_0\in\mathring{\Sigma}$. Consider $\chi(\tau,\underline{x},\underline{y})=\chi_1(\tau)\chi_2(\underline{x},\underline{y})\in C^\infty_0(\bR\times\Sigma\times\Sigma)$ be such that $\chi(\tau_0,\underline{x}_0,\underline{y}_0)=1$. Hence, denoting by $\mathcal{F}_\tau$ the Fourier transform along the $\tau$-direction, 
	$$\mathcal{F}_\tau(\chi\mathcal{W}^\otimes_\kappa)=\mathcal{F}_\tau(\chi_1)\star(\chi_2(\underline{x},\underline{y})\mathcal{F}_\tau(\mathcal{W}^\otimes_\kappa))=\mathcal{F}_\tau(\chi_1)\star(\chi_2(\underline{x},\underline{y})\Theta(k^0)\widehat{G}_\kappa(k^0)),$$
	where we used Equation \eqref{Eq: Fourier Identity ground 2-point function}. Since $\Theta(k^0)\widehat{G}(k^0)$ vanishes for $k^0<0$, it turns out the only $\mathcal{F}_\tau(\chi_1)$ contributes to the negative frequencies of 
	$\mathcal{F}_\tau(\chi\mathcal{W}^\otimes_\kappa)$ which entails rapid decrease along such directions. In other words $k_{\tau_0}>0$. Switching back to $\omega_{2,\kappa}$ this entails that, if $(x,k_x,y,-k_y)\in\mathrm{WF}(\omega_{2,\kappa})$ then $k_x\triangleright 0$. Let us now assume that $(y,k_y,y,k^\prime_y)\in\mathrm{WF}(\omega_{2,\kappa})$.  Being the underlying metric stationary, it holds that $(k_y)_0+(k^\prime_y)_0=0$. Since $k_y\triangleright 0$, then $k^\prime_y$ must be past pointing. 
	
	At this point let $\omega_{2,\kappa}^T$ be the bi-distribution such that $\omega_{2,\kappa}^T(x,y)=\omega_{2,\kappa}(y,x)$. If $(y,k_y,y,k^\prime_y)\in\mathrm{WF}(\omega_{2,\kappa})$, $k_y$ and $k^\prime_y$ are future and past pointing, respectively. This entails that $(y,k^\prime_y,y,k_y)\in\mathrm{WF}(\omega_{2,\kappa}^T)$ but $(y,k_y,y,k^\prime_y)\notin\mathrm{WF}(\omega_{2,\kappa}^T)$. To draw the sought conclusion, we recall that, as a consequence of Definition \ref{Def: Robin Hadamard two-point function} $\omega_{2,\kappa}-\omega_{2,\kappa}^T=iG_\kappa$. Hence, given $(y,k_y,y,k^\prime_y)\in\mathrm{WF}(\omega_{2,\kappa})$, then $(y,k_y,y,k^\prime_y)\in\mathrm{WF}(G_\kappa)$ since this point does not lie in the wavefront set of $\omega_{2,\kappa}^T$. Yet, as established in Step 4. in the proof of Theorem \ref{Thm: WF of Propagators}, $(y,k_y,y,k^\prime_y)\in\mathrm{WF}(G_\kappa)$ if and only if $k^\prime_y=-k_y$. The converse statement follows by a similar argument. Consider $(y,k_y,y,-k_y)\in T^*(\mathring{\mcM}\times\mathring{\mcM})\setminus\{\boldsymbol{0}\}$ with $k_y\triangleright 0$. As in the previous part we can infer that $(y,k_y,y,-k_y)\in\mathrm{WF}(G_\kappa)$ and $(y,k_y,y,-k_y)\notin\mathrm{WF}(\omega_{2,\kappa}^T)$ since $k_y$ ought to be past-directed. Hence, using that $\omega_{2,\kappa}=iG_\kappa+\omega_{2,\kappa}^T$, we can conclude that $(y,k_y,y,-k_y)\in\mathrm{WF}(\omega_{2,\kappa})$.
	
	Having established that Step 4. in the proof of Theorem \ref{Thm: WF of Propagators} is valid also for $\omega_{2,\kappa}$, we can consider a point $(x,k_x,y,-k_y)\in\mathrm{WF}(\omega_{2,\kappa})$. It holds that $k_x\triangleright 0$ and $(w,k_w,y,-k_y)\in\mathrm{WF}(\omega_{2,\kappa})$ whenever $(x,k_x)\sim_{\mathrm b} (w,k_w)$. Using the no glancing hypothesis and the lack of Zeno points, we can consider without loss of generality $w\in\mathring{\mcM}$ such that the time component $t(w)=t(y)$. If $w\in\partial\mcM$, we can propagate $(y,k_y)$ along the bicharacteristic flow to a different point so to avoid this degenerate case. Once more the identity $\omega_{2,\kappa}-\omega_{2,\kappa}^T=iG_\kappa$ in combination with $(w,k_w,y,-k_y)\notin\mathrm{WF}(\omega^T_{2,\kappa})$ since $k_w\triangleright 0$ entails that $(w,k_w,y,-k_y)\in\mathrm{WF}(G_\kappa)$. Yet, since $t(w)=t(y)$, this can happen only if $w=y$, which forces, in turn, $k_w=k_y$. This entails that $\mathrm{WF}(\omega_{2,\kappa})\subseteq\mathcal{C}_{\mathrm b}^\triangleright(\mathring{\mcM})$. The reverse inclusion can be shown as follows. Suppose that $(x,k_x,y,-k_y)\in \mathcal{C}_{\mathrm b}^\triangleright(\mathring{\mcM})$ but it does not lie in $\mathrm{WF}(\omega_{2,\kappa})$. By propagation of singularities as per \cite{Melrose_1978}, $(w,k_w,y,-k_y)\notin\mathrm{WF}(\omega_{2,\kappa})$ for all $(w,k_w)\sim_{\mathrm b}(y,k_y)$. In particular, this entails $(y,k_y,y,-k_y)\notin\mathrm{WF}(\omega_{2,\kappa})$. Yet the identity $\omega_{2,\kappa}-\omega_{2,\kappa}^T=iG_\kappa$ combined with $(y,k_y,y,-k_y)\notin\mathrm{WF}(\omega_{2,\kappa}^T)$ since $k_y$ is future directed entails that $(y,k_y,y,-k_y)\notin\mathrm{WF}(G_\kappa)$ which is a contradiction.
\end{proof}

\begin{remark}
	Observe that, in Proposition \ref{Prop: Ground 2-point function is Hadamard}, it would suffice to require that $\sigma(A_{m,\widetilde{\kappa}})\subseteq (0,\infty)$. This guarantees the absence of infrared divergences, but, at the same time, it excludes physically significant examples such as a massless scalar field in half-Minkowski spacetime. This is the reason why we consider a less stringent class of assumptions.
\end{remark}

\begin{remark}\label{Rem: Ground state branchwise local}
Under the assumptions of Proposition \ref{Prop: Ground 2-point function is Hadamard} and Theorem \ref{Thm: Global-to-Local}, the ground-state two-point correlation function $\omega_{2,\kappa}$ is an example of a branchwise local Robin-Hadamard two-point correlation function on $\mathcal{N}$. As a matter of fact, Proposition \ref{Prop: Ground 2-point function is Hadamard} establishes Equation \eqref{Eq: Radzikowski counterpart WF}, hence the conclusion follows from Theorem \ref{Thm: Global-to-Local}.
\end{remark}

\subsection{Deformation Argument}\label{Sec: Deformation Argument}

In this last section of the paper, we prove existence of Hadamard states removing the assumption of the underlying spacetime being static as in Proposition \ref{Prop: Ground 2-point function is Hadamard}. This is obtained extending to the case in hand a deformation argument, first appeared in \cite{Fulling:1981cf}, which is traditionally used in the mathematical physics literature on globally hyperbolic spacetimes with empty boundary. At a geometric level we follow here the realization of this argument as presented by M\"uller in \cite{Muller}.

\begin{definition}\label{Def: Deformation}
	Let $(\mcM,g)$ be a globally hyperbolic spacetime with timelike boundary. We call it {\bf deformable} if, taking into account the isometric identification $\mcM\simeq\bR\times\Sigma$, there exists two globally hyperbolic spacetimes with timelike boundary $(\mcM,g_0)$ and $(\mcM,g_I)$ such that $(\mcM,g_0)$ is standard static and there exists $t_0,t_1,\in\bR$ with $t_0<t_1$ for which $g_I=g_0$ for all $t<t_0$, while $g_I =g$ for all $t>t_1$. 
\end{definition}

As a first step, we show that any infinitesimally convex, globally hyperbolic spacetime with timelike boundary can be deformed to a static background with Cauchy surfaces of bounded geometry, while preserving Assumption \ref{Ass: Infinitesimally convex}, see Section \ref{Sec: Fundamental Solutions on Static Spacetimes}.

\begin{proposition}\label{Prop: Existence of geometric Deformation}
Let $(\mcM,g)$ be a globally hyperbolic, infinitesimally convex, spacetime with timelike boundary. Then $(\mcM,g)$ is deformable in the sense of Definition \ref{Def: Deformation}. Moreover, the static metric $g_0$ may be chosen so that its Cauchy surface is of bounded geometry as per Definition \ref{Def: Bounded_geometry_submanifold}, while the interpolating metric $g_I$ may be chosen so that $\partial\mathcal M$ remains weakly infinitesimally convex.
\end{proposition}

\begin{proof}
In view of Equation \eqref{Eq: line element psi*g}, we can identify isometrically $\mcM$ with $\bR\times\Sigma$, with line element $ds^2=-\beta^2dt^2+h_t$. On $\Sigma$, fix a complete Riemannian metric $h_0$ of bounded geometry such that $\partial\Sigma$ is infinitesimally convex, as per Assumption \ref{Ass: Infinitesimally convex}. The metric $g_0$ with line element $ds_0^2=-dt^2+h_0$ identifies $(\mcM,g_0)$, an ultrastatic globally hyperbolic spacetime with timelike boundary.

We construct a preliminary interpolation between $g_0$ and $g$. Choose $\chi\equiv\chi(t)\in C^\infty(\bR,[0,1])$ such that $\chi=0$ for $t<t_0$ and $\chi=1$ for $t>t_1$, where $t_0<t_1$. Setting
\begin{equation}
\lambda=(1-\chi)+\chi\beta^2\quad\textrm{and}\quad\widehat{h}_t=(1-\chi)h_0+\chi h_t.
\end{equation}
Per construction the Lorentzian metric $\widehat{g}$ whose line element is $ds^2=-\lambda dt^2+\widehat{h}_t$ is such that $\widehat g=g_0$ for $t<t_0$ while $\widehat g=g$ for $t>t_1$. Yet, is it not guaranteed that $(\mcM,\widehat{g})$ is globally hyperbolic. To enforce this property, we employ the argument of \cite[Thm.~1]{Muller}. Following the same rational used in the proof of Proposition \ref{Prop: Existence of Solutions}, for every $(t,x)\in[t_0,t_1]\times\Sigma$, we can exploit that $\widehat h_t(x)$ and $h_0(x)$ induce equivalent inner products. Hence, one can choose
$R\in C^\infty(\bR\times\Sigma,(0,\infty))$ such that on $[t_0,t_1]\times\Sigma$
\begin{equation}\label{Eq: Muller domination}
R\,\widehat{h}_t>\max\{1,\lambda\}\,h_0
\end{equation}
where the inequality is in the sense of quadratic form, being the metrics Riemannian. Let $\eta\in C^\infty_0(\bR,[0,1])$ be such that there exists $\epsilon>0$ for which $\eta=1$ on $(t_0-\epsilon,t_1+\epsilon)$. Setting $F=1+\eta R$, we can consider then the metric
\begin{equation}\label{Eq: Muller metric}
g_I=-\lambda dt^2+F\widehat h_t.
\end{equation}
Here $t$ is a temporal function and every causal curve can be parametrized as
\begin{equation}
\gamma(t)=(t,x(t))
\end{equation}
where $x$ are coordinates on $\Sigma$. Since $g_M(\dot{\gamma},\dot{\gamma})\leq 0$, 
\begin{equation}
F\widehat{h}_t(\dot x,\dot x)\leq\lambda.
\end{equation}
On the region where $\eta=1$, it holds $F=1+R>R$. Equation \eqref{Eq: Muller domination} yields
\begin{equation}\label{Eq: Estimate on h0}
\max\{1,\lambda\}\,h_0(\dot{x},\dot{x})<F\widehat{h}_t(\dot{x},\dot{x})\leq\lambda.\Longrightarrow
h_0(\dot x,\dot x)\leq\frac{\lambda}{\max\{1,\lambda\}}\leq 1.
\end{equation}
It follows that, for $s,t\in[t_0,t_1]$,
\begin{equation}\label{Eq: Muller distance estimate}
d_{h_0}(x(t),x(s))\leq\int_s^t\sqrt{h_0(\dot x,\dot x)}\,d\tau	\leq |t-s|.
\end{equation}

To prove that $(\mcM,g_I)$ is globally hyperbolic, we can fine tune $\eta$ so that there exists $t_-<t_0<t_1<t_+$ for which $\operatorname{supp}\eta\subset(t_-,t_+)$. Suppose that a
future-inextensible causal curve $\gamma(t)=(t,x(t))$ is such that 
\begin{equation}
	b\doteq\sup t(\gamma)<+\infty.
\end{equation}
If $b\in[t_0,t_1]$, Equation \eqref{Eq: Muller distance estimate} implies that there exists $x_b\in\{b\}\times\Sigma$ such that $x(t)\longrightarrow x_b$. This point might even lie at the boundary. Hence $\gamma$ has an endpoint $(b,x_b)$ but it can be continuously extended along the timelike $t$-direction, in contradiction with being future inextensible. If $b>t_1$, the curve $\gamma$ has entered the region where $g_I$ coincides with $g$ and, having an endpoint on $(b,x_b)$ would contradict the hypothesis of $(\mcM, g)$ being globally hyperbolic. Consider now the region $t>t_1$. There $\lambda=\beta^2$ and
$\widehat h_t=h_t$, so that
\begin{equation}
	g_I =-\beta^2dt^2+Fh_t.
\end{equation}
Since $F\geq1$, every $g_I$-causal vector is also $g$-causal. Hence, if $b>t_1$, a final segment of $c$ is also a future-directed $g$-causal curve. If such a segment had an endpoint it would also be
extendible as a $g_I$-causal curve, since the $t$-direction is $g_I$-timelike. It is therefore future-inextensible also as a $g$-causal curve. This is impossible because the level sets of $t$
are Cauchy surfaces for $(\mcM,g)$, and hence every future-inextensible $g$-causal curve has $t\to+\infty$. If $b<t_0$ one can apply a similar argument to $(\mcM,g_0)$, hence concluding that $(\mcM,g_I)$ is globally hyperbolic.

It remains to ensure that the interpolation preserves infinitesimal convexity of the boundary. Let $K_0$ and $K_1$ denote the second fundamental forms of $\partial\mcM$ with respect to $g_0$
and $g$, respectively and set
\begin{equation}
	K=(1-\chi)K_0+\chi K_1.
\end{equation}
Per linearity $K$ abides by Assumption \ref{Ass: Infinitesimally convex}. To infer the same result for $g_I$, we can exploit that infinitesimal convexity is a property at the boundary and we can therefore exploit the freedom to modify $g_I$ in a small neighbourhood of $\partial\mcM$. To this end, fixing a point $p\in\partial\mcM$ and using Gaussian normal coordinates $(z,y)$ centred thereat, it holds that $z=0$ on $\partial\mcM$, $\partial_z$ coincides with the inward-pointing unit normal and, locally,
\begin{equation}
g_I=dz^2+\alpha_z,
\end{equation}
where $\alpha_z$ is a $z$-dependent Lorentzian metric. It holds in addition that $K_{g_I}	=-\frac{1}{2}\partial_z\alpha_z\big|_{z=0}$. Set 
\begin{equation}
	\Delta K=K-K_{g_M}.
\end{equation}
Choose a smooth cutoff $\rho\equiv\rho(z)$ with $\rho|_{z=0}=1$ and replace $\alpha_z$ by
\begin{equation}
\widetilde\alpha_z=\alpha_z-2z\rho\,\Delta K\Longrightarrow -\frac{1}{2}	\partial_z\widetilde\alpha_z\big|_{z=0}=K_{g_I}+\Delta K=K.
\end{equation}
This entails that Assumption \ref{Ass: Infinitesimally convex} holds true. 

Observe that, although the construction is performed locally, it can be patched by means of a locally finite partition of unity. Equation \eqref{Eq: Muller domination} entails that Equation \eqref{Eq: Estimate on h0} is left unaffected by small perturbations of the metric. This is the case for the correction $-2z\rho\,\Delta K$ and the argument used above to establish global hyperbolicity remains valid.
\end{proof}

\noindent To conclude we establish existence of Hadamard states and we follow the same procedure as in \cite[Thm 5.2]{Dappiaggi-Marta_2021}. This is based on two results, the first of which is also known as {\em (classical) time-slice axiom}.

\begin{lemma}\label{Lem: Time-Slice Axiom}
	Let $(\mcM,g)$ be globally hyperbolic spacetime with timelike boundary, on top of which we consider the Klein-Gordon operator $P$ with Robin boundary conditions. Denoting by $\mathcal{G}^\pm_\kappa$ the advanced and retarded Green operators as per Theorem \ref{Thm: Existence of Propagators} and by $\mathcal{N}$ a Cauchy neighbourhood as per Definition \ref{Def: Cauchy neighbourhood}, here defined with a slight generalization as being of the form $(t_-,t_+)\times\Sigma$, then for every $f\in C^\infty_0(\mcM)$, there exists $h\in C^\infty_0(\mathcal{N})$ such that $\mathcal{G}_\kappa h=\mathcal{G}_\kappa f$ with $\mathcal{G}_\kappa \doteq \mathcal{G}^-_\kappa-\mathcal{G}^+_\kappa$.
\end{lemma}

\begin{proof}
	Consider $\chi\equiv\chi(t)\in C^\infty(\bR;[0,1])$ be such that $\chi=1$ if $t>t_1$ while $\chi=0$ if $t<t_0$. Let $h \doteq P\chi\mathcal{G}_\kappa f$. This is smooth and of compact support since $\mathcal{G}_\kappa$ is causally supported as per Proposition \ref{Prop: Causal Support of Solutions}. Observing that $\chi\mathcal{G}_\kappa f \in C^\infty_{\mathrm{pc}}(\mcM)$, see Equation \eqref{Eq: Extend Domain} and the following discussion, $\mathcal{G}^-_\kappa(P\chi\mathcal{G}_\kappa f)=\chi\mathcal{G}_\kappa f$ while $$\mathcal{G}^+_\kappa(P\chi\mathcal{G}_\kappa f)=\mathcal{G}^+_\kappa (P\mathcal{G}_\kappa f-P(\chi-1)\mathcal{G}_\kappa f)=\mathcal{G}^+_\kappa(P(1-\chi)\mathcal{G}_\kappa f)=(1-\chi)\mathcal{G}_\kappa f.$$ In the last equation we have used that $(1-\chi)\mathcal{G}_\kappa f\in C^\infty_{\mathrm{fc}}(\mcM)$. In other words we have proven that $\mathcal{G}_\kappa h=\mathcal{G}_\kappa f$.
\end{proof}

The second ingredient that we need is the counterpart of \cite[Lem 5.3]{Dappiaggi-Marta_2021} which is in turn based on \cite[Lem 5.10]{Wrochna}. Since the proof is \emph{mutatis mutandis} the same as in this last reference, we omit it.

\begin{lemma}\label{Lem: Extending the Hadamard form}
	Let $(\mcM,g)$ be globally hyperbolic spacetime with timelike boundary abiding by Assumptions \ref{Ass: Infinitesimally convex} and \ref{Ass: No Zeno points}, on top of which we consider the Klein-Gordon operator $P$ with Robin boundary conditions. Let $\mathcal{N}$ be a Cauchy neighbourhood as per Definition \ref{Def: Cauchy neighbourhood}, here defined with a slight generalization as being of the form $(t_-,t_+)\times\Sigma$. If $\omega_{2,\kappa}\in\mathcal{D}^\prime(\mcM\times \mcM)$ abides by Equations \eqref{Eq: Robin eqs of motion} and \eqref{Eq: positivity + CCR} while Equation \eqref{Eq: Hadamard WF set Robin} holds true in $\mathcal{N}$, then it is a Robin Hadamard two-point correlation function.
\end{lemma}

To conclude, we can combine this two lemmas to prove the sought deformation argument, adapting to the case in hand \cite[Thm 5.2]{Dappiaggi-Marta_2021}. In view of the relevance of the result, we report the proof.

\begin{proposition}\label{Prop: Deformation Argument}
Let $(\mcM,g)$ be a globally hyperbolic spacetime with timelike boundary satisfying Assumptions \ref{Ass: Infinitesimally convex} and \ref{Ass: No Zeno points}, and let $P$ be the Klein--Gordon operator on $(\mcM,g)$ subject to Robin boundary conditions. Let $(\mcM,g_0)$ and $(\mcM,g_I)$ be, respectively, the static reference spacetime and the interpolating spacetime constructed in Proposition \ref{Prop: Existence of geometric Deformation}. Assume, in addition, that both auxiliary spacetimes satisfy Assumption \ref{Ass: No Zeno points} and that a ground state in the sense of Proposition \ref{Prop: Ground 2-point function is Hadamard} exists on $(\mcM, g_0)$. Then there exists a Robin-Hadamard two-point correlation function for $P$ on $(\mcM,g)$.
\end{proposition}

\begin{proof}
On account of Proposition \ref{Prop: Existence of geometric Deformation}, there exists a globally hyperbolic spacetime $(\mcM,g_I)$ such that $(\mcM,g_0)$ is standard static and there exists $t_0,t_1\in\bR$ with $t_0 <t_1$ for which $g_I=g_0$ for all $t<t_0$, while $g_I=g$ for all $t>t_1$. Let $P_0$, $P_I$ and $P$ denote the Klein-Gordon operators associated with $g_0$, $g_I$ and $g$, respectively. 

On account of Proposition \ref{Prop: Ground 2-point function is Hadamard}, on $(\mcM,g_0)$ we can identify a Robin Hadamard two-point correlation function which we denote by $\omega^{(0)}_{2,\kappa}$. Calling $M_0$ the static region of $(\mcM, g_I)$, the restriction $\omega^{(0)}_{2,\kappa}|_{M_0\times M_0}$ identifies a two-point distribution of Hadamard form. Consider a Cauchy neighbourhood $\mathcal{N}$ of the form $(t_-,t_+)\times\Sigma$ such that $t_+<t_0$, namely $g_I|_{\mathcal{N}}=g_0$. Using the time-slice axiom in Lemma \ref{Lem: Time-Slice Axiom}, for any pair of test-functions $f,f^\prime\in\mathcal{D}(\mcM)$, we set $h=P_I(\chi\mathcal{G}_{I,\kappa} f)$ and $h^\prime=P_I(\chi\mathcal{G}_{I,\kappa} f^\prime)$ where $\mathcal{G}_{I,\kappa}=\mathcal{G}^-_{I,\kappa}-\mathcal{G}^+_{I,\kappa}$ is the Pauli-Jordan operator associated with $P_I$ on $(\mcM,g_I)$, while $\chi=\chi(t)$ is any smooth function chosen as in the proof of Lemma \ref{Lem: Time-Slice Axiom}. We set
$$\omega^\prime_{2,\kappa}(f,f^\prime)=\omega_{2,\kappa}^{(0)}(h,h^\prime).$$
Observe that $h,h^\prime\in C^\infty_0(\mathcal{N})$ and therefore the right-hand side of this identity is well-defined. In addition, since $\mathcal{G}_{I, \kappa}$ is a continuous map, see Theorem \ref{Thm: Existence of Propagators}, sequential continuity entails that $\omega^\prime_{2,\kappa}\in\mathcal{D}'(\mcM\times\mcM)$. In addition, per construction, it abides by Equations \eqref{Eq: Robin eqs of motion} and \eqref{Eq: positivity + CCR} on $(\mcM,g_I)$. Since $\omega_{2,\kappa}^{(0)}$ is of Hadamard form on $\mathcal{N}$ as it is precisely the ground state of Proposition \ref{Prop: Ground 2-point function is Hadamard}, Lemma \ref{Lem: Extending the Hadamard form} yields that $\omega^\prime_{2,\kappa}$ is of Hadamard form. To conclude it suffices to restrict $\omega_{2,\kappa}^\prime$ on a Cauchy neighbourhood $\mathcal{N}' \Subset \{t > t_1\}$ such that $g_I|_{\mathcal{N}^\prime}=g$ and apply once more Lemma \ref{Lem: Extending the Hadamard form} as detailed above.
\end{proof}

\section*{Acknowledgements}
\addcontentsline{toc}{section}{Acknowledgements}

We are grateful to Felix Finster, Onirban Islam, Nicola Pinamonti and Micha{\l} Wrochna for enlightening discussions. The work of BC has been supported by a fellowship of the University of Pavia. BC and CD both acknowledge the support of the INFN Sezione di Pavia and of Gruppo Nazionale di Fisica Matematica, part of INdAM. The work of BC is supported in part by a fellowship of the "Progetto Giovani GNFM 2025" under the project "Hadamard states for linearized Yang-Mills theories" fostered by Gruppo Nazionale di Fisica Matematica -- INdAM. BAJ-A is supported by the EPSRC Open Fellowship EP/Y014510/1. The authors acknowledge the use of ChatGPT (OpenAI) as an auxiliary tool during the preparation of this manuscript, in particular for linguistic editing, improving exposition, and discussing technical aspects of some arguments. All mathematical statements and proofs were independently checked and derived by the authors, who take full responsibility for the content of the paper.

\vskip.2cm

\noindent\textbf{Data availability statement}. Data sharing is not applicable to this article as no new data were created or analysed in this study.

\vskip .2cm

\noindent\textbf{Conflict of interest statement.} The authors certify that they have no affiliations with or involvement in any
organization or entity with any financial interest or non-financial interest in the subject matter discussed in
this manuscript.

\appendix

\section{Proof of Proposition \ref{Prop: Causal Support of Solutions}}\label{App: A}

The proof is divided into two main steps. The first, and more substantial, step is devoted to establishing a suitable energy estimate. The second then applies this estimate to derive the desired finite propagation speed.

\paragraph{Energy estimate:} We start by setting the regions of $\mcM$ that we shall consider. Given $t_0, t_1 \in \mathbb{R}$, with $t_0 < t_1$, fix $p\in \mcM$ such that $t_p=t(p) \geq t_1$ where $t:\mcM\to\bR$ denotes the underlying time function. Consider the domain
\begin{equation}\label{Eq: domain D}
	D_p\doteq J^-(p) \cap \{t_0 \le t \le t_1\}.
\end{equation}
where $J^-$ denotes the causal past in $\mcM$. Furthermore we choose $p$ so that $\operatorname{supp}(f) \subset D_p$. We observe that, since $(\mcM,g)$ is globally hyperbolic, $D_p$ is compact. Thereon, and more generally on $\mcM$, we introduce an \emph{auxiliary Riemannian metric} $g_E$ whose line element reads
\begin{equation}\label{Eq: aux Riemannian metric element}
	ds_E^2 = \beta dt^2 + h_t, 
\end{equation}
together with $\partial_\chi \doteq \beta^{-\frac{1}{2}} \partial_t$ which is the unit vector, normal to $\Sigma_t$, \textit{i.e.}, $g_E(\partial_\chi, \partial_\chi) = 1$. 

\begin{definition}\label{Def: euclidean kinetic energy}
	Let $(\mcM, g)$ be a globally hyperbolic spacetime with timelike boundary, as per Definition \ref{Def: Globally Hyperbolic} and let $u: \mcM \to \mathbb{R}$. We call {\bf Euclidean kinetic energy density} for $u$ 
	\begin{equation}\label{Eq: Euclidean kinetic energy}
		|\nabla u|_E^2 \doteq g_E(\nabla u, \nabla u) = (\partial_\chi u)^2 + |\nabla_{h_t} u|^2, \quad |\nabla_{h_t} u|^2 \doteq h^{ij}_t \nabla_i u \nabla_j u,
	\end{equation}
	where $g_E$ is the auxiliary Riemannian metric as in Equation \eqref{Eq: aux Riemannian metric element}. 
\end{definition}

Having introduced the necessary preliminary ingredients, let us consider the massless stress-energy tensor associated to a solution $u$ of the mixed initial-boundary value problem in Equation \eqref{Eq: Cauchy problem on non-static spacetime}: 
\begin{equation}
	\label{Eq: massless SET}
	T_{\mu \nu} [u] = T_{\nu \mu}[u] \doteq \nabla_\mu u \nabla_\nu u - \frac{1}{2} g_{\mu \nu} \nabla^\lambda u \nabla_\lambda u.
\end{equation}
Calling $X \doteq \partial_t = \beta^{\frac{1}{2}} \partial_\chi \in \Gamma(T\mcM)$, we can define the current 
\begin{equation*}
	J_\mu [u]= T_{\mu \nu}[u] X^\nu. 
\end{equation*}
A direct computation entails that 
\begin{equation}
	\label{Eq: div of current}
	\nabla^\mu J_\mu[u] = \Box_g u \partial_t u + T_{\mu \nu}[u] \nabla^\mu X^\nu = \underbrace{(f + m^2 u) \partial_t u}_{(1)} + \underbrace{T_{\mu \nu}[u] \nabla^\mu X^\nu}_{(2)}, 
\end{equation}
where we used Equation \eqref{Eq: Cauchy problem on non-static spacetime}. We now seek estimates for terms $(1)$ and $(2)$ in Equation \eqref{Eq: div of current}. Recalling that $m^2 = m^2(x)$ is smooth and $D_p$ is a compact domain, we introduce the constants \begin{equation}
\label{Eq: min max mass on Dp}
    M^2_{0, p} = \min_{x \in D_p} \{|m^2(x)|\}, \qquad M^2_{1,p} = \max_{x \in D_p} \{|m^2(x)|\}. 
\end{equation}
Then, starting from term $(1)$ and using a Young type inequality, it holds
\begin{equation*}
	|(f+m^2u) \partial_t u |\leq \frac{1}{2} \left[ |f|^2 + |\partial_t u|^2 + M^2_{1,p} |u|^2 + M^2_{1,p} |\partial_t u|^2 \right]. 
\end{equation*}
On $D_p$, we can write 
\begin{equation}\label{Eq: LongLiveAllEstimates}
	|\partial_t u|^2 = \beta |\partial_\chi u|^2 \le \beta_D |\nabla u|_E^2,
\end{equation}
where $\beta_D \doteq \max\limits_{x\in D_p}\{\beta(x)\}$, while $\vert \nabla u\vert_E^2$ is as per Equation \eqref{Eq: Euclidean kinetic energy}. Combining these two estimates, we get 
\begin{equation}
	\label{Eq: estimate term (1)}
	|(f + m^2u) \partial_t u | \le C_D \left[ |f|^2 + |u|^2 + |\nabla u|^2_E \right], 
\end{equation}
where the constant $C_D$ is chosen in such a way that $C_D \ge \max \left\{\frac{1}{2}, \frac{M^2_{1,p}}{2}, \frac{(1+M^2_{1,p})}{2} \beta_D \right\}$.

Turning our attention to the second term in Equation \eqref{Eq: div of current}, we use the fact that for any symmetric $2$-tensor $B_{\mu \nu}$, one can write 
$$B^2_E = |g_E^{\mu \gamma} g_E^{\nu \lambda} B_{\mu \nu} \, B_{\gamma \lambda} \vert.$$
Consider the bundle map $\mathfrak R: T^*\mcM \to T^*\mcM$, whose action on a given $g$-orthonormal coframe $(\chi^\flat, e^1, \ldots, e^{d-1})$ is defined by
$$\mathfrak R(\chi^\flat) = -\chi^\flat, \qquad \mathfrak R(e^i) = e^i, \quad i=1,\ldots,d-1.$$ Hence, with respect to this coframe, the matrix representation of $\mathfrak R$ coincides with $(\eta^{ij})_{i,j=0}^{d-1}$, where $\eta$ is the Minkowski metric with signature $(-, +, \ldots, +)$. Noting that $(g_E)_{a\gamma} = g_{\mu \gamma} \eta^{\mu}_a$, we can write 
\begin{equation*}
	|T_{\mu \nu}[u] \nabla^\mu X^\nu | = | ((\mathfrak R \otimes \mathfrak R)(T[u]), \nabla X)_E| \le |(\mathfrak R \otimes \mathfrak R)(T[u])|_E |\nabla X|_E, 
\end{equation*}
where $(\cdot, \cdot)_E$ is the Euclidean scalar product induced by the auxiliary Riemannian metric $g_E$, while $[(\mathfrak R \otimes \mathfrak R) (T[u])]^{a b} = \eta^a_i \eta^b_j T^{ij}[u]$. In the last step above, we used Cauchy-Schwartz inequality. After a few algebraic manipulations, one finds that $|(\mathfrak R \otimes \mathfrak R)(T[u])|_E = |T[u]|_E$. Therefore, \begin{equation}
	\label{Eq: Tmunu est 1}
	|T_{\mu \nu}[u] \nabla^\mu X^\nu | \le \underbrace{|T[u]|_E}_{\text{(a)}} \underbrace{|\nabla X|_E}_{\text{(b)}}.
\end{equation}
For what concerns term (b), $\vert \nabla X|^2_E \doteq g_E^{ab} g_E^{cd} \nabla_a X_c \nabla_b X_d$ encompasses only smooth geometric quantities, which are fixed on $D_p$. Consequently, there exists $M_D^2 \doteq \max\limits_{x\in D_p} |\nabla X|^2_E \ge 0$. The analysis of term (a) is slightly more complicated since it comprises, in turn, three main contributions, see Equation \eqref{Eq: massless SET}: 
\begin{equation*}
	\begin{split}
		|g_E^{ab} g_E^{cd} \, \nabla_a u \nabla_b u \nabla_c u \nabla_d u|& = | \nabla u|^4_E, \\
		\frac{1}{2} |g_E^{ab} g_E^{cd} g_{bd} \, \nabla_a u \nabla_c u \nabla^\lambda u \nabla_\lambda u| &\le \frac{1}{2} |\nabla u|_E^4, \\
		\frac14 |g_E^{ab} g_E^{cd} g_{ac} g_{bd} \, \nabla_\rho u \nabla_\rho u \nabla^\lambda u \nabla_\lambda u| &\le \frac{d}{4} |\nabla u|_E^4, 
	\end{split}
\end{equation*}
where in the second and third lines we used the identities $g_{ab} = \eta^{c}_a (g_E)_{cb}$ and $g_E^{ab} g_E^{cd} g_{bd} = g^{ac}$ as well as the estimate $|g^{ac} \nabla_a u \nabla_c u| \le |\partial_\chi u|^2 + |h^{ij}_t \nabla_i u \nabla_j u|$. Thus, we can conclude that 
$$|T[u]|^2_E \le \left(\frac{d}{4} + 2 \right) |\nabla u|^4_E.$$ 
Combining together the estimates obtained for terms (a) and (b), it holds that
\begin{equation}\label{Eq: estimate term (2)}
	|T_{\mu \nu}[u] \nabla^\mu X^\nu| \leq M_D |\nabla u|_E^2,
\end{equation}
where we reabsorbed all numerical factors into a redefinition of the constant $M_D$. 

Finally, putting together Equations \eqref{Eq: estimate term (1)} and \eqref{Eq: estimate term (2)}, we can conclude that there exists a constant $\tilde{C}_D \ge 0$ such that, on $D_p$, 
\begin{equation}\label{Eq: combining estimates (1) and (2)}
	\vert \nabla_\mu J^\mu [u] | \le \tilde{C}_D \left[ |f|^2 + |u|^2 + |\nabla u|_E^2 \right]. 
\end{equation}
Let us now consider the domain $D_t \doteq D_p \cap \{q\in D_p\;|\; t_0\leq t(q)\leq t\}$, where, with a slight abuse of notation, we write $t(q)$ to refer to the time coordinate of $q$, while $t\in\bR$ is a fixed value such that $t_0<t<t_1$. Up to corner points which play no role, the boundary decomposes as 
\begin{equation}\label{Eq: boundary of Dt}
	\partial D_t = \Omega_{t} \cup \Omega_{t_0} \cup N_t \cup B_t,  
\end{equation}
where $\Omega_{i} \doteq D_p \cap \Sigma_{i}$, $i=t_0,t$, are the spacelike portions, $N_t$ is the lightlike portion up to $\Sigma_t$, while $B_t \doteq \partial \mcM \cap D_t$ is the part contained in the timelike boundary. As a preliminary step, we integrate Equation \eqref{Eq: combining estimates (1) and (2)} over this domain, thereby obtaining an estimate which will be useful for the ensuing discussion. More precisely, we have that there exists $C^\prime_D>0$ such that
\begin{flalign}
	\label{Eq: estimate for Stokes}
	\bigg \vert \int_{D_t} \nabla^\mu J_\mu [u] \, d\mu_g \bigg \vert \le \int_{D_t} \vert  \nabla^\mu J_\mu [u] \vert \, d\mu_g \le C'_D \int_{t_0}^t \left[ \|f\|^2_{L^2_s} + \|u\|^2_{L^2_s} + E_{\mathrm{der}}(s) \right] \, ds, 
\end{flalign}
where we denote by $L^2_s \doteq L^2(\Sigma_s \cap D_p)$ and by 
$$E_{\mathrm{der}} (s) \doteq \|\nabla_E u\|^2_{L^2_s}.$$
The strictly positive lapse function $\beta$, whose square root would otherwise appear in the metric-induced volume measure, is estimated by $L^\infty$-norm over the integration domain and absorbed into the constant $C^\prime_D$. Applying Stokes' theorem over $D_t$, we obtain
\begin{flalign}\label{Eq: Stokes}
	\int_{D_t} \nabla^\mu J_\mu[u] \, d\mu_g &= -\int_{\Omega_t} J_\mu [u] \mathfrak{n}^\mu \, d\mu_{h_t} - \int_{\Omega_{t_0}} J_\mu [u] \mathfrak{n}^\mu \, d\mu_{h_{t_0}}  + \int_{B_t} J_\mu [u] \mathfrak{n}^\mu \, d\mu_{B_t} + \int_{N_t} \iota_J(d\mu_g)
\end{flalign}
where, with a slight abuse of notation, we use the same symbol $\mathfrak n$ to denote the outward pointing normal vector field to the different portions of the boundary. The relevant choices will be specified case by case below. The only exception is $N_t$ which, being lightlike, behaves differently and its contribution is written as $\iota_J(d\mu_g)$, the restriction to $N_t$ of the contraction of the metric induced volume form against the current vector field $J$. The lower bound in the estimate of Equation \eqref{Eq: estimate for Stokes} entails 
\begin{flalign}\label{Eq: consequence Stokes}
	\int_{\Omega_t} J_\mu [u] \mathfrak{n}^\mu \, d\mu_{h_t} &\le- \int_{\Omega_{t_0}} J_\mu [u] \mathfrak{n}^\mu \, d\mu_{h_{t_0}}  + \int_{B_t} J_\mu [u] \mathfrak{n}^\mu \, d\mu_{B_t} + \int_{N_t} \iota_J(d\mu_g) \notag \\ &+C^\prime_D \int_{t_0}^t \left[ \|f\|^2_{L^2_s} + \|u\|^2_{L^2_s} + E_{\mathrm{der}}(s)\right] \, ds.
\end{flalign}
In the following, we shall estimate each contribution arising on the right-hand side of Equation \eqref{Eq: Stokes} separately. 

Let us consider the integral over $\Omega_t \doteq \Sigma_t \cap D_p$ and recall that the unit normal vector field to $\Omega_t$ is $\mathfrak n = \chi = \beta^{-\frac{1}{2}} X$. 
\begin{equation*}
	\int_{\Omega_t} J_\mu [u] \mathfrak{n}^\mu \, d\mu_{h_t} = \int_{\Omega_t} T_{\mu \nu}[u] X^\nu \chi^\mu \, d\mu_{h_t} = \beta^{\frac{1}{2}} \int_{\Omega_t} T_{\mu \nu}[u] \chi^\nu \chi^\mu 
\end{equation*}
Using that $T_{\mu \nu}[u] \chi^\mu \chi^\nu = \frac{1}{2} |\nabla u|_E^2$, we obtain 
\begin{equation} \label{Eq: estimate Stokes 1} 
	\int_{\Omega_t} J_\mu [u] \mathfrak{n}^\mu \, d\mu_{h_t} = \frac{1}{2} \int_{\Omega_t} |\nabla u|^2_E \, \beta^{\frac{1}{2}} d\mu_{h_t} \doteq E_{\mathrm{der}}^\beta(t), 
\end{equation} 
where $d\mu_{h_t}$ denotes the volume form on $\Sigma_t$ induced by the Riemannian metric $h_t$. The analysis of the contribution relative to $\Omega_{t_0}$ is completely analogous, barring setting $\mathfrak n = - \chi$. Thus, we obtain 
\begin{equation}\label{Eq: estimate Stokes 2}
	\int_{\Omega_{t_0}} J_\mu [u] \mathfrak{n}^\mu \, d\mu_{h_t} = -\frac{1}{2} \int_{\Omega_{t_0}} |\nabla u|^2_E \, \beta^{\frac{1}{2}} d\mu_{h_t} \doteq -E_{\mathrm{der}}^\beta(t_0). 
\end{equation}
Let us now turn our attention to the contribution corresponding to $N_t$, the lightlike portion of the boundary. In this case, consider a future directed lightlike vector field on $N_t$ of the form $\mathfrak{l} = a (\chi + S)$. The constraint $g(\mathfrak{l}, \mathfrak{l}) = 0$ combined with $g(\mathfrak{l}, \chi) < 0$, entails that $\mathfrak{l}$ is future pointing if $a>0$. Furthermore, $S$ is a spacelike vector such that $g(S,S) = 1$ while $g(S,\chi)=0$.  Observe that that $T_{\mu \nu}[u] X^\mu \chi^\nu = \frac{\beta^{\frac{1}{2}}}{2} |\nabla u|_E^2$ and, since $X=\beta^{\frac{1}{2}}\chi$, we have
\begin{equation*}
	T_{\mu\nu}[u]X^\nu S^\mu = \beta^{\frac{1}{2}}\partial_\chi u \,\partial_S u.
\end{equation*}
Combining these identities we obtain
\begin{equation*}
	T_{\mu\nu}[u]X^\nu \mathfrak{l}^\mu=a\beta^{\frac{1}{2}}\left(\frac{1}{2}|\nabla u|_E^2+\partial_\chi u \, \partial_Su\right).
\end{equation*}
Since $a \beta^\frac12>0$, we prove that the expression within brackets is non-negative. Recalling that
\begin{equation*}
	|\nabla u|_E^2=(\partial_\chi u)^2+h^{ij}_t \partial_i u \,\partial_j u,
\end{equation*}
we obtain
\begin{equation*}
	\frac{1}{2}|\nabla u|_E^2+\partial_\chi u \,\partial_Su =
	\frac{1}{2}(\partial_\chi u)^2 +\frac{1}{2}h^{ij}_t \partial_i u \,\partial_j u +\partial_\chi u\,\partial_Su.
\end{equation*}
Choosing an orthonormal frame $(\chi,S,e_2,\ldots,e_{d-1})$ with respect to $g$, we have
\begin{equation*}
	h^{ij}_t\partial_i u \,\partial_j u =(\partial_Su)^2 +\sum\limits_{k=2}^{d-1}[e_k(u)]^2, ,
\end{equation*}
Therefore,
\begin{equation*}
	\frac{1}{2}|\nabla u|_E^2+\partial_\chi u \,\partial_Su =\frac{1}{2}\left[(\partial_\chi u+\partial_Su)^2	 +\sum\limits_{k=2}^{d-1}[e_k(u)]^2\right] \geq 0.
\end{equation*}
Gathering everything, we end up with
\begin{gather*}
	\int_{N_t} \iota_J(d\mu_g) =  -\int_{N_t} a \left[T_{\mu \nu} [u] X^\nu \chi^\mu + T_{\mu \nu}[u] X^\nu S^\mu \right] \, d\mu_{N_t} =\\
	-\int_{N_t} a \left[\frac{\beta^{\frac{1}{2}}}{2} |\nabla u|_E^2 + T_{\mu \nu}[u] X^\nu S^\mu \right] d\mu_{N_t},
\end{gather*}
where $d\mu_{N_t}$ is the volume density induced by $\chi+S$. We conclude that 
\begin{equation}\label{Eq: estimate Stokes 3}
	\int_{N_t} \iota_J(d\mu_g) \leq 0. 
\end{equation}
In view of these estimates Equation \eqref{Eq: consequence Stokes} can be recast in the form 
\begin{gather}
	E_{\mathrm{der}}^\beta(t) \leq E_{\mathrm{der}}^\beta(t) - \int_{N_t} \iota_J(d\mu_g) \leq \notag\\ E_{\mathrm{der}}^\beta(t_0) + \int_{B_t} J_\mu [u] \mathfrak{n}^\mu \, d\mu_{B_t} + C_D' \int_{t_0}^t \left[ \|f\|^2_{L^2_s} + \|u\|^2_{L^2_s} + E_{\mathrm{der}}^\beta(s) \right]\, ds.\label{Eq: aux eq 2}
\end{gather}

Focusing on the boundary integral over $B_t \doteq \partial \mcM \cap D_t$, we can choose Gaussian normal coordinates adapted to the boundary $(t,z, x^2, \ldots, x^{d-1}) \equiv (t,z, \mathsf{x})$ so that $\partial \mcM$ corresponds to the locus $\{z = 0\}$ while, in the interior $z\geq 0$. In this case, the outward pointing, unit normal to $B_t$ is $\mathfrak n = -\partial_z$ and we can write
\begin{equation*}
	\int_{B_t} J_\mu [u] \mathfrak{n}^\mu \, d\mu_{B_t} = -\int_{t_0}^t \int_{\partial \Sigma} \partial_s u \, \partial_z u  \, \beta^\frac{1}{2} ds \, d\mu_{h_s}^\partial,
\end{equation*}
where $d \mu_{B_t} = \beta^{\frac{1}{2}} dt d\mu_{h_t}^\partial$ with $d\mu_{h_t}^\partial = \sqrt{|h_t \vert_{\partial{\mcM}}} d \mathsf{x}$. Using the Robin boundary conditions $\partial_z u = -\kappa u$, we obtain   
\begin{flalign*}
	\int_{B_t} J_\mu [u] \mathfrak{n}^\mu \, d\mu_{B_t} &= \frac{\kappa}{2} \int_{t_0}^t \int_{\partial \Sigma}  \partial_s (u^2) \, \beta^\frac{1}{2} ds \, d\mu_{h_s}^\partial \\&= \frac{\kappa}{2} \left(\int_{\partial \Sigma} u^2(s) \, \beta^\frac{1}{2} d\mu_{h_s}^\partial \right)\bigg \vert_{t_0}^t - \frac{\kappa}{2} \int_{t_0}^t \int_{\partial \Sigma} u^2 \, \partial_s (\beta^\frac{1}{2} \sqrt{|h_s \vert_{\partial \mcM}}) ds \, d\mu_{\mathsf x} \doteq \mathcal{B}(t).
\end{flalign*}
To control $|\mathcal{B}(t)|$, consider
\begin{equation}\label{Eq: aux Eq 1}
	\bigg \vert \frac{\kappa}{2} \int_{\partial \Sigma} \beta^\frac{1}{2} u^2(t) \, d\mu_{h_t}^\partial \bigg \vert \le \frac{|\kappa|}{2} \int_{\partial \Sigma} \beta^\frac{1}{2} |u(t)|^2 \, d\mu_{h_t}^\partial \le \frac{|\kappa|}{2} \beta_{\mathrm{max}}^{\frac{1}{2}}  \|u(t)\|^2_{L^2(\partial \Sigma)},
\end{equation}
where the norm is taken over $L^2(\partial \Sigma, d\mu_{h_t}^\partial)$ and we controlled $\beta$ with its maximum $\beta_{\mathrm{max}}$ over the integration domain. Observe that, with a slight abuse of notation we write $\partial\Sigma=\Sigma_t\cap B_t$. Since per hypothesis, for any but fixed $t \in \mathbb{R}$, $\Omega_t \doteq \Sigma_t \cap D_p$ is compact with a smooth, compact boundary $\partial \Sigma_t$, the $\epsilon$-trace norm estimate holds, namely, for every $\epsilon > 0$, there exists a constant $C_\epsilon > 0$ such that
\begin{equation}\label{Eq: trace norm estimate}
	\|u(t)\|^2_{L^2(\partial \Sigma_t)} \le \epsilon \|\nabla_{h_t} u(t) \|^2_{L^2(\Omega_t)} + C_\epsilon \|u(t)\|^2_{L^2(\Omega_t)}, 
\end{equation}
where $|\nabla_{h_t} u(t)|^2_{h_t} \doteq h_t^{ij} \partial_i u(t) \partial_j u(t)$ in local coordinates on $\Omega_t$. Recalling the definition of $$E_{\mathrm{der}}^\beta (t) \doteq \frac{1}{2} \int_{\Omega_t} \beta^{\frac{1}{2}} |\nabla u|_E^2 \, d\mu_{h_t} = \frac{1}{2} \int_{\Omega_t} \beta^{\frac{1}{2}} \left( (\partial_\chi u)^2 + |\nabla_{h_t} u|^2 \right) \, d\mu_{h_t},$$ and since $\beta$ admits a minimum $\beta_{\mathrm{min}} > 0$ on the integration domain, it follows that 
$$\|\nabla_{h_t} u(t)\|^2_{L^2(\Omega_t)} \doteq \int_{\Omega_t} |\nabla_{h_t} u|^2 \, d\mu_{h_t} \le 2 \frac{E^\beta_{\mathrm{der}}(t)}{\beta_{\mathrm{min}}^\frac{1}{2}}.$$ 
Plugging the above estimates into Equation \eqref{Eq: aux Eq 1}, we obtain 
\begin{equation*}
	\bigg \vert \frac{\kappa}{2} \int_{\partial \Sigma} \beta^\frac{1}{2} u^2(t) \, d\mu_{h_t}^\partial \bigg \vert \le \frac{|\kappa|}{2} \beta_{\mathrm{max}}^\frac{1}{2} \left( \epsilon \|\nabla_{h_t} u(t)\|^2_{L^2(\Omega_t)} + C_\epsilon \|u(t)\|^2_{L^2(\Omega_t)} \right) \le \tilde{C} \epsilon E_{\mathrm{der}}^\beta(t) + C'_\epsilon \|u(t)\|^2_{L^2(\Omega_t)}, 
\end{equation*}
where $\tilde{C} \doteq \left(\frac{\beta_{\mathrm{max}}}{\beta_{\mathrm{min}}} \right)^{\frac{1}{2}} |\kappa|$, while $C'_{\epsilon} \doteq  \frac{|\kappa|}{2} \beta_{\mathrm{max}}^\frac{1}{2} C_\epsilon$. An analogous estimate applies also to the term 
$$ \bigg \vert \frac{\kappa}{2} \int_{\partial \Sigma} \beta^\frac{1}{2} u^2(t_0) \, d\mu_{h_{t_0}}^\partial \bigg \vert \le  \tilde{C}_0 \epsilon_0E_{\mathrm{der}}^\beta(t_0) + C'_{\epsilon_0} \|u(t_0)\|^2_{L^2(\Omega_{t_0})}.$$ 
For what concerns the last integral appearing in the definition of $\mathcal{B}(t)$, exploiting the fact that $\partial \Sigma$ is compact, there exists $2C \doteq \max\limits_{B_t} \{ \vert \partial_s \ln (\beta^\frac{1}{2} \sqrt{|h_s \vert_{\partial \mcM}}) \vert\}$ and, thus, 
\begin{equation*}
	\bigg \vert \frac{\kappa}{2} \int_{t_0}^t \int_{\partial \Sigma} u^2 \, \partial_s (\beta^\frac{1}{2} \sqrt{|h_s \vert_{\partial \mcM}}) ds \, d\mu_{\mathsf x} \bigg \vert \le |\kappa| C \beta_{\mathrm{max}}^\frac12 \int_{t_0}^t \|u(s)\|^2_{L^2(\partial \Sigma)} \, \, ds,
\end{equation*}
where $\|u(s)\|^2_{L^2(\partial \Sigma)}$ can in turn be controlled by means of the trace norm estimate in Equation \eqref{Eq: trace norm estimate}. Denoting by $\widetilde{C}>0$ an overall constant incorporating all the numerical factors in the previous estimates, we can thus conclude that
\begin{flalign*}
	|\mathcal{B}(t)| \le & \widetilde{C}\left(\epsilon E_{\mathrm{der}}^\beta(t) + \widetilde{C}_\epsilon \|u(t)\|^2_{L^2(\Omega_t)}  +   E_{\mathrm{der}}^\beta(t_0) +\|u(t_0)\|^2_{L^2(\Omega_{t_0})} \right. \\ &\left.+ \int_{t_0}^t  \left(E_{\mathrm{der}}^\beta(s) + \|u(s)\|^2_{L^2(\Omega_s)}\right)\, ds \right),
\end{flalign*}
where $\widetilde{C}_\epsilon:=\frac{C^\prime_\epsilon}{\widetilde{C}}$. Combining this with Equation \eqref{Eq: aux eq 2} and choosing $\epsilon$ so that $\widetilde{C}\epsilon<1$, we obtain 
\begin{equation}\label{Eq: Ebeta Inequality}
	E^\beta_{\mathrm{der}}(t)\leq C_\epsilon \|u(t)\|^2_{L^2(\Omega_t)} + \widetilde{C}\left( E_{\mathrm{der}}^\beta(t_0) +\|u(t_0)\|^2_{L^2(\Omega_{t_0})}+ \int_{t_0}^t  (E_{\mathrm{der}}^\beta(s) + \|u(s)\|^2_{L^2(\Omega_s)} +\|f\|^2_{L^2_s})\,ds \right),
\end{equation}
where we have implicitly reabsorbed all multiplicative constant in $C_\epsilon$ and $\widetilde{C}$ respectively. At this point, using the identity $u(t)=u(t_0)+\int\limits_{t_0}^t \partial_s u(s)ds$ we can use Equation \eqref{Eq: LongLiveAllEstimates} to control with two constants $\mathfrak{c},\mathfrak{c}^\prime >0$ the following estimate:
$$\|u(t)\|^2_{L^2(\Omega_t)}\leq\mathfrak{c}\left(\|u(t_0)\|^2_{L^2(\Omega_{t_0})}+\int\limits_{t_0}^t \|\partial_s u(s)\|^2_{L^2(\Omega_s)}\right)\leq\mathfrak{c}\|u(t_0)\|^2_{L^2(\Omega_{t_0})}+\mathfrak{c}^\prime\int_{t_0}^t E^\beta_{der}(s)ds.$$
Inserting this inequality in Equation \eqref{Eq: Ebeta Inequality} justifies defining the {\em total energy} as
\begin{equation}\label{Eq: total energy}
	E(s) \doteq E^\beta_{\mathrm{der}}(s) + K \|u(s)\|^2_{L^2(\Omega_s)}, 
\end{equation}
where $K$ is any but fixed positive number. This entails that there exists $\tilde{K}>0$ so that
\begin{equation}\label{Eq: Energy estimate 1}
	E(t) \le \tilde{K} \left(E(t_0) + \int_{t_0}^t E(s) \, ds + \int_{t_0}^t  \|f\|^2_{L^2_s} \, ds \right).
\end{equation}
Applying Gr\"onwall lemma to Equation \eqref{Eq: Energy estimate 1} and reabsorbing in $\widetilde{K}$ all multiplicative constants, we thus obtain
\begin{equation}\label{Eq: Energy estimate 2}
	E(t) \leq \widetilde{K} \left(E(t_0) + \int_{t_0}^t \|f\|^2_{L^2_s} \, ds  \right).
\end{equation}

\paragraph{Finite Propagation Speed:} In view of our preceding discussion, we only need to reap its fruits. Consider now a compact set $K=\operatorname{supp}(u_0)\cup\operatorname{supp}(u_1)\cup\operatorname{supp}(f)$ and let $q\notin J^+(K)$. Consequently, $K\cap J^-(q)=\emptyset$. Fix $t_0,t_1$ such that $t_0<t_1<t(q)$ and let $D_q$ be as per Equation \eqref{Eq: domain D}. Therein both the source and the initial data vanish. Hence Equation \eqref{Eq: Energy estimate 2} combined with Equations \eqref{Eq: estimate Stokes 1} and \eqref{Eq: total energy} entails
$$E(t_1)=0\Longrightarrow  E^\beta_{\mathrm{der}}(t_1)=\|u(t_1)\|^2_{L^2(\Omega_{t_1})}=0.$$
This entails that $u(t_1)=0$. Since both $q$ and $t_1$ are arbitrary and $(\mcM,g)$ is globally hyperbolic, we can conclude that the solution is vanishing everywhere outside $J^+(K)$ and, \emph{mutatis mutandis}, also on $J^-(K)$.

\section{Proof of Proposition \ref{Prop: Existence of Solutions}}\label{App: B}

We divide the argument into separate logical steps. Yet, before delving into them, we recall that $\mcM$ can be isometrically identified with $\bR\times\Sigma$, see Proposition \ref{Prop: Globally Hyperbolic}.  

\vskip .2cm

\paragraph{Step 1 -- Weak Formulation:} Henceforth, until Step 5., the Cauchy surface $\Sigma$ is assumed to be compact and we consider the strip
$$D\doteq I \times \Sigma,\quad I:=[t_0,t_1].$$ 
Considering $u\in C^\infty(D)$ and Equation \eqref{Eq: line element psi*g}, it holds 
\begin{equation}
	\label{Eq: KG expansion}
	(\Box_g + m^2) u = \underbrace{-\frac{1}{\beta^{\frac{1}{2}} \sqrt{|h_t|}} \partial_t (\beta^{-\frac{1}{2}} \sqrt{|h_t|} \partial_t u)}_{(1)} + \underbrace{\frac{1}{\beta^{\frac{1}{2}}} \operatorname{div}_{h_t}(\beta^{\frac{1}{2}} \nabla u)}_{(2)} + m^2 u = f
\end{equation}
where we assume $\textrm{supp}(f)\subset D$ and where, for later convenience, we recall that 
$$\frac{1}{\beta^{\frac{1}{2}} \sqrt{|h_t|}} \partial_i (\beta^{-\frac{1}{2}} \sqrt{|h_t|} h^{ij}_t \partial_j u) = \frac{1}{\beta^{\frac{1}{2}}} \operatorname{div}_{h_t} (\beta^{\frac{1}{2}} \nabla u).$$ 
Considering $v\in C^\infty(\Sigma)$ the term $(1)$ in Equation \eqref{Eq: KG expansion} reads
\begin{flalign*}
	-\int_{\Sigma} \frac{1}{\sqrt{|h_t|}} \partial_t (\beta^{-\frac{1}{2}} \sqrt{|h_t|} \partial_t u) v \, d\mu_t =  -\int_{\Sigma} \partial_t (\beta^{-\frac{1}{2}} \sqrt{|h_t|} \partial_t u) v \, dx = - \frac{d}{dt} \int_{\Sigma} \beta^{-\frac{1}{2}} \partial_t u \, v \, d\mu_t,
\end{flalign*}
where, being both $u$ and $v$ smooth, we can take the derivative outside the integral and where $d\mu_t \doteq \sqrt{|h_t|} dx$ denotes the metric induced volume form. Concerning term $(2)$ in Equation \eqref{Eq: KG expansion}, integrating by parts we obtain
\begin{flalign*}
	\int_{\Sigma} \operatorname{div}_{h_t} (\beta^{\frac{1}{2}} \nabla u) \, v \, d\mu_t = \int_{\partial \Sigma} \beta^{\frac{1}{2}} (-\nabla_n u) v \, d\sigma_t - \int_{\Sigma} \beta^{\frac{1}{2}} h_t(\nabla u, \nabla v) \, d\mu_t=\\
	\kappa \int_{\partial \Sigma} \beta^{\frac{1}{2}}  u v \, d\sigma_t - \int_{\Sigma} \beta^{\frac{1}{2}} h_t(\nabla u, \nabla v) \, d\mu_t.
\end{flalign*}
where, with a slight abuse of notation, $n$ refers to the inward pointing unit normal vector field to $\partial \Sigma$ at time $t$, while we denote by $d\sigma_t \equiv d\mu_{h_t}^\partial$ the volume form on $\partial \Sigma$ and by $h_t(\nabla u, \nabla v) \doteq h_t^{ij} \nabla_i u \nabla_j v$. Observe that, in the last equality, we used the Robin boundary condition as in Equation \eqref{Eq: Robin Boundary Conditions}. In other words we can rewrite Equation \eqref{Eq: KG equation} in the form
\begin{equation}\label{Eq: weak formulation 2}
\frac{d}{dt} \mathrm{K}(\partial_t u,v) + \mathrm{A}(u, v) = \mathrm{S}_f(v),
\end{equation}
where
\begin{itemize}
	\item[\ding{104}] the \textbf{kinetic form} 
	\begin{equation}
		\label{Eq: kinetic form}
		\mathrm{K}(u,v) \doteq \int_{\Sigma} \beta^{-\frac{1}{2}} u v \, d\mu_t.
	\end{equation}
	\item[\ding{104}] the \textbf{potential form}
\begin{equation}\label{Eq: potential form}
\mathrm{A}(u,v) \doteq \int_{\Sigma} \beta^\frac{1}{2} h_t(\nabla u, \nabla v) \, d\mu_t - \kappa \int_{\partial \Sigma} \beta^{\frac{1}{2}} u v \, d\sigma_t - m^2 \int_{\Sigma} \beta^{\frac{1}{2}} uv \, d\mu_t.
\end{equation}
	\item[\ding{104}] the \textbf{source term}
	\begin{equation}
		\label{Eq: source term}
		\mathrm{S}_f(v) \doteq - \int_{\Sigma} \beta^\frac{1}{2} f \, v \, d\mu_t.
	\end{equation}
\end{itemize}
We take Equation \eqref{Eq: weak formulation 2} as the weak formulation of Equation \eqref{Eq: KG equation}. Therefore, henceforth, we look for a solution $u\in L^2(I;H^1(\Sigma, d\mu_t))$ such that $\partial_t u\in L^2(I;L^2(\Sigma, d\mu_t))$, where Equation \eqref{Eq: weak formulation 2} holds in $\mathcal{D}^\prime(I)$ for every $v\in H^1(\Sigma)$. 

\subparagraph{Analytic Properties:} Before proving existence of a solution of the weak form of the Klein-Gordon equation, we establish continuity and coercivity estimates connected to the quadratic forms $\mathrm{A}$ and $\mathrm{K}$ whereas for the source term $\mathrm{S}_f$ we only require a continuity bound. To this end and for later convenience, we observe that, being $\beta$ and $m^2$ as in Equation \eqref{Eq: Effective mass} smooth and $D$ a compact subset, we can introduce the constants
\begin{equation}\label{Eq: Constants from beta}
C_0=\min\limits_{x\in D}\{\beta^{-\frac{1}{2}}, \beta^{\frac{1}{2}}\}\quad\textrm{and}\quad	C_1=\max\limits_{x\in D}\{\beta^{-\frac{1}{2}}, \beta^{\frac{1}{2}}\},
\end{equation}
\begin{equation}\label{Eq: Constants from the mass}
	M^2_0=\min\limits_{x\in D}\{|m^2(x)|\}\quad\textrm{and}\quad	M^2_1=\max\limits_{x\in D}\{|m^2(x)|\}.
\end{equation}
Adopting the notation $L^2_t \doteq L^2(\Sigma, d\mu_t)$, $H^1_t \doteq H^1(\Sigma, d\mu_{t})$  and $L^2_{\partial_t} \doteq L^2(\partial \Sigma, d\sigma_t)$ we observe that, Equation \eqref{Eq: kinetic form} can be bounded using the Cauchy-Schwartz inequality as
\begin{equation}
	\label{Eq: estimate for K}
	|\mathrm{K}(u,v) | \le C_1 \|u\|_{L^2_t} \, \|v\|_{L^2_t}, 
\end{equation}
whereas, for the same reason, each term in Equation \eqref{Eq: potential form} can be controlled as
\begin{flalign*}
	\bigg \vert \int_{\Sigma} \beta^\frac{1}{2} h_t(\nabla u, \nabla v) \, d\mu_t \bigg \vert &\leq C_1 \|\nabla u\|_{L^2_t} \, \|\nabla v\|_{L^2_t} \le C_1 \|u\|_{H^1_t} \, \|v\|_{H^1_t},\\
	\bigg \vert \kappa \int_{\partial \Sigma} \beta^{\frac{1}{2}} u v \, d\sigma_t \bigg \vert & \leq |\kappa| \, C_1 \|u\|_{L^2_{\partial_t}} \, \|v\|_{L^2_{\partial_t}} \leq C_\partial \|u\|_{H^1_t} \, \|v\|_{H^1_t},\\
	\bigg \vert m^2 \int_{\Sigma} \beta^{\frac{1}{2}} uv \, d\mu_t \bigg \vert & \leq C_1M^2_1 \|u\|_{L^2_t} \, \|v\|_{L^2_t} \leq C_1M^2_1 \|u\|_{H^1_t} \, \|v\|_{H^1_t}.
\end{flalign*}
In the last equality of the second line,  we have used that, being $\Sigma$ compact, the Lions-Magenes trace $$\gamma_0 \in \mathcal{B}(H^1(\Sigma, d\mu_t), H^\frac{1}{2}(\partial \Sigma, d\sigma_t)).$$ 
Combining this fact with $H^\frac{1}{2}(\partial \Sigma, d\sigma_t)$ being continuously embedded into $L^2_\partial$, we can infer that, for each $t\in I$, there exists $C_t>0$ such that $\|u\|_{L^2_\partial} \leq C_t \|u\|_{H^1_t}$. This entails that $C_\partial=C_1\sup\limits_{t\in I}(C_t)$. Hence, combining all these estimates, there exists $C^\prime_1>0$ for which the following continuity bound for $\mathrm{A}$ holds: 
\begin{equation}\label{Eq: estimate for A}
	|\mathrm{A}(u,v) | \le C_1^\prime \|u\|_{H^1_t} \, \|v\|_{H^1_t}. 
\end{equation}

Focusing instead on establishing lower bounds, we start from Equation \eqref{Eq: potential form} and, once more, we need to focus on each term separately, though setting $u=v$:  
\begin{enumerate}
\item The first and the third can be controlled once more by the Cauchy-Schwartz inequality:
\begin{equation*}
\int_\Sigma \beta^\frac{1}{2} |\nabla v|_{h_t}^2 \, d\mu_t \geq C_0 \|\nabla v\|^2_{L^2_t},
\end{equation*}
and 
\begin{equation*}
- m^2 \int_{\Sigma} \beta^{\frac{1}{2}} u^2\geq -M^2_0 C_1 \|v\|^2_{L^2_t}.
\end{equation*}
\item The boundary contribution can be controlled combining the Peter-Paul inequality with the boundedness of $\gamma_0$. This entails that, for every $\epsilon>0$, there exists $C_\epsilon>0$ such that  $\|v\|_{L^2_\partial}^2 \leq\epsilon \|\nabla v\|^2_{L^2_t} + C_{\epsilon} \|v\|_{L^2_t}^2$. Consequently, it holds
\begin{equation*}
	-\kappa \int_{\partial \Sigma} \beta^{\frac{1}{2}} v^2 \, d\sigma_t \ge - |\kappa| C_1   \|v\|_{L^2_\partial}^2 \ge - |\kappa| C_1  \left(  \epsilon \|\nabla v\|^2_{L^2_t} + C_{\epsilon}  \|v\|_{L^2_t}^2\right).
\end{equation*}
\end{enumerate}
Combining all these estimates together we end up with
\begin{equation}\label{Eq: lower bound A}
	\mathrm{A}(v,v) \geq  (C_0-|\kappa| C_1\epsilon) \|\nabla v\|_{L^2_t}^2 - C_1(M^2_0 +|\kappa| C_\epsilon) \|v\|^2_{L^2_t},
\end{equation}
which, choosing $\epsilon \leq \frac{C_0}{2 |\kappa| C_1}$ and setting $C_2 \doteq C_1(M^2_0 +|\kappa| C_\epsilon)$, yields the estimate
\begin{equation}\label{Eq: Garding estimate}
\mathrm{A}(v,v) \ge \frac{C_0}{2} \|\nabla v\|_{L^2_t}^2 - C_2 \|v\|^2_{L^2_t}.
\end{equation}
It is important to observe that we can compensate the negative contribution to Equation \eqref{Eq: Garding estimate} by adding to $A$ a multiple of the kinetic form $\mathrm K$ since $\mathrm{K}(v,v) \geq C_0 \|v\|_{L^2_t}^2$. In other words, choosing $\lambda=\frac{1}{2}+\frac{C_2}{C_0}$, we set
\begin{equation}\label{Eq: coercive estimate 1}
B(v,v):=\mathrm{A}(v,v) + \lambda \mathrm{K}(v,v)\geq\frac{C_0}{2} \|\nabla v\|_{L^2_t}^2 + (\lambda C_0 - C_2) \|v\|^2_{L^2_t}\geq \frac{C_0}{2} \|v\|_{H^1_t}^2,  
\end{equation}
which is the sought {\em coercivity estimate}. 

\paragraph{Step 2 -- Galerkin Approximation:} Having established a weak formulation of the underlying problem in Equation \eqref{Eq: weak formulation 2}, in this second step we employ the method of Galerkin, see \cite[Ch. 7.2]{Evans}, to establish first of all a sequence of approximating solutions, proving ultimately their convergence in a suitable topology. To this end, we introduce the following two-notable ingredients
\begin{itemize}
	\item[\ding{104}] Since $D$ is compact, we can consider a reference Riemannian metric on $\Sigma$ which we choose without loss of generality as $h_0\equiv h$, the one on $\{t_0\}\times\Sigma$. Hence, for all $t\in I=[t_0,t_1]$, it holds that $h_t \asymp h$. We recall that the symbol $\asymp$ entails that there exists two constants $0<d_0\leq d_1<\infty$ such that $d_0 h\leq h_t\leq d_1 h$ uniformly in $D$. In turn, this entails an equivalence  between the $L^2$- and $H^1$-norms on $\Sigma$ induced by $h_t$ and by $h$. Henceforth we employ the notation $L^2(\Sigma)\equiv L^2(\Sigma, d\mu_{t_0})$ and $H^1(\Sigma)\equiv H^1(\Sigma,d\mu_{t_0})$.
	\item[\ding{104}] We denote by $\Delta_N$ the self-adjoint extension on $L^2(\Sigma)$ of the Laplace-Beltrami operator $\Delta_h$ induced by $h$ with Neumann boundary conditions on $\partial\Sigma$. The associated quadratic form is 
\begin{equation}\label{Eq: quadratic form Neumann Laplacian}
\mathfrak{q}_N (v) \doteq \int_{\Sigma} |\nabla v|^2_h d\mu_h,
\end{equation}
with $v$ lying in the domain of $\mathfrak{q}_N$, which coincides with $H^1(\Sigma)$.
\end{itemize}

\begin{remark}
It is worth mentioning that one could have chosen also a self-adjoint extension of $\Delta_h$ with a fixed Robin boundary condition on $(\Sigma,h)$, leading to no appreciable difference in the upcoming discussion since the domain of the underlying quadratic form would still be $H^1(\Sigma)$. On the contrary, we cannot choose Dirichlet boundary conditions since the domain of the counterpart of Equation \eqref{Eq: quadratic form Neumann Laplacian} would be $H^1_0(\Sigma)$. This is not suitable for our purposes since it would lead to constructing solutions vanishing at $\partial\mcM$, which is incompatible with Equation \eqref{Eq: Robin Boundary Conditions}. 
\end{remark}

\begin{itemize}
	\item[\ding{104}] Since $\Sigma$ is compact $-\Delta_N$ has pure point spectrum and therefore there exists a basis of $L^2(\Sigma)$ built out of eigenvectors of the Neumann Laplacian, namely $\{\psi_i\}_{i\in\mathbb{N}}$ with $\psi_i\in H^k(\Sigma)$ for all $k\in\mathbb{N}\cup\{0\}$ and $-\Delta_N\psi_i=\lambda_i\psi_i$ with $\lambda_i\geq 0$ for all $i\in\mathbb{N}$. Henceforth we consider an arbitrary but fixed ordering of the eigenvectors and we denote by $V_N\doteq\mathrm{span}_{\bR}\{\psi_1,\dots,\psi_N\}$. Observe that $\overline{\bigcup\limits_{N\in\mathbb{N}}V_N}=L^2(\Sigma)$ and we introduce the orthogonal projector
\begin{equation*}
P_N:L^2(\Sigma)\to V_N.
\end{equation*}
\end{itemize}

We can start with the approximation procedure setting $u_N \doteq \sum_{j=1}^N c_j(t) \psi_j$ and inserting in Equation \eqref{Eq: weak formulation 2}:
\begin{equation}\label{Eq: weak formulation N}
\frac{d}{dt} \mathrm{K}(\partial_t u_N, v) + \mathrm{A}(u_N, v) = \mathrm{S}_f(v).
\end{equation}
Here $v$ is chosen also to lie in $V_N$, following Galerkin procedure. Denoting by $\mathsf{K},\mathsf{A}$ and $\mathsf{S}_f$ the time-dependent matrices whose entries are
$$\mathsf{K}_{ij}(t) \doteq \mathrm{K} (\psi_i, \psi_j), \quad \mathsf{A}_{ij}(t) \doteq \mathrm{A}(\psi_i, \psi_j), \quad \mathsf{S}_{f,i} (t)\doteq \mathrm{S}_f(\psi_i),$$ 
we observe that, in view of Equations \eqref{Eq: kinetic form} and \eqref{Eq: Constants from beta}, $\mathsf{K}>0$, hence it is invertible. Consequently, Equation \eqref{Eq: weak formulation N} becomes
\begin{equation*}
	\mathsf{K} (t) \ddot{c}(t) + \dot{\mathsf{K}}(t) \dot{c}(t) + \mathsf{A}(t) c(t) = \mathsf{S}_f(t)\Longrightarrow\ddot{c}(t) + \mathsf{K}^{-1} (t) \dot{\mathsf{K}}(t) \dot{c}(t) + \mathsf{K}^{-1} (t) \mathsf{A}(t) c(t) = \mathsf{K}^{-1} (t) \mathsf{S}_f(t)
\end{equation*}
This is a system of $N$ second order, ordinary differential equations, admitting by Picard-Lindel\"of thereom a smooth solution $c \in C^\infty(I, \mathbb{R}^N)$ for prescribed initial data. Comparing with Equation \eqref{Eq: Cauchy problem on non-static spacetime}, these are chosen starting from $u_0,u_1\in C^\infty_0(\mathring{\Sigma})$ and replacing them with $P_N(u_0),P_n(u_1)\in V_N$. For later convenience observe that 
\begin{equation*}
\lim_{N \to \infty} \|P_N u_0 - u_0\|_{H^1} = 0, \quad \textrm{and} \quad \lim_{N \to \infty} \|P_N u_1 - u_1\|_{L^2} = 0,
\end{equation*}
where the same holds true replacing $L^2,H^1$ with their counterparts $L^2_t,H^1_t$.

\subparagraph{Boundedness of the sequences $u_N,\partial_t u_N$:} Equation \eqref{Eq: coercive estimate 1} suggests to consider the following ($N$-){\em finite energy}
\begin{equation}\label{Eq: finite energy}
\mathsf{E}_N [u_N](t) \doteq \frac{1}{2} \mathsf{K}(\partial_t u_N, \partial_t u_N) + \frac{1}{2} \mathsf{B}(u_N, u_N), 
\end{equation}
where $u_N (t) \in V_N$ for any but fixed $t \in [t_0, t_1]$ and where $\mathsf{B}=\mathsf{A}+\lambda\mathsf{K}$. Equations \eqref{Eq: estimate for K} and \eqref{Eq: estimate for A} can be combined to obtain an estimate from above, namely there exists $\mathfrak C > 0$ such that 
\begin{equation}\label{Eq: estimate of the finite energy from above}
\mathsf{E}_N [u](t)\leq \mathfrak C (\|\partial_t u_N\|^2_{L^2_t} + \|u_N\|^2_{H^1_t}). 
\end{equation}
To find a suitable coercivity estimate bounding the finite energy from below, first of all we recall that
$$\mathsf{K}(\partial_t u_N, \partial_t u_N) \geq C_0 \|\partial_t u_N\|_{L^2_t}^2,$$
which, combined with Equation \eqref{Eq: coercive estimate 1}, entails that there exists $\mathfrak c>0$  such that
\begin{equation}\label{Eq: estimate of the finite energy from below}
	\mathsf{E}_N [u_N](t)  \ge \mathfrak{c} (\|\partial_t u_N\|^2_{L^2_t} + \|u_N\|^2_{H^1_t}). 
\end{equation}
Expanding the time derivative in Equation \eqref{Eq: weak formulation N}, we can rewrite it as
\begin{equation}\label{Eq: Aux2}
	\mathsf{K}(\partial_t^2 u_N, v) + \dot{\mathsf{K}} (\partial_t u_N, v) + \mathsf{A}(u_N, v) = \mathsf{S}_f(v), \qquad \dot{\mathsf{K}}(u,v) \doteq \int_{\Sigma} \frac{\partial_t(\beta^{-\frac{1}{2}}\sqrt{|h|_t})}{\sqrt{|h|_t}} u \, v \, d\mu_t, 
\end{equation}
Since $v\in V_N$, we an set $v \doteq \{\partial_t u_N (t)\}_{t \in \mathbb{R}} \in V_N$ and, using the identities
\begin{flalign*}
	\frac{d}{dt} \mathsf{K}(\partial_t u_N, \partial_t u_N) = 2 \mathsf{K}(\partial_t^2 u_N, \partial_t u_N) + \dot{\mathsf{K}}(\partial_t u_N, \partial_t u_N),\;	\frac{d}{dt} \mathsf{A}(u_N, u_N) = 2 \mathsf{A}(u_N, \partial_t u_N) + \dot{\mathsf{A}}(u_N, u_N),
\end{flalign*}
Equation \eqref{Eq: Aux2} can rewritten as
\begin{equation*}
	\frac{1}{2} \frac{d}{dt} \left[ \mathsf{K}(\partial_t u_N, \partial_t u_N) + \mathsf{A} (u_N, u_N) \right] = \mathsf{S}_f(\partial_t u_N) + \frac{1}{2} \dot{\mathsf{A}}(u_N, u_N) - \frac{1}{2} \dot{\mathsf{K}}(\partial_t u_N, \partial_t u_N). 
\end{equation*}
From Equation \eqref{Eq: finite energy}, we get
\begin{equation}
	\label{Eq: time derivative of finite energy}
	\frac{d}{dt} \mathcal{E}_N [u] (t) = \mathrm{S}_f(\partial_t u_N) + \frac{1}{2} \dot{\mathrm{A}}(u_N, u_N) - \frac{1}{2} \dot{\mathrm{K}}(\partial_t u_N, \partial_t u_N) + \lambda \mathrm{K}(u_N, \partial_t u_N) + \frac{1}{2} \lambda \dot{\mathrm{K}}(u_N, u_N). 
\end{equation}
We shall now derive a suitable upper bound for the right-hand side of Equation \eqref{Eq: time derivative of finite energy} and, to this end, set
\begin{equation*}
	C_2=\max\limits_{x\in D}\left\{\frac{\partial_t(\beta^{-\frac{1}{2}}\sqrt{|h_t|})}{\sqrt{|h_t|}},\frac{\partial_t(\beta^\frac{1}{2} \sqrt{|h_t|})}{\sqrt{|h_t|}},\right\}
\end{equation*}
Having already continuity bounds on both $\mathsf{K}$ and $\mathsf{A}$, we need to focus only on $\dot{\mathsf{K}}$ and $\dot{\mathsf{A}}$. Starting from the former, the Cauchy-Schwartz and Young inequalities entail
\begin{equation*}
	|\dot{\mathsf{K}}(u_N,v_N) |\leq C_2 \|u_N\|_{L^2_t} \|v_N\|_{L^2_t} \leq \frac{C_2}{2} \left(\|u_N\|^2_{L^2_t} + \|v_N\|^2_{L^2_t} \right),.
\end{equation*}
Focusing on $\dot{\mathsf{A}}(u_N,v_N)$, with reference to Equation \eqref{Eq: potential form}, we consider first the mass and boundary contributions:
\begin{flalign*}
	\bigg \vert \left(m^2 \int_{\Sigma} \frac{1}{\sqrt{|h_t|}}\partial_t(\beta^\frac{1}{2} \sqrt{|h_t|}) u_N \, v_N \, d\mu_t \right) \bigg \vert &\leq C_2 M^2_1  \|u_N\|_{H^1_t} \|v_N\|_{H^1_t}\\
	\bigg \vert \left(\kappa \int_{\Sigma} \frac{\partial_t(\beta^\frac{1}{2}\sqrt{|h_{\partial,t}|})}{\sqrt{|h_{\partial,t}|}} u_N \, v_N \, d\sigma_t \right) \bigg \vert & \leq C_{2,\partial} |\kappa| \, \|u_N\|_{H^1_t} \, \|v_N\|_{H^1_t},  
\end{flalign*}
where $|h_{\partial,t}|$ is the determinant of the metric $h_t$ pulled-back to $\partial\Sigma_t$, while $C_{2,\partial}>0$ is a constant obtained combining $C_2$ with the same estimate leading to Equation \eqref{Eq: estimate for A}. The last contribution coming from Equation \eqref{Eq: potential form} takes the form
\begin{flalign*}
\int_{\Sigma} \frac{1}{\sqrt{|h_t|}}\partial_t(\beta^\frac{1}{2} \sqrt{|h_t|}) h_t(\nabla u_N, \nabla v_N) d\mu_t  + \int_{\Sigma} \beta^\frac{1}{2} \, \partial_t h_t (\nabla u_N, \nabla v_N) d\mu_t,
\end{flalign*}
where $\partial_t h_t (\nabla u_N, \nabla v_N)= h_t(\nabla u_N, H \nabla v_N)$ with $H_l^j \doteq (h_t)_{lk} (\partial_t h_t)^{kj}$. Since $H$ is uniformly bounded on $D$, we can infer that also the absolute value of this last bit of $\dot{\mathsf{A}}(u_N,v_N)$ can be controlled in terms of the product $\|u_N\|_{H^1_t} \|v_N\|_{H^1_t}$. Putting together all these estimates, we can infer that there exists $\tilde{C}_2>0$ such that
\begin{equation}\label{Eq: Estimate on dotA}
	|\dot{\mathsf{A}}(u_N,v_N)| \leq \tilde{C}_2 \|u_N\|_{H^1_t} \|v_N\|_{H^1_t}. 
\end{equation}
At last we consider the source term which can be bounded from above by 
\begin{equation}
	|\mathsf{S}_f(v_N)| \le C_1\|f\|_{L^2_t} \|v_N\|_{L^2_t} \leq \frac{1}{2} C_1 \left( \|f\|^2_{L^2_t} + \|v_N\|^2_{L^2_t} \right). 
\end{equation}
Putting together all these results we end up with the existence of a constant $C_{\mathrm{max}}>0$ such that
\begin{equation}\label{Eq: eq 0}
\bigg \vert \frac{d}{dt} \mathsf{E}_N[u_N](t) \bigg \vert \leq C_{\mathrm{max}} \left[ \|f\|^2_{L^2_t} + \|\partial_t u_N\|^2_{L^2_t} + \|u_N\|^2_{L^2_t} \right]\le C \|f\|^2_{L^2_t} + \frac{C_{\max}}{\mathfrak{c}} \mathsf{E}_N[u_N](t).,
\end{equation}
where we used the coercive estimate in Equation \eqref{Eq: estimate of the finite energy from below}. Applying now Gr\"onwall lemma and recalling Equation \eqref{Eq: estimate of the finite energy from below}, it descends that, for all $t\in I$
\begin{equation*}
 \mathfrak{c} (\|\partial_t u_N\|^2_{L^2_t} + \|u_N\|^2_{H^1_t})\leq\mathsf{E}_N[u_N](t) \leq\mathsf{E}_N[u_N](t_0) + \int_{t_0}^t \|f\|^2_{L^2_s} \, ds,
\end{equation*}
where the right hand side is a finite quantity being the initial data $P_N(u_0),P_N(u_1)\in V_N$ and being convergent to $u_0,u_1$, respectively. Since all estimates are uniform in the time variable, we can summarize the outcome of the preceding analysis in the following lemma. 

\begin{lemma}\label{Lem: uN and derivative are bounded}
Let $u_N$ be a solution of Equation \eqref{Eq: weak formulation N} and let $\mathsf{E}_N[u_N]$ be the corresponding finite-energy as per Equation \eqref{Eq: finite energy}. It holds that, given $I=[t_0,t_1]$,
\begin{equation*}
\sup_{t \in I} \sup_{N \in \mathbb{N}} \mathsf{E}_N[u_N](t) < \infty,
\end{equation*}
which implies, in turn, that 
\begin{equation*}
\sup_{t \in I} \sup_{N \in \mathbb{N}} \left(\|\partial_t u_N\|^2_{L^2_t} + \|u_N\|^2_{H^1_t} \right) < \infty,
	\end{equation*}
namely, $u_N \in L^\infty(I; H^1_t)$ and $\partial_t u_N \in L^\infty(I; L^2_t).$
\end{lemma}

\noindent At this stage we can exploit the equivalence between the $L^2$- and $H^1$-norms on $\Sigma$ induced by $h_t$ and by $h$ to conclude that  
$$u_N \in L^\infty(I; H^1(\Sigma))\quad\textrm{and}\quad\partial_t u_N \in L^\infty(I; L^2(\Sigma)).$$
Varying over $N$, we have that boundedness of the sequences $u_N,\partial_t u_N$, $N\in\mathbb{N}$, entails that we can apply the Banach-Alaoglu-Bourbaki theorem concluding that we can extract subsequences, still denoted by the same symbols with a slight abuse of notation, which are weakly $\ast$-converging:
\begin{equation*}
u_N \overset{\ast}{\rightharpoonup} u\in L^\infty(I;H^1(\Sigma)),\quad\textrm{and}\quad \partial_t u_N \overset{\ast}{\rightharpoonup} v\in L^\infty(I;L^2(\Sigma)).
\end{equation*}

\noindent \emph{A priori} there is no connection between $u$ and $v$ and the following corollary established its existence.

\begin{corollary}\label{Cor: Convergence of Galerkin}
Let $\{u_N\}_{N \in \mathbb{N}} \subset L^\infty(I; H^1_t)$ be as in Lemma \ref{Lem: uN and derivative are bounded} and let $u\in L^\infty(I;H^1(\Sigma))$ and $v\in L^\infty(I;L^2(\Sigma))$ be the weak $*$-limits of $u_N$ and $\partial_t u_N$ respectively. It holds that $v = \partial_t u$ in $\mathcal{D}^\prime(I;L^2(\Sigma))$ and that $u\in C^0(I;L^2(\Sigma))$.
\end{corollary}

\begin{proof}
For all $\rho \in C^\infty_0(I)$ and $h \in L^2(\Sigma)\subset H^{-1}(\Sigma)$, integration by parts yields
	\begin{equation*}
		\int_I \langle \partial_s u_N (s), h \rangle \rho(s) \, ds = - \int_{I} \langle u_N(s), h \rangle \dot{\rho}(s) \, ds.
	\end{equation*}
	Taking the limit as $N \to \infty$ on both sides, we obtain 
	\begin{equation*}
		\int_I \langle v(s) ,h \rangle \rho(s) \, ds = - \int_{I} \langle u(s), h \rangle \dot{\rho}(s) \, ds = \int_I \langle \partial_s u(s), h \rangle \rho(s) \, ds.
	\end{equation*}
	This entails that $v = \partial_t u$ in $\mathcal{D}'(I, L^2(\Sigma))$. Since per hypothesis $v \in L^\infty(I; L^2(\Sigma))$, it follows that $u \in W^{1, \infty} (I, L^2(\Sigma))\subset C^{0,1}(I; L^2(\Sigma))$. Hence, $\partial_t u$ exists for almost all values of $t\in I$ and it coincides with $v(t)$. 
\end{proof}

\noindent To conclude this step, we need to prove that the limit function constructed in Lemma \eqref{Lem: uN and derivative are bounded} is still a solution of the weak formulation of the Klein-Gordon equation, in a sense specified by the following lemma.

\begin{lemma}\label{Lem: Weak solution to the problem}
Let $u \in C^0(I; L^2(\Sigma))$ be constructed as in Corollary \ref{Cor: Convergence of Galerkin}. Then, it is a solution of Equation \eqref{Eq: weak formulation 2} in $\mathcal{D}'(I; H^{-1}(\Sigma))$. 
\end{lemma}

\begin{proof}
For each $N \in \mathbb{N}$ and for all $v\in V_N$, it holds that
\begin{equation}\label{Eq: aux VN}
\frac{d}{dt} \mathsf{K}(\partial_t u_N,v) + \mathsf{A}(u_N, v) = \mathsf{S}_f(v).
\end{equation}
Observe that, if $v \in V_{N_0}$ for any but fixed $N_0 \in \mathbb{N}$, then $v \in V_N$ for all $N \ge N_0$. Testing Equation \eqref{Eq: aux VN} against $\rho\in C^\infty_0(I)$ yields, upon integration by parts,
\begin{gather*}
\int_I  \frac{d}{ds} \mathsf{K}(\partial_s u_N,v) \, \rho(s) \, ds + \int_I [\mathsf{A}(u_N, v) - \mathsf{S}_f(v)] \, \rho(s) \, ds =\\
-\int_I  \mathsf{K}(\partial_s u_N,v) \, \dot{\rho}(s) \, ds+ \int_I [\mathsf{A}(u_N, v) - \mathsf{S}_f(v)] \, \rho(s) \, ds=0.
\end{gather*}
This last identity holds true for all $v\in\bigcup\limits_{N\in\mathbb{N}}V_N$ and, in view of the continuity of the forms $\mathsf{K}, \mathsf{A}$ and of the source term $\mathsf{S}_f$ combined with the weak-$*$ convergence of the sequences $u_N$ and $\partial_t u_N$, we can take the limit as $N\to\infty$
\begin{equation}\label{Eq: aux VN 1}
- \int_I  \mathrm{K}(\partial_s u,v) \, \dot{\rho}(s) \, ds +\int_I [\mathrm{A}(u, v) - \mathrm{S}_f(v)] \, \rho(s) \, ds = 0.
\end{equation}
Recalling that the union of the spaces $V_N$ is dense in $H^1(\Sigma)$, the desired result follows once again by continuity.
\end{proof}

\begin{remark}\label{Rem: Galerkin Final Nail}
The initial conditions are inherited from the overall procedure. Although weak-$*$ convergence alone does not preserve pointwise values of the solution in time, the uniform bound on $u_N$ and on $\partial_tu_N$ yields uniform continuity of $u_N$ in $L^2(\Sigma)$. Together with $u_N(t_0)=P_N(u_0)$ converging to $u_0$, this implies $u(t_0)=u_0$. We do not discuss the implementation of $u_1$ at this stage since we have only proven continuity of the solution in time. Similarly, the Robin boundary condition is encoded in the weak formulation through the boundary term entering the quadratic form $\mathrm{A}$. Once we will prove smoothness of the solution, Green's identity allows one to recover $(\nabla_n+\kappa)u \vert_{\partial\mcM}=0$.
\end{remark}

\paragraph{Step 3 -- Improving Regularity:} We shall now prove that the solution constructed in Step 2 is actually smooth. We start by improving the regularity along the time direction, considering the spatial ones only at a later stage. As a preliminary observation, we recall that we are considering smooth and compactly supported source terms $f$ such that $\textrm{supp}(f)\subset D$ and, without loss of generality, we can engineer $D$ so that, the exists $\varepsilon>0$ for which $f = 0$ on $[t_0, t_0+ \epsilon]\times \Sigma$. 

\subparagraph{Enhancing the regularity in the $t$-variable:} This is a recursive procedure. We discuss here the first iteration, commenting on the others at the end of the analysis, highlighting only the minor differences which occur. We start from the weak formulation and, in particular, following the Galerkin approximation, $u_N$ is a solution of Equation \eqref{Eq: weak formulation N}. Taking and expanding the time derivative of this identity we obtain  
\begin{equation}\label{Eq: Delta1}
\mathsf{K}(\partial_t^3 u_N, v) + 2\dot{\mathsf{K}}(\partial_t^2 u_N, v) + \ddot{\mathsf{K}}(\partial_t u_N, v) = \mathsf{S}_{\dot f}(v) - \mathsf{A}(\partial_t u_N, v) -  \dot{\mathsf{A}}(u_N, v), 
\end{equation}
where 
$$\ddot{\mathsf{K}}(\partial_t^2 u_N, v)=\int_{\Sigma} \frac{\partial^2_t(\beta^{-\frac{1}{2}}\sqrt{|h|_t})}{\sqrt{|h|_t}} u \, v \, d\mu_t. $$
Mimicking Step 2 we consider the updated energy functional
\begin{equation}\label{Eq: Updated Finite Energy}
\mathsf{E}_{N,1}[u_N](t) \doteq \frac{1}{2} \mathsf{K}(\partial_t^2 u_N, \partial_t^2 u_N) + \frac{1}{2} \mathsf{B}(\partial_t u_N, \partial_t u_N),
\end{equation}
where $\mathsf{B}=\mathsf{A}+\lambda\mathsf{K}$. We can repeat the same analysis leading first to Equation \eqref{Eq: coercive estimate 1} and then to Equation \eqref{Eq: estimate of the finite energy from below} with $u_N$ replaced by $\partial u_N$. Recall that the former is a finite linear combination of eigenfunctions of the Neumann Laplacian with coefficients which are smooth in time. Hence there exists a positive constant $\mathfrak{c}_1>0$ such that 
\begin{equation}\label{Eq: Coercive Estimate Updated Energy}
\mathsf{E}_{N,1}[u_N](t) \geq\mathfrak{c}_1\left(\|\partial_t^2 u_N\|^2_{L^2_t} + \|\partial_t u_N\|^2_{H^1_t}\right).
\end{equation}
Setting $v = \partial_t^2 u_N$ in Equation \eqref{Eq: Delta1} and taking the time derivative of Equation \eqref{Eq: Updated Finite Energy}, we obtain
\begin{flalign*}
	\frac{d}{dt} \mathsf{E}_{N,1}[u](t) &= -\frac32 \dot{\mathsf{K}}(\partial_t^2 u_N, \partial_t^2 u_N) + \frac{1}{2} \dot{\mathsf{A}} (\partial_t u_N, \partial_t u_N) + \lambda \mathsf{K}(\partial_t u_N, \partial_t^2 u_N) \\ &+\frac{\lambda}{2} \dot{\mathsf{K}}(\partial_t u_N, \partial_t u_N) + \mathsf{S}_{\partial_t f} (\partial_t^2 u_N) - \ddot{\mathsf K}(\partial_t u_N, \partial_t^2 u_N) - \dot{\mathsf{A}}(u_N, \partial_t^2 u_N). 
\end{flalign*}
In order to repeat the procedure as in Step 2, we can infer that the only problematic term on the right hand side is the last one since the coercive estimate in Equation \eqref{Eq: Coercive Estimate Updated Energy} allows only a control on $\|\partial_t^2 u_N\|_{L^2_t}$, but not $\|\partial_t^2 u_N\|_{H^1_t}$. To bypass this hurdle observe that
\begin{equation*}
	\frac{d}{dt} \dot{\mathsf{A}} (u_N, \partial_t u_N) - \dot{\mathsf A}(\partial_t u_N, \partial_t u_N) - \ddot{\mathsf A} (u_N, \partial_t u_N) = \dot{\mathsf{A}}(u_N, \partial_t^2 u_N), 
\end{equation*}
which entails, that calling $\mathsf{F}_{N,1} \doteq \mathsf{E}_{N,1}+\dot{\mathsf{A}}_{N,1}$,
\begin{flalign}
	\label{Eq: aux eq step 3}
	\notag \frac{d}{dt} \left[ \mathsf{F}_{N,1}[u_N](t)\right] = &-\frac32 \dot{\mathsf{K}}(\partial_t^2 u_N, \partial_t^2 u_N) + \frac{3}{2}\dot{\mathsf{A}} (\partial_t u_N, \partial_t u_N) + \lambda \mathsf{K}(\partial_t u_N, \partial_t^2 u_N) \\ &\notag + \frac{\lambda}{2} \dot{\mathsf{K}}(\partial_t u_N, \partial_t u_N) + \mathsf{S}_{\partial_t f} (\partial_t^2 u_N) - \ddot{\mathsf K}(\partial_t u_N, \partial_t^2 u_N) \\ &+ \ddot{\mathsf{A}} (u_N, \partial_t u_N). 
\end{flalign}
Here $\ddot{\mathsf A}$ is constructed out of Equation \eqref{Eq: potential form} and it comprises of three terms whose explicit form is not relevant for this analysis and, hence, we avoid writing them explicitly. Equation \eqref{Eq: Estimate on dotA} can be adapted to the case in hand in two ways, useful in the following. In the first we combine it with the Young inequality, while, in the second, with a specific instance of the Peter-Paul counterpart, obtained fixing $\epsilon=\frac{1}{2}$. These entail that there exists $\tilde{c}_0,\tilde{c}^\prime_0,\tilde{c}_1>0$ such that
\begin{subequations}
\begin{equation}\label{Eq: 0 step 3}
|\dot{\mathsf{A}} (u_N, \partial_t u_N)| \leq\frac{\widetilde{C}_2}{2}\left(\|u_N\|^2_{H^1_t} + \|\partial_t u_N\|^2_{H^1_t} \right) \leq \tilde{c}_0\mathsf{E}_{N}[u](t) + \tilde{c}_1\mathsf{E}_{N,1}[u](t),
\end{equation}
\begin{equation}\label{Eq: 0 step 3 v2}
|\dot{\mathsf{A}} (u_N, \partial_t u_N)| \leq\tilde{c}^\prime_0 \mathsf{E}_{N}[u_N](t) + \frac{1}{2} \mathsf{E}_{N,1}[u_N](t)
\end{equation}
\end{subequations}
where, in the second inequality we used Equations \eqref{Eq: estimate of the finite energy from above} and \eqref{Eq: Coercive Estimate Updated Energy}. Focusing on Equation \eqref{Eq: aux eq step 3} we can follow the same idea as in Step 2. This means that, first of all, we bound each term on the right hand side from above and, in turn, we bound the ensuing expression using the coercive estimate. Each inequality below would require a separate proportionality constant, but, to avoid burdening an already heavy notation, we use a single collective constant $\mathfrak{C}_1>0$. Since previously unspecified intermediate inequalities are necessary in two instances, also in this case we use a single collective constant $\widetilde{\mathfrak{C}}_1>0$: 
\begin{flalign*}
\vert \ddot{\mathsf{K}}(\partial_t u_N, \partial_t^2 u_N) \vert &\leq\widetilde{\mathfrak{C}}_1\left(\|\partial_t u_N\|^2_{H^1_t} + \|\partial_t^2 u_N\|^2_{L^2_t} \right) \le \mathfrak{C}_1\, \mathsf{E}_{N,1}[u_N](t) \\
\vert \mathsf{S}_{\partial_t f}(\partial_t^2 u_N) \vert &\leq \|\partial_t f\|_{L^2_t} \, \|\partial_t^2 u_N\|_{L^2_t} \le \frac{1}{2} \|\partial_t f\|_{L^2_t}^2 + \mathfrak{C}_1\, \mathsf{E}_{N,1}[u_N](t) \\
\vert \dot{\mathsf{K}}(\partial_t^2 u_N, \partial_t^2 u_N)\vert & \leq C_2 \|\partial_t^2 u_N\|^2_{L^2_t} \leq \mathfrak{C}_1 \mathsf{E}_{N,1}[u_N](t) \\
\vert {\mathsf{K}}(\partial_t u_N, \partial_t u_N)\vert  & \leq C_1 \|\partial_t u_N\|^2_{H^1_t} \leq\mathfrak{C}_1 \mathsf{E}_{N,1}[u_N](t) \\
\vert \ddot{\mathsf A}(u_N, \partial_t u_N)\vert  & \leq \widetilde{\mathfrak{C}}_1\|\partial_t u_N\|^2_{H^1_t} \leq\mathfrak{C}_1 \mathsf{E}_{N,1}[u_N](t) \\
\vert \mathsf{K}(\partial_t u_N, \partial_t^2 u_N) \vert &\leq C_1 \left(\|\partial_t u_N\|^2_{H^1_t} + \|\partial_t^2 u_N\|^2_{L^2_t}\right) \leq\mathfrak{C}_1 \mathsf{E}_{N,1}[u](t). 
\end{flalign*}
Combining all the above estimates with Equation \eqref{Eq: 0 step 3}, Equation \eqref{Eq: aux eq step 3} leads to 
\begin{equation}\label{Eq: star step 5}
\frac{d}{dt}\mathsf{F}_{N,1}[u_N](t)  \leq \mathfrak{C}_2 \mathsf{E}_{N,1}[u_N] (t) + \mathfrak{C}_2^\prime \mathsf{E}_{N}[u_N](t) + \frac{1}{2} \|\partial_t f\|^2_{L^2_t},
\end{equation}
where $\mathfrak{C}_2,\mathfrak{C}^\prime_2$ are two suitable positive constants. Since Equation \eqref{Eq: 0 step 3 v2} entails that $\dot{\mathsf{A}} (u_N, \partial_t u_N) \geq -  \tilde{c}^\prime_0 \mathsf{E}_{N}[u_N](t) - \frac{1}{2} \mathsf{E}_{N,1}[u_N](t)$, adding to both sides $\mathsf{E}_{N,1}[u_N](t)$ yields
\begin{equation}\label{Eq: Delta step 5}
	\mathsf{E}_{N,1}[u_N](t) \le 2\mathsf{F}_{N,1}[u_N](t) + 2\tilde{c}^\prime_0 \, \mathsf{E}_N[u_N](t).
\end{equation}
Combined with Equation \eqref{Eq: star step 5} this leads to
\begin{equation*}
	\frac{d}{dt}\mathsf{F}_{N,1}[u_N](t)  \le 2\mathfrak{C}_2\mathsf{F}_{N,1}[u_N](t) + (2\tilde{c}^\prime_0+\mathfrak{C}_1) \, \mathsf{E}_N[u_N](t)  + \frac{1}{2} \|\partial_t f\|^2_{L^2_t}.
\end{equation*}
By Gr\"onwall lemma, 
\begin{equation*}
	\mathsf{F}_{N,1}[u_N](t) \le e^{2\mathfrak{C}_2(t-t_0)}\mathsf{F}_{N,1}[u_N](t_0)+\int_{t_0}^t e^{2\mathfrak{C}_2(t-s)} \left[(2\tilde{c}^\prime_0+\mathfrak{C}_1)\mathsf{E}_N[u_N](s) + \frac{1}{2} \|\partial_s f\|^2_{L^2_s} \right] \, ds.
\end{equation*}
\begin{remark}\label{Rem: Initial Data}
At this stage it is tempting to conclude that $\mathsf{F}_{N,1}[u_N](t_0)$ is a bounded quantity determined by the initial data. This is true with one caveat, namely $\mathsf{F}_{N,1}$ contains in particular $\mathsf{E}_{1,N}$ which is in turn determined by $\partial_t u_N$ and $\partial^2_t u_N$. At the $t=t_0$ the first term is determined by the initial datum in Proposition \ref{Prop: Existence of Solutions}, namely $\partial_t u_N(t_0)=P_N(u_1)$. In order to determine $\partial^2_t u_N[t_0]$ we focus on Equation \eqref{Eq: Aux2}. At time $t=t_0$, it holds
$$\mathsf{K}(\partial_t^2 u_N(t_0), v) = - \dot{\mathsf{K}} (\partial_t u_N(t_0), v) - \mathsf{A}(u_N(t_0), v) + \mathsf{S}_f(v).$$
Since this identity holds true for all $v\in H^1(\Sigma)$ and since $\mathsf{K}$ is invertible, we can extract $$\partial^2_t u_N[t_0]=\mathsf{K}^{-1}\left( - \dot{\mathsf{K}} (\partial_t u_N(t_0))- \mathsf{A}(u_N(t_0))+\mathsf{S}_f\right),$$
where all time dependent quantities are evaluated at $t=t_0$.
\end{remark}

\noindent Taking into account Step 2 and Equation \eqref{Eq: Delta step 5}, we can conclude
\begin{equation*}
\sup_{t \in I} \sup_{N \in \mathbb{N}} \mathsf{E}_{N}[u_N](t) < \infty,\Longrightarrow
\sup_{t \in I} \sup_{N \in \mathbb{N}} \mathsf{F}_{N,1}[u_N](t) < \infty\Longrightarrow
\sup_{t \in I} \sup_{N \in \mathbb{N}} \mathsf{E}_{N,1}[u_N](t) < \infty.
\end{equation*}
Using the coercive estimate in Equation \eqref{Eq: Coercive Estimate Updated Energy}, we can infer 
\begin{equation}\label{Eq: additional regularity of Galerkin approx}
\sup_{N \in \mathbb{N}} \|\partial_t u_N\|_{L^{\infty}(I; H^1_t)} < \infty\quad\textrm{and}\quad\sup_{N \in \mathbb{N}} \|\partial_t^2 u_N\|_{L^{\infty}(I; L^2_t)} < \infty.
\end{equation}
In other words, the sequences $\{\partial_t u_N\}_{N\in\mathbb{N}}$ and $\{\partial_t^2 u_N\}_{N\in\mathbb{N}}$ are uniformly bounded in $L^{\infty}(I;H^1_t)$ and $L^{\infty}(I;L^2_t)$, respectively. At this stage we can exploit the equivalence between the $L^2$- and $H^1$-norms on $\Sigma$ induced by $h_t$ and by the reference metric $h$. Hence, using the same rationale as in Step 2, we can infer that there exists a weakly-$*$ convergent subsequence, still denoted by the same symbol with a mild abuse of notation, such that 
\begin{equation*}
	\partial_t u_N \overset{\ast}{\rightharpoonup} v_1 \in L^{\infty}(I;H^1(\Sigma)), \qquad \partial_t^2 u_N \overset{\ast}{\rightharpoonup} v_2 \in L^{\infty}(I;L^2(\Sigma)).
\end{equation*}
Yet, in Step 2. we established that $\partial_t u_N \overset{\ast}{\rightharpoonup} \partial_t u \in L^{\infty}(I; L^2(\Sigma))$. Uniqueness of the limit combined with the continuous embedding $H^1(\Sigma)\hookrightarrow L^2(\Sigma)$ entails that
\begin{equation}\label{Eq: additional regulatity of u}
v_1 = \partial_t u \in L^\infty(I; H^1(\Sigma)). 
\end{equation}
At the same time, one can adapt with minor modifications Corollary \ref{Cor: Convergence of Galerkin} to this case establishing that $\partial_t^2 u = v_2$ in $\mathcal{D}'(I; L^2_t)$. Yet, since $v_2 \in L^\infty(I; L^2(\Sigma))$ and, thus, $$\partial_t^2 u = v_2 \in L^\infty(I, L^2(\Sigma)),$$ it follows that $u \in W^{1, \infty} (I; H^1(\Sigma))\subset C^0(I; H^1(\Sigma))$ and $\partial_t u \in W^{1, \infty}(I; L^2(\Sigma))\subset C^0(I; L^2(\Sigma))$. We can therefore conclude that 
\begin{equation}\label{Eq: regularity of u C1}
u \in C^1(I; L^2(\Sigma)) \cap C^0(I; H^1(\Sigma)).
\end{equation}

\noindent The same procedure used to control the second derivative $\partial^2_t u_N(t)$ can be applied with minor modifications to $\partial^k u_N(t)$ for all $k>2$. For this reason, we avoid repeating it and we limit ourselves to pointing out that it yields
\begin{equation}\label{Eq: regularity of u higher order}
u \in C^k(I; L^2(\Sigma)) \cap C^{k-1}(I; H^1(\Sigma)), \; \forall k \in \mathbb{N} \Rightarrow u \in C^\infty(I; H^1(\Sigma)). 
\end{equation}

\subparagraph{Enhancing the spatial regularity:} In order to gain better control along the spatial directions, our strategy consists of regarding the weak formulation of the Klein-Gordon equation as a weak elliptic problem on $\Sigma$ with Robin boundary conditions. To this end, consider an arbitrary but fixed solution $u\in C^\infty(I; H^1(\Sigma))$ of Equation \eqref{Eq: weak formulation 2} and, for all $v\in H^1(\Sigma)$, let
$$g_0(v)(t) \doteq \mathrm{S}_f(v)-\frac{d}{dt}\mathsf K(\partial_t u,v).$$
Hence Equation \eqref{Eq: weak formulation 2} reads
\begin{equation}\label{Eq: Weak Elliptic Robin problem}
\mathsf{A}_u(v)=g_0(v)(t),
\end{equation}
where the subscript $u$ recalls that this element cannot vary. Furthermore its principal symbol reads
\begin{equation*}
	\sigma_p(\mathrm{A}) = \sqrt{\beta} h^{ij}_t k_i k_j, 
\end{equation*}
which is uniformly elliptic on $I\times\Sigma$. Since from the preceding analysis we know that $g_0(t)\in H^1(\Sigma)$ uniformly in $t\in I$, we can apply elliptic regularity to infer that there exists a time-independent $C>0$ such that
$$\|u(t)\|_{H^3}\leq C\left(\|g_0(t)\|_{H^1}+\|u(t)\|_{L^2}\right),$$
which entails that $u\in L^\infty(I; H^3(\Sigma))$.

To show the interplay with smoothness in time, we differentiate Equation \eqref{Eq: Weak Elliptic Robin problem} with respect to $t$. For every fixed $v\in H^1(\Sigma)$ and $t\in I$, we obtain
\begin{equation*}
	\mathrm{A}(\partial_t u,v)
	=
	-\dot{\mathrm{A}}(u,v)
	+
	\frac{d}{dt}g_0(v).
\end{equation*}
Observe that the right-hand side defines, for each fixed $t$, the source term of an elliptic problem for $\partial_tu$. 
Since $u\in L^\infty(I;H^3(\Sigma))$ and all coefficients depend smoothly on $t$, both the bulk source and the boundary datum allow to apply elliptic regularity, which yields 
\begin{equation*}
	\partial_tu\in L^\infty(I;H^3(\Sigma)).
\end{equation*}
Repeating the same argument for higher time derivatives gives
\begin{equation*}
	\partial_t^k u\in L^\infty(I;H^3(\Sigma)),\qquad \forall k\in\mathbb N .
\end{equation*}
Combining this result with the smoothness in time previously established in $H^1(\Sigma)$, we obtain
\begin{equation*}
	\partial_t^k u\in L^\infty(I;H^3(\Sigma))\cap C^\infty(I;H^1(\Sigma)),\qquad \forall k\in\mathbb N .
\end{equation*}
Hence, for any arbitrary but fixed $k$ we can infer that $\partial_t^k u\in W^{1,\infty}(I; H^3(\Sigma))\subset C^0(I; H^3(\Sigma))$. In other words $u\in C^\infty(I; H^3(\Sigma))$. Repeating the elliptic regularity procedure iteratively, we end up with $u\in C^\infty(I; H^\infty(\Sigma))$, hence, by Sobolev embedding,
\begin{equation*}
	u \in C^\infty(I \times \Sigma). 
\end{equation*}
Since every $p\in\mcM$ lies in a cylinder of the form $I\times\Sigma$, we can extrapolate $u\in C^\infty(\mcM)$.

\paragraph{Step 4 -- Uniqueness:} In this step, we complement the preceding one by establishing uniqueness of the solution of Equation \eqref{Eq: weak formulation 2}. By linearity, this is tantamount to showing that vanishing initial data $u_0=u_1=0$ and vanishing source $f=0$ imply $u=0$. We can consider Proposition \ref{Prop: Causal Support of Solutions} and, in particular, the energy estimate in Equation \eqref{Eq: Energy estimate 2} to conclude, as for causal propagation, that $u$ must vanish.

\paragraph{Step 5 -- Dropping Compactness of $\Sigma$:} 

In this final step, we discuss how to remove the compactness hypothesis of $\Sigma$. Let $I=[t_0,t_1]$ and recall that
$u_0,u_1\in C^\infty_0(\mathring{\Sigma}_{t_0})$, while $f\in C^\infty_0(\mcM)$. Choose $q\in I\times\Sigma$ as well as $p\in I^+(q)$ with $t(p)>t_1$ defining
\begin{equation}\label{Eq: Dp}
D_p:=J^-(p)\cap(I\times\Sigma).
\end{equation}
Being $(\mcM,g)$ globally hyperbolic this is a compact set. Choose $\mathcal{C}\subset\Sigma$ to be a smooth and compact set such that $D_p\subset I\times\mathring{\mathcal{C}}$.

Observe that, in comparison to Steps 1--4, here, we have an additional portion of a timelike boundary in addition to $(I\times\mathcal{C})\cap\partial\mcM$. Therein we impose for mere convenience the same Robin boundary conditions as on $\partial\mcM$. Barring minor modifications, we can employ the Step 1--4 to infer that there exists a unique smooth solution of Equation \eqref{Eq: Cauchy problem on non-static spacetime}
$$u_{p,\mathcal C}\in C^\infty(I\times\mathcal C),$$
where the subscripts serve the purpose to recall the underlying choices. To establish that the solution does not depend on $\mathcal{C}$ consider a second compact enlargement $\mathcal{C}\subset\mathcal{C}^\prime\subset\Sigma$. Per construction $D_p\subset I\times\mathring{\mathcal{C}}^\prime$. Consider  $u_{p,\mathcal{C}^\prime}$ solution of Equation \eqref{Eq: Cauchy problem on non-static spacetime} with the same initial data and source. The difference $u_{p,\mathcal{C}^\prime}-u_{p,\mathcal{C}}$ restricted to $D_p$ solves the homogeneous Klein-Gordon equation with vanishing initial data thereon. Yet, the energy estimate in Appendix \ref{App: A} entail that the only smooth solution must be the vanishing one. In other words,
$$u_{p,\mathcal{C}}|_{D_p}=u_{p,\mathcal{C}^\prime}|_{D_p}.$$
We may consequently denote this common restriction simply by $u_p$. Let us now allow $p$ to vary and let us choose $p^\prime\in\mcM$ such that $t(p')>t_1$. Letting $D_{p^\prime}$ be defined as per Equation \eqref{Eq: Dp}, if $r\in D_p\cap D_{p^\prime}\neq\emptyset$, both $u_p$ and $u_{p^\prime}$ solve the same initial-boundary value problem on $D_r$. Once more, using Appendix \ref{App: A}, we can infer that they must coincide in a neighbourhood of $r$, that is $u_p=u_{p^\prime}$. For every $x\in I\times\Sigma$, choose $p$ with $t(p)>t_1$ such that $x\in I^-(p)$ and, consequently, $x\in D_p$. Here $I^-$ denotes the chronological past. The family $\{D_p\}$ as $p$ varies therefore covers $I\times\Sigma$. Define $u(x)\doteq u_p(x)$, where $p$ is any point such that $x\in D_p$. Since $u_p=u_{p'}$ whenever $D_p\cap D_{p^\prime}\neq\emptyset$ this definition is independent of the choice of $p$. To conclude, we observe that the chronological past $I^-(p)$ is open. Hence there must exists an open neighbourhood $\mathcal{U}_x$ of $x$ such that
$\mathcal U_x\subset I^-(p)\cap(I\times\Sigma)\subset D_p$. In other words
$$u|_{\mathcal{U}_x}=u_p|_{\mathcal U_x}\in C^\infty(\mathcal{U}_x).$$
Since $x\in I\times \Sigma$ is arbitrary, it follows that
$$u\in C^\infty(I\times\Sigma).$$
Furthermore, by construction,
$$Pu=f\quad\textrm{and}\quad(\nabla_n+\kappa)u|_{\partial\mcM}=0,$$
with the prescribed initial data. As at the end of Step 3 arbitrariness of the interval $I=[t_0,t_1]$ and uniqueness entail that we can patch up the solutions to conclude that
$$u\in C^\infty(J^+(\Sigma_{t_0})).$$
At last, the finite-propagation estimate of Proposition \ref{Prop: Causal Support of Solutions} yields
$$\operatorname{supp}u\subset J^+\left(\operatorname{supp}(u_0)\cup\operatorname{supp}(u_1)\cup\operatorname{supp}(f)\right).$$
Here we have considered forward propagation, but the same reasoning applies to the counterpart backwards in time yielding
$$u\in C^\infty(J^+(\Sigma_{t_0}))\cap C^\infty(J^-(\Sigma_{t_0}))\equiv C^\infty(\mcM).$$


\begin{thebibliography}{100}
\addcontentsline{toc}{section}{References}

\bibitem[AGS19]{AmmanGrosseNistor}
B. Amman, N. Gro\ss e and V. Nistor, 	
\emph{``Well-posedness of the Laplacian on manifolds with boundary and bounded geometry''},
Math. Nachr. {\bf 292} (2019), no. 6, 1213.
arXiv:1611.00281 [math-AP]

\bibitem[AC04]{Albuq_2004}
L.~C.~de Albuquerque and R.~M.~Cavalcanti,
{\it ``Casimir effect for the scalar field under Robin boundary conditions: A functional integral approach,''}
J. Phys. A: Math. Gen. \textbf{37} (2004), 7039, arXiv:hep-th/0311052

\bibitem[Bar15]{Bar15}
C.~B\"ar,
{\it``Green-hyperbolic operators on globally hyperbolic spacetimes''},
Comm. Math. Phys. \textbf{333}, (2015), 1585,
arXiv:1310.0738 [math-ph]

\bibitem[BGP07]{Baer_2007}
C.~B\"ar, N.~Ginoux and F.~Pf\"affle,
Wave Equations on Lorentzian Manifolds and Quantization,
ESI Lectures in Mathematics and Physics, EMS (2007), 194p

\bibitem[BSS15]{Bellucci_2015}
S.~Bellucci, A.~A.~Saharian and N.~A.~Saharyan,
{\it ``Casimir effect for scalar current densities in topologically nontrivial spaces,''}
Eur. Phys. J. C \textbf{75} (2015), 378, arXiv:1507.08832 [hep-th]


\bibitem[BFS+02]{Bordag_2002}
M.~Bordag, H.~Falomir, E.~M.~Santangelo and D.~V.~Vassilevich,
{\it ``Boundary dynamics and multiple reflection expansion for Robin boundary conditions,''}
\emph{Phys. Rev. D} \textbf{65} (2002), 064032, arXiv: hep-th/0111073

\bibitem[BDF+15]{brunetti2}
R.~Brunetti, C.~Dappiaggi, K.~Fredenhagen and J.~Yngvason, 
Advances in algebraic quantum field theory, 
{\it Springer} (2015), 455p

\bibitem[BFK96]{Brunetti_1996}
R.~Brunetti, K.~Fredenhagen and M.~Kohler, 
{\it ``The microlocal spectrum condition and Wick polynomials of the free fields on curved spacetimes''}, 
Commun. Math. Phys. {\bf 180} (1996), 633--652, 	arXiv:gr-qc/9510056

\bibitem[Cap13]{Caponio13}
E. Caponio
{\it ``Infinitesimal and local convexity of a hypersurface in a semi-Riemannian manifold'',}
Springer Proceedings in Mathematics \& Statistics \textbf{26} (2013), 163, arXiv:1201.0147 [math.DG]

\bibitem[CJS11]{Caponio}
E. Caponio, M. A. Javaloyes and Miguel S\`{a}nchez,
{\it ``On the interplay between Lorentzian Causality and Finsler metrics of Randers type'',}
Rev. Mat. Iberoamericana {\bf 27}, 3, (2011) 919, arXiv:0903.3501 [math.DG]

\bibitem[Cha75]{Chazarin}
J. Chazarain, 
{\it ``Paramétrix du problème mixte pour l'équation des ondes à l'intérieur d'un domaine convexe pour les bicaractéristiques'',}
Journées Équations aux Dérivées Partielles (1975), 165

\bibitem[CB23]{ClarkBloch_2023}
W.~Clark and A.~Bloch,
{\it ``Invariant forms in hybrid and impact systems and a taming of Zeno''},
Arch. Ration. Mech. Anal. \textbf{247} (2023), Art.~13, arXiv:2101.11128 [math.DS]

\bibitem[CS26]{Contini_2026}
A.P. Contini and A. Strohmaier, 
{\it ``Hadamard states for spacetimes with timelike boundaries''}, arXiv: 2609.13358 [math.AP],

\bibitem[CDG26]{CosteriDappiaggiGoi2026}
B.~Costeri, C.~Dappiaggi and M.~Goi,
{\it ``Conservation Law and Trace Anomaly for the Stress-Energy Tensor of a Self-interacting Scalar Field''}, 
Ann. Henri Poincar\'e, \textbf{27} (2026), 1843,
arXiv:2411.07109 [math-ph]

\bibitem[CDJ+25]{Costeri_25}
B.~Costeri, C.~Dappiaggi, B. A.~Ju\'arez-Aubry and R.D.~Singh,
{\it ``The Hadamard parametrix on half-Minkowski with Robin boundary conditions: Fundamental solutions and Hadamard states'',} arXiv:2509.26035 [math-ph], to appear on Ann. Henri Poincar\'e.

\bibitem[DDF20]{Dappiaggi-Drago_2019}
C.~ Dappiaggi, N.~ Drago and H.~R.~C.~ Ferreira, 
{\it  ``Fundamental solutions for the wave operator on static Lorentzian manifolds with timelike boundary''}, 
Lett. Math. Phys. \textbf{109} 10 (2019), 2157, 
arXiv:1804.03434 [math-ph]

\bibitem[DF17]{Dappiaggi:2017wvj}
C.~Dappiaggi and H.~R.~C.~Ferreira,
{\it ``On the algebraic quantization of a massive scalar field in anti-de-Sitter spacetime''},
Rev. Math. Phys. \textbf{30} (2017) no.02, 1850004, arXiv:1701.07215 [math-ph]

\bibitem[DJS26]{Dappiaggi_26}
C.~Dappiaggi, B.~A.~Ju\'arez-Aubry and R.~D.~Singh, 
In preparation

\bibitem[DM21]{Dappiaggi-Marta_2021}
C.~Dappiaggi and A.~Marta,
{\it ``Fundamental solutions and Hadamard states for a scalar field with arbitrary boundary conditions on an asymptotically AdS spacetimes''},
Math. Phys. Anal. Geom. \textbf{24} (2021) no.3, 28, arXiv:2101.10290 [math-ph]

\bibitem[DM20]{Dappiaggi-Marta_2020}
C.~Dappiaggi and A.~Marta,
{\it``A generalization of the propagation of singularities theorem on asymptotically anti-de Sitter spacetimes''},
Math. Nachr. \textbf{295} (2022) no.10, 1934-1968, arXiv:2006.00560 [math-ph]

\bibitem[DNP16]{Dappiaggi-Nosari_2016}
C.~Dappiaggi, G.~Nosari and N.~Pinamonti,
{\it ``The Casimir effect from the point of view of algebraic quantum field theory''},
Math. Phys. Anal. Geom. \textbf{19} (2016) no.2, 12,
arXiv:1412.1409 [math-ph]

\bibitem[DF05]{Decanini:2005eg}
Y.~Decanini and A.~Folacci,
{\it ``Hadamard renormalization of the stress-energy tensor for a quantized scalar field in a general spacetime of arbitrary dimension''},
Phys. Rev. D \textbf{78}, 044025 (2008),
arXiv:gr-qc/0512118 [gr-qc]

\bibitem[DH72]{Duistermaat_1972}
J.J. Duistermaat and L. H\"ormander, 
Fourier Integral Operators II,
{\it Acta Mathematica} 128 (1972), 183 -- 269 

\bibitem[DW19]{Dybalski:2018egv}
W.~Dybalski and M.~Wrochna,
\textit{``A mechanism for holography for non-interacting fields on anti-de Sitter spacetimes''},
Class. Quant. Grav. \textbf{36} (2019) no.8, 085006,
arXiv:1809.05123 [math-ph]

\bibitem[Eich91]{Eich91}
J. Eichhorn,
\emph{``The Banach manifold structure of the space of metrics on noncompact manifolds,''}
Diff. Geom. Appl. {\bf 1} (1991) 89

\bibitem[EUW26]{Enciso_2026}
A. Enciso, G. Uhlmann and M. Wrochna, M, 
\textit{``The Dirichlet-to-Neumann map on asymptotically anti-de Sitter spaces and holography''}, arXiv:2512.03587 [math.AP]

\bibitem[Eva10]{Evans}
L. C. Evans,
Partial Differential Equations,
2nd ed., Graduate Studies in Mathematics, Vol. 19,
American Mathematical Society, Providence, RI, (2010), 749p

\bibitem[FS22]{Fang_2022}
Y.L. Fang and A. Strohmaier, 
{\it ``A Mathematical Analysis of Casimir Interactions I: The Scalar Field''}, 
Ann. Henri Poincar{\'e}, \textbf{23} 4 (2022),
arXiv:2104.09763 [math-ph]

\bibitem[FV13]{Fewster:2013lqa}
C.~J.~Fewster and R.~Verch,
\textit{``The Necessity of the Hadamard Condition''},
Class. Quant. Grav. \textbf{30} (2013), 235027
arXiv:1307.5242 [gr-qc]

\bibitem[FS21]{Fournodavlos:2021eye}
G.~Fournodavlos and J.~Smulevici,
{\em ``The Initial Boundary Value Problem in General Relativity: The Umbilic Case,''}
Int. Math. Res. Not. \textbf{2023} (2023) no.5, 3790
arXiv:2104.08851 [gr-qc]

\bibitem[Fri75]{Friedlander_1975}
F.G. Friedlander, 
The Wave Equation on a Curved Space-Time,
Cambridge University Press (1975), 282p

\bibitem[FNW81]{Fulling:1981cf}
S.~A.~Fulling, F.~J.~Narcowich and R.~M.~Wald,
{\em ``Singularity Structure of the Two Point Function in Quantum Field Theory in Curved Space-time. {II},''}
Annals Phys. \textbf{136} (1981), 243--272

\bibitem[GW22]{Gannot_2022}
O.~Gannot and M.~Wrochna,
{\it``Propagation of Singularities on AdS Spacetime for General Boundary Conditions and the Holographic Hadamard Condition''},
J. Inst. Math. Jussieu \textbf{21} (2022) no.1, 67,
arXiv:1812.06564 [math.AP]

\bibitem[Gar64]{Garabedian_1964}
P. Garabedian, 
Partial differential equations,
New York, Wiley (1964), 682p 

\bibitem[GOW17]{Gérard}
C. Gérard, O. Oulghazi and M. Wrochna,
{\it ``Hadamard States for the Klein–Gordon Equation on Lorentzian Manifolds of Bounded Geometry''}
Comm. Math. Phys. \textbf{352} (2017) no.2, 519, arXiv:1602.00930 [math-ph]

\bibitem[GM22]{Ginoux_2022}
N. Ginoux and S. Murro,
{\it ``On the Cauchy problem for Friedrichs systems on globally hyperbolic manifolds with timelike boundary''},
Adv. Differ. Equ., \textbf{27}, (2022), 7-8,
arXiv:2007.02544 [math.AP]

\bibitem[Gor73]{Gordon73}
W.~B.~Gordon,
{\it ``An analytical criterion for the completeness of Riemannian manifolds''},
Proc. Amer. Math. Soc. \textbf{37} (1973), 221

\bibitem[Gor74]{Gordon74}
W.~B.~Gordon,
{\it ``Corrections to ``An analytical criterion for the completeness of	Riemannian manifolds''},
Proc. Amer. Math. Soc. \textbf{45} (1974), 130

\bibitem[GM18]{Grosse}
N. Gro\ss e and S. Murro,
\textit{``The well-posedness of the Cauchy problem for the Dirac operator on globally hyperbolic manifolds with timelike boundary'',}
Doc. Math. {\bf 25}, (2020), 737,
arXiv:1806.06544 [math.DG]

\bibitem[Gru68]{Grubb68}
G. Grubb, 
\emph{``A characterization of the non-local boundary value problems associated with an elliptic operator''}, 
Ann. Sc. Norm. Sup. Pisa (3) {\bf 22} (1968) 425

\bibitem[HFS20]{Ak_Hau_2020}
L. A. Hau, J. L. Flores, and M. Sánchez, 
{\it  ``Structure of globally hyperbolic spacetimes with timelike boundary''},
Rev. Mat. Iberoam. \textbf{37}, no.1 (2021) 45--94, arXiv:1808.04412 [gr-qc]

\bibitem[HU19]{HintzUhlmann}
P.~Hintz and G.~Uhlmann,
\emph{Reconstruction of Lorentzian manifolds from boundary light observation sets},
Int. Math. Res. Not. IMRN \textbf{2019}, no.~22, 6949--6987, arXiv:1705.01215 [math.DG]

\bibitem[HS26]{Strohmaier_2026}
A. Hofmann and A. Strohmaier,
{\it ``Relative trace formulas for obstacle scattering with Neumann and transmission boundary conditions''}, (2026), arXiv: 2605.21201 [math-ph]

\bibitem[HW01]{Hollands_2001}
S.~Hollands and R.~M.~Wald,
{\it ``Local Wick polynomials and time ordered products of quantum fields in curved spacetime''},
Commun. Math. Phys. \textbf{223} (2001), 289, 	arXiv:gr-qc/0103074

\bibitem[HW02]{Hollands_2002}
S.~Hollands and R.~M.~Wald,
{\it ``Existence of local covariant time ordered products of quantum fields in curved spacetime''}
Commun. Math. Phys. \textbf{231} (2002), 309--345, arXiv: gr-qc/0111108

\bibitem[H\"or90]{Hormander_1990}
L. H\"ormander, 
The Analysis of Linear Partial Differential Operators I, Springer (1990), 440p

\bibitem[Joh82]{FritzJohn}
F. John,
Partial Differential Equations,
4th ed, Springer (1982), 249p

\bibitem[IS24]{Islam_2024}
O. Islam and A. Strohmaier,
{\it ``On microlocalisation and the construction of Feynman Propagators for normally hyperbolic operators''}, 
Comm. in Anal. and Geom. \textbf{32} no.~7 (2024), 1811, arXiv:2012.09767 [math.AP]

\bibitem[Jun95]{Junker}
W.~Junker,
{\em``Adiabatic vacua and Hadamard states for scalar quantum fields on curved space-time,''}
PhD thesis, U. of Hamburg (1995), arXiv:hep-th/9507097 [hep-th]

\bibitem[Kay92]{Kay1992FLocality}
B.~S. Kay, \emph{The Principle of Locality and Quantum Field Theory on (non globally hyperbolic) Curved Spacetimes}, Rev. Math. Phys. \textbf{4} (Spec. Iss.) (1992), 167--195

\bibitem[KW91]{Kay_Wald:1991}
B.~S.~Kay and R.~M.~Wald,
{\it ``Theorems on the Uniqueness and Thermal Properties of Stationary, Nonsingular, Quasifree States on Space-Times with a Bifurcate Killing Horizon''},
Phys. Rept. \textbf{207} (1991), 49--136

\bibitem[Lee13]{Lee_2013}
J.M.~Lee, 
Introduction to smooth manifolds,
Graduate Texts in Mathematics, Springer-Verlag, \textbf{218} (2013), 726p

\bibitem[Lee18]{Lee_2018}
J.~M.~Lee, 
Introduction to Riemannian manifolds,
Graduate Texts in Mathematics, Springer-Verlag, \textbf{176} (2018), 447p

\bibitem[MS78]{Melrose_1978}
R.B. Melrose and J. Sj\"ostrand, 
{\it ``Singularities of boundary value problems. I.''}, 
Comm. Pure Appl. Math., \textbf{31}: 593-617 (1978)

\bibitem[MS82]{Melrose_1982}
R.B. Melrose and J. Sj\"ostrand, 
{\it ``Singularities of boundary value problems. II.''}, 
Comm. Pure Appl. Math., \textbf{35}: 129-168 (1982)

\bibitem[Mor00]{Moretti_2000}
V.~Moretti,
{\it ``Proof of the symmetry of the off diagonal Hadamard / Seeley-deWitt's coefficients in C$^\infty$ Lorentzian manifolds by a 'local Wick rotation' ''},
Comm. Math. Phys. \textbf{212} (2000), 165, arXiv:gr-qc/9908068 [gr-qc]

\bibitem[Mor03]{Moretti2003}
V.~Moretti,
{\it ``Comments on the stress-energy tensor operator in curved spacetime''}
Commun. Math. Phys. \textbf{232} (2003), 189--221, arXiv: gr-qc/0109048

\bibitem[Mor17]{Moretti_2017}
V.~Moretti,
Spectral Theory and Quantum Mechanics, Springer, 3rd ed. (2017) 950p

\bibitem[M\"ul12]{Muller}
O.~M\"uller,
{\it ``Asymptotic flexibility of globally hyperbolic manifolds''},
Comptes Rendus Mathématique \textbf{350} (2012), 421,
arXiv:1110.1037 [math.DG]

\bibitem[NO61]{NomizuOzeki}
K.~Nomizu and H.~Ozeki,
{\it ``The existence of complete Riemannian metrics'',}
Proc. Amer. Math. Soc. \textbf{12} (1961), 889.

\bibitem[O'N83]{ONeill83}
B. O'Neill,
Semi-Riemannian Geometry With Applications to Relativity,
Academic Press (1983), 468p

\bibitem[Pet98]{Petersen}
P. Petersen,
Riemannian Geometry,
Springer (1998), 432p

\bibitem[PS17]{Petkov}
V. Petkov and L. Stoyanov, 
Geometry of the Generalized Geodesic Flow and Inverse Spectral Problems, 2nd ed., Wiley (2017), 432p

\bibitem[PPV11]{Poisson_2011}
E. Poisson, A. Pound and I. Vega, 
{\it  ``The Motion of Point Particles in Curved Spacetime''},
Living Reviews in Relativity, \textbf{14} (2011), 1, arXiv:1102.0529 [gr-qc]

\bibitem[Rad96a]{Radzikowski_1996}
M.~J.~Radzikowski,
{\it ``Micro-local approach to the Hadamard condition in quantum field theory on curved space-time''},
Comm. Math. Phys. \textbf{179} (1996), 529

\bibitem[Rad96b]{Radzikowski_1996_1}
M.~J.~Radzikowski,
{\it ``A Local to global singularity theorem for quantum field theory on curved space-time''},
Comm. Math. Phys. \textbf{180} (1996), 1

\bibitem[Rej16]{Rejzner}
K.~Rejzner, 
Perturbative Algebraic Quantum Field Theory, 
{\it Springer}, (2016), 180p

\bibitem[SV00]{Sahlmann:2000fh}
H.~Sahlmann and R.~Verch,
{\em``Passivity and microlocal spectrum condition,''}
Commun. Math. Phys. \textbf{214} (2000), 705, arXiv:math-ph/0002021 [math-ph]

\bibitem[San07]{Sanchez:2007rx}
M.~Sanchez,
{\it ``Remarks on the notion of global hyperbolicity,''}
EAS Publ. Ser. \textbf{30} (2008), 201, arXiv:0712.1933v2 [gr-qc]

\bibitem[Sch01]{Schick}
T. Schick,
\emph{``Manifolds with boundary and of bounded geometry''}, Math. Nachr. {\bf 233} (2001) 103, arXiv: math/0001108

\bibitem[SZ17]{Sogge}
C. D. Sogge and S. Zelditch
\emph{``Sup norms of Cauchy data of eigenfunctions on manifolds with concave boundary''},
Comm. PDE {\bf 42} (2017), 1249,
arXiv:1411.1035 [math.AP]

\bibitem[Tay78]{Taylor_1978}
M.~E.~Taylor,
{\it ``Propagation, reflection, and diffraction of singularities of solutions to wave equations''},
Bull. Amer. Math. Soc. \textbf{84} (1978), 589

\bibitem[Tre67]{Treves}
F.~Treves, Topological Vector Spaces, Distributions and Kernels, (1967) Academic Press, 565p

\bibitem[Wro21]{Wrochna}
M.~Wrochna,
{\it ``The holographic Hadamard condition on asymptotically anti-de Sitter spacetimes''},
Lett. Math. Phys. \textbf{107} (2017) no.12, 2291
arXiv:1612.01203 [math-ph]

\bibitem[WYZ24]{Wunsch}
J.~Wunsch, M.~Yang amd Y.~Zou
{\it ``The Morse index theorem for mechanical systems with reflections''},
Nonlinearity {\bf 37} (2024), 085006, arXiv:2308.16162 [math.DG]


\end{thebibliography}
\end{document}